%% file: main.tex
\newif\iffull
\fulltrue

\newif\ifdiff
\difffalse

\newif\ifapp
\apptrue
\newif\ifextrarules
\extrarulesfalse

\iffull
\documentclass[acmsmall, screen, nonacm]{acmart}
\else
\documentclass[acmsmall, screen]{acmart}
\fi

\iffull
\setcopyright{none}
\else
\setcopyright{rightsretained}
\acmJournal{PACMPL}
\acmYear{2024} \acmNumber{POPL} \acmMonth{1}
\fi

\usepackage{booktabs}   %% For formal tables:
\usepackage{multirow}
\usepackage{subcaption} %% For complex figures with subfigures/subcaptions
\usepackage{wrapfig}
\usepackage[final]{listings}
\usepackage{amsmath, amsthm, stmaryrd, mathtools, mathrsfs}
\usepackage{mathpartir}
\usepackage{hyperref}
\usepackage[margin=false,inline=true]{fixme}

\usepackage{footnote}
\makesavenoteenv{tabular}

\newtheorem{lemma}{Lemma}
\newtheorem{corollary}{Corollary}
\newtheorem{theorem}{Theorem}

\theoremstyle{definition}

\input{macros}

\begin{document}

%\title{Graph Types for Recursive Data Structures}
\title{Analyzing Collections and Pipelines of Futures with Graph Types}
\title{Pipelines and Beyond: Graph Types for ADTs with Futures}
\iffull\titlenote{This is an extended version of a paper that appeared at POPL 2024.}\fi
%\title{Predicting Pipelined Parallelism with Graph Types}
%\iffull
%\subtitle{Extended Supplementary Version}
%\fi

%\email{webmaster@marysville-ohio.com}
%\affiliation{%
%  \institution{Institute for Clarity in Documentation}
%  \streetaddress{P.O. Box 1212}
%  \country{USA}

\author{Francis Rinaldi}
\email{frinaldi@hawk.iit.edu}
\affiliation{
  \institution{Illinois Institute of Technology}
  \country{USA}
}

\author{june wunder}
\email{jwunder@bu.edu}
\affiliation{
  \institution{Boston University}
  \country{USA}
}

\author{Arthur Azevedo de Amorim}
\email{aaavcs@rit.edu}
\affiliation{
  \institution{Rochester Institute of Technology}
  \country{USA}
}

\author{Stefan K. Muller}
\email{smuller2@iit.edu}
\affiliation{
  \institution{Illinois Institute of Technology}
  \country{USA}
}


\begin{abstract}
  \input{abstract.tex}
\end{abstract}

\iffull
\else


%% %%
%% %% Keywords. The author(s) should pick words that accurately describe
%% %% the work being presented. Separate the keywords with commas.
\fi
\maketitle

\input{intro}
\input{overview}
%\input{prelim}
%\input{warmup}
%\input{dag}

\input{lang}
\input{soundness}

\input{infer}

\input{impl}
%\input{case}
%\input{disc}
\input{related}
\input{conclusion}

%% \begin{acks}
%%   The author would like to thank the anonymous reviewers, as well as Umut
%%   Acar, Jatin Arora, Guy Blelloch, and Sam Westrick for helpful discussions at
%%   various points in this work.
%%   %
%%   This work was partially supported by the National Science Foundation under
%%   grant CCF-2107289.
%% \end{acks}

\bibliographystyle{ACM-Reference-Format}
\bibliography{main}

\clearpage

\ifapp
\appendix
\input{app}
\fi

\end{document}

%% file: macros.tex
\newcommand{\defeq}{\triangleq}
\newcommand{\langname}{\ensuremath{\lambda^{G\mu}}}
\newcommand{\oldlangname}{\ensuremath{\lambda^G}}

\newcommand{\toolname}{\textsc{GML}}
\newcommand{\newtoolname}{\textsc{GML$^\mu$}}
\newcommand{\bnfdef}{\mathop{::=}}
\newcommand{\bnfalt}{\mid}

\newcommand{\rulename}[1]{\textsc{#1}}
\newcommand{\Rule}[4][]{\ensuremath{\inferrule[{\footnotesize (#2)}]{#3}{#4}}}

\newcommand{\kw}[1]{\ensuremath{\mathsf{#1}}}
\newcommand{\ppt}{\ell}
\newcommand{\ppts}{\vec{\ppt}}

\newcommand{\graph}{\ensuremath{G}}

\newcommand{\kwtriv}{\langle \rangle}
\newcommand{\kwpair}[2]{(#1, #2)}

\newcommand{\kwfst}[1]{\kw{fst}~#1}
\newcommand{\kwsnd}[1]{\kw{snd}~#1}
\newcommand{\kwinl}[1]{\kw{inl}~#1}
\newcommand{\kwinr}[1]{\kw{inr}~#1}
\newcommand{\kwcase}[5]{\kw{case}~#1~\{#2. #3; #4. #5\}}
\newcommand{\kwfun}[5]{\kw{fun}[#1; #2]~#3~#4 = #5}
\newcommand{\kwapp}[3]{{#2}~#3}
\newcommand{\kwtapp}[3]{#1[#2;#3]}
\newcommand{\kwpar}[2]{\kw{par} (#1, #2)}
\newcommand{\kwfuture}[2]{\kw{future}[#1]~#2}

\newcommand{\kwforce}[1]{\kw{touch}~#1}
\newcommand{\kwnewf}[3]{\kw{new}~\hastycl{#1}{#2}. #3}

\newcommand{\kwroll}[1]{\kw{roll}~#1}
\newcommand{\kwunroll}[1]{\kw{unroll}~#1}
\newcommand{\kwhandle}[2]{\kw{handle}[#1]~#2}

\newcommand{\kwsingletapp}[2]{#1[#2]}

\newcommand{\kwunit}{\kw{unit}}
\newcommand{\kwarrow}[3]{#1 \xrightarrow{#3} #2}
\newcommand{\kwprod}[2]{#1 \times #2}
\newcommand{\kwsum}[2]{#1 + #2}
\newcommand{\kwfutt}[2]{#1~\kw{future}[#2]}

\newcommand{\kwpi}[3]{\Pi [#1; #2]. #3}

\newcommand{\con}{\ensuremath{c}}
\newcommand{\convar}{\alpha}
\newcommand{\kwrec}[2]{\mu #1 . #2}
\newcommand{\kwprec}[4]{\mu (#1; #2. #3; #4)}
\newcommand{\kwxi}[2]{\lambda  #1. #2}
\newcommand{\kwvapp}[2]{#1~#2}

\newcommand{\kind}{\kappa}
\newcommand{\kwtykind}{\mathsf{Ty}}
\newcommand{\kwkindarr}[2]{#1 \rightarrow #2}

\newcommand{\vs}{\ensuremath{U}}

\newcommand{\vvar}{\ensuremath{u}}
\newcommand{\genseed}{\ensuremath{\vec{u}}}
\newcommand{\vertgen}[2]{\ensuremath{\kw{gen}(#2)}}

\newcommand{\vspath}{\ensuremath{\mathit{VP}}}

\newcommand{\vsnormal}{~\kw{normal}}
\newcommand{\vsnormalneg}{~\kw{normal^{-}}}

\newcommand{\vsty}{\mathcal{U}}
\newcommand{\kwvty}{\blacklozenge}
\newcommand{\kwunitty}{1}
\newcommand{\kwvsrec}[2]{\mu #1 . #2}
\newcommand{\kwvscorec}[2]{\nu #1 . #2}
\newcommand{\vstyvar}{\ensuremath{t}}

\newcommand{\avail}{\ensuremath{a}}
\newcommand{\isav}{\blacksquare}
\newcommand{\isunav}{\square}
\newcommand{\vsprod}[4]{\kwprod{#1^{#2}}{#3^{#4}}}

\newcommand{\emptygraph}{\bullet}
\newcommand{\parcomp}{\mathbin{\otimes}}
\newcommand{\seqcomp}{\mathbin{\oplus}}
\newcommand{\dagor}{\lor}
\newcommand{\dagrec}[3]{\mu #1 . #2}
\newcommand{\gvar}{\ensuremath{\gamma}}
\newcommand{\leftcomp}[2]{#1\!\sswarrow_{#2}}
\newcommand{\touchcomp}[1]{\mathop{\strut^{#1}\!\ssearrow}}
\newcommand{\dagnew}[2]{\kw{new}~#1. #2}
\newcommand{\dagpi}[3]{\Pi [#1; #2]. #3}

\newcommand{\vertex}{u}
\newcommand{\verts}{\vs} %\vec{u}
\newcommand{\graphkind}{\ensuremath{\kappa_{G}}}
\newcommand{\kgraph}{*}

\newcommand{\vertsof}[1]{\mathit{Verts(}#1\mathit{)}}

\newcommand{\cgraph}{g}
\newcommand{\dagq}[4]{(#1, #2, #3, #4)}
\newcommand{\vertices}{V}
\newcommand{\edges}{E}
\newcommand{\startv}{s}
\newcommand{\sinkv}{t}

\newcommand{\ctx}{\Gamma}
\newcommand{\gctx}{\Delta}
\newcommand{\ectx}{\cdot}
\newcommand{\uspctx}{\Omega}
\newcommand{\utctx}{\Psi}
\newcommand{\hastype}[2]{#1 : #2}
\newcommand{\hastycl}[2]{#1\!:\!#2}
\newcommand{\hastypewithdag}[3]{\hastype{#1}{#2 \mid #3}}
\newcommand{\tywithdag}[7]{#1; #2; #3; #4 \vdash \hastypewithdag{#5}{#6}{#7}}
\newcommand{\affinetyped}[3]{#1 \vdash_A \hastype{#2}{#3}}

\newcommand{\unannvaltype}[2]{\vdash \hastype{#1}{#2}}

\newcommand{\dagwf}[5]{#1; #2; #3 \vdash #4 : #5}

\newcommand{\isctx}[4]{#1; #3 \vdash #4}
\newcommand{\fresh}{~\kw{fresh}}
\newcommand{\eval}[4]{#2 \Downarrow #3 \mid #4}

\newcommand{\sub}[3]{#1[#2/#3]}
\newcommand{\gsub}[3]{#1[#2/#3]}
\newcommand{\tsub}[3]{#1[#2/#3]}
\newcommand{\vsub}[3]{#1[#2/#3]}

\newcommand{\getgraphs}[1]{\mathit{Exp}(#1)}
\newcommand{\unrstep}{\hookrightarrow}

\newcommand{\uctx}{\Sigma}
\newcommand{\vsistype}[4]{#1; #2 \vdash \hastype{#3}{#4}}

\newcommand{\vstctx}{\Upsilon}
\newcommand{\haskind}[2]{#1 :: #2}
\newcommand{\iskind}[6]{#1; #3; #4 \vdash \haskind{#5}{#6}}

\newcommand{\unannkind}[3]{#1 \vdash \haskind{#2}{#3}}

\newcommand{\vseval}[2]{#1 \Downarrow #2}
\newcommand{\coneq}[6]{#1; #2; #3 \vdash \haskind{#4 \equiv #5}{#6}}
\newcommand{\vseq}[4]{#1 \vdash \hastype{#2 \equiv #3}{#4}}

\newcommand{\vstysubt}[2]{#1 \sqsubseteq #2}

\newcommand{\uspstep}[2]{#1 \rightsquigarrow #2}
\newcommand{\uspmerge}[2]{#1 \boxplus #2}
\newcommand{\uspsplit}[3]{\uspstep{#1}{\uspmerge{#2}{#3}}}

\newcommand{\vstystep}[2]{#1 \rightsquigarrow #2}
\newcommand{\vstymerge}[2]{#1 \boxplus #2}
\newcommand{\vstysplit}[3]{\vstystep{#1}{\vstymerge{#2}{#3}}}

\newcommand{\bnnf}[1]{\mathsf{NBNF}(#1)}

\newcommand{\guspctx}{{\Omega^\circ}}
\newcommand{\gutctx}{\Psi^\circ}

\newcommand{\gexistsst}[3]{there exists $#1$ such that $#2 \unrstep^* #1$ and $#3 \in \getgraphs{\bnnf{#1}}$}
\newcommand{\uctxexistsst}[4]{there exists $#2$ such that $\genseed \in #2$ implies $\genseed \fresh$
													and $\tywithdag{\ectx}{\ectx}{#1, #2}{\ectx}{#3}{#4}{\emptygraph}$}

\newcommand{\gsubvsunrollexst}[5]{\exists #3. #1 \unrstep^* #3 \wedge \vsub{#3}{#4}{#5} = #2}
\newcommand{\gsubvsunrollexstTwo}[7]{\exists #3. #1 \unrstep^* #3 \wedge \vsub{\vsub{#3}{#4}{#5}}{#6}{#7} = #2}

\newcommand{\vstysplitexst}[5]{there exists $#1$ such that $\vstysplit{#2}{#1}{#4}$ and $\vstysplit{#1}{#3}{#5}$}
\newcommand{\uspsplitexst}[5]{there exists $#1$ such that $\uspsplit{#2}{#1}{#4}$ and $\uspsplit{#1}{#3}{#5}$}
\newcommand{\uspsplitbothexst}[7]{there exists $#1$ and $#2$ such that $\uspsplit{#3}{#1}{#2}$ and $\uspsplit{#1}{#4}{#6}$ and $\uspsplit{#2}{#5}{#7}$}

\newcommand{\splitLprodexst}[6]{there exists $#2$ and $#3$ such that $\vstysplit{#1}{#2}{#3}$ and 
															$\vstysubt{\vsprod{#2}{\isav}{#4}{\isunav}}{#5}$ and 
															$\vstysubt{\vsprod{#3}{\isav}{#4}{\isunav}}{#6}$}
\newcommand{\splitRprodexst}[6]{there exists $#2$ and $#3$ such that $\vstysplit{#1}{#2}{#3}$ and 
															$\vstysubt{\vsprod{#4}{\isunav}{#2}{\isav}}{#5}$ and 
															$\vstysubt{\vsprod{#4}{\isunav}{#3}{\isav}}{#6}$}

\newcommand{\splitLfullprodexstA}[6]{There exists $#2$ and $#3$ such that $\vstysplit{#1}{#2}{#3}$ and 
															$\vstysubt{\vsprod{#2}{\isav}{#4}{\isav}}{#5}$ and 
															$\vstysubt{\vsprod{#3}{\isav}{#4}{\isunav}}{#6}$}
\newcommand{\splitLfullprodexstB}[6]{There exists $#2$ and $#3$ such that $\vstysplit{#1}{#2}{#3}$ and 
															$\vstysubt{\vsprod{#2}{\isav}{#4}{\isunav}}{#5}$ and 
															$\vstysubt{\vsprod{#3}{\isav}{#4}{\isav}}{#6}$}
\newcommand{\splitRfullprodexstA}[6]{There exists $#2$ and $#3$ such that $\vstysplit{#1}{#2}{#3}$ and 
															$\vstysubt{\vsprod{#4}{\isav}{#2}{\isav}}{#5}$ and 
															$\vstysubt{\vsprod{#4}{\isunav}{#3}{\isav}}{#6}$}
\newcommand{\splitRfullprodexstB}[6]{There exists $#2$ and $#3$ such that $\vstysplit{#1}{#2}{#3}$ and 
															$\vstysubt{\vsprod{#4}{\isunav}{#2}{\isav}}{#5}$ and 
															$\vstysubt{\vsprod{#4}{\isav}{#3}{\isav}}{#6}$}
\newcommand{\splitbothprodexst}[8]{There exists $#2$ and $#3$ and $#5$ and $#6$ such that $\vstysplit{#1}{#2}{#3}$ and 
															$\vstysplit{#4}{#5}{#6}$ and 
															$\vstysubt{\vsprod{#2}{\isav}{#5}{\isav}}{#7}$ and 
															$\vstysubt{\vsprod{#3}{\isav}{#6}{\isav}}{#8}$}

\newcommand{\splituspLexst}[6]{$#3 = #1, \hastype{#5}{#6}$ and $\uspsplit{#2}{#1}{#4}$}
\newcommand{\splituspRexst}[6]{$#4 = #1, \hastype{#5}{#6}$ and $\uspsplit{#2}{#3}{#1}$}
\newcommand{\splituspLexstfull}[6]{$#3 = #1, \hastype{#5}{#6}$ and $\uspsplit{#2}{#1}{#4}$
														and $#5 \notin #4$}
\newcommand{\splituspRexstfull}[6]{$#4 = #1, \hastype{#5}{#6}$ and $\uspsplit{#2}{#3}{#1}$
														and $#5 \notin #3$}
\newcommand{\splituspbothexst}[9]{$\vstysplit{#7}{#8}{#9}$ and
															$#4 = #1, \hastype{#6}{#8}$ and
															$#5 = #2, \hastype{#6}{#9}$ and
															$\uspsplit{#3}{#1}{#2}$}
\newcommand{\splituspbothcrossexst}[9]{$\vstysplit{#7}{#8}{#9}$ and
															$#4 = #1, \hastype{#6}{#9}$ and
															$#5 = #2, \hastype{#6}{#8}$ and
															$\uspsplit{#3}{#1}{#2}$}

\newcommand{\vstytouspsplitext}[6]{there exists $#2$ and $#3$ such that
															$\uspsplit{#1}{#2}{#3}$ and
															$\vsistype{#2}{\ectx}{#4}{#5}$ and
															$\vsistype{#3}{\ectx}{#4}{#6}$}

\newcommand{\undecfutty}[1]{#1~\kw{future}}
\newcommand{\undecty}{\sigma}
\newcommand{\undece}{\epsilon}
\newcommand{\undechandle}[1]{\kw{handle}~#1}
\newcommand{\futurify}[5]{#1 \vdash_{#2} #3 \leadsto #4; #5}
\newcommand{\futurifye}[3]{\vdash_{#1} #2 \leadsto #3}

\newcommand{\annvstctx}[1]{\mathsf{Ann}(#1)}
\newcommand{\unannvstctx}[1]{\mathsf{Unann}(#1)}

\FXRegisterAuthor{aaa}{anaaa}{\color{cyan}AAA}

\FXRegisterAuthor{skm}{anskm}{\color{red}SKM}
\newcommand{\skm}[1]{\skmnote{#1}}
\FXRegisterAuthor{jw}{anjw}{\color{orange}JW}

\FXRegisterAuthor{fr}{anfr}{\color{green}FR}

\FXRegisterAuthor{mg}{anmg}{\color{blue}MG}

%% file: abstract.tex
%% DO NOT EDIT THIS FILE-it is produced automatically from the file 'abstract'
Parallel programs are frequently modeled as {\em dependency} or {\em cost} graphs, which can be used to detect various bugs, or simply to visualize the parallel structure of the code. However, such graphs reflect just one particular execution and are typically constructed in a {\em post-hoc} manner. {\em Graph types}, which were introduced recently to mitigate this problem, can be assigned statically to a program by a type system and compactly represent the family of all graphs that could result from the program.

Unfortunately, prior work is restricted in its treatment of {\em futures}, an increasingly common and especially dynamic form of parallelism. In short, each instance of a future must be statically paired with a vertex name. Previously, this led to the restriction that futures could not be placed in collections or be used to construct data structures. Doing so is not a niche exercise: such structures form the basis of numerous algorithms that use forms of pipelining to achieve performance not attainable without futures. All but the most limited of these examples are out of reach of prior graph type systems.

In this paper, we propose a graph type system that allows for almost arbitrary combinations of futures and recursive data types. We do so by indexing datatypes with a type-level {\em vertex structure}, a codata structure that supplies unique vertex names to the futures in a data structure. We prove the soundness of the system in a parallel core calculus annotated with vertex structures and associated operations. Although the calculus is annotated, this is merely for convenience in defining the type system. We prove that it is possible to annotate arbitrary recursive types with vertex structures, and show using a prototype inference engine that these annotations can be inferred from OCaml-like source code for several complex parallel algorithms.

%% file: intro.tex
\section{Introduction}
\label{sec:intro}

Decades of work on reasoning about parallel programs have focused
on {\em computation} or {\em cost graphs}, directed graphs that
represent the dependencies of threads.
Computation graphs are a convenient target for analysis because they abstract
away details of the program, language, and even the parallelism features that
were used, while still capturing enough information about the relationships
between threads to perform many useful analyses.
For example, computation graphs have been used to study
deadlock~\citep{CogumbreiroHuMaYo18},
data races~\citep{BannerjeeBlMaPe06},
priority inversions~\citep{BabaogluMaSc93} and
evaluation cost~\citep{BlellochGr95, BlellochGr96}.

To analyze such properties, it is desirable to calculate the computation graph
of a program statically, at compile time or analysis time.
Doing so is often possible in languages and threading libraries for
{\em coarse-grained parallelism}, such as \texttt{pthreads},
where thread creation and synchronization are expensive and rare.
Much recent interest in parallel programming, however, has been in the area
of {\em fine-grained parallelism}, in which threads are created cheaply and
eagerly, often based on runtime conditions.
For example, a program might fork at each level of a divide-and-conquer
algorithm, or a web server might spawn a new thread to handle every incoming
request asynchronously.
Reasoning statically about the dependency structure of fine-grained parallel
programs is difficult because of the highly dynamic nature of thread creation
and synchronization in these programs.

This difficulty is compounded when programs use {\em futures} and related
abstractions for fine-grained parallelism, which are becoming increasingly
popular and have been made available in Python, Scala, Rust, and the most recent
release of OCaml~\citep{SivaramakrishnanDoWhJaKeSaPaDhMa20}, among other
languages.
Essentially, a future is a first-class handle to an asynchronous computation.
The result of the computation can be demanded via a
{\em force} or {\em touch} operation, which blocks if the result
is not yet available.
Because futures run in separate threads,
we can model each future as its own vertex
$u$ in the computation graph of a program. Edges leading into $u$ track the
intermediate results used to compute the future, and when we touch the future,
we add an edge from $u$ to the thread where the touch happened.
Futures may be passed around a program arbitrarily and end up being touched in a
very different part of the program from where it was spawned, leading to great
power and flexibility but also complex computation graphs which are difficult to
reason about.

To address the difficulty of predicting parallel dependences in fine-grained
parallel programs, especially those with futures,
\citet{Muller22} introduced the notion of {\em graph
  types}, which statically overapproximate the set of computation graphs that
might result from running a program.
A {\em graph type system} statically assigns graph types to programs, and its
soundness theorem ensures that the actual computation graph resulting from any
execution of a well-typed program is described by the program's graph type.

Much of the complexity of the graph type system centers around futures.
Because futures can be touched in an entirely different part of the program from
where they are created, each future type is annotated with a distinguished {\em
  vertex name}, so that the graph type system can refer to the correct vertex
when tracking the dependencies of touch operations.  (Explicit vertex names are
not needed in simpler parallelism models such as fork-join, because it is clear
what thread is being synchronized.)
To avoid tracking spurious dependencies, the graph type system ensures that each
vertex name is associated with at most one future during execution.  More
precisely, when spawning a new future, the graph type system annotates the type
of the result with a fresh vertex name, which is tracked in a separate affine
context to prevent reuse.

This treatment of futures leads to a
significant limitation in prior work: it is difficult or
impossible to build useful data structures containing futures.
Even an expression as simple as
\lstinline{[future e1; future e2]} (a list containing two new futures) cannot
be assigned a type.
The reason is that the two elements of this list must have
types~$\kwfutt{\tau}{u_1}$ and~$\kwfutt{\tau}{u_2}$, respectively,
where~$u_1$ and~$u_2$ are distinct vertex names and~$\kwfutt{\tau}{u}$ is the
type of a future returning a value of type~$\tau$ with the vertex
named~$u$---these two elements can't
be placed in a list because prior work supports only homogenous lists.
Although this example is simple and artificial,
much of the power of futures, as opposed to more limited
parallelism models such as fork-join, comes from the ability to
program with data structures that contain an unbounded number
of futures, such as lists and trees.
As examples, \citet{BlellochRe97} describe a number of algorithms and data
structures that use futures in complex ways to pipeline computations, resulting
in asymptotic improvements over the best known fork-join implementations.
These programming idioms exercise the full complexity of futures, motivating the
need for techniques to reason statically about the computation graphs of these
programs.

In this paper, we develop a graph type
system, and accompanying inference algorithm, that can handle complex
data structures using futures.
As a motivating example, consider a function that produces a pipeline of
increasingly precise approximations of~$\pi$.
This could be, for example, the first stage in a graphics or simulation
pipeline.
%later stages would take these iterative approximations and generate
%increasingly accurate results until some time limit or precision threshold is
%reached.
%
We wish to compute the approximations asynchronously so that earlier
approximations can be used while later ones are still being computed.
Figure~\ref{fig:intro-pi} shows two possible implementations of such a function.
The implementation on the left produces a list of futures with the intermediate
results.  The function~$\kw{list\_pi}$ takes a number~$\kw{k}$ and a {\em
  future}~$\kw{a}$, which computes the~$(k-1)$st approximation.
Each iteration of~$\kw{list\_pi}$ spawns a new future to compute the~$\kw{k}$th
term of the Gregory series multiplied by~4, adds it to the running total being
computed by~$\kw{a}$, and adds the new future (which is completing the new
running total) to a list, then calls~$\kw{list\_pi}$ recursively to compute the
remaining terms.
To illustrate a use of this structure, the \lstinline{main} function
takes the second approximation from the list.
Note that the~$\kw{list\_pi}$ function, as written, doesn't terminate.

Because the function~$\kw{list\_pi}$ produces a list of futures, it cannot be
given a graph type under prior work~\cite{Muller22}.\footnote{%
  Actually, \citet{Muller22} does discuss a similar pipelining example in his
  system; cf. Figure~10 and Section~6. However, that example is expressible
  precisely because it does not accumulate the intermediate results in a list.}
This is a shame, because its computation graph would have revealed a
subtle but fatal bug: despite the futures, there is no real asynchrony
or pipelining because almost the entire list of approximations (which, in
this example, is infinite) must be constructed before the program
proceeds.
This can be seen in the visualization on the left side of
Figure~\ref{fig:intro-pi-graphs}, which is produced automatically by our
implementation from the inferred graph type.
In the figure, vertices in the graph, notated with either a text label or a
small circle, represent pieces of computation.
The vertices with labels like~$n1 \bullet 1$ are the final vertices of
a future, and these labels are the vertex names assigned to the future.
The reason for these particular labels will become clear later in the paper.
Edges represent dependences: an edge out of a labeled vertex indicates a
touch of the corresponding future, and other edges represent sequential
dependences within a thread or the spawning of a future.
A path of edges in the graph therefore represents a chain of sequential
dependences and two vertices with no path between them indicate opportunities
for parallelism.
Long paths indicate a lack of parallelism.

The figure shows a visualization of the graph type of the program, with
the recursion of \lstinline{list_pi} unrolled a fixed number of times to make
the recursive structure visually clear.
A vertex labeled~$\dots$ indicates a recursive call that has been elided
because of the cutoff on number of unrollings.
The vertex representing the \lstinline{touch} operation in
\lstinline{main} is circled in red: we can see that there is a long chain
of dependences on the critical path to reach this operation, which means the
operation will be significantly delayed when running the program.
Indeed, the topmost~$\dots$ appears on the critical path, indicating a
potentially (and, in this case, actually) infinite critical path.

The second implementation in Figure~\ref{fig:intro-pi} instead uses
a new data structure \lstinline{'a pipe} which resembles a lazy list: the head
of the list is computed eagerly and may be used immediately but the tail of
the list is computed asynchronously in a future.
The function \lstinline{pipeline_pi} takes the running total~\lstinline{a}
(now as an actual float, rather than a future) and~\lstinline{k}.
It adds the~\lstinline{k}th approximation to the running total,
then returns the new running total as well as a future to
call~\lstinline{pipeline_pi} recursively to compute the remainder of the pipeline.
This is reminiscent of the ``producer'' example of~\citet{BlellochRe97}.
%, who
%use futures in this general pattern to construct a wide variety of
%pipelined data structures.
%
As we can see from the visualization on the right side of
Figure~\ref{fig:intro-pi-graphs}, the graphs corresponding
to~\lstinline{pipeline_pi} exhibit much more parallelism than the previous version.
Here, the \lstinline{touch} operation in \lstinline{main} (again circled in red)
occurs in parallel with the computation of the remainder of the list
and there are only a small, finite number of operations on its critical path.

\begin{figure}
\begin{minipage}{0.49\textwidth}
\begin{lstlisting}
let rec list_pi (a, k) : float future list =
  let a' =
     future ((-1.0) ** (k +. 1.0)
             *. 4.0 /. (2. *. k -. 1.0)
             +. touch a)
  in
  a'::(list_pi (a', k +. 1))

let main () =
  touch (hd (tl (list_pi (0.0, 1.0))))
\end{lstlisting}
\end{minipage}%
\begin{minipage}{0.48\textwidth}
\begin{lstlisting}
type 'a pipe = Pipe of 'a * 'a pipe future

let rec pipeline_pi (a, k) : float pipe =
  let a' = a +. (-1.0) ** (k +. 1.0)
           *. 4.0 /. (2. *. k -. 1.0)
  in
  Pipe (a', future (pipeline_pi (a', k +. 1.)))

let main () =
  let Pipe (_, f1) = pipeline_pi (0.0, 1.0) in
  let Pipe (pi2, _) = touch f1 in pi2
\end{lstlisting}
\end{minipage}%
\caption{Two implementations of a function that iteratively computes~$\pi$
  with futures.}
\label{fig:intro-pi}
\end{figure}

\iffull\else
\begin{wrapfigure}{i}{0.45\textwidth}
%\begin{figure}
  \includegraphics[clip,scale=0.4]{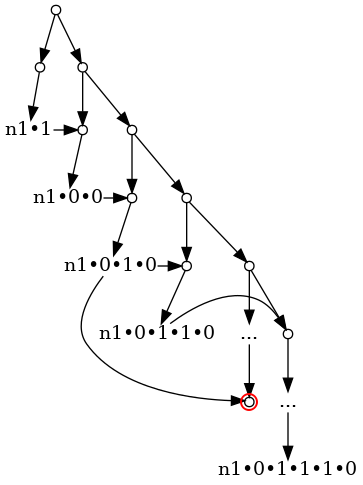}%
  \includegraphics[clip,scale=0.4]{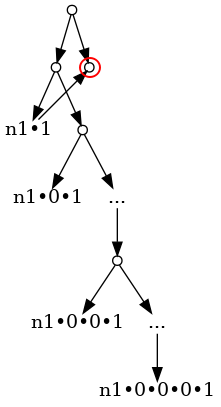}
  \caption{Visualizations for \texttt{list\_pi} (left) and \texttt{pipeline\_pi}
    (right) showing the differences in parallelization strategies. In both
    figures, the node corresponding to the \texttt{touch} operation in
    \texttt{main} is circled in red for emphasis.
  }
  \label{fig:intro-pi-graphs}
  %\vspace{-0.5cm}
  \end{wrapfigure}
%\end{figure}
\fi

In this paper, we present a graph type system that can statically compute graph
types for the above pipelining examples (and many more), thus allowing us to
detect and repair the parallelism bugs we discussed.
We present the type system in~$\langname$, a core calculus containing
both futures and recursive types.
The key to our approach is parameterizing recursive data structures involving
futures with a source of fresh vertex names called a {\em vertex structure} (or
VS, for short).
Conceptually, one can think of a VS as a separate structure of the same shape
as the program's recursive data structure, containing unique vertex names.
For example, both the~\lstinline{float future list} and the~\lstinline{float pipe} of
Figure~\ref{fig:intro-pi} would be parameterized by a stream of vertices.
The two functions~\lstinline{list_pi} and~\lstinline{pipeline_pi} would take this
vertex stream, let's call it~$\vs$,
as an implicit parameter and, at each iteration, use the next
vertex in the stream ($\kwfst{\vs}$) to spawn the new future and pass the
rest of the stream ($\kwsnd{\vs}$) to the recursive call.
As a result, the returned list (resp., pipe) will be ``zipped'' together with
the vertex stream in the sense that the first future in the list (resp., pipe)
will use the first vertex of the stream, and so on.
In this way, we need not ``unroll'' the vertex structure at compile time:
the types will refer to projections of a VS parameter.\footnote{We do, however,
  unroll a VS when unrolling the corresponding graph types, e.g., to create
  the visualizations in Figure~\ref{fig:intro-pi-graphs}. In the figure,
  $n1$ is the ``root'' of a VS and the notations that follow are projections
  of this VS.}
Vertex structures are not limited to streams: in general, a VS can be an
infinite (corecursive) tree with arbitrary branching patterns.
We show that this allows us to construct a VS corresponding to
arbitrary recursive data structures.

\iffull
\begin{wrapfigure}{i}{0.45\textwidth}
%\begin{figure}
  \includegraphics[clip,scale=0.4]{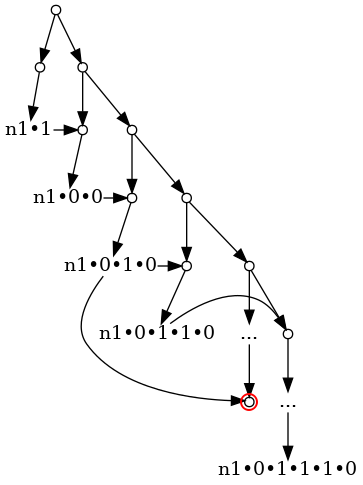}%
  \includegraphics[clip,scale=0.4]{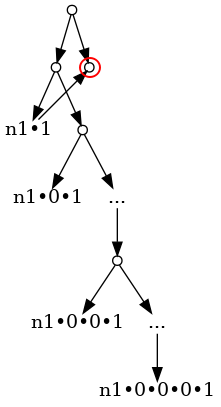}
  \caption{Visualizations for \texttt{list\_pi} (left) and \texttt{pipeline\_pi}
    (right) showing the differences in parallelization strategies. In both
    figures, the node corresponding to the \texttt{touch} operation in
    \texttt{main} is circled in red for emphasis.
  }
  \label{fig:intro-pi-graphs}
  %\vspace{-0.5cm}
  \end{wrapfigure}
%\end{figure}
\fi

Recall that we require vertex names to be unique.
The vertices contained by a vertex structure are all unique,
but ensuring that each vertex is used at most once is non-trivial
and requires numerous extensions to the graph type system.
One source of complexity is that types can perform significant computation
on vertex structures.
As an example, as we discussed above, the same vertex~$u$ cannot be used as
a name for two futures.
In the presence of computation on VSs, this is not a simple restriction to
enforce syntactically: if~$\vs$ is a vertex structure,
then under a reasonable semantics for
vertex structures,~$\kwfst{(\kwfst{\vs}, \kwsnd{\vs})}$ and~$\kwfst{\vs}$ refer
to the same vertex name and therefore cannot both be used to spawn futures.

The~$\langname$ calculus assumes that data structures are annotated with
vertex structures and makes explicit many of the manipulations of VSs
described above.
However,~$\langname$ should be seen as an intermediate representation---an
inference algorithm can infer all necessary annotations from unannotated code
in a high-level source language.
As a proof of concept, we extend GML~\citep{Muller22}, a graph type checker for
a subset of OCaml (but which did not previously support lists containing
futures, or any sort of user-defined algebraic data types) with support for
user-defined algebraic data types containing futures.
Our graph type checker is able to infer annotations and produce graph
types from the examples in Figure~\ref{fig:intro-pi}, as well as all of the
other examples contained in this paper, with no additional
annotations or programmer burden, as well as to produce visualizations of their
graph types.
As shown in the example above, such visualizations can allow programmers to
identify errors in the parallelization of their code, and can also be used
to reason about parallel complexity and other features.
Prior work~\citep{Muller22} also explains how graph types can be used to
aid other analyses, such as deadlock detection.
A formal presentation of the inference algorithm is out of the scope of this
paper, but much of the challenge in our extension is constructing the
vertex structure corresponding to an arbitrary user-defined algebraic data
type.
We describe this process formally and prove some metatheoretic results about
it.

In sum, our contributions are:
\begin{itemize}
\item \langname{}, a parallel calculus with a graph type system
  supporting recursive data types (Section~\ref{sec:lang}).
\item A soundness result for \langname{}, guaranteeing that the graph type of a
  program correctly describes the computation graph that arises when running the
  program (Section~\ref{sec:soundness}).
\item An algorithm for inferring the shape of a vertex structure that will
  provide the necessary vertex names for an arbitrary recursive data structure,
  and results showing (among other things) that such a VS exists for any
  valid recursive data type (Section~\ref{sec:infer}).
\item A prototype implementation of graph type inference for an OCaml-like
  source language, including OCaml-style user-defined algebraic
  data types mixed with futures (Section~\ref{sec:impl}).
  %. We have
  %used the implementation to visualize the graph types of the examples
  %described in this introduction, as well as several others .
\end{itemize}

\iffull
Technical details of the type system and many of the proofs are contained
in the appendix, which we reference throughout the paper.
\else
Due to space limitations, we defer some of the technical details
and many of the proofs to the full version of the paper~\citep{full}.
\fi
We begin with an overview of graph types as well as a high-level description of
our extensions.

%% file: overview.tex
\section{Overview}\label{sec:overview}

We begin with an overview of graph types but refer the interested reader to
the original paper~\citep{Muller22} for a more thorough presentation; we
indicate using footnotes where we diverge from that paper's presentation.
Our motivating example is a parallel implementation of Quicksort using
futures~(Figure~\ref{fig:qsort}).
The code is supplemented with annotations in gray that are inserted during type
inference and used in the formal presentation of~{\langname}, but are not
written in actual code; these annotations will be explained later, in
Section~\ref{sec:lang}.
The implementation returns immediately in the case of an empty list.
On a non-empty list, the first element is selected as a pivot and used to
partition the list using a sequential function~$\kw{partition}$, whose
implementation we omit.
A future is spawned to sort the first list recursively while the
second list is sorted in the main thread.
When the second list is sorted, we touch the future to retrieve the sorted
first list, and append the lists.

The type of the function in {\langname} is given below the code; the type
indicates that~$\kw{qsort}$ accepts and returns an~\lstinline{'a list}.
As is common in presentations of type-and-effect systems, we write an
annotation over the arrow indicating the effects performed by running the
function.
In this case, the ``effect'' is the {\em graph type} of the function; that is,
a graph type describing the family of computation graphs which might arise
from executing~$\kw{qsort}$.
The prefix~$\mu \gvar.$ indicates that~$\graph$ binds a recursive instance of
itself as~$\gvar$---this notation is taken from standard presentations of
recursive types.
The body of~$\graph$ is
a disjunction of two families of graphs, indicated by the~$\vee$ symbol.
This notation appears when the code executes a conditional or pattern
match and indicates two possible families of graphs:~$\graph_1 \dagor \graph_2$
indicates that the graph can take a form indicated by either~$\graph_1$
or~$\graph_2$.
In the example,
the first graph type, $\emptygraph$, indicates a sequential computation
and corresponds to executing the base case.
The second graph type corresponds to the recursive case, and indicates that a
future is spawned.
In order to refer to this future later in the graph type, futures are assigned
unique names.
By convention, these names are assumed to refer to a vertex ``attached'' to the
computation graph of a future as a final vertex.
We will refer to this vertex as the ``sink'' vertex of the future,
borrowing a term from graph
theory because, until the future is touched, it has no outgoing edges.
The notation~$\dagnew{\hastycl{\vvar}{\kwvty}}{}$, which appears in
corresponding locations in the graph type~$\graph$ and as an annotation in the
code, indicates binding a new {\em vertex variable}~$\vvar$ which locally refers
to a new, fresh vertex name; $\kwvty$ is the type of this variable and means
that $\vvar$ refers to a single vertex.\footnote{This type annotation is not
  used or needed in prior work, where only single vertices are bound. We
  introduce it here for consistency with the syntax used in the rest of this
  paper.}
When the annotated program is evaluated (we do this only to prove soundness;
the vertex name annotations have no runtime meaning in the actual program)
or the graph type is unrolled (e.g., to produce the visualizations of
Figure~\ref{fig:intro-pi-graphs}),~$\vvar$ will be instantiated with a new,
fresh vertex name.\footnote{In the terminology of this paper, it will actually
  be instantiated with a new vertex structure of type~$\kwvty$.}

\newcommand{\figannot}[1]{{\color{gray}#1}}

%\begin{figure}
%\begin{minipage}{0.55\textwidth}
\begin{wrapfigure}{i}{0.55\textwidth}
  \hspace{0.05\textwidth}
  \begin{minipage}{0.5\textwidth}
    \begin{lstlisting}
let rec qsort (l: 'a list) : 'a list =
  match l with
  | [] -> []
  | p::t ->
    let (lt, ge) = partition p t in
    $\figannot{\kwnewf{\vvar}{\kwvty}{}}$
    let future_sort_lt = future$\figannot{[\vvar]}$ (qsort lt) in
    let sort_ge = qsort ge in
    let sort_lt = touch future_sort_lt in
    sort_lt @ [p] @ sort_ge
\end{lstlisting}
%  \end{minipage}%
%  \begin{minipage}{0.45\textwidth}
    \[
    \begin{array}{r c l}
    \kw{qsort} & : & \kwarrow{\kw{'a~list}}{\kw{'a~list}}{\graph}\\
    \multicolumn{3}{l}{\text{where}}\\
    \graph & = &
    \dagrec{\gvar}{[\emptygraph \dagor (\dagnew{\hastycl{\vvar}{\kwvty}}
        {\leftcomp{\gvar}{\vvar} \seqcomp \gvar \seqcomp
          \touchcomp{\vvar}})]}{}
    \end{array}
    \]
%  \end{minipage}
  \caption{Code and types for parallel-recursive Quicksort using futures.
    %(the definition of the partition function is omitted).
    Code annotations
    in gray are shown for convenience; these are not written by the
    programmer.}
  \label{fig:qsort}
  \end{minipage}
  %\end{figure}
\end{wrapfigure}

The sequential composition of two graph types is
denoted~$\graph_1 \seqcomp \graph_2$, indicating that the program performs
a computation described by~$\graph_1$ followed by one described by~$\graph_2$.
In our example graph type, the graph type corresponding to the recursive
case is the sequential composition of three operations.
The graph type~$\leftcomp{\gvar}{\vvar}$ indicates that~$\vvar$ is the sink
of a future whose graph is described by the graph type~$\gvar$
(which, recall, is a recursive
instance of~$\graph$ corresponding to a recursive call to~$\kw{qsort}$).
In general,~$\leftcomp{\graph}{\vvar}$ indicates a future whose computation
graph can be described by~$\graph$ and whose sink vertex is given the
name~$\vvar$.
%
%\aaa{It is not clear what it means for a vertex to be attached to a graph; is
%it attached at any particular position (e.g. the sink)?}%
In~{\langname}, spawns using the \lstinline{future} keyword are also annotated
with the vertex that is used; this annotation is shown in gray in the code.
The spawn in the graph type
is then sequentially composed
with another instance of~$\gvar$
for the other recursive call, and finally a touch of the future whose sink
is~$\vvar$ (a touch of vertex~$\vvar$ is denoted~$\touchcomp{\vvar}$).

Note that the vertex~$\vvar$ in the Quicksort example exists only within the
scope of the binding and so, in particular,
cannot be allowed to escape the scope.
If futures are, e.g., returned from a function, the vertices for those futures
must be created outside and passed as parameters to the function.
As an example, take the~$\kw{pipeline\_pi2}$ function in
Figure~\ref{fig:spawn2}, which returns a future.  This function is similar to
the analogous function of Section~\ref{sec:intro}, but limited to two
approximations of $\pi$.
The vertex parameter is made explicit in the~{\langname} annotations, and
also appears in the type of the function (shown on the
right side of the figure) as a~$\Pi$ binding.
This construct in a graph type binds two parameters.
Both parameters stand for {\em vertex structures} (VSs), type-level
(co)data structures containing vertices,
and both are annotated with {\em vertex structure types}
indicating their shapes.
The first parameter,~$\vvar$, will contain the vertices the function may use to
spawn futures.
In the case of~\lstinline{pipeline_pi2}, it is annotated with VS
type~$\kwprod{\kwvty}{\kwvty}$, indicating a pair of vertices (recall
that~$\kwvty$ is a VS type representing a single vertex).\footnote{Note that prior work had similar notation for
vertex parameters but, as it did not have vertex structures, allowed~$\Pi$s to
bind arbitrary-length vectors of vertex parameters.
The introduction of vertex structures in the present work, which we will use
later to more substantial effect, also simplifies this notation and makes it
more uniform.}
The second parameter contains the vertices the function may touch; in the
case of~\lstinline{pipeline_pi2}, it is empty as indicated by the unit VS
type~$\kwunitty$.\footnote{The theory of {\langname} does not include the
VS type~$\kwunitty$ to keep the calculus minimal, but it is included in our
implementation and would be a straightforward addition to the calculus.}

The function $\kw{pipeline\_pi2}$ returns a future (spawned using the first
component of the vertex structure~$\vvar$) that produces a pair of a float and
another future, spawned with the second component.
Note that the types of futures explicitly indicate the vertices with which
the future was spawned.
As in the Quicksort code, these vertices also appear as annotations on the
\lstinline{future} keyword which are inferred during type checking.
Finally, the graph type~$\graph$ of the function body shows that the function
spawns a future using the vertex~$\kwfst{\vvar}$, which in turn spawns a future
using the vertex~$\kwsnd{\vvar}$, which finally does not spawn further threads.

The code in Figure~\ref{fig:spawn2} also shows a function that calls
$\kw{pipeline\_pi2}$ and touches the two futures.
The graph type of this function binds a new vertex structure for
the two futures, which no longer need to escape the function.
As in Quicksort, the new vertex structure is bound using a binding of the
form~$\dagnew{\hastycl{\vvar}{\vsty}}{\graph}$, where the vertex structure
type is now the product~$\kwprod{\kwvty}{\kwvty}$ instead of~$\kwvty$.
The call to~$\kw{pipeline\_pi2}$ instantiating the bound vertex structure
variable~$\vvar$ with the new VS~$\vvar'$ is indicated by substituting~$\vvar'$
for~$\vvar$ in~$\graph$.

\begin{figure}
  \begin{minipage}{0.4\textwidth}
    \begin{lstlisting}
let pipeline_pi2$\figannot{[\hastycl{\vvar}{\kwprod{\kwvty}{\kwvty}}; \hastycl{\_}{\kwunitty}]}$ () =
  future$\figannot{[\kwfst{\vvar}]}$ (3.1, future$\figannot{[\kwsnd{\vvar}]}$ 3.14)

let use_pi () =
  $\figannot{\kwnewf{\vvar}{\kwprod{\kwvty}{\kwvty}}{}}$
  let (pi1, pi2_fut) =
    touch (pipeline_pi2 ())
  in touch pi2_fut
\end{lstlisting}
  \end{minipage}%
  \begin{minipage}{0.55\textwidth}
    \[
    \begin{array}{r c l}
      \kw{pipeline\_pi2} & \!:\! &
      \dagpi{\hastycl{\vvar}{\kwprod{\kwvty}{\kwvty}}}
            {\hastycl{\_}{\kwunitty}}
            {\kwarrow{\kwunit}
              {\kwfutt{\tau}{\kwfst{\vvar}}}
              {\graph}}\\
      \kw{use\_pi} & \!:\! &
      \kwarrow{\kwunit}
              {\kw{float}}
              {\graph'}
              \end{array}
              \]
              where
              \[
              \begin{array}{r c l}
              %\multicolumn{3}{l}{\text{where}}\\
              \tau & = & \kwprod{\kw{float}}
                  {\kwfutt{\kw{float}}{\kwsnd{\vvar}}}\\
    \graph & = &
    \leftcomp{(\leftcomp{\emptygraph}{\kwsnd{\vvar}})}{\kwfst{\vvar}}\\
    \graph' & = &
    \dagnew{\hastycl{\vvar'}{\kwprod{\kwvty}{\kwvty}}}
           {\sub{\graph}{\vvar'}{\vvar} \seqcomp \touchcomp{\kwfst{\vvar}}
             \seqcomp \touchcomp{\kwsnd{\vvar}}}
    \end{array}
    \]
  \end{minipage}
  \caption{A function that iteratively approximates~$\pi$
    twice in a pipelined manner}
  \label{fig:spawn2}
\end{figure}

Before proceeding, we make one additional note about the
graph type system.
We have referred to~$\vertex$ and similar as {\em unique} vertex
names---each vertex name can be used to spawn a future at most once,
otherwise the resulting graph will be ambiguous (if two futures have~$\vertex$
as a sink vertex, there is no way to know to which future a
touch~$\touchcomp{\vertex}$ refers).
The graph type system (both ours and that of prior work) enforce this using
an affine type system that restricts the use of vertex names.

\begin{figure}
  \begin{minipage}{0.42\textwidth}
    \begin{lstlisting}
type 'a pipe = Pipe of 'a * 'a pipe future

let rec pipeline_pi$\figannot{[\hastycl{\vvar}\kw{vstream}; \hastycl{\_}{\kwunitty}]}$ (a, k) =
  let a' = a +. (-1.0) ** (k +. 1.0)
           *. 4.0 /. (2. *. k -. 1.0)
  in
  Pipe (a', future$\figannot{[\kwfst{\vvar}]}$
        (pipeline_pi$\figannot{[\kwsnd{\vvar}]}$ (a', k +. 1.)))

let rec nth$\figannot{[\hastycl{\_}{\kwunitty}; \hastycl{\vvar}{\kw{vstream}}]}$
  ((pipe, n) : 'a pipe$\figannot{[\vvar]}$ * int) =
  let Pipe (a, f) = pipe in
  if n <= 0 then a
  else nth$\figannot{[\kwsnd{\vvar}]}$ (touch f, n - 1)

let main () =
  $\figannot{\kwnewf{\vvar}{\kw{vstream}}{}}$
  nth$\figannot{[\vvar]}$ (pipeline_pi$\figannot{[\vvar]}$ (0.0, 1.0)) 1000
\end{lstlisting}
  \end{minipage}%
  \begin{minipage}{0.55\textwidth}
    \[
    \begin{array}{r l l}
      \kw{pipeline\_pi} & \!:\! &
      \dagpi{\hastycl{\vvar}{\kw{vstream}}}
            {\hastycl{\_}{\kwunitty}}
            {}\\
            & & \quad \kwarrow{\kwprod{\kw{float}}{\kw{float}}}
              {\kw{float~pipe}[\vvar]}
              {\graph}\\
      \kw{nth} & \!:\! &
      \dagpi{\hastycl{\_}{\kwunitty}}{\hastycl{\vvar}{\kw{vstream}}}
            {}\\
            & & \quad \kwarrow{\kwprod{\kw{float~pipe}[\vvar]}{\kw{int}}}
              {\kw{float}}
              {\graph'}\\
    \kw{main} & : & \kwarrow{\kwunit}{\kw{float}}
       {\graph''}
       \end{array}
       \]
       where
       \[
       \begin{array}{r l l}
    \graph & \!=\! &
    \dagrec{\gvar}
           {\dagpi{\hastycl{\vvar}{\kw{vstream}}}
             {\hastycl{\_}{\kwunitty}}
             {\leftcomp{(\kwsingletapp{\gvar}{\kwsnd{\vvar}})}{\kwfst{\vvar}}}}
           {}\\
    \graph' & \!=\! &
    \dagrec{\gvar}
           {\dagpi{\hastycl{\_}{\kwunitty}}{\hastycl{\vvar}{\kw{vstream}}}
             {\emptygraph \dagor \touchcomp{\kwfst{\vvar}} \seqcomp
               \kwsingletapp{\gvar}{\kwsnd{\vvar}}}}
           {}\\
    \graph'' & \!=\! &
    \dagnew{\hastycl{\vvar}{\kw{vstream}}}
           {\kwtapp{\graph}{\vvar}{\kwtriv} \seqcomp
             \kwtapp{\graph'}{\kwtriv}{\vvar}}\\
    \multicolumn{3}{l}{\kw{vstream}  = 
    \kwvscorec{\vstyvar}{\kwprod{\kwvty}{\vstyvar}}}
    \end{array}
    \]
  \end{minipage}
  \caption{Code and types for a function that iteratively approximates~$\pi$
    indefinitely in a pipelined manner.}
  \label{fig:pi}
\end{figure}

Thus far, we have discussed examples that are within the capabilities of prior
work.
Now suppose we wish to generalize~$\kw{pipeline\_pi2}$ to continue producing
iterative approximations indefinitely.
The code in Figure~\ref{fig:pi} does this,
producing a value of type~$\kw{float~pipe}$, also defined in the figure, which
is a recursive type containing an approximation and a future to
continue the pipeline.
This is reminiscent of the ``producer'' example of~\citet{BlellochRe97}, who
use futures in this general pattern to construct a wide variety of
pipelined data structures.
As in the Introduction, each iteration
computes the~$\kw{k}$th term in the approximation, adds it to a
running total~$\kw{a}$, and returns the new running total as well as a future to
call~$\kw{pipeline\_pi}$ recursively to compute the~$\kw{k+1}$st term.

Useful instances of the recursive data type~$\kw{'a~pipe}$ cannot be typed with
the existing graph type system, because doing so would require an infinite
sequence of new vertex names and a way of associating each future in the
pipeline with successive vertex names.
One (incorrect but illustrative) approach would be to instantiate each future
in the type with a fresh vertex name using, for example, an existential.
The pipe type would then be annotated as follows:
\begin{lstlisting}
type 'a pipe = Pipe of 'a * pipe future$\figannot{[\exists \vvar : \kwvty. \vvar]}$
\end{lstlisting}
This is still not useful, however, because it doesn't allow any vertex name
to escape the scope of the single future type, just as the vertex in
the~$\kw{qsort}$ function was confined to the~$\kw{qsort}$ function.
The type~$\kwfutt{\tau}{\exists \vvar:\kwvty. \vvar}$ tells us that the
future is spawned with {\em some} vertex, but gives no information about
{\em which}, an untenable loss of precision when we try to touch this future
and add an edge to its vertex.
As an illustration of this loss of precision, consider the following program,
and suppose we wish to use its graph type to check for deadlocks (simply put,
a program may deadlock if its graph type can unroll to a cyclic graph):
\begin{lstlisting}
let rec f n =
  if n <= 0 then [future (fun () -> 0)]
  else
    let l' = f (n - 1) in
    (future (fun () -> touch (List.hd l')))::l'
\end{lstlisting}
Each future in the list contains a function that touches the following
future in the list.
This is a fairly clear structure and a
visualization or suitable analysis of the graph type produced by our
system could show that the program is deadlock-free.
However, if the type of the output list were given as
$\kwfutt{(\kwarrow{\kw{unit}}{\kw{int}}{\touchcomp{\exists
    u. u}})}{\exists u. u}$,
the most precise thing that could be said about this list
is that it is a list of thunks under futures, each of which touches
any future in the list (or, indeed, without further information, any
future in the {\em program}), including itself.
Thus, a sound deadlock
detector would have to conclude that the program might deadlock.

As a more precise solution to the problem of generating unique vertex names
for elements in a data structure,
we introduce {\em vertex structures}, mentioned above, which we
allow to be (co)recursive and thus serve as the source or collection of vertex
names we need.
In the code annotations and the type of~$\kw{pipeline\_pi}$ on the right side
of the figure, the function takes a vertex structure parameter of vertex
structure type~$\kw{vstream}$, which is defined in the lower right side of
the figure to be a corecursive type of an infinite list or stream of vertices.
The return type of the function is~$\kw{float~pipe}[\vvar]$, where the
recursive~$\kw{pipe}$ type is now parameterized by a vertex structure.
This vertex structure is threaded through the recursive structure of the
pipeline data type such that successive futures in the data type are associated
with corresponding vertices from~$\vvar$.
The details of this are technical and so we defer them, as well as the formal
presentation of recursive data types in~{\langname}, to the next section.
As in~$\kw{pipeline\_pi2}$, the first future
uses the vertex~$\kwfst{\vvar}$, which appears as an annotation in
the code and on the graph type ($G$ contains~$\leftcomp{}{\kwfst{\vvar}}$,
indicating a spawn of~$\kwfst{\vvar}$).\footnote{We treat
  corecursive vertex structure types as equi(co)recursive, so no unrolling is
  needed.}
However, now the function calls itself recursively to generate
the rest of the pipeline.
Because the function takes a vertex parameter, this recursive call must
instantiate the vertex parameter with a~$\kw{vstream}$, and it does so with
the tail of the stream,~$\kwsnd{\vvar}$.
This appears in the graph type
as~$\kwsingletapp{\gvar}{\kwsnd{\vvar}}$.

We complete this overview with a demonstration of how the pipeline can be
consumed, which shows how vertex structures link individual futures to their
touches.
The~$\kw{nth}$ function in Figure~\ref{fig:pi} consumes a pipeline recursively,
returning the~$\kw{n}^{th}$ value.
It also takes a parameter~$\vvar$ of vertex structure type~$\kw{vstream}$, but
this time as the second parameter, because the function uses these vertices
to touch futures and does not spawn futures.
The use of~$\vvar$ as the parameter to the~$\kw{pipe}$ data type indicates that
vertices for futures in the pipeline will be drawn from the stream~$\vvar$,
which is enough information to infer in the graph type~$\graph'$ that the
touch~$\touchcomp{\kwfst{\vvar}}$ targets the first vertex of~$\vvar$.
The function then calls itself recursively with the tail of the vertex stream,
which also appears in the recursive instantiation of~$\gvar$ in the graph
type.
Finally,~$\kw{main}$ calls~$\kw{nth}$ with the~$\pi$ pipeline.
As we have seen before, the vertex structure~$\vvar$ is bound here so that its
scope covers its uses both for spawns (in~$\kw{pipeline\_pi}$) and for
touches (in~$\kw{nth}$).
The calls to both functions instantiate the vertex structure parameter with
the same vertex structure~$\vvar$, linking the spawns and touches in the
graph types.
The graph type for~$\kw{main}$ composes the graph types of the producer and the
consumer and links the spawns and touches by instantiating both graph types
with the same vertex structure.

%% file: lang.tex
%% TREELIKE EXPLANATION

%% Figure \ref{fig:treelike} describes rules for determining if a VS type is "treelike." The judgement $\vsttree{\vstctx}{\vsty}$ indicates that $\vsty$ is treelike under the context $\vstctx$, which lists treelike VS type variables $\vstyvar$. A VS type is considered "treelike" if it is $\kwvty$, a product of treelike VS types, a corecursive VS type of a treelike VS type, or a VS type variable in $\vstctx$. A treelike VS type can be imagined as a tree, where $\kwvty$ is a leaf, a pair VS type is a parent of its two components VS types, and a corecursive VS type is a recursive instance of the tree. 

%% Treelike VS types are important because they represent the treelike way that vertex generators generate vertices; consequently, vertex generators must have treelike VS types. For a vertex generator $\vertgen{\vsty}{\genseed}$, if $\vsty$ is $\kwvty$, it is a vertex; if it is a pair VS type, then the subtrees representing its components can be traversed using $\vgenl{\vertgen{\vsty}{\genseed}}$ and $\vgenr{\vertgen{\vsty}{\genseed}}$; if it is a corecursive VS type, a corecursive layer of the tree can be unrolled and further traversed via $\vgenproj{\vertgen{\vsty}{\genseed}}{\vvar}$. Since the $\kw{new}$ clause declares a new VS variable that is replaced with a vertex generator during runtime, \langname requires that the VS type of the declared VS variable is treelike.

\section{Graph Types with Vertex Structures}\label{sec:lang}

%There are seven components in \langname: expressions $e$, type constructors $\con$, kinds $\kind$, graph types $\graph$, graph kinds $\graphkind$, vertex structures (VSs) $\vs$, and vertex structure types (VS types) $\vsty$. 

%In \langname, valid expressions are statically assigned a type and a graph type. Any expressions that utilize futures, as well as their types and graph types, must utilize VSs for two reasons: graph types need VSs to reference the futures used in expressions, and constructors require VSs to type expressions using them. To enforce VS type safety, constructors and graph types require that VS parameters and variable bindings be annotated with a VS type.

%\begin{wrapfigure}{i}{0.6\textwidth}
\begin{figure}
%  \begin{minipage}{0.6\textwidth}
  \begin{mathpar}
  \begin{array}{l r l l}
    \mathit{Vertex~Structures} & \vs & \bnfdef &
    %\vertgen{\vsty}{\genseed} \bnfalt
    \vvar \bnfalt
    \kwpair{\vs}{\vs} \bnfalt
    \kwfst{\vs} \bnfalt
    \kwsnd{\vs} %\bnfalt
%    \kwinl{\vs} \bnfalt
%    \kwinr{\vs}
    \\

    %% \mathit{Vertex~Path} & \vspath & \bnfdef &
    %% \vertgen{\vsty}{\genseed} \bnfalt
    %% \kwfst{\vspath} \bnfalt
    %% \kwsnd{\vspath}
    %% \\

    %% \mathit{Normal~VSs} & \vsv & \bnfdef &
    %% \vspath \bnfalt
    %% \kwpair{\vsv}{\vsv} \bnfalt
    %% \kwinl{\vsv} \bnfalt
    %% \kwinr{\vsv} \bnfalt
    %% \kwvsroll{\vsv} \bnfalt
    %% \overline{\vsv}
    %% \\

    %% \mathit{Aux.~Normal~VSs} & \overline{\vsv} & \bnfdef &
    %% \vvar \bnfalt
    %% \kwfst{\overline{\vsv}} \bnfalt
    %% \kwsnd{\overline{\vsv}} \bnfalt
    %% \kwvsunroll{\overline{\vsv}}\\

    \mathit{Vertex~Structure~Types} & \vsty & \bnfdef &
    \kwvty \bnfalt
    \vsprod{\vsty}{\avail}{\vsty}{\avail} \bnfalt
%    \kwsum{\vsty}{\vsty} \bnfalt
    \vstyvar \bnfalt
%    \kwvsrec{\vstyvar}{\vsty} \bnfalt
    \kwvscorec{\vstyvar}{\vsty}
    \\

    \mathit{Availability} & \avail & \bnfdef &
    \isav \bnfalt
    \isunav
    \\

    \mathit{Graph~Types} & \graph & \bnfdef &
    \gvar \bnfalt
    \emptygraph \bnfalt
    \graph \seqcomp \graph \bnfalt
%%    \graph \parcomp \graph \bnfalt
    \graph \dagor \graph \bnfalt
    \dagrec{\gvar}{\graph}{} \bnfalt
    \leftcomp{\graph}{\vs} \bnfalt
    \touchcomp{\vs} \bnfalt \\ & & &
    \dagpi{\hastycl{\vvar}{\vsty}}{\hastycl{\vvar}{\vsty}}{\graph} \bnfalt
    \kwtapp{\graph}{\vs}{\vs} \bnfalt
    \dagnew{\hastycl{\vvar}{\vsty}}{\graph} %\bnfalt\\ & & &
%    \kwvcase{\vs}{\vvar}{\graph}{\vvar}{\graph}
    \\

    \mathit{Graph~Kinds} & \graphkind & \bnfdef &
    \kgraph \bnfalt
    \dagpi{\hastycl{\vvar_f}{\vsty_f}}{\hastycl{\vvar_t}{\vsty_t}}{\graphkind}
    \\

    \mathit{Kinds} & \kind & \bnfdef &
    \kwtykind \bnfalt
    \kwkindarr{\vsty}{\kwtykind}
    \\

    \mathit{Type~Constructors} & \con & \bnfdef &
    \kwunit \bnfalt
    \kwpi{\hastycl{\vvar}{\vsty}}{\hastycl{\vvar}{\vsty}}{\kwarrow{\con}{\con}{\graph}} \bnfalt
    \kwprod{\con}{\con} \bnfalt
    \kwsum{\con}{\con} \bnfalt
    \kwfutt{\con}{\vs} \bnfalt
    \convar \bnfalt\\ & & &
    \kwprec{\convar}{\hastycl{\vvar}{\vsty}}{\con}{\vs} \bnfalt
    \kwxi{\hastycl{\vvar}{\vsty}}{\con} \bnfalt
    \kwvapp{\con}{\vs} %\bnfalt
%    \kwvcase{\vs}{\vvar}{\con}{\vvar}{\con}
    \\
    
    \mathit{Expressions} & e & \bnfdef &
    x \bnfalt
%%    v \bnfalt
    \kwtriv \bnfalt
    \kwfun{\vvar}{\vvar}{f}{x}{e} \bnfalt
    \kwapp{}{\kwtapp{e}{\verts}{\verts}}{e} \bnfalt
    \kwpair{e}{e} \bnfalt
    \kwfst{e} \bnfalt
    \kwsnd{e} \bnfalt\\ & & &
%%    \kwpar{e}{e} \bnfalt
    \kwinl{e} \bnfalt
    \kwinr{e} \bnfalt
    \kwcase{e}{x}{e}{y}{e} \bnfalt
    \kwroll{e} \bnfalt
    \kwunroll{e} \bnfalt\\ & & &
    \kwfuture{\vs}{e} \bnfalt
    \kwforce{e} \bnfalt
    \kwnewf{\vvar}{\vsty}{e} %\bnfalt\\ & & &
%    \kwcase{\vs}{\vvar}{e}{\vvar}{e}
%%    \\
%%
%%    \mathit{Values} & v & \bnfdef &
%%    \kwtriv \bnfalt
%%    \kwfun{\vvar}{\vvar}{f}{x}{e} \bnfalt
%%    \kwpair{v}{v} \bnfalt
%%    \kwinl{v} \bnfalt
%%    \kwinr{v} \bnfalt
%%    \kwroll{v} \bnfalt\\ & & &
%%    \kwhandle{\vspath}{v}
  \end{array}
  \end{mathpar}
  \caption{Syntax of {\langname}.}
  \label{fig:lang}
  %\end{minipage}
\end{figure}
%\end{wrapfigure}

This section provides a formal presentation of \langname{}, whose syntax
is given in Figure~\ref{fig:lang}.
%
%In some respects, the language is similar to the one introduced by
%\citet{Muller22}.  In particular, expressions ($e$) in \langname{} are assigned
%both types ($\tau$) and graph types ($\graph$), and graph types describe the
%computation graphs that can arise when running the expression.
%
In the remainder of this section, we describe the features of the language
in detail, focusing on the main novelties of \langname{} compared to prior
work: \emph{vertex structures} (VSs) and \emph{recursive types}.
%
%% As illustrated in the examples we have seen so far, recursive types such as
%% trees or lists are crucial for expressing common programming idioms for
%% fine-grained parallelism.  However, these idioms require data types that can
%% store a potentially infinite number of distinct futures, which is not possible
%% in prior work~\citep{Muller22}.  The VSs of \langname{} provide a way for data
%% types to refer to infinitely many distinct futures, thereby circumventing this
%% limitation.
%
%In the remainder of this section, we discuss these two features in detail.
%(Figure~\ref{fig:lang} summarizes the syntax of the language for reference.)

\subsection{Vertex Structures and Their Types}

Vertex Structures~($\vs$) contain vertices that represent futures in
computation graphs.  As shown in Figure~\ref{fig:lang}, VSs appear in
annotations within expressions, type constructors, and graph types;
these annotations are not inserted into real code by programmers, but
are filled in during type inference.

Vertex structures are classified with VS types ($\vsty$).  The
VS type~$\kwvty$ represents a single vertex, and only VSs of type $\kwvty$ can be used
to name futures.  The product type
$\vsprod{\vsty_1}{\avail_1}{\vsty_2}{\avail_2}$ represents pairs of VSs in
$\vsty_1$ and $\vsty_2$.  The availability annotations $\avail_1$ and $\avail_2$
indicate whether the corresponding component is available ($\isav$) or
unavailable ($\isunav$) for spawning new futures; their use is
inspired by record types in the Cogent language~\cite{OConnorCRJAKMSK21}.
The need for availability will become clearer later, when we
discuss the type system.  Finally, we can also form
corecursive VS types $\kwvscorec{\vstyvar}{\vsty}$, which we will use to generate
graph types that require a potentially unbounded number of vertices.

Figure~\ref{fig:vert-typing} presents the rules for assigning VS types
to VSs. The judgment
$\vsistype{\uspctx}{\utctx}{\vs}{\vsty}$ denotes that the VS $\vs$ has VS type
$\vsty$, where $\uspctx$ and $\utctx$ are contexts that map VS variables to
their types.
Vertices, and thus VSs, are treated in an affine manner to ensure that any
vertex is used at most once to spawn a future---this affine treatment leads
to the use of two contexts.
The first,~$\uspctx$, is an affine context storing vertices that may be used
to spawn futures and the second,~$\utctx$, is an unrestricted context for
vertices that may be used to touch futures (we may touch a vertex any number
of times).
Because we wish to be able
to touch any vertex we spawn, the set of variables in $\uspctx$ will always be a
subset of that in $\utctx$.
%, though they may have different mappings to VS types due to $\uspctx$-splitting (discussed later).
A VS variable is well-typed if it is in either $\uspctx$ or $\utctx$, and
%, and a vertex generator is well-typed if its seed is in either context. 
we assume that $\uspctx$ does not contain multiple
mappings for the same variable.

\input{fig-vert-typing}
\input{fig-vert-subtyping}

VSs can be variables~($\vvar$),\footnote{
	Unlike the original
  presentation~\citep{Muller22}, $\vvar$ refers to a variable instead of a
  vertex.} pairs, and
projections.  As seen in the rules $\rulename{U:Fst}$ and $\rulename{U:Snd}$,
only available components can be projected. For example, if $\vvar$ has VS type
$\vsprod{\vsty_1}{\isav}{\vsty_2}{\isunav}$, then $\kwfst{\vvar}$ is safe to
use, but $\kwsnd{\vvar}$ is not.
Rule~\rulename{U:Subtype} is a subsumption rule for the subtyping
relation on VS types, denoted $\vstysubt{\vsty'}{\vsty}$ and defined
in Figure~\ref{fig:vert-subtyping}.
%As seen in the subsumption rule
%\rulename{U:Subtype}, this means that a VS of VS type $\vsty'$ may be used wherever
%one of VS type $\vsty$ is expected.
%
We allow three forms of subtyping: first (\rulename{UT:Corec1} and
\rulename{UT:Corec2}), we can freely roll and unroll corecursive VS types.
Second
(\rulename{UT:ProdLeft} and \rulename{UT:ProdRight}),
it is safe to take an available component and treat it as unavailable.
Third, the types of
unavailable components of VSs may be changed at will, which is safe since those
sides can never be used.
%% For example, we can prove the following subtype
%% judgment:

%% \[\vstysubt{
%%     \vsprod{\kwvty}{\isav}{
%%       (\kwvscorec{\vstyvar}{\vsprod{\kwvty}{\isav}{\vstyvar}{\isav}})}
%%     {\isav}
%%   }
%%   {
%%     \vsprod{\kwvty}{\isunav}{
%%       (\vsprod{\kwvty}{\isav}{\kwvty}{\isunav})}
%%     {\isav}}.\]

\input{fig-vert-type-split}
\input{fig-uspctx-split}

The VS typing rule \rulename{U:Pair} uses an auxiliary splitting relation
$\uspsplit{\bullet}{\bullet}{\bullet}$~\cite{OConnorCRJAKMSK21}, which is
described in Figure~\ref{fig:uspctx-split}. 
This relation is responsible for
enforcing the affine treatment of $\uspctx$ contexts. The judgment
$\uspsplit{\uspctx}{\uspctx_1}{\uspctx_2}$ states that $\uspctx$ splits into the
disjoint contexts $\uspctx_1$ and $\uspctx_2$. It is important that~$\uspctx_1$
and~$\uspctx_2$ be disjoint so that futures spawned under $\uspctx_1$ and under
$\uspctx_2$ have distinct vertices. However, we allow a variable with a product
VS type to appear in both contexts, as long as the availability of the products is
in turn split between the two. This is allowed by \rulename{OM:VarTypeSplit}:
$\uspctx, \hastype{\vvar}{\vsty}$ may split to
$\uspctx_1, \hastype{\vvar}{\vsty_1}$ and $\uspctx, \hastype{\vvar}{\vsty_2}$ if
$\vstysplit{\vsty}{\vsty_1}{\vsty_2}$ holds. Intuitively, $\vsty$, $\vsty_1$,
and $\vsty_2$ are the same types but with different availabilities: if a
component of a product VS type is available in $\vsty$, then that component is
available in $\vsty_1$ or $\vsty_2$ or neither, but not both. For example,
if~$\hastycl{\vvar}{\vsprod{\vsty_1}{\isav}{\vsty_2}{\isav}}$ appears
in~$\uspctx$, we may
have~$\hastycl{\vvar}{\vsprod{\vsty_1}{\isav}{\vsty_2}{\isunav}}$ in~$\uspctx_1$
and~$\hastycl{\vvar}{\vsprod{\vsty_1}{\isunav}{\vsty_2}{\isav}}$ in~$\uspctx_2$,
but we cannot have~$\hastycl{\vvar}{\vsprod{\vsty_1}{\isav}{\vsty_2}{\isav}}$
appear in either $\uspctx_1$ and $\uspctx_2$.

The VS type splitting judgment $\vstysplit{\vsty}{\vsty_1}{\vsty_2}$ is defined in
Figure \ref{fig:vert-type-split}. The core mechanism of VS type splitting is
\rulename{US:Prod}, which states a VS type $\vsprod{\vsty_1}{\isav}{\vsty_2}{\isav}$
may split into $\vsprod{\vsty_1}{\isav}{\vsty_2}{\isunav}$ and
$\vsprod{\vsty_1}{\isunav}{\vsty_2}{\isav}$. Most of the other rules are
``search'' rules allowing applications of \rulename{US:Prod} in nested VS types. The
other significant VS type splitting rule is \rulename{US:Subtype}, which allows for
``weakening'' the VS types resulting from a split (by turning available sides of
product VS types to unavailable).

\subsection{Graph Types and Type Constructors}

\input{fig-dag-wf-abbr}
\input{fig-kinds}

There are two kinding judgments to
characterize well-formed types: one for graph
types, and another one for type constructors.
For graph types, the judgment
$\dagwf{\gctx}{\uspctx}{\utctx}{\graph}{\graphkind}$ states that the graph type
$\graph$ has graph kind $\graphkind$
(Figure~\ref{fig:dag-wf-abbr}).
%The most notable change is the use of VSs; in
%particular, the contexts $\uspctx$ and $\utctx$ can bind variables of different
%VS types, whereas \citet{Muller22} only allowed variables denoting individual
%vertices.
The context $\gctx$ maps graph type variables to their kinds, and is
used to check that recursive graph types are well-formed
(see \rulename{DW:RecPi}).
\ifapp
The full set of graph type formation rules are given in Figure \ref{fig:dag-wf} in the appendix.
\else
We show selected graph type formation rules for space reasons; others are similar and are given in
the supplementary appendix.
\fi

For type constructors, the judgment
$\iskind{\gctx}{}{\utctx}{\vstctx}{\con}{\kind}$ states that the type
constructor $\con$ is well-formed and has kind $\kind$ (Figure~\ref{fig:kinds}).
%Compared to prior work, the two main changes are the use of kinds $\kind$ and
%the context $\vstctx$.
%\footnote{%
%  In prior work~\cite{Muller22}, there was no need for $\vstctx$ or for multiple
%  kinds of type constructors, so the analog of this judgment was simply written
%  $\gctx; \utctx \vdash \con~\mathsf{ok}$.}
The context $\vstctx$ maps type
variables to their kinds.
%, and is used to define recursive types, to which
%we will return shortly.
%
Type constructors can be ordinary types, which are given the kind~$\kwtykind$
(and for which we sometimes use the metavariable~$\tau$).
%when a type constructor is indeed a type, we will sometimes use the
%notation~$\tau$ for it).
%
Types may also be parameterized by vertex structures.
This allows, for example, a type of lists of futures which is parameterized
by the VS providing the vertices for the futures.\footnote{We could use the
  same type parameter mechanism to allow types to be parameterized by other
  types, as in the ML type \texttt{'a list}, but this is orthogonal and we
  do not consider it in the formalism to streamline the presentation.}
The kind $\kwkindarr{\vsty}{\kwtykind}$ classifies
type-level functions that take a VS of type $\vsty$ and return a type
constructor of kind $\kwtykind$.  Rule \rulename{K:Lambda} describes how to
assign the kind $\kwkindarr{\vsty}{\kwtykind}$ to such functions.

As mentioned earlier, the motivation behind~{\langname} is to allow
recursive types containing futures.  We achieve this by
parameterizing recursive types by VSs containing the vertices for these
futures. %%  In principle, we could employ different types of VSs for a given
%% recursive type, but it suffices to consider an infinite version of the recursive
%% type we want to parameterize. For example, a type describing lists of integer
%% futures is parameterized by an infinite stream of vertices: for a list of size
%% $n$, we would use the first $n$ vertices of the stream to name each future on
%% the list.  As a first approximation, we might define this type as follows:
%% \[\kwrec{\convar}{\kwxi{\hastycl{\vvar}{\kw{vstream}}}{(\kwsum{\kwunit}{(\kwprod{\kwfutt{\kw{int}}{\kwfst{\vvar}}}{\kwvapp{\convar}{(\kwsnd{\vvar})}})})}},
%% \]
%% where $\kw{vstream} = \kwvscorec{\vstyvar}{\kwprod{\kwvty}{\vstyvar}}$ and
%% $\convar$ is a recursive type binding.%
%% \aaa{The issue is not clear to me. Is the problem that type constructors are
%%   isorecursive?}  However, we cannot use this definition because there is no way
%% of applying this type constructor to a VS, since the outermost type constructor
%% is a recursive binding instead of a type-level VS function. To solve this
%% problem, we use {\em parameterized} recursive data types
The syntax for a parameterized recursive data type is
$\kwprec{\convar}{\hastycl{\vvar}{\vsty}}{\con}{\vs}$, where
$\hastycl{\vvar}{\vsty}.\con$ is a type level VS function (equivalent to
$\kwxi{\hastycl{\vvar}{\vsty}}{\con}$), $\convar$ is a recursive binding of
$\hastycl{\vvar}{\vsty}.\con$, and $\vs$ is the argument applied to
$\hastycl{\vvar}{\vsty}.\con$ when the recursive type is unrolled. The formation
of parameterized recursive types is performed by \rulename{K:Rec}, in which
$\convar$ has kind $\kwkindarr{\vsty}{\kwtykind}$ within $\con$ since it
represents a type-level function that passes a VS argument to the VS argument of
the recursive instance (seen in more detail below). By applying a sub-VS of $\vs$ to an instance of
$\convar$ in $\con$, the type
$\kwprec{\convar}{\hastycl{\vvar}{\vsty}}{\con}{\vs}$ is able to recur over
$\vs$. (Note that non-parameterized recursive types do not require special
syntax because we can parameterize them in a trivial way by using a dummy VS as
the parameter.) We can now implement the type constructor for a list of integer
futures as
\[\kwxi{\hastycl{\vvar'}{\kw{vstream}}}{\kwprec{\convar}{\hastycl{\vvar}{\kw{vstream}}}{(\kwsum{\kwunit}{(\kwprod{\kwfutt{\kw{int}}{\kwfst{\vvar}}}{\kwvapp{\convar}{(\kwsnd{\vvar})}})})}{\vvar'}}.\]

\begin{figure*}
  \small
  \centering
  \def \MathparLineskip {\lineskip=0.43cm}
  \begin{mathpar}
    \Rule{UE:FstPair}
         {\vseq{\utctx}{\vs}{\kwpair{\vs_1}{\vs_2}}{\vsprod{\vsty_1}{\isav}{\vsty_2}{\avail}}}
         {\vseq{\utctx}{\kwfst{\vs}}{\vs_1}{\vsty_1}}
    \and
    \Rule{UE:SndPair}
         {\vseq{\utctx}{\vs}{\kwpair{\vs_1}{\vs_2}}{\vsprod{\vsty_1}{\avail}{\vsty_2}{\isav}}}
         {\vseq{\utctx}{\kwsnd{\vs}}{\vs_2}{\vsty_2}}
  \end{mathpar}
  \caption{Selected rules for vertex structure equivalence.}
  \label{fig:select-vert-equiv}
\end{figure*}

\begin{figure*}
  \small
  \centering
  \def \MathparLineskip {\lineskip=0.43cm}
  \begin{mathpar}
    \Rule{CE:Fut}
         {\coneq{\gctx}{\utctx}{\vstctx}{\con}{\con'}{\kwtykind}\\
			\vseq{\utctx}{\vs}{\vs'}{\kwvty}}
         {\coneq{\gctx}{\utctx}{\vstctx}{\kwfutt{\con}{\vs}}{\kwfutt{\con'}{\vs'}}{\kwtykind}}
    \and
    \Rule{CE:BetaEq}
         {\iskind{\gctx}{}{\utctx, \hastype{\vvar}{\vsty}}{\vstctx}{\con}{\kwtykind}\\
			\vsistype{\ectx}{\utctx}{\vs}{\vsty}}
         {\coneq{\gctx}{\utctx}{\vstctx}{\kwvapp{(\kwxi{\hastycl{\vvar}{\vsty}}{\con})}{\vs}}
           {\vsub{\con}{\vs}{\vvar}}{\kwtykind}}
  \end{mathpar}
  \caption{Selected rules for type constructor equivalence.}
  \label{fig:select constructor evaluation}
\end{figure*}

Because VSs can occur in types, type checking \langname{} programs may require
performing some type-level computation, notably when projecting vertices out of
a VS.  To address this, we introduce two judgments: an equivalence judgment on
VSs, $\vseq{\utctx}{\vs}{\vs'}{\vsty}$ (see Figure~\ref{fig:select-vert-equiv}), and an equivalence
judgment on type constructors,
$\coneq{\gctx}{\utctx}{\vstctx}{\con_1}{\con_2}{\kind}$
(see Figure~\ref{fig:select constructor evaluation}).  Regarding VSs, the most notable
rules are \rulename{UE:FstPair} and \rulename{UE:SndPair}, which
extract (available) components of a pair.
For type constructors, equivalence has two
purposes: performing a type-level VS function application, which is performed by
\rulename{CE:BetaEq}; and changing VSs within types to equivalent VSs according
to the VS equivalence rules (\rulename{CE:Future} is given as an example).
\ifapp
The full set of VS equivalence and type equivalence rules are given in Figures \ref{fig:vert-equiv} and
\ref{fig:constructor evaluation} respectively in the appendix.
\else
Other rules are relatively straightforward and are deferred to the supplementary
appendix for space reasons.
\fi

\subsection{Graph Type System}

\input{fig-statics}

Figure \ref{fig:statics} presents the type system for \langname{}, which assigns
graph types (and types) to expressions. The judgment
$\tywithdag{\gctx}{\uspctx}{\utctx}{\ctx}{e}{\tau}{\graph}$ states that the
expression $e$ has type $\tau$ and graph type $\graph$. In addition to the 
graph type context $\gctx$ and
VS contexts~$\uspctx$ and~$\utctx$, this judgment uses the
context~$\ctx$ which, as usual, maps expression variables to their types.
%As
%mentioned earlier, assigning a graph type $\graph$ to $e$ means that any
%computation graph resulting from executing $e$ will be in the family of
%computation graphs represented by $\graph$.

We give a quick overview of the graph type system: $\emptygraph$ represents
expressions that execute purely sequentially; $\graph_1 \seqcomp \graph_2$
represents expressions that execute an expression with graph type $\graph_1$
followed by an expression with graph type $\graph_2$;
$\graph_1 \dagor \graph_2$ represents expressions that execute an
expression with graph type $\graph_1$ or an expression with graph type
$\graph_2$; $\dagrec{\gvar}{\graph}{}$ is a recursive graph type;
$\leftcomp{\graph}{\vs}$ represents spawning a future at vertex $\vs$ that
executes an expression with graph type $\graph$ in parallel;
and $\touchcomp{\vs}$ represents touching the future at vertex $\vs$.
A parameterized graph type
$\dagpi{\hastycl{\vvar_f}{\vsty_f}}{\hastycl{\vvar_t}{\vsty_t}}{\graph}$
accepts two VSs as arguments ($\vvar_f$ has VS type
$\vsty_f$ and is added to $\uspctx$ and $\utctx$ while $\vvar_t$ has VS type
$\vsty_t$ and is only added to $\utctx$).
Such graph type functions are applied
with the syntax~$\kwtapp{\graph}{\vs}{\vs}$.
Finally, $\dagnew{\hastycl{\vvar}{\vsty}}{\graph}$ binds a new
VS variable of VS type $\vsty$ and represents expressions that do the same.

Rule \rulename{S:Fun} types function expressions
$\kwfun{\vvar_f}{\vvar_t}{f}{x}{e}$ where $f$ is the name of the function, $x$
is an expression parameter, $\vvar_f$ and $\vvar_t$ are two VS parameters, and
$e$ is the body of the function. Excluding bindings of new VS variables within
$e$, the only VS variable that can be used for spawning futures in $e$ is
$\vvar_f$, while future touches can use any VS variable within the context
(including $\vvar_f$ and $\vvar_t$), hence the function having two VS
parameters. The type of functions is
\[\kwpi{\hastycl{\vvar_f}{\vsty_f}}{\hastycl{\vvar_t}{\vsty_t}}
  {\kwarrow{\tau_1}{\tau_2}
    {\kwtapp{(\dagrec{\gvar}{\dagpi{\hastype{\vvar_f}{\vsty_f}}{\hastype{\vvar_t}{\vsty_t}}
          {\graph}}{})}{\vvar_f}{\vvar_t}}}\] where $\tau_1$ is the type of the
parameter $x$, $\tau_2$ is the type of the function body, $\graph$ is the graph
type of the function body, and
$\kwtapp{(\dagrec{\gvar}{\dagpi{\hastycl{\vvar_f}{\vsty_f}}{\hastycl{\vvar_t}{\vsty_t}}{\graph}}{})}{\vvar_f}{\vvar_t}$
is the graph type representing graphs produced by applying the function. The
function type is parameterized by the VS parameters $\vvar_f$ and $\vvar_t$ (the same
ones from the function expression), which have VS types $\vsty_f$ and $\vsty_t$
respectively. The graph type
$\kwtapp{(\dagrec{\gvar}{\dagpi{\hastycl{\vvar_f}{\vsty_f}}{\hastycl{\vvar_t}{\vsty_t}}{\graph}}{})}{\vvar_f}{\vvar_t}$
contains a recursive binding to a graph type function whose body is $\graph$
(note that this function binds a new, separate $\vvar_f$ and $\vvar_t$ within
$\graph$), and this recursive binding is applied to the $\vvar_f$ and $\vvar_t$
bound within the type. This allows $\graph$ to pass different VS arguments to
recursive instances of itself (this is useful, for example, when recursive
instances access deeper levels of a vertex stream). In addition to adding the
function expression's parameters to the context when typing the function body, we add $f$, the
recursive binding of the function; and $\gvar$, the recursive binding of the
graph type function whose body is $\graph$.

%Compared to prior work~\cite{Muller22}, the most significant additions to the
%type system come from rules \rulename{S:Type-Eq}, \rulename{S:Roll}, and
%\rulename{S:Unroll}.
Rule \rulename{S:Type-Eq} ensures that typing respects type
constructor equivalence.
%: if $\tau_1$ is equivalent to $\tau_2$, then all
%expressions of type $\tau_1$ also have type $\tau_2$.
Rules \rulename{S:Roll} and
\rulename{S:Unroll} roll and unroll parameterized recursive data types:
$\kwprec{\convar}{\hastycl{\vvar}{\vsty}}{\tau}{\vs}$ unrolls to
$\sub{\sub{\tau}{\vs}{\vvar}}{\kwxi{\hastycl{\vvar'}{\vsty}}{\kwprec{\convar}{\hastycl{\vvar}{\vsty}}{\tau}{\vvar'}}}{\convar}$
by applying $\vs$ to $\hastycl{\vvar}{\vsty}.\tau$ (hence
$\sub{\tau}{\vs}{\vvar}$) and then substituting itself recursively. Instead of
replacing the $\convar$ with an instance of the recursive type, we replace it
with a type-level VS function that passes its argument to the recursive type.

%% file: fig-vert-typing.tex
\begin{figure*}
  \small
  \centering
  \def \MathparLineskip {\lineskip=0.43cm}
  \begin{mathpar}
    \Rule{U:OmegaVar}
         {\strut}
         {\vsistype{\uspctx, \hastype{\vvar}{\vsty}}{\utctx}{\vvar}{\vsty}}
    \and
    \Rule{U:PsiVar}
         {\strut}
         {\vsistype{\uspctx}{\utctx, \hastype{\vvar}{\vsty}}{\vvar}{\vsty}}
%%    \and
%%    \Rule{U:OmegaGen}
%%         {\strut}
%%         {\vsistype{\uspctx, \hastype{\genseed}{\vsty}}{\utctx}{\vertgen{\vsty}{\genseed}}{\vsty}}
%%    \and
%%    \Rule{U:PsiGen}
%%         {\strut}
%%         {\vsistype{\uspctx}{\utctx, \hastype{\genseed}{\vsty}}{\vertgen{\vsty}{\genseed}}{\vsty}}
    \and
    \Rule{U:Pair}
         {\uspsplit{\uspctx}{\uspctx_1}{\uspctx_2}\\
		  \vsistype{\uspctx_1}{\utctx}{\vs_1}{\vsty_1}\\
		  \vsistype{\uspctx_2}{\utctx}{\vs_2}{\vsty_2}}
         {\vsistype{\uspctx}{\utctx}
			{\kwpair{\vs_1}{\vs_2}}{\vsprod{\vsty_1}{\isav}{\vsty_2}{\isav}}}
	\and
    \Rule{U:OnlyLeftPair}
         {\vsistype{\uspctx}{\utctx}{\vs_1}{\vsty_1}\\
			\vsistype{\ectx}{\utctx'}{\vs_2}{\vsty_2}}
         {\vsistype{\uspctx}{\utctx}
				{\kwpair{\vs_1}{\vs_2}}{\vsprod{\vsty_1}{\isav}{\vsty_2}{\isunav}}}
    \and
    \Rule{U:OnlyRightPair}
         {\vsistype{\ectx}{\utctx'}{\vs_1}{\vsty_1}\\
			\vsistype{\uspctx}{\utctx}{\vs_2}{\vsty_2}}
         {\vsistype{\uspctx}{\utctx}
				{\kwpair{\vs_1}{\vs_2}}{\vsprod{\vsty_1}{\isunav}{\vsty_2}{\isav}}}
    \and
    \Rule{U:Fst}
         {\vsistype{\uspctx}{\utctx}{\vs}{\vsprod{\vsty_1}{\isav}{\vsty_2}{\avail}}}
         {\vsistype{\uspctx}{\utctx}{\kwfst{\vs}}{\vsty_1}}
    \and
    \Rule{U:Snd}
         {\vsistype{\uspctx}{\utctx}{\vs}{\vsprod{\vsty_1}{\avail}{\vsty_2}{\isav}}}
         {\vsistype{\uspctx}{\utctx}{\kwsnd{\vs}}{\vsty_2}}
    \and
    \Rule{U:Subtype}
         {\vsistype{\uspctx}{\utctx}{\vs}{\vsty'}\\
				\vstysubt{\vsty'}{\vsty}}
         {\vsistype{\uspctx}{\utctx}{\vs}{\vsty}}

  \end{mathpar}
  \caption{Vertex structure type system for {\langname}.}
  \label{fig:vert-typing}
\end{figure*}

%% file: fig-vert-subtyping.tex
\begin{figure*}
  \small
  \centering
  \def \MathparLineskip {\lineskip=0.43cm}
  \begin{mathpar}
	\Rule{UT:ProdLeft}
         {\strut}
         {\vstysubt{\vsprod{\vsty_1}{\avail_1}{\vsty_2}{\avail_2}}
			{\vsprod{\vsty_3}{\isunav}{\vsty_2}{\avail_2}}}
    \and
	\Rule{UT:ProdRight}
         {\strut}
         {\vstysubt{\vsprod{\vsty_1}{\avail_1}{\vsty_2}{\avail_2}}
			{\vsprod{\vsty_1}{\avail_1}{\vsty_3}{\isunav}}}
    \and
	\Rule{UT:Prod}
         {\vstysubt{\vsty_1}{\vsty_1'}\\
			\vstysubt{\vsty_2}{\vsty_2'}}
         {\vstysubt{\vsprod{\vsty_1}{\avail_1}{\vsty_2}{\avail_2}}
			{\vsprod{\vsty_1'}{\avail_1}{\vsty_2'}{\avail_2}}}
%    \and
%	\Rule{UT:Sum}
%         {\vstysubt{\vsty_1}{\vsty_1'}\\
%			\vstysubt{\vsty_2}{\vsty_2'}}
%         {\vstysubt{\kwsum{\vsty_1}{\vsty_2}}
%			{\kwsum{\vsty_1'}{\vsty_2'}}}
%    \and
%	\Rule{UT:Rec1}
%         {\strut}
%         {\vstysubt{\kwvsrec{\vstyvar}{\vsty}}
%			{\vsub{\vsty}{\kwvsrec{\vstyvar}{\vsty}}{\vstyvar}}}
%    \and
%	\Rule{UT:Rec2}
%         {\strut}
%		{\vstysubt{\vsub{\vsty}{\kwvsrec{\vstyvar}{\vsty}}{\vstyvar}}
%			{\kwvsrec{\vstyvar}{\vsty}}}
    \and
	\Rule{UT:Corec1}
         {\strut}
         {\vstysubt{\kwvscorec{\vstyvar}{\vsty}}
			{\vsub{\vsty}{\kwvscorec{\vstyvar}{\vsty}}{\vstyvar}}}
    \and
	\Rule{UT:Corec2}
         {\strut}
		{\vstysubt{\vsub{\vsty}{\kwvscorec{\vstyvar}{\vsty}}{\vstyvar}}
			{\kwvscorec{\vstyvar}{\vsty}}}
    \and
	\Rule{UT:Reflexive}
         {\strut}
         {\vstysubt{\vsty}{\vsty}}
    \and
	\Rule{UT:Transitive}
         {\vstysubt{\vsty}{\vsty''}\\
			\vstysubt{\vsty''}{\vsty'}}
         {\vstysubt{\vsty}{\vsty'}}
  \end{mathpar}
  \caption{VS subtyping.}
  \label{fig:vert-subtyping}
\end{figure*}

%% file: fig-vert-type-split.tex
\begin{figure*}
  \small
  \centering
  \def \MathparLineskip {\lineskip=0.43cm}
  \begin{mathpar}
	\Rule{US:Prod}
         {\strut}
         {\vstysplit{\vsprod{\vsty_1}{\isav}{\vsty_2}{\isav}}
			{\vsprod{\vsty_1}{\isav}{\vsty_2}{\isunav}}
			{\vsprod{\vsty_1}{\isunav}{\vsty_2}{\isav}}}
    \and
    \Rule{US:SplitBoth}
         {\vstysplit{\vsty_1}{\vsty_1'}{\vsty_1''}\\
			\vstysplit{\vsty_2}{\vsty_2'}{\vsty_2''}}
         {\vstysplit{\vsprod{\vsty_1}{\isav}{\vsty_2}{\isav}}
			{\vsprod{\vsty_1'}{\isav}{\vsty_2'}{\isav}}
			{\vsprod{\vsty_1''}{\isav}{\vsty_2''}{\isav}}}
    \and
    \Rule{US:SplitLeft}
         {\vstysplit{\vsty_1}{\vsty_1'}{\vsty_1''}}
         {\vstysplit{\vsprod{\vsty_1}{\isav}{\vsty_2}{\avail}}
			{\vsprod{\vsty_1'}{\isav}{\vsty_2}{\avail}}
			{\vsprod{\vsty_1''}{\isav}{\vsty_2}{\isunav}}}
    \and
    \Rule{US:SplitRight}
         {\vstysplit{\vsty_2}{\vsty_2'}{\vsty_2''}}
         {\vstysplit{\vsprod{\vsty_1}{\avail}{\vsty_2}{\isav}}
			{\vsprod{\vsty_1}{\avail}{\vsty_2'}{\isav}}
			{\vsprod{\vsty_1}{\isunav}{\vsty_2''}{\isav}}}
%    \and
%    \Rule{US:Recursive}
%         {\vstysplit{\vsub{\vsty}{\kwvsrec{\vstyvar}{\vsty}}{\vstyvar}}{\vsty_1}{\vsty_2}}
%         {\vstysplit{\kwvsrec{\vstyvar}{\vsty}}{\vsty_1}{\vsty_2}}
    \and
    \Rule{US:Corecursive}
         {\vstysplit{\vsub{\vsty}{\kwvscorec{\vstyvar}{\vsty}}{\vstyvar}}{\vsty_1}{\vsty_2}}
         {\vstysplit{\kwvscorec{\vstyvar}{\vsty}}{\vsty_1}{\vsty_2}}
    \and
    \Rule{US:Subtype}
         {\vstysplit{\vsty}{\vsty_1'}{\vsty_2}\\
			\vstysubt{\vsty_1'}{\vsty_1}}
         {\vstysplit{\vsty}{\vsty_1}{\vsty_2}}
    \and
    \Rule{US:Commutative}
         {\vstysplit{\vsty}{\vsty_2}{\vsty_1}}
         {\vstysplit{\vsty}{\vsty_1}{\vsty_2}}
  \end{mathpar}
  \caption{Vertex structure type splitting.}
  \label{fig:vert-type-split}
\end{figure*}

%% file: fig-uspctx-split.tex
\begin{figure*}
  \small
  \centering
  \def \MathparLineskip {\lineskip=0.43cm}
  \begin{mathpar}
	\Rule{OM:Empty}
         {\strut}
         {\uspsplit{\ectx}{\ectx}{\ectx}}
    \and
    \Rule{OM:Commutative}
         {\uspsplit{\uspctx}{\uspctx_2}{\uspctx_1}}
         {\uspsplit{\uspctx}{\uspctx_1}{\uspctx_2}}
    \and
    \Rule{OM:Var}
         {\uspsplit{\uspctx}{\uspctx_1}{\uspctx_2}}
         {\uspsplit{\uspctx, \hastype{\vvar}{\vsty}}{\uspctx_1, \hastype{\vvar}{\vsty}}{\uspctx_2}}
%%    \and
%%    \Rule{OM:Gen}
%%         {\uspsplit{\uspctx}{\uspctx_1}{\uspctx_2}}
%%         {\uspsplit{\uspctx, \hastype{\genseed}{\vsty}}{\uspctx_1, \hastype{\genseed}{\vsty}}{\uspctx_2}}
    \and
	\Rule{OM:VarTypeSplit}
         {\uspsplit{\uspctx}{\uspctx_1}{\uspctx_2}\\
			\vstysplit{\vsty}{\vsty_1}{\vsty_2}}
         {\uspsplit{\uspctx, \hastype{\vvar}{\vsty}}
			{\uspctx_1, \hastype{\vvar}{\vsty_1}}
			{\uspctx_2, \hastype{\vvar}{\vsty_2}}}
%%    \and
%%	\Rule{OM:GenTypeSplit}
%%         {\uspsplit{\uspctx}{\uspctx_1}{\uspctx_2}\\
%%			\vstysplit{\vsty}{\vsty_1}{\vsty_2}}
%%         {\uspsplit{\uspctx, \hastype{\genseed}{\vsty}}
%%			{\uspctx_1, \hastype{\genseed}{\vsty_1}}
%%			{\uspctx_2, \hastype{\genseed}{\vsty_2}}}
  \end{mathpar}
  \caption{$\uspctx$ context splitting.}
  \label{fig:uspctx-split}
\end{figure*}

%% file: fig-dag-wf-abbr.tex
\begin{figure*}
  \small
  \centering
  \def \MathparLineskip {\lineskip=0.43cm}
  \begin{mathpar}
    \Rule{DW:Spawn}
         {\uspsplit{\uspctx}{\uspctx_1}{\uspctx_2}\\
		  \dagwf{\gctx}{\uspctx_1}{\utctx}{\graph}{\kgraph}\\
           \vsistype{\uspctx_2}{\ectx}{\vs}{\kwvty}\\
%           \vertex \not\in \vertices
         }
         {\dagwf{\gctx}{\uspctx}{\utctx}
           {\leftcomp{\graph}{\vs}}{\kgraph}}
    \and
    \Rule{DW:Touch}
         {\vsistype{\ectx}{\utctx}{\vs}{\kwvty}}
         {\dagwf{\gctx}{\uspctx}{\utctx}{\touchcomp{\vs}}{\kgraph}}
    \and
    \Rule{DW:New}
         {\dagwf{\gctx}{\uspctx, \hastype{\vvar}{\vsty}}{\utctx, \hastype{\vvar}{\vsty}}
             {\graph}{\kgraph}\\
%           \vsttree{\ectx}{\vsty}
		}
         {\dagwf{\gctx}{\uspctx}{\utctx}{\dagnew{\hastycl{\vvar}{\vsty}}{\graph}}{\kgraph}}
    \and
    %\iffull
    \Rule{DW:RecPi}
         {\dagwf{\gctx,\hastype{\gvar}
             {\dagpi{\hastycl{\vvar_f}{\vsty_f}}{\hastycl{\vvar_t}{\vsty_t}}{\kgraph}}}
           {\hastype{\vvar_f}{\vsty_f}}{\utctx, \hastype{\vvar_f}{\vsty_f}, \hastype{\vvar_t}{\vsty_t}}{\graph}
           {\kgraph}}
         {\dagwf{\gctx}{\uspctx}{\utctx}
           {\dagrec{\gvar}{\dagpi{\hastycl{\vvar_f}{\vsty_f}}{\hastycl{\vvar_t}{\vsty_t}}{\graph}}{}}
           {\dagpi{\hastycl{\vvar_f}{\vsty_f}}{\hastycl{\vvar_t}{\vsty_t}}{\kgraph}}}
    \and
    %\fi
    \Rule{DW:App}
         {\uspsplit{\uspctx}{\uspctx_1}{\uspctx_2}\\
		   \dagwf{\gctx}{\uspctx_1}{\utctx}{\graph}
           {\dagpi{\hastycl{\vvar_f}{\vsty_f}}{\hastycl{\vvar_t}{\vsty_t}}{\kgraph}}\\
             \vsistype{\uspctx_2}{\ectx}{\vs_f}{\vsty_f}\\
		\vsistype{\ectx}{\utctx}{\vs_f}{\vsty_f}\\
             \vsistype{\ectx}{\utctx}{\vs_t}{\vsty_t}}
         {\dagwf{\gctx}{\uspctx}{\utctx}
           {\kwtapp{\graph}{\vs_f}{\vs_t}}
           {\kgraph}
             %\tsub{\tsub{\graphkind}{\verts_f'}{\verts_f}}{\verts_t'}{\verts_t}}
         }
  \end{mathpar}
  \caption{Selected rules for graph type formation.}
  \label{fig:dag-wf-abbr}
\end{figure*}

%% file: fig-kinds.tex
\begin{figure*}
  \small
  \centering
  \def \MathparLineskip {\lineskip=0.43cm}
  \begin{mathpar}
    \Rule{K:Unit}
         {\strut}
         {\iskind{\gctx}{\uspctx}{\utctx}{\vstctx}
           {\kwunit}{\kwtykind}}
    \and
    \Rule{K:Var}
         {\strut}
         {\iskind{\gctx}{\uspctx}{\utctx}{\vstctx, \haskind{\convar}{\kind}}
           {\convar}{\kind}}
    \and
    \Rule{K:Fun}
         {\iskind{\gctx}{\uspctx_1}{\utctx, \hastype{\vvar_f}{\vsty_f}, \hastype{\vvar_t}{\vsty_t}}{\vstctx}
             {\con_1}{\kwtykind}\\
           \iskind{\gctx}{\uspctx_2}{\utctx, \hastype{\vvar_f}{\vsty_f}, \hastype{\vvar_t}{\vsty_t}}{\vstctx}
             {\con_2}{\kwtykind}\\
           \dagwf{\gctx}{\ectx}{\utctx}{\graph}
                 {\dagpi{\hastycl{\vvar_f}{\vsty_f}}{\hastycl{\vvar_t}{\vsty_t}}{\kgraph}}}
         {\iskind{\gctx}{\uspctx}{\utctx}{\vstctx}
           {\kwpi{\hastycl{\vvar_f}{\vsty_f}}{\hastycl{\vvar_t}{\vsty_t}}
             {\kwarrow{\con_1}{\con_2}{\kwtapp{\graph}{\vvar_f}{\vvar_t}}}}{\kwtykind}}
    \and
    \Rule{K:Sum}
         {\iskind{\gctx}{\uspctx}{\utctx}{\vstctx}
             {\con_1}{\kwtykind}\\
           \iskind{\gctx}{\uspctx}{\utctx}{\vstctx}
             {\con_2}{\kwtykind}}
         {\iskind{\gctx}{\uspctx}{\utctx}{\vstctx}
           {\kwsum{\con_1}{\con_2}}{\kwtykind}}
    \and
    \Rule{K:Prod}
         {\iskind{\gctx}{\uspctx_1}{\utctx}{\vstctx}
             {\con_1}{\kwtykind}\\
           \iskind{\gctx}{\uspctx_2}{\utctx}{\vstctx}
             {\con_2}{\kwtykind}}
         {\iskind{\gctx}{\uspctx_1, \uspctx_2}{\utctx}{\vstctx}
             {\kwprod{\con_1}{\con_2}}{\kwtykind}}
    \and
    \Rule{K:Fut}
         {\iskind{\gctx}{\uspctx}{\utctx}{\vstctx}
             {\con}{\kwtykind}\\
           \vsistype{\ectx}{\utctx}{\vs}{\kwvty}}
         {\iskind{\gctx}{\uspctx}{\utctx}{\vstctx}
             {\kwfutt{\con}{\vs}}{\kwtykind}}
    \and
    \Rule{K:Rec}
         {\iskind{\gctx}{\uspctx}{\utctx, \hastype{\vvar}{\vsty}}
             {\vstctx, \haskind{\convar}{\kwkindarr{\vsty}{\kwtykind}}}{\con}{\kwtykind}\\
           \vsistype{\ectx}{\utctx}{\vs}{\vsty}}
         {\iskind{\gctx}{\uspctx}{\utctx}{\vstctx}
             {\kwprec{\convar}{\hastycl{\vvar}{\vsty}}{\con}{\vs}}{\kwtykind}}
    \and
    \Rule{K:Lambda}
         {\iskind{\gctx}{\hastype{\vvar_f}{\vsty_f}}{\utctx, \hastype{\vvar}{\vsty}}{\vstctx}
             {\con}{\kwtykind}}
         {\iskind{\gctx}{\uspctx}{\utctx}{\vstctx}
             {\kwxi{\hastycl{\vvar}{\vsty}}{\con}}
             {\kwkindarr{\vsty}{\kwtykind}}}
    \and
    \Rule{K:App}
         {\iskind{\gctx}{\uspctx_1}{\utctx}{\vstctx}
             {\con}{\kwkindarr{\vsty}{\kwtykind}}\\
           \vsistype{\ectx}{\utctx}{\vs}{\vsty}}
         {\iskind{\gctx}{\uspctx_1, \uspctx_2}{\utctx}{\vstctx}
             {\kwvapp{\con}{\vs}}{\kwtykind}}
%    \and
%    \Rule{K:Let}
%         {\iskind{\gctx}{}{\utctx, \hastype{\vvar_1}{\vsty_1}, \hastype{\vvar_2}{\vsty_2}}{\vstctx}
%             {\con}{\kind}\\
%           \vsistype{\ectx}{\utctx}{\vs}{\kwprod{\vsty_1}{\vsty_2}}}
%         {\iskind{\gctx}{}{\utctx}{\vstctx}
%             {\kwletpair{\vvar_1}{\vvar_2}{\vs}{\con}}{\kind}}
%    \and
%    \Rule{K:Case}
%         {\vsistype{\ectx}{\utctx}{\vs}{\kwsum{\vsty_1}{\vsty_2}}\\
%           \iskind{\gctx}{\uspctx_2, \hastype{\vvar_1}{\vsty_1}}
%             {\utctx, \hastype{\vvar_1}{\vsty_1}}{\vstctx}
%             {\con_1}{\kwtykind}\\
%           \iskind{\gctx}{\uspctx_2, \hastype{\vvar_2}{\vsty_2}}
%             {\utctx, \hastype{\vvar_2}{\vsty_2}}{\vstctx}
%             {\con_2}{\kwtykind}}
%         {\iskind{\gctx}{\uspctx_1, \uspctx_2}{\utctx}{\vstctx}
%             {\kwvcase{\vs}{\vvar_1}{\con_1}{\vvar_2}{\con_2}}{\kwtykind}}
  \end{mathpar}
  \caption{Type constructor kinding rules.}
  \label{fig:kinds}
\end{figure*}

%% file: fig-statics.tex
\begin{figure*}
  \small
  \centering
  \def \MathparLineskip {\lineskip=0.43cm}
  \begin{mathpar}
    \Rule{S:Var}
         {\strut}
         {\tywithdag{\gctx}{\uspctx}{\utctx}{\ctx, \hastype{x}{\tau}}{x}{\tau}
           {\emptygraph}}
    \and
    \Rule{S:Unit}
         {\strut}
         {\tywithdag{\gctx}{\uspctx}{\utctx}{\ctx}
           {\kwtriv}{\kwunit}{\emptygraph}}
    \and
    \Rule{S:Fun}
         {\gctx' = \gctx,\hastype{\gvar}{\dagpi{\hastycl{\vvar_f}{\vsty_f}}{\hastycl{\vvar_t}{\vsty_t}}{\kgraph}}\\
		\ctx' = \ctx, \hastype{f}
	             {\kwpi{\hastycl{\vvar_f}{\vsty_f}}{\hastycl{\vvar_t}{\vsty_t}}
	               {\kwarrow{\tau_1}{\tau_2}{\kwtapp{\gvar}{\vvar_f}{\vvar_t}}}},
	             \hastype{x}{\tau_1}\\
           \gvar\fresh\\
		\utctx' = \utctx, \hastype{\vvar_f}{\vsty_f}, \hastype{\vvar_t}{\vsty_t}\\
           \tywithdag{\gctx'}{\hastype{\vvar_f}{\vsty_f}}{\utctx'}{\ctx'}{e}{\tau_2}{\graph}\\
           \iskind{\gctx}{\ectx}{\utctx'}{\ectx}{\tau_1}{\kwtykind}\\
           \iskind{\gctx}{\ectx}{\utctx'}{\ectx}{\tau_2}{\kwtykind}}
         {\tywithdag{\gctx}{\uspctx}{\utctx}{\ctx}
           {\kwfun{\vvar_f}{\vvar_t}{f}{x}{e}}
           {\kwpi{\hastycl{\vvar_f}{\vsty_f}}{\hastycl{\vvar_t}{\vsty_t}}
             {\kwarrow{\tau_1}{\tau_2}
               {\kwtapp{(\dagrec{\gvar}{\dagpi{\hastype{\vvar_f}{\vsty_f}}{\hastype{\vvar_t}{\vsty_t}}
               {\graph}}{})}{\vvar_f}{\vvar_t}}}}
           {\emptygraph}
         }
    \and
    \Rule{S:App}
         {\uspsplit{\uspctx}{\uspctx_1}{\uspctx'}\\
		   \uspsplit{\uspctx'}{\uspctx_2}{\uspctx_3}\\
		   \tywithdag{\gctx}{\uspctx_1}{\utctx}{\ctx}{e_1}
           {\kwpi{\hastycl{\vvar_f}{\vsty_f}}{\hastycl{\vvar_t}{\vsty_t}}
             {\kwarrow{\tau_1}{\tau_2}{\kwtapp{\graph_3}{\vvar_f}{\vvar_t}}}}
           {\graph_1}\\
           \tywithdag{\gctx}{\uspctx_2}{\utctx}{\ctx}{e_2}
                     {\tsub{\tsub{\tau_1}{\vs_f}{\vvar_f}}{\vs_t}{\vvar_t}}
                     {\graph_2}\\
           \vsistype{\uspctx_3}{\ectx}{\vs_f}{\vsty_f}\\
           \vsistype{\ectx}{\utctx}{\vs_f}{\vsty_f}\\
           \vsistype{\ectx}{\utctx}{\vs_t}{\vsty_t}}
         {\tywithdag{\gctx}{\uspctx}{\utctx}{\ctx}
           {\kwapp{}{\kwtapp{e_1}{\vs_f}{\vs_t}}{e_2}}
           {\tsub{\tsub{\tau_2}{\vs_f}{\vvar_f}}{\vs_t}{\vvar_t}}
           {\graph_1 \seqcomp \graph_2 \seqcomp
             \kwtapp{\graph_3}{\vs_f}{\vs_t}}
         }
    \and
    \Rule{S:Pair}
         {\uspsplit{\uspctx}{\uspctx_1}{\uspctx_2}\\
		   \tywithdag{\gctx}{\uspctx_1}{\utctx}{\ctx}{e_1}{\tau_1}{\graph_1}\\
           \tywithdag{\gctx}{\uspctx_2}{\utctx}{\ctx}{e_2}{\tau_2}{\graph_2}}
         {\tywithdag{\gctx}{\uspctx}{\utctx}{\ctx}
           {\kwpair{e_1}{e_2}}{\kwprod{\tau_1}{\tau_2}}
           {\graph_1 \seqcomp \graph_2}}
    \and
    \Rule{S:Fst}
         {\tywithdag{\gctx}{\uspctx}{\utctx}{\ctx}{e}
           {\kwprod{\tau_1}{\tau_2}}{\graph}}
         {\tywithdag{\gctx}{\uspctx}{\utctx}{\ctx}{\kwfst{e}}
           {\tau_1}{\graph}}
%%    \and
%%    \Rule{S:Par}
%%         {\uspsplit{\uspctx}{\uspctx_1}{\uspctx_2}\\
%%		   \tywithdag{\gctx}{\uspctx_1}{\utctx}{\ctx}{e_1}{\tau_1}{\graph_1}\\
%%           \tywithdag{\gctx}{\uspctx_2}{\utctx}{\ctx}{e_2}{\tau_2}{\graph_2}}
%%         {\tywithdag{\gctx}{\uspctx}{\utctx}{\ctx}{\kwpar{e_1}{e_2}}
%%           {\kwprod{\tau_1}{\tau_2}}
%%           {\graph_1 \parcomp \graph_2}}
         \and
    \Rule{S:Snd}
         {\tywithdag{\gctx}{\uspctx}{\utctx}{\ctx}{e}
           {\kwprod{\tau_1}{\tau_2}}{\graph}}
         {\tywithdag{\gctx}{\uspctx}{\utctx}{\ctx}{\kwsnd{e}}
           {\tau_2}{\graph}}
         \and
    \Rule{S:InL}
         {\tywithdag{\gctx}{\uspctx}{\utctx}{\ctx}{e}{\tau_1}{\graph}\\\\
           \iskind{\gctx}{}{\utctx}{\ectx}{\tau_2}{\kwtykind}
         }
         {\tywithdag{\gctx}{\uspctx}{\utctx}{\ctx}{\kwinl{e}}
           {\kwsum{\tau_1}{\tau_2}}{\graph}}
         \and
    \Rule{S:InR}
         {\tywithdag{\gctx}{\uspctx}{\utctx}{\ctx}{e}{\tau_2}{\graph}\\\\
           \iskind{\gctx}{}{\utctx}{\ectx}{\tau_1}{\kwtykind}
         }
         {\tywithdag{\gctx}{\uspctx}{\utctx}{\ctx}{\kwinr{e}}
           {\kwsum{\tau_1}{\tau_2}}{\graph}}
         \and
    \Rule{S:Case}
         {\uspsplit{\uspctx}{\uspctx_1}{\uspctx_2}\\
		   \tywithdag{\gctx}{\uspctx_1}{\utctx}{\ctx}{e_1}
           {\kwsum{\tau_1}{\tau_2}}{\graph_1}\\
           \tywithdag{\gctx}{\uspctx_2}{\utctx}{\ctx,\hastype{x}{\tau_1}}{e_2}
                     {\tau'}{\graph_2}\\
           \tywithdag{\gctx}{\uspctx_2}{\utctx}{\ctx,\hastype{y}{\tau_2}}{e_3}
                     {\tau'}{\graph_3}}
         {\tywithdag{\gctx}{\uspctx}{\utctx}{\ctx}
           {\kwcase{e_1}{x}{e_2}{y}{e_3}}{\tau'}
           {\graph_1 \seqcomp (\graph_2 \dagor \graph_3)}}
         \ifextrarules
  \end{mathpar}
  \caption{Graph type system for {\langname} (Part 1 of 2).}
  \label{fig:statics}
\end{figure*}
\begin{figure*}
  \small
  \centering
  \def \MathparLineskip {\lineskip=0.43cm}
  \begin{mathpar}
    \else
    \and
    \fi
   % \and
    \Rule{S:Roll}
         {\tywithdag{\gctx}{\uspctx}{\utctx}{\ctx}{e}
           {\sub{\sub{\tau}{\vs}{\vvar}}
             {\kwxi{\hastycl{\vvar'}{\vsty}}{\kwprec{\convar}{\hastycl{\vvar}{\vsty}}{\tau}{\vvar'}}}{\convar}}{\graph}\\
           \iskind{\gctx}{}{\utctx}{\ectx}{\kwprec{\convar}{\hastycl{\vvar}{\vsty}}{\tau}{\vs}}{\kwtykind}}
         {\tywithdag{\gctx}{\uspctx}{\utctx}{\ctx}{\kwroll{e}}
           {\kwprec{\convar}{\hastycl{\vvar}{\vsty}}{\tau}{\vs}}{\graph}}
    \and
    \Rule{S:Unroll}
         {\tywithdag{\gctx}{\uspctx}{\utctx}{\ctx}{e}
           {\kwprec{\convar}{\hastycl{\vvar}{\vsty}}{\tau}{\vs}}{\graph}}
         {\tywithdag{\gctx}{\uspctx}{\utctx}{\ctx}{\kwunroll{e}}
           {\sub{\sub{\tau}{\vs}{\vvar}}
             {\kwxi{\hastycl{\vvar'}{\vsty}}{\kwprec{\convar}{\hastycl{\vvar}{\vsty}}{\tau}{\vvar'}}}{\convar}}{\graph}}
    \and
    \Rule{S:Future}
         {\uspsplit{\uspctx}{\uspctx_1}{\uspctx_2}\\\\
		   \tywithdag{\gctx}{\uspctx_1}{\utctx}{\ctx}{e}{\tau}{\graph}\\
           \vsistype{\uspctx_2}{\ectx}{\vs}{\kwvty}\\
           \vsistype{\ectx}{\utctx}{\vs}{\kwvty}}
         {\tywithdag{\gctx}{\uspctx}{\utctx}{\ctx}
           {\kwfuture{\vs}{e}}
           {\kwfutt{\tau}{\vs}}
           {\leftcomp{\graph}{\vs}}}
    \and
%%    \Rule{S:Handle}
%%         {\tywithdag{\gctx}{\ectx}{\utctx}{\ctx}{v}{\tau}{\emptygraph}\\
%%           \vsistype{\ectx}{\utctx}{\vspath}{\kwvty}}
%%         {\tywithdag{\gctx}{\uspctx}{\utctx}{\ctx}{\kwhandle{\vspath}{v}}
%%           {\kwfutt{\tau}{\vspath}}
%%           {\emptygraph}}
%%    \and
    \Rule{S:Touch}
         {\tywithdag{\gctx}{\uspctx}{\utctx}{\ctx}{e}
           {\kwfutt{\tau}{\vs}}{\graph}}
         {\tywithdag{\gctx}{\uspctx}{\utctx}{\ctx}
           {\kwforce{e}}{\tau}
           {\graph \seqcomp \touchcomp{\vs}}}
    \and
    \Rule{S:New}
         {\tywithdag{\gctx}{\uspctx, \hastype{\vvar}{\vsty}}
           {\utctx, \hastype{\vvar}{\vsty}}{\ctx}{e}
           {\tau}{\graph}\\
           \iskind{\gctx}{\uspctx}{\utctx}{\ectx}{\tau}{\kwtykind}\\
%           \vsttree{\ectx}{\vsty}
         }
         {\tywithdag{\gctx}{\uspctx}{\utctx}{\ctx}
           {\kwnewf{\vvar}{\vsty}{e}}
           {\tau}
           {\dagnew{\hastycl{\vvar}{\vsty}}{\graph}}}
    \and
    \Rule{S:Type-Eq}
         {\coneq{\gctx}{\utctx}{\ectx}{\tau_1}{\tau_2}{\kwtykind}\\
            \tywithdag{\gctx}{\uspctx}{\utctx}{\ctx}{e}
           {\tau_1}{\graph}}
         {\tywithdag{\gctx}{\uspctx}{\utctx}{\ctx}{e}
           {\tau_2}{\graph}}
  \end{mathpar}
  \caption{\ifextrarules Graph type system for {\langname} (Part 2 of 2).
    \else Graph type system for {\langname}.\fi}
  \label{fig:statics}
\end{figure*}

%% file: soundness.tex
\section{Soundness}\label{sec:soundness}

The goal of this section is to prove the soundness of the graph type system
for {\langname}; that is, that the computation graph of a program is described
by its graph type.
In order to prove this theorem, we must first formalize 1) the notion of a computation graph being ``described by'' a graph type
and 2) the operational
semantics by which a program evaluates to produce a computation graph.
The first notion is one of {\em normalization}~\citep{Muller22}, a process
for constructing the set of computation graphs corresponding to a given
graph type.
The second is a {\em cost semantics}, which we present as a big-step semantics
that evaluates an expression to a value and a computation graph.

The rest of this section is structured as follows.
In Section~\ref{sec:soundness-generators}, we discuss how we represent the
creation of new vertices (which occurs during both normalization, as~$\kw{new}$
bindings are normalized, and during evaluation, as they are evaluated).
We then formalize normalization (Section~\ref{sec:soundness-norm}), and
finally present the cost semantics and prove soundness
(Section~\ref{sec:soundness-cost}).

\subsection{Generation of Runtime Vertex Names}\label{sec:soundness-generators}

\begin{figure*}
  \small
  \centering
  \def \MathparLineskip {\lineskip=0.43cm}
  \begin{mathpar}
    \Rule{UV:Var}
         {\strut}
         {\vseval{\vvar}{\vvar}}
    \and
    \Rule{UV:Path}
         {\strut}
         {\vseval{\vertgen{\vsty}{\genseed}}{\vertgen{\vsty}{\genseed}}}
    \and
    \Rule{UV:Pair}
         {\vseval{\vs_1}{\vs_1'}\\
           \vseval{\vs_2}{\vs_2'}}
         {\vseval{\kwpair{\vs_1}{\vs_2}}{\kwpair{\vs_1'}{\vs_2'}}}
    \and
    \Rule{UV:FstPair}
         {\vseval{\vs}{\kwpair{\vs_1}{\vs_2}}}
         {\vseval{\kwfst{\vs}}{\vs_1}}
    %\and
    %\Rule{UV:SndPair}
    %     {\vseval{\vs}{\kwpair{\vs_1}{\vs_2}}}
    %     {\vseval{\kwsnd{\vs}}{\vs_2}}
    \and
    \Rule{UV:FstNotPair}
         {\vseval{\vs}{\vs'}\\
		\vs' \neq \kwpair{\vs_1}{\vs_2}}
         {\vseval{\kwfst{\vs}}{\kwfst{\vs'}}}
    %\and
    %\Rule{UV:SndNotPair}
    %     {\vseval{\vs}{\vs'}\\
%		\vs' \neq \kwpair{\vs_1}{\vs_2}}
%         {\vseval{\kwsnd{\vs}}{\kwsnd{\vs'}}}
  \end{mathpar}
  \caption{Selected rules for vertex structure normalization.}
  \label{fig:vs-eval-abbr}
\end{figure*}

Recall that the constructs~$\dagnew{\hastycl{\vvar}{\vsty}}{\graph}$ in graph
types and~$\kwnewf{\vvar}{\vsty}{e}$ in expressions
bind ``fresh'' vertex structure (VS) variables to be used in the graph type
and the
expression, respectively.
Because VSs can be infinite and of arbitrary type,
some care must be taken in how to represent them at
``runtime'', i.e., in normalization
and the cost semantics.
The key insight is that VSs
are (possibly infinite) trees with unique vertices at each leaf.
It is thus possible to uniquely identify a vertex by the VS it
comes from, and the path taken to reach it from the root of the VS.
Paths in a vertex structure~$\vs$ are already represented in our syntax
as sequences of projections, e.g.,~$\kwfst{\kwsnd{\vs}}$, so most of the
new conceptual work is in representing the roots of the vertex structures.

We use the syntax~$\genseed$ and variants to represent a unique vertex
name, called a {\em generator},
which serves as the root of a VS~$\vertgen{\vsty}{\genseed}$.
Generators are included in the contexts~$\uspctx$ and~$\utctx$ like VS
variables, but are meaningful runtime symbols representing unique
VSs.
We will use the notation~$\guspctx$ and~$\gutctx$ to refer to contexts that
contain only generators, and no variables.
These are the only contexts that will exist at runtime for typing top-level
terms and graph types, as such terms and graph types contain no free variables.
We refer to these terms and graph types as {\em closed}
even though they may contain free vertex names in
the form of generators.
The judgments for VS typing and $\uspctx$ context splitting are extended with
rules for generators that resemble the rules for variables.

Equipped with a way to represent the roots of new vertex structures, we turn
our attention again to paths from the root to a vertex.
Currently, the same path can be represented in multiple ways,
for example~$\vseq{\utctx}{\kwfst{\kwpair{\kwfst{\vs}}{\kwsnd{\vs}}}}{\kwfst{\vs}}{\vsty}$.
It will be useful to have a normal form for paths,
so we introduce a normalization operation on VSs with the
%which we introduce by
%defining a normalization operation on vertex structures.
%
judgment~$\vseval{\vs}{\vs}$, defined in
Figure~\ref{fig:vs-eval-abbr} (rules symmetric to these are omitted).
Intuitively, the operation~$\beta$-reduces any projections of pairs (but
leaves alone projections of vertex structures that are not syntactically
pairs, e.g.~$\kwfst{\vertgen{\vsty}{\genseed}}$).

%% Furthermore, the next lemma shows that normalization is deterministic up to
%% equivalence of vertex structures.
%% %
%% As a corollary (because if~$\vseval{\vs}{\vs'}$, then~$\vseval{\vs'}{\vs'}$),

%% \begin{lemma}\label{lem:vs-eval-transitive}
%% 	If~$\vseval{\vs}{\vs''}$
%% 		and $\vseq{\utctx}{\vs}{\vs'}{\vsty}$,
%%   	then~$\vseval{\vs'}{\vs''}$.
%% \end{lemma}

\begin{figure}
  \begin{mathpar}
    \begin{array}{l r l l}
      \mathit{VSs} & \vs & \bnfdef &
      \dots \bnfalt
      \vspath
    \\

    \mathit{Vertex~Path} & \vspath & \bnfdef &
    \vertgen{\vsty}{\genseed} \bnfalt
    \kwfst{\vspath} \bnfalt
    \kwsnd{\vspath}
	\\

    \mathit{Expressions} & e & \bnfdef &
      \dots \bnfalt
    v
    \\

    \mathit{Values} & v & \bnfdef &
    \kwtriv \bnfalt
    \kwfun{\vvar}{\vvar}{f}{x}{e} \bnfalt
    \kwpair{v}{v} \bnfalt
    \kwinl{v} \bnfalt
    \kwinr{v} \bnfalt
    \kwroll{v} \bnfalt
    \kwhandle{\vspath}{v}
    \end{array}
  \end{mathpar}
  \caption{Extended syntax for generators, vertex paths, and values}
  \label{fig:ext-syntax}
\end{figure}

\begin{figure}
  \centering
  \def \MathparLineskip {\lineskip=0.43cm}
  \begin{mathpar}
    \Rule{S:Handle}
         {\tywithdag{\gctx}{\ectx}{\utctx}{\ctx}{v}{\tau}{\emptygraph}\\
           \vsistype{\ectx}{\utctx}{\vspath}{\kwvty}}
         {\tywithdag{\gctx}{\uspctx}{\utctx}{\ctx}{\kwhandle{\vspath}{v}}
           {\kwfutt{\tau}{\vspath}}
           {\emptygraph}}
  \end{mathpar}
  \caption{Rule for typing handles.}
  \label{fig:handle-statics}
\end{figure}

%\begin{lemma}\label{lem:vs-paths-unique}\strut
%\begin{enumerate}
%\item If~$\vseval{\vs'}{\kwfst{\vs}}$ or
%	$\vseval{\vs'}{\kwsnd{\vs}}$
%	then there exists a $\vs''$
%	such that $\vseval{\vs''}{\vs}$.
%	\label{lem:vs-proj-vals-contain-vals}
%\item If~$\vseq{\utctx}{\vs}{\kwpair{\vs_1}{\vs_2}}{\vsty}$
%	then $\vs \neq \vspath$.
%	\label{lem:no-vs-path-pairs}
%\item If~$\vseval{\vs_0}{\vs}$
%		and either $\vseq{\utctx}{\vs}{\vspath}{\vsty}$
%		or $\vseq{\utctx}{\vspath}{\vs}{\vsty}$,
%  	then~$\vs = \vspath$.
%	\label{lem:vs-paths-unique-main}
%\end{enumerate}
%\end{lemma}

We use the term {\em vertex paths}, and the notation~$\vspath$,
to refer to closed, normal VSs.
Figure~\ref{fig:ext-syntax} extends the syntax for VSs with generators and
gives the syntax for vertex paths.
We now have a way of producing references to unique vertices: generators give
rise to unique, non-intersecting VSs, and unique vertex paths in
a given VS refer to unique vertices.
When evaluating a~$\kw{new}$ binding, we will simply create a fresh generator.
The remainder of the vertices in the VS are then created
implicitly, and will be accessed by the program as it traverses the VS.

\subsection{Normalization}\label{sec:soundness-norm}

\input{fig-unroll-graph-abbr}
\input{fig-bnnf}
\input{fig-mk-graphs}

Figures~\ref{fig:unroll-graph-abbr}--\ref{fig:mk-graphs} present the machinery for
normalization.
Because a recursive graph type represents an infinite set of graphs (as it
can be unrolled any number of times), we stage the construction of these sets
so that every set constructed is finite.
Constructing a set of graphs consists of three operations, each of which
performs some of the required tasks.
First, recursive graph types are unrolled a desired number of times,
yielding another graph type that is equivalent up to unrollings of
recursive~$\mu$ bindings.
Next, the graph type is reduced to ``New-$\beta$ normal form'' (NBNF),
which ``evaluates'' any exposed ``new'' bindings by generating and substituting
fresh vertex structures.
This process also performs any applicable~$\beta$ reductions on exposed
applications.
At this point, we are left with a valid graph type, but one with no exposed
``new'' bindings or applications.
Finally, the resulting graph type is {\em expanded} into the set of graphs;
because there are no exposed ``new'' bindings, this process does not involve
generating any new vertex names or structures.
We will now discuss each of these steps in more detail.

Figure~\ref{fig:unroll-graph-abbr} gives a small-step semantics for unrolling
recursive bindings in graph types.
The bulk of the work is done by rule \rulename{UR:Rec}, which steps
a binding~$\dagrec{\gvar}{\graph}{}$ to~$\graph[\dagrec{\gvar}{\graph}{}/\gvar]$.
\ifapp
The remaining rules ``search'' the graph type for recursive bindings, so
we defer the full set of rules to Figure \ref{fig:unroll-graph} in the appendix.
\else
The remaining rules ``search'' the graph type for recursive bindings, so
we defer most of these rules to the supplementary appendix.
\fi
Note that, unlike in a standard left-to-right (or right-to-left) operational
semantics, the rules \rulename{UR:Seq1} and \rulename{UR:Seq2} allow any
instance of recursion in the type to be unrolled at any time in a
nondeterministic fashion.
For example, both steps below are valid:
\[
\begin{array}{c}
(\dagrec{\gvar}{\gvar \seqcomp \gvar}{}) \seqcomp
(\dagrec{\gvar}{\gvar \dagor \gvar}{})
\unrstep
((\dagrec{\gvar}{\gvar \seqcomp \gvar}{})
\seqcomp
(\dagrec{\gvar}{\gvar \seqcomp \gvar}{})) \seqcomp
(\dagrec{\gvar}{\gvar \dagor \gvar}{})\\
(\dagrec{\gvar}{\gvar \seqcomp \gvar}{}) \seqcomp
(\dagrec{\gvar}{\gvar \dagor \gvar}{})
\unrstep
(\dagrec{\gvar}{\gvar \seqcomp \gvar}{}) \seqcomp
((\dagrec{\gvar}{\gvar \dagor \gvar}{}) \seqcomp
(\dagrec{\gvar}{\gvar \dagor \gvar}{}))
\end{array}
\]

The rules for reducing to NBNF, given in Figure~\ref{fig:bnnf},
eliminate ``new'' bindings by substituting fresh vertex structure generators,
and perform any exposed applications.
Evaluation proceeds recursively through sequential compositions
and alternatives, but not under binders.
As a result, closed sub-graph-types of NBNF graph types are themselves NBNF.

Finally, Figure~\ref{fig:mk-graphs} gives the rules for expanding a graph
type into the set of graphs it represents.
We represent a graph~$\cgraph$
formally as a 4-tuple~$(\vertices, \edges, \startv, \sinkv)$
containing the sets of vertices~$\vertices$ and edges~$\edges$, as well as
a designated ``start'' vertex~$\startv$ and ``end'' vertex~$\sinkv$.
We use shorthands for combining graphs sequentially and in parallel; these
shorthands use many of the same operators as graph type composition, but
should not be confused.
Figure~\ref{fig:dag} gives formal definitions for these shorthands; for more
description of them, the reader is referred to prior work~\citep{Muller22}.
In brief, sequential composition~$\seqcomp$ joins the end vertex of the first
graph to the start vertex of the second graph.
The ``left composition'' operator~\citep{spoonhower09},
written~$\leftcomp{\cgraph}{\vvar}$, adds a subgraph~$\cgraph$
corresponding to a future to the graph, with an edge representing the spawn.
It also adds the vertex~$\vvar$ as the sink of the future's graph.
The ``touch'' operator~$\touchcomp{\vvar}$ adds an edge from~$\vvar$.

\input{fig-dag}

Sequential compositions are expanded by expanding both
subgraphs, and then sequentially composing the resulting graphs.
Alternation simply takes the union of the two sets of graphs.
Expansion does not perform any additional unrolling, so the expansion of
a recursive graph type is the empty set of graphs.
Expansion of the future~$\leftcomp{\graph}{\vs}$
is performed by left-composing all of the resulting graphs with the
vertex~$\vs$, and touches~$\touchcomp{\vs}$ simply expand to the singleton
graph consisting of the touch.
%
%Note that it is not necessary to expand ``new'' bindings
%because the graph type is in NBNF.
%
Note that because NBNF has already expanded all ``new'' bindings,
there is no rule for expanding these, and
expansion does not generate new vertex names.
The latter is a key property in guaranteeing that expansion results in
well-formed graphs.

These three operations combine to form a normalization process that is correct:
any set of graphs that results from unrolling, normalizing, and expanding a
well-formed graph type is well-formed. This result is formalized by Theorem \ref{thm:norm-correct},
\ifapp
which is proven in Appendix \ref{app:soundness-proofs} alongside several necessary technical lemmas in Appendices \ref{app:lang-proofs} and \ref{app:soundness-proofs}.
\else
which is proven in the attached supplementary appendix alongside several necessary technical lemmas.
\fi

\begin{theorem}\label{thm:norm-correct}
  If~$\dagwf{\ectx}{\uspctx}{\utctx}{\graph}{\graphkind}$
  and~$\graph \unrstep^* \graph'$,
  then~$\bnnf{\graph'}$ exists and if~$\cgraph \in \getgraphs{\bnnf{\graph'}}$,
  then~$\cgraph$ is a well-formed graph.
\end{theorem}

\subsection{Cost Semantics and Soundness}\label{sec:soundness-cost}

We equip {\langname} with a cost semantics,
a big-step operational semantics that evaluates an expression and also produces
the computation graph that represents the execution.
The judgment is~$\eval{}{e}{v}{\cgraph}$, meaning that expression~$e$
evaluates to value~$v$, producing the cost graph~$\cgraph$.
The rules for this judgment are in Figure~\ref{fig:cost}, 
and the syntax for values are in Figure~\ref{fig:ext-syntax}.
In~\rulename{C:Future}, the body of the future is evaluated (in a real
execution, the body of the future will be evaluated in parallel, but the
big-step cost semantics deliberately abstracts away evaluation order)
and the future evaluates to a handle, a new syntactic form which records the
result of the future.
In addition, we evaluate the vertex structure~$\vs$ used to spawn the future
to a vertex path~$\vspath$, which is recorded by the handle.
The~\rulename{C:Touch} rule extracts both the vertex path and future result
from the handle.
In~\rulename{C:New}, a new generator~$\genseed$ is created and used to
generate a vertex structure which instantiates the variable~$\vvar$.

\input{fig-cost}

The soundness theorem for the graph type system of {\langname} is that
if a program~$e$ has a graph type~$\graph$ and evaluates to produce
a graph~$\cgraph$, then~$\cgraph$ is described by~$\graph$ (that is,~$\cgraph$
should be in the set of graphs obtained by
normalizing~$\graph$ using the machinery in the previous subsection).
This is stated formally as Theorem~\ref{thm:soundness}.
The formal statement of the theorem also includes a context~$\gutctx$ containing
generators created during execution which may be captured in the result
value~$v$.

\begin{theorem}\label{thm:soundness}
  If~$\tywithdag{\ectx}{\uspctx}{\utctx}{\ectx}{e}{\tau}{\graph}$
  and~$\eval{}{e}{v}{\cgraph}$,
  then there exists a~$\gutctx$
  such that $\genseed \in \gutctx$ implies $\genseed \fresh$
  and $\tywithdag{\ectx}{\ectx}{\utctx, \gutctx}{\ectx}{v}{\tau}{\emptygraph}$
  and there exists a~$\graph'$ such that~$\graph \unrstep^* \graph'$
  and~$\cgraph \in \getgraphs{\bnnf{\graph'}}$.
\end{theorem}

\ifapp
The proof of the theorem can be found in Appendix \ref{app:soundness-proofs}
alongside several necessary technical lemmas in Appendices \ref{app:lang-proofs} and \ref{app:soundness-proofs}.
\else
The proof of the theorem, as well as statements and proofs of several
necessary technical lemmas, appears in the full version of the paper~\citep{full}.
These lemmas include:
\begin{itemize}
\item A number of substitution results for expressions, types,
  vertex structures, etc.
\item If $\tywithdag{\gctx}{\uspctx}{\utctx}{\ctx}{e}{\tau}{\graph}$
  then~$\iskind{\gctx}{}{\utctx}{\ectx}{\tau}{\kwtykind}$
  and~$\dagwf{\gctx}{\uspctx}{\utctx}{\graph}{\kgraph}$
  (assuming~$\ctx$ contains only well-kinded types).
\item A ``framing'' property that allows us to repeatedly unroll a sub-graph
  type while keeping the rest of the graph type the same,
  for example, if~$\graph_1 \unrstep^* \graph_1'$
    then~$\graph_1 \seqcomp \graph_2 \unrstep^* \graph_1' \seqcomp \graph_2$.
\end{itemize}
\fi

%% file: fig-unroll-graph-abbr.tex
\begin{figure*}
  \centering
  \def \MathparLineskip {\lineskip=0.43cm}
  \begin{mathpar}
    \Rule{UR:Seq1}
         {\graph_1 \unrstep \graph_1'}
         {\graph_1 \seqcomp \graph_2
           \unrstep
           \graph_1' \seqcomp \graph_2}
    \and
    \Rule{UR:Seq2}
         {\graph_2 \unrstep \graph_2'}
         {\graph_1 \seqcomp \graph_2
           \unrstep
           \graph_1 \seqcomp \graph_2'}
    \and
    \Rule{UR:Rec}
         {\strut}
         {\dagrec{\gvar}{\graph}{} \unrstep
           \gsub{\graph}{\dagrec{\gvar}{\graph}{}}{\gvar}
         }
  \end{mathpar}
  \caption{Selected rules for graph type unrolling.}
  \label{fig:unroll-graph-abbr}
\end{figure*}

%% file: fig-bnnf.tex
\begin{figure}
  \small
  \[
  \begin{array}{r l l l}
    \bnnf{\emptygraph} & \defeq & \emptygraph &\\
    \bnnf{\graph_1 \seqcomp \graph_2} & \defeq &
    \bnnf{\graph_1} \seqcomp \bnnf{\graph_2}\\
%%    \bnnf{\graph_1 \parcomp \graph_2} & \defeq &
%%    \bnnf{\graph_1} \parcomp \bnnf{\graph_2}\\
    \bnnf{\graph_1 \dagor \graph_2} & \defeq &
    \bnnf{\graph_1} \dagor \bnnf{\graph_2} &\\
    \bnnf{\dagrec{\gvar}{\graph}{\uctx}} & \defeq &
    \dagrec{\gvar}{\graph}{\uctx}\\
    \bnnf{\leftcomp{\graph}{\vs}} & \defeq &
    \leftcomp{\bnnf{\graph}}{\vs'}
    & \vseval{\vs}{\vs'}\\%, \vs' \vsnormal\\
    \bnnf{\touchcomp{\vs}} & \defeq &
    \touchcomp{\vs'}
    & \vseval{\vs}{\vs'}\\%, \vs' \vsnormal\\
    
    \bnnf{\dagnew{\hastycl{\vvar}{\vsty}}{\graph}} & \defeq &
    \bnnf{\gsub{\graph}{\vertgen{\vsty}{\genseed}}{\vvar}} & \genseed \fresh\\
    \bnnf{\dagpi{\hastype{\vvar_f}{\vsty_f}}{\hastype{\vvar_t}{\vsty_t}}{\graph}} & \defeq &
    \dagpi{\hastype{\vvar_f}{\vsty_f}}{\hastype{\vvar_t}{\vsty_t}}{\graph}\\
    \bnnf{\kwtapp{\graph}{\vs_f}{\vs_t}} & \defeq &
    \bnnf{\gsub{\graph'}{\vs_f, \vs_t}{\vvar_f, \vvar_t}}
    & \bnnf{\graph} = \dagpi{\hastycl{\vvar_f}{\vsty_f}}{\hastycl{\vvar_t}{\vsty_t}}{\graph'}\\
    \bnnf{\kwtapp{\graph}{\vs_f}{\vs_t}} & \defeq &
    \kwtapp{\graph}{\vs_f}{\vs_t} & \text{o.w.}\\
    %% \bnnf{\kwletpair{\vvar_1}{\vvar_2}{\vs}{\graph}} & \defeq &
    %% \bnnf{\gsub{\graph}{\vs_1, \vs_2}{\vvar_1, \vvar_2}} &
    %% \vseval{\vs}{\kwpair{\vs_1}{\vs_2}}\\
    %% \bnnf{\kwletpair{\vvar_1}{\vvar_2}{\vs}{\graph}} & \defeq &
    %% \bnnf{\gsub{\graph}{\vgenl{\vspath},
    %%   \vgenr{\vspath}}{\vvar_1, \vvar_2}} &
    %% \vseval{\vs}{\vspath}\\
%    \bnnf{\kwvcase{\vs}{\vvar_1}{\graph_1}{\vvar_2}{\graph_2}} & \defeq &
%    \bnnf{\gsub{\graph_1}{\vs'}{\vvar_1}} &
%    \vseval{\vs}{\kwinl{\vs'}}\\
%    \bnnf{\kwvcase{\vs}{\vvar_1}{\graph_1}{\vvar_2}{\graph_2}} & \defeq &
%    \bnnf{\gsub{\graph_2}{\vs'}{\vvar_2}} &
%    \vseval{\vs}{\kwinr{\vs'}}\\
  \end{array}
  \]
  \caption{New-$\beta$ normal form.}
  \label{fig:bnnf}
\end{figure}

%% file: fig-mk-graphs.tex
\begin{figure}
  \[
  \begin{array}{r l l}
    \getgraphs{\emptygraph} & \defeq & \{\emptygraph\}\\
  \getgraphs{\graph_1 \seqcomp \graph_2} & \defeq &
  \{\cgraph_1' \seqcomp \cgraph_2' \mid
  \cgraph_1' \in \getgraphs{\graph_1}, \cgraph_2' \in \getgraphs{\graph_2}\}\\
%%  \getgraphs{\graph_1 \parcomp \graph_2} & \defeq &
%%  \{\cgraph_1' \parcomp \cgraph_2' \mid
%%  \cgraph_1' \in \getgraphs{\graph_1}, \cgraph_2' \in \getgraphs{\graph_2} &\\
  \getgraphs{\graph_1 \dagor \graph_2} & \defeq &
  \getgraphs{\graph_1} \cup \getgraphs{\graph_2}\\
  \getgraphs{\dagrec{\gvar}{\graph}{\uctx}} & \defeq & \emptyset\\
  \getgraphs{\kwtapp{\graph}{\vs_f}{\vs_t}} & \defeq & \emptyset\\
  \getgraphs{\leftcomp{\graph}{\vs}} & \defeq &
  \{\leftcomp{\cgraph'}{\vs} \mid \cgraph' \in \getgraphs{\graph}\}\\
  \getgraphs{\touchcomp{\vs}} & \defeq &
  \{\touchcomp{\vs}\}\\
  %\getgraphs{\kwtapp{\graph}{\vs_f}{\vs_t}} & \defeq &
  %\emptyset
  \end{array}
  \]
  \caption{Graph type expansion.}
  \label{fig:mk-graphs}
\end{figure}

%% file: fig-dag.tex
\begin{figure*}
  \small
    \centering
\begin{mathpar}
  \begin{array}{l l l l}
    \emptygraph & \defeq & \dagq{\{\vertex\}}{\emptyset}{\vertex}{\vertex} &
    \vertex \fresh\\
    \dagq{\vertices_1}{\edges_1}{\startv_1}{\sinkv_1} \seqcomp
  \dagq{\vertices_2}{\edges_2}{\startv_2}{\sinkv_2} & \defeq &
  \dagq{\vertices_1 \cup \vertices_2}
       {\edges_1 \cup \edges_2 \cup \{(\sinkv_1, \startv_2)\}}
       {\startv_1}{\sinkv_2} & \vertices_1 \cap \vertices_2 = \emptyset\\
%%  \dagq{\vertices_1}{\edges_1}{\startv_1}{\sinkv_1} \parcomp
%%  \dagq{\vertices_2}{\edges_2}{\startv_2}{\sinkv_2} & \defeq &
%%  (\vertices_1 \cup \vertices_2 \cup \{\vertex_1, \vertex_2\},
%%    & \vertex_1, \vertex_2 \fresh,\\
%%    & & ~~\edges_1 \cup \edges_2 \cup
%%         \{(\vertex_1, \startv_1), (\vertex_1, \startv_2),
%%         (\sinkv_1, \vertex_2), (\sinkv_2, \vertex_2)\}, \vertex_1, \vertex_2)
%%         & \vertices_1 \cap \vertices_2 = \emptyset
%%       \\
  \leftcomp{\dagq{\vertices}{\edges}{\startv}{\sinkv}}{\vertex} & \defeq &
  \dagq{\vertices \cup \{\vertex, \vertex'\}}
       {\edges \cup \{(\vertex', \startv), (\sinkv, \vertex)\}}
       {\vertex'}{\vertex'} & \vertex' \fresh, \vertex \not\in \vertices\\
  \touchcomp{\vertex} & \defeq &
  \dagq{\{\vertex'\}}{\{(\vertex, \vertex')\}}{\vertex'}{\vertex'} &
  \vertex'\fresh
       
\end{array}
\end{mathpar}
\caption{Shorthands for combining graphs.}
\label{fig:dag}
\end{figure*}

%% file: fig-cost.tex
\begin{figure*}
  \small
  \centering
  \def \MathparLineskip {\lineskip=0.43cm}
  \begin{mathpar}
    \Rule{C:Value}
         {\strut}
         {\eval{\ppts}{v}{v}{\emptygraph}}
    \and
    %\Rule{C:Unit}
    %     {\strut}
    %     {\eval{\ppts}{\kwtriv}{\kwtriv}{\emptygraph}}
    %\and
    %\Rule{C:Fun}
    %     {\strut}
    %     {\eval{\ppts}{\kwfun{\vvar_f}{\vvar_t}{f}{x}{e}}
    %       {\kwfun{\vvar_f}{\vvar_t}{f}{x}{e}}{\emptygraph}}
    %\and
    \Rule{C:App}
         {\eval{\ppts}{e_1}{\kwfun{\vvar_f}{\vvar_t}{f}{x}{e}}{\cgraph_1}\\
           \eval{\ppts}{e_2}{v'}{\cgraph_2}\\\\
           \eval{\ppts,\ppt}
                {\sub{\tsub{\tsub{\sub{e}{\kwfun{\vvar_f}{\vvar_t}{f}{x}{e}}{f}}{\vs_f}{\vvar_f}}{\vs_t}{\vvar_t}}{v'}{x}}{v}{\cgraph_3}}
         {\eval{\ppts}{\kwapp{\ppt}{\kwtapp{e_1}{\vs_f}{\vs_t}}{e_2}}{v}
           {\cgraph_1 \seqcomp \cgraph_2 \seqcomp \cgraph_3}}
    \and
    \Rule{C:Pair}
         {\eval{\ppts}{e_1}{v_1}{\cgraph_1}\\
           \eval{\ppts}{e_2}{v_2}{\cgraph_2}}
         {\eval{\ppts}{\kwpair{e_1}{e_2}}{\kwpair{v_1}{v_2}}
           {\cgraph_1 \seqcomp \cgraph_2}}
    \and
    \Rule{C:Fst}
         {\eval{\ppts}{e}{\kwpair{v_1}{v_2}}{\cgraph}}
         {\eval{\ppts}{\kwfst{e}}{v_1}{\cgraph}}
         \and
         \ifextrarules
    \Rule{C:Snd}
         {\eval{\ppts}{e}{\kwpair{v_1}{v_2}}{\cgraph}}
         {\eval{\ppts}{\kwsnd{e}}{v_2}{\cgraph}}
         \and
         \fi
    %% \Rule{C:Par}
    %%      {\eval{\ppts}{e_1}{v_1}{\cgraph_1}\\
    %%        \eval{\ppts}{e_2}{v_2}{\cgraph_2}}
    %%      {\eval{\ppts}{\kwpar{e_1}{e_2}}{\kwpair{v_1}{v_2}}
    %%        {\cgraph_1 \parcomp \cgraph_2}}
    %% \and
    \Rule{C:InL}
         {\eval{\ppts}{e}{v}{\cgraph}}
         {\eval{\ppts}{\kwinl{e}}{\kwinl{v}}{\cgraph}}
    \and
    \ifextrarules
    \Rule{C:InR}
         {\eval{\ppts}{e}{v}{\cgraph}}
         {\eval{\ppts}{\kwinr{e}}{\kwinr{v}}{\cgraph}}
    \and
    \fi
    \Rule{C:CaseL}
         {\eval{\ppts}{e}{\kwinl{v}}{\cgraph_1}\\
           \eval{\ppts}{\sub{e_1}{v}{x}}{v'}{\cgraph_2}}
         {\eval{\ppts}{\kwcase{e}{x}{e_1}{y}{e_2}}{v'}
           {\cgraph_1 \seqcomp \cgraph_2}}
         \and
         \ifextrarules
    \Rule{C:CaseR}
         {\eval{\ppts}{e}{\kwinr{v}}{\cgraph_1}\\
           \eval{\ppts}{\sub{e_2}{v}{y}}{v'}{\cgraph_2}}
         {\eval{\ppts}{\kwcase{e}{x}{e_1}{y}{e_2}}{v'}
           {\cgraph_1 \seqcomp \cgraph_2}}
         \and
         \fi
    \Rule{C:Roll}
         {\eval{\ppts}{e}{v}{\cgraph}}
         {\eval{\ppts}{\kwroll{e}}{\kwroll{v}}{\cgraph}}
    \and
    \Rule{C:Unroll}
         {\eval{\ppts}{e}{\kwroll{v}}{\cgraph}}
         {\eval{\ppts}{\kwunroll{e}}{v}{\cgraph}}
    \and
    \Rule{C:Future}
         {\eval{\ppts}{e}{v}{\cgraph}\\
           \vseval{\vs}{\vspath}}
         {\eval{\ppts}{\kwfuture{\vs}{e}}{\kwhandle{\vspath}{v}}
           {\leftcomp{\cgraph}{\vspath}}}
    \and
    %% \Rule{C:Handle}
    %%      {\strut}
    %%      {\eval{\ppts}{\kwhandle{\vspath}{v}}{\kwhandle{\vspath}{v}}{\emptygraph}}
    %% \and
    \Rule{C:Touch}
         {\eval{\ppts}{e}{\kwhandle{\vspath}{v}}{\cgraph}}
         {\eval{\ppts}{\kwforce{e}}{v}{\cgraph \seqcomp \touchcomp{\vspath}}}
    \and
    \Rule{C:New}
         {\genseed \fresh\\\\\eval{\ppts}{\sub{e}{\vertgen{\vsty}{\genseed}}{\vvar}}{v}{\cgraph}}
         {\eval{\ppts}{\kwnewf{\vvar}{\vsty}{e}}
           {v}{\cgraph}}
%    \and
%    \Rule{C:LetVertsPair}
%         {\eval{\ppts}{\sub{e}{\vs_1, \vs_2}{\vvar_1, \vvar_2}}{v}{\cgraph}\\
%           \vseval{\vs}{\kwpair{\vs_1}{\vs_2}}}
%         {\eval{\ppts}{\kwletpair{\vvar_1}{\vvar_2}{\vs}{e}}
%           {v}{\cgraph}}
%    \and
%    \Rule{C:LetVertsGen}
%         {\eval{\ppts}{\sub{e}{\vgenl{\vspath}, \vgenr{\vspath}}
%             {\vvar_1, \vvar_2}}{v}{\cgraph}\\
%           \vseval{\vs}{\vspath}}
%         {\eval{\ppts}{\kwletpair{\vvar_1}{\vvar_2}{\vs}{e}}
%           {v}{\cgraph}}
%    \and
%    \Rule{C:CaseVertsL}
%         {\eval{\ppts}{\sub{e_1}{\vs'}{\vvar_1}}{v}{\cgraph}\\
%           \vseval{\vs}{\kwinl{\vs'}}}
%         {\eval{\ppts}{\kwcase{\vs}{\vvar_1}{e_1}{\vvar_2}{e_2}}
%           {v}{\cgraph}}
%    \and
%    \Rule{C:CaseVertsR}
%         {\eval{\ppts}{\sub{e_2}{\vs'}{\vvar_2}}{v}{\cgraph}\\
%           \vseval{\vs}{\kwinr{\vs'}}}
%         {\eval{\ppts}{\kwcase{\vs}{\vvar_1}{e_1}{\vvar_2}{e_2}}
%           {v}{\cgraph}}

  \end{mathpar}
  \caption{Cost Semantics for {\langname} (selected).
    Rules symmetric to these are omitted.}
  \label{fig:cost}
\end{figure*}

%% file: infer.tex
\section{Elaboration of Recursive Types with Vertex Structures}\label{sec:infer}

Thus far, we have presented the annotated language~{\langname} containing
recursive data types~$\kwprec{\convar}{\hastycl{\vvar}{\vsty}}{\con}{\vspath}$,
annotated with a vertex path~$\vspath$ of type~$\vsty$ that provides
vertex names for data structures of the recursive type.
We have motivated that~$\vspath$ should have a structure that in some sense
``maps on'' to the recursive structure of the list so that any futures in
the structure have a corresponding vertex name.
As examples, a list data type corresponds to an infinite stream of vertices,
and a binary tree data type corresponds to an infinite binary tree of vertices.
%and a 2-3 tree (where each node may have~2 or~3 children) corresponds to an
%infinite ternary tree of vertices.
%
As discussed, the annotated language is provided merely as a core calculus
for expressing the ideas of the graph type system; the annotations can be
inferred from unannotated code by our implementation.

Other than the addition of vertex structures, the general structure of the
algorithm for inferring these annotations is similar to that of
GML~\citep{Muller22}, and the details of the algorithm are largely outside the
scope of this paper.
However, one important and non-obvious fact for inferring annotations
for~$\langname$ is that it is indeed possible to annotate any recursive
data structure with a corresponding vertex path.
Showing this fact is the goal of this section.
We do so by defining a set of rules for annotating unannotated types and
values with vertex structure annotations.
For simplicity, the system we present in this section is declarative and not
algorithmic, so it still abstracts away many of the complexities of our
inference algorithm, but we show that the rules are complete  and thus that any
recursive type may be so annotated.

We first define a syntax for unannotated types~$\undecty$ and
unannotated values~$\undece$.
\[
\begin{array}{r l l}
  \undecty & \bnfdef &
  \convar \bnfalt
  \kwunit \bnfalt
  \kwpi{\hastycl{\vvar_t}{\vsty_t}}
                    {\hastycl{\vvar_f}{\vsty_f}}
                    {\kwarrow{\con_1}{\con_2}{}} \bnfalt
    \kwprod{\undecty}{\undecty} \bnfalt
    \kwsum{\undecty}{\undecty} \bnfalt
    \undecfutty{\undecty} \bnfalt

    \mu \convar. \undecty\\

    \undece & \bnfdef &
    \kwtriv \bnfalt
    \kwfun{\vvar}{\vvar}{f}{x}{e} \bnfalt
    \kwpair{\undece}{\undece} \bnfalt
    \kwinl{\undece} \bnfalt
    \kwinr{\undece} \bnfalt
    \kwroll{\undece} \bnfalt
    \undechandle{\undece}
\end{array}
\]
Unannotated types consist of the unit type, functions, products, and sums,
as well as an unannotated future type and an unannotated recursive type.
Note that the annotation of functions is orthogonal to the annotation of
recursive data types; we assume that function types and function values have
already been annotated and include annotated function types and annotated
function values as unannotated types and unannotated values, respectively.
The unannotated future type~$\undecfutty{\undecty}$ is similar to the annotated
future type~$\kwfutt{\con}{\vs}$ but is not annotated with a VS.
Similarly, the unannotated recursive type~$\kwrec{\convar}{\undecty}$
binds a type variable but does not bind a VS variable and
does not take a VS as an argument.
Because unannotated types do not interact with vertex structures, there is
no type-level lambda and all unannotated types have kind~$\kwtykind$ (and so
we do not distinguish between ``unannotated types'' and ``unannotated type
constructors'').
Unannotated values differ from values~$v$ only in that future handles are
not annotated with vertex paths.

Figure~\ref{fig:annot-vs-types} defines the
judgment~$\futurify{\vstctx}{\vs}{\undecty}{\con}{\vsty}$.
This indicates that~$\undecty$ may be annotated to the type constructor~$\con$
(which will, by construction, have kind~$\kwtykind$).
It also returns a vertex structure type~$\vsty$ that ``corresponds'' to the
type~$\undecty$.
For recursive types~$\kwrec{\convar}{\undecty}$, the VS type~$\vsty$ is the
type of the VS annotation for the recursive type (that is,
$\kwrec{\convar}{\undecty}$ will be annotated to be
$\kwprec{\convar}{\hastycl{\vvar}{\vsty}}{\con}{\vs}$
for some~$\con$ and some~$\vs$).
As an example, if~$\undecty$ is the type of int future
lists,
then~$\vsty$ will be (equivalent to) the type of vertex
streams,~$\kwvscorec{\vstyvar}{\vsprod{\kwvty}{\isav}{\vstyvar}{\isav}}$.
The judgment takes a type variable context~$\vstctx$ mapping type variables
to kinds (these will be annotated types and so their kinds will not be~$\kwtykind$).
It is also parameterized by a vertex structure~$\vs$ to use for annotations.
When annotating a closed unannotated type~$\undecty$, this parameter will
simply be instantiated with a fresh vertex path~$\vertgen{}{\genseed}$ to
derive~$\futurify{\ectx}{\vertgen{}{\genseed}}{\undecty}{\tau}{\vsty}$.
The returned type~$\tau$ would be annotated with projections of~$\vertgen{}{\genseed}$.
The returned VS type~$\vsty$ would be the type that~$\genseed$
should be assigned in order for~$\tau$ to be well-kinded.

Rule~\rulename{F:TyVar} looks up the type variable~$\convar$ in the context.
By construction, its kind will be of the form~$\kwkindarr{\vsty}{\kwtykind}$,
indicating that to properly annotate the use of the variable~$\convar$,
it must be applied to a VS of type~$\vsty$.
We use the VS~$\vs$ for the annotation and return the
type~$\vsty$ as the required type of~$\vs$.
The unit and function types do not require additional annotations, and so
are simply returned.\footnote{The returned VS type is~$\kwvty$,
  which will result in the
  addition of unnecessary vertices to the final VS; it would
  be straightforward to add a multiplicative unit to VS types,
  which would be the most appropriate VS type to return here, but
  we have not done so this far to keep the VS type language as
  simple as possible.}
Rule~\rulename{F:Prod} takes a VS~$\vs$ and annotates
the first component~$\undecty_1$ with the left projection of~$\vs$
and the second component~$\undecty_2$ with the right projection.
The type required for~$\vs$ is thus the product of the two returned types.
Rule~\rulename{F:Sum}, somewhat counterintuitively, also returns a product
of the two VS types.
This is because if a data structure can take one of two forms, the
corresponding VS must offer either
possibility.\footnote{As an optimization, we could take a ``maximum'' over the
  two VS types. For example, in a 2-3 tree, where each node
  may have two or three children, the corresponding VS could
  always offer three branches and a 2-node would use the first two.}
Rule~\rulename{F:Fut} takes~$\vs$ to be a product whose second component
is a single vertex, which it uses to annotate the
future; the first component is used to annotate the future's return type.
Finally, rule~\rulename{F:Rec} annotates a recursive
type~$\kwrec{\convar}{\undecty}$.
It begins by adding~$\convar$ to the context with
kind~$\kwkindarr{\kwvscorec{\vstyvar}{\vsty}}{\kwtykind}$
(this is the only truly non-algorithmic feature of these rules;
we do not discuss how to construct~$\vsty$).
With this context, it annotates~$\undecty$.
The resulting VS type is rolled back into the corecursive
type~$\kwvscorec{\vstyvar}{\vsty}$, which is the type required for~$\vs$.

\paragraph{Example.}
We can represent the type of a list of integer futures as an unannotated
type~$\undecty$:
\[\undecty = \kwrec{\convar}{\kwsum{\kwunit}
  {(\kwprod{\undecfutty{\kw{int}}}{\convar})}}
\]
Using the rules of Figure~\ref{fig:annot-vs-types},
we can infer the following annotation for~$\undecty$:
\[
\futurify{\ectx}{\vs}{\undecty}
         {\kwprec{\convar}{\hastycl{\vvar}{\vsty}}
           {\kwsum{\kwunit}
             {(\kwprod{\kwfutt{\kw{int}}{\kwfst{\kwsnd{\vvar}}}}
               {\kwvapp{\convar}{(\kwsnd{\kwsnd{\vvar}})}})}}
           {\vs}}
         {\vsty}
\]
where~$\vsty = \kwvscorec{\vstyvar}{\vsprod{\kwvty}{\isav}
  {(\vsprod{\kwvty}{\isav}{\vstyvar}{\isav})}{\isav}}$.

The VS corresponding to~$\undecty$ is a stream of vertices
(note that because we treat VS types equi-corecursively, the VS type above is
equivalent to~$\kwvscorec{\vstyvar}{\vsprod{\kwvty}{\isav}{\vstyvar}{\isav}}$
but unrolled slightly).
In the body of the recursive annotated type, which is
$\kwsum{\kwunit}
{(\kwprod{\kwfutt{\kw{int}}{\kwfst{\kwsnd{\vvar}}}}
  {\kwvapp{\convar}{(\kwsnd{\kwsnd{\vvar}})}})}$,
the first vertex of the stream is discarded (this is an effect of mapping the
type~$\kwunit$ to the VS type~$\kwvty$ even though it does not need a vertex),
the second vertex of the stream (the first vertex of the tail) is used for the
future and the remainder (the tail of the tail) is passed to the recursive
instance of the type.

%% As another example, we would represent the type \texttt{int pipe} from
%% Section~\ref{sec:overview} as
%% \[\undecty = \kwrec{\convar}{\kwprod{\kw{int}}{\undecfutty{\convar}}}
%% \]
%% We can infer the following annotation:
%% \[
%% \futurify{\ectx}{\vs}{\undecty}
%%          {\kwprec{\convar}{\hastycl{\vvar}{\vsty}}
%%            {\kwprod{\kw{int}}{\kwfutt{(\kwvapp{\convar}{(\kwfst{\kwsnd{\vs}})})}
%%                {\kwsnd{\kwsnd{\vs}}}}}
%%            {\vs}}
%%          {\vsty}
%% \]
%% where~$\vsty = \kwvscorec{\vstyvar}{\vsprod{\kwvty}{\isav}
%%   {(\vsprod{\vstyvar}{\isav}{\kwvty}{\isav})}{\isav}}$
%% which is also isomorphic to a stream of vertices.

\begin{figure*}
  \small
  \centering
  \def \MathparLineskip {\lineskip=0.43cm}

\begin{mathpar}
  \Rule{F:TyVar}
       {\strut}
       {\futurify{\vstctx,\haskind{\convar}{\kwkindarr{\vsty}{\kwtykind}}}
         {\vs}
         {\convar}{\kwvapp{\convar}{\vs}}
         {\vsty}
       }
  \and
  \Rule{F:Unit}
       {\strut}
       {\futurify{\vstctx}{\vs}{\kwunit}{\kwunit}{\kwvty}}
  \and
  \Rule{F:Fun}
       {
         \strut
       }
       {
         \futurify{\vstctx}{\vs}
                  {\kwpi{\hastycl{\vvar_f}{\vsty_f}}{\hastycl{\vvar_t}{\vsty_t}}
                    {\kwarrow{\con_1}{\con_2}{\graph}}
                  }
                  {\kwpi{\hastycl{\hastycl{\vvar_f}{\vsty_f}}{\vvar_t}{\vsty_t}}
                    {\kwarrow{\con_1}{\con_2}{\graph}}}
                  {\kwvty}
       }
  \and
  \Rule{F:Prod}
       {
         \futurify{\vstctx}{\kwfst{\vs}}{\undecty_1}{\con_1}{\vsty_1}\\
         \futurify{\vstctx}{\kwsnd{\vs}}{\undecty_2}{\con_2}{\vsty_2}
       }
       {
         \futurify{\vstctx}{\vs}{\kwprod{\undecty_1}{\undecty_2}}
                  {\kwprod{\con_1}{\con_2}}
                  {\vsprod{\vsty_1}{\isav}{\vsty_2}{\isav}}
       }
  \and
  \Rule{F:Sum}
       {
         \futurify{\vstctx}{\kwfst{\vs}}{\undecty_1}{\con_1}{\vsty_1}\\
         \futurify{\vstctx}{\kwsnd{\vs}}{\undecty_2}{\con_2}{\vsty_2}
       }
       {
         \futurify{\vstctx}{\vs}{\kwsum{\undecty_1}{\undecty_2}}
                  {\kwsum{\con_1}{\con_2}}
                  {\vsprod{\vsty_1}{\isav}{\vsty_2}{\isav}}
       }
  \and
  \Rule{F:Fut}
       {\futurify{\vstctx}{\kwfst{\vs}}{\undecty}{\con}{\vsty}
       }
       {\futurify{\vstctx}{\vs}{\undecfutty{\undecty}}
         {\kwfutt{\con}{\kwsnd{\vs}}}
         {\vsprod{\vsty}{\isav}{\kwvty}{\isav}}
       }
  \and
  \Rule{F:Rec}
       {
         \futurify{\vstctx,
           \haskind{\convar}{\kwkindarr{\kwvscorec{\vstyvar}{\vsty}}{\kwtykind}}}
                  {\vvar}
                  {\undecty}{\con}
                  {\sub{\vsty}{\kwvscorec{\vstyvar}{\vsty}}{\vstyvar}}
       }
       {
         \futurify{\vstctx}{\vs}
                  {\kwrec{\convar}{\undecty}}
                  {\kwprec{\convar}
                    {\hastype{\vvar}{\kwvscorec{\vstyvar}{\vsty}}}
                    {\con}{\vs}}
                  {\kwvscorec{\vstyvar}{\vsty}}
       }
\end{mathpar}
  \caption{Annotating types with vertex structures.}
  \label{fig:annot-vs-types}
\end{figure*}

The judgment described above declaratively shows a correspondence between
unannotated types and the vertex structure types required to annotate them.
Later in this section, we show that this relation is complete with respect to
well-kinded unannotated types, and thus that any type has a corresponding
VS type.
We next wish to show that a VS of the returned VS type actually
does suffice to provide all necessary vertices for a data structure of the
given type.
We do this using another judgment that annotates unannotated values.
This judgment is defined in Figure~\ref{fig:annot-vs-values}
and takes the form~$\futurifye{\vspath}{\undece}{v}$,
where~$\vspath$ is a vertex path to use for annotation (similar to the type
annotation judgment above),~$\undece$ is an unannotated value, and~$v$
is the annotated value.
We restrict annotations of values to vertex paths since the value $\kwhandle{\vspath}{v}$ may only use vertex paths $\vspath$ as the handle.
Otherwise, annotation of values proceeds in much the same way as annotation of types.
Rule~\rulename{FE:Pair} uses the two components of~$\vspath$
to annotate the components of the pair.
Rules~\rulename{FE:InL} and~\rulename{FE:InR} use the first
and second components, respectively, of~$\vspath$ to annotate left and right
injections (recall that, for a sum type,~$\vspath$ is given a product
type so that the two components of~$\vspath$ may be used for the two
injections).
Finally, just as~\rulename{F:Fut} uses the first component of~$\vs$
to annotate the type of the future's payload and the second component as the
vertex for the future, rule~\rulename{FE:Handle} uses~$\kwfst{\vspath}$ to
annotate the payload and~$\kwsnd{\vspath}$ to annotate the handle itself.

\begin{figure*}
  \small
  \centering
  \def \MathparLineskip {\lineskip=0.43cm}
  \begin{mathpar}
    \Rule{FE:Unit}
         {\strut}
         {\futurifye{\vspath}{\kwtriv}{\kwtriv}}
    \and
    \Rule{FE:Fun}
         {\strut}
         {\futurifye{\vspath}{\kwfun{\vvar_f}{\vvar_t}{f}{x}{e}}
           {\kwfun{\vvar_f}{\vvar_t}{f}{x}{e}}
         }
    \and
    \Rule{FE:Pair}
         {\futurifye{\kwfst{\vspath}}{\undece_1}{v_1}\\
           \futurifye{\kwsnd{\vspath}}{\undece_2}{v_2}
         }
         {
           \futurifye{\vspath}{\kwpair{\undece_1}{\undece_2}}
                     {\kwpair{v_1}{v_2}}
         }
    \and
    \Rule{FE:InL}
         {
           \futurifye{\kwfst{\vspath}}{\undece}{v}
         }
         {
           \futurifye{\vspath}{\kwinl{\undece}}{\kwinl{v}}
         }
    \and
    \Rule{FE:InR}
         {
           \futurifye{\kwsnd{\vspath}}{\undece}{v}
         }
         {
           \futurifye{\vspath}{\kwinr{\undece}}{\kwinr{v}}
         }
    \and
    \Rule{FE:Roll}
         {
           \futurifye{\vspath}{\undece}{v}
         }
         {
           \futurifye{\vspath}{\kwroll{\undece}}{\kwroll{v}}
         }
    \and
    \Rule{FE:Handle}
         {
           \futurifye{\kwfst{\vspath}}{\undece}{v}
         }
         {
           \futurifye{\vspath}{\undechandle{\undece}}
                     {\kwhandle{\kwsnd{\vspath}}{v}}
         }
  \end{mathpar}
  \caption{Annotating values with vertex structures.}
  \label{fig:annot-vs-values}
\end{figure*}

\paragraph{Example.}
Consider the list of integer futures from above.
We claimed that the correct VS type for this type
is~$\vsty = \kwvscorec{\vstyvar}{\vsprod{\kwvty}{\isav}{\vstyvar}{\isav}}$.
The rules of Figure~\ref{fig:annot-vs-values} provide a ``recipe'' for
constructing a future list using a vertex path of VS type~$\vsty$.
As an example, consider the unannotated value
\[\kwroll{\kwinr{\kwpair{\undechandle{1}}
      {\kwroll{\kwinr{\kwpair{\undechandle{2}}{
            \kwroll{\kwinl{\kwtriv}}}}}}}}
\]
which represents the list containing two future handles, one returning~1 and
the other returning~2.
Applying the rules, we get the expression
\[\kwroll{\kwinr{\kwpair{\kwhandle{\kwfst{\kwsnd{\vspath}}}{1}}{
      \kwroll{\kwinr{\kwpair{\kwhandle{\kwfst{\kwsnd{\kwsnd{\kwsnd{\vspath}}}}}{2}}{
            \kwroll{\kwinl{\kwtriv}}}}}}}}
\]
As described above, the futures take
consecutive odd vertices from the stream~$\vspath$.

The main result of this section has three components.
First, any well-kinded unannotated type may be matched with an annotated type
by the rules of Figure~\ref{fig:annot-vs-types}.
Second, if a well-kinded unannotated type is annotated with a VS of the VS type
returned by the annotation judgment, then the annotated type is also well-kinded.
Third, if an unannotated type~$\undecty$ is annotated with a vertex path
(that is, if~$\futurify{\ectx}{\vspath}{\undecty}{\tau}{\vsty}$),
then any well-typed unannotated value of type~$\undecty$
may be annotated with $\vspath$ by the rules of Figure~\ref{fig:annot-vs-values},
and the annotated value is well-typed when $\vspath$ has type $\vsty$.
Moreover, to show that $\vspath$ has ``enough'' vertices to fully
annotate the value with unique vertices, we show that the annotated value
is well-typed under a new typing judgment that uses only an affine context
for vertices.
Usually, values would be typed with the unrestricted
context~$\utctx$, because a data structure is allowed to contain multiple
handles to the same future, but in this case, we wish to show that we
{\em can} restrict data structures to use new vertices for each handle.
We write the new judgment~$\affinetyped{\uspctx}{v}{\tau}$.
The rules are similar to the standard
typing rules, but use the affine context~$\uspctx$ for typing handles.
This rule for typing handle values
is given in Figure \ref{fig:affine-handle-statics}.
Values always have the graph type~$\emptygraph$, so we omit the graph
type from the judgment.
\ifapp
The full set of affine typing rules for values is given in Figure \ref{fig:affine-statics} in the appendix.
\else
The remaining rules are straightforward and are deferred to the supplementary
appendix for space reasons.
\fi

\begin{figure}
\begin{minipage}{0.45\textwidth}
  \centering
  \def \MathparLineskip {\lineskip=0.43cm}
  \begin{mathpar}
    \Rule{SV:Handle}
         {\uspsplit{\uspctx}{\uspctx_1}{\uspctx_2}\\\\
           \affinetyped{\uspctx_1}{v}{\tau}\\
           \vsistype{\uspctx_2}{\ectx}{\vspath}{\kwvty}}
         {\affinetyped{\uspctx}{\kwhandle{\vspath}{v}}
           {\kwfutt{\tau}{\vspath}}}
  \end{mathpar}
  \caption{Affine typing rule for handle values.}
  \label{fig:affine-handle-statics}
\end{minipage}%
\begin{minipage}{0.55\textwidth}
  \[
  \begin{array}{r l l l}
	\unannvstctx{\ectx} & \defeq & \ectx &\\
    \unannvstctx{\vstctx, \haskind{\convar}{\kind}} & \defeq
		& \unannvstctx{\vstctx}, \haskind{\convar}{\kwtykind} &\\
	\strut\\
    \annvstctx{\ectx} & \defeq & \ectx &\\
    \annvstctx{\vstctx, \haskind{\convar}{\kwtykind}} & \defeq
		& \annvstctx{\vstctx}, \haskind{\convar}{\kwkindarr{\vstyvar}{\kwtykind}} & \vstyvar \fresh
  \end{array}
  \]
  \caption{Annotating and unannotating kinds in $\vstctx$.}
  \label{fig:ann-vstctx}
\end{minipage}
\end{figure}

%\begin{figure}

%\end{figure}

Theorem~\ref{thm:annot-complete} formalizes the main result of this section,
that is, that 1) type annotation is complete with respect to well-kinded
unannotated types, 2) type annotation annotates well-kinded unannotated types
into well-kinded types, and 3) annotating values with vertex structures of the
returned VS type results in well-typed values.
In order to show this, we introduce a kinding judgment for unannotated types,
$\unannkind{\vstctx}{\undecty}{\kwtykind}$,
and a typing judgment for unannotated values, $\unannvaltype{\undece}{\undecty}$.
\ifapp
The rules for these judgments are similar to those for annotated types
and expressions and can be found in Figures \ref{fig:unannotated-kinding} and 
\ref{fig:unannotated-statics} of the appendix, respectively.
\else
The rules for these judgments are similar to those for annotated types
and expressions, so we defer them to the supplementary appendix.
\fi
%standard, so we have allocated them to the appendix.
%\fr{Is it fine to reference the appendix like this?}}.
%
Since the kinds of type variables
bound by unannotated and annotated recursive types are different
($\kwtykind$ and $\kwkindarr{\vsty}{\kwtykind}$ respectively),
we need some way to change the kinds that these type variables are bound to.
We address this with the functions $\mathsf{Unann}$ and $\mathsf{Ann}$. 
$\unannvstctx{\vstctx}$ takes a context $\vstctx$ suitable for annotating types 
and kinding annotated types 
(where type variables can, and will always,
have kind $\kwkindarr{\vsty}{\kwtykind}$),
and return an {\em unannotated context},
one suitable for kinding unannotated types
(where every type variable has kind $\kwtykind$).
$\annvstctx{\vstctx}$ performs this process in reverse, where the VS type
expected by every type variable in $\vstctx$ is a fresh VS type variable 
unique to that type variable (each which can be substituted with the desired VS type).

\ifapp
The proof of Theorem \ref{thm:annot-complete}, as well as statements and proofs of several
necessary technical lemmas, appears in Appendix \ref{app:infer-proofs}.
\else
The proof of Theorem \ref{thm:annot-complete}, as well as statements and proofs of several
necessary technical lemmas, appears in the attached supplementary appendix.
\fi

\begin{theorem}\label{thm:annot-complete}\strut
  \begin{enumerate}
  \item
    For an unannotated context $\vstctx$,
	if~$\unannkind{\vstctx}{\undecty}{\kwtykind}$,
    then for any~$\vs$,
    there exist~$\tau$ and~$\vsty$
	such that~$\futurify{\annvstctx{\vstctx}}{\vs}{\undecty}{\tau}{\vsty}$.
	\label{lem:annot-progress}
  \item
    If~$\futurify{\vstctx}{\vs}{\undecty}{\tau}{\vsty}$
		and $\unannkind{\unannvstctx{\vstctx}}{\undecty}{\kwtykind}$
		and $\vsistype{\ectx}{\utctx}{\vs}{\vsty}$,
    then $\iskind{\ectx}{}{\utctx}{\vstctx}{\tau}{\kwtykind}$.
	\label{lem:annot-preservation}
  \item
    If~$\futurify{\ectx}{\vspath}{\undecty}{\tau}{\vsty}$
    and~$\unannvaltype{\undece}{\undecty}$
	and $\unannkind{\ectx}{\undecty}{\kwtykind}$,
    then there exists~$v$ such that~$\futurifye{\vspath}{\undece}{v}$
    and for any~$\uspctx$ such that~$\vsistype{\uspctx}{\ectx}{\vspath}{\vsty}$.
    we have~$\affinetyped{\uspctx}{v}{\tau}$.
	\label{lem:annot-values}
  \end{enumerate}
\end{theorem}

%% file: impl.tex
\section{Implementation and Examples}\label{sec:impl}

We have implemented a prototype graph inference algorithm for \langname{} on top
of \toolname{}~\citep{Muller22}, an existing graph type inference algorithm.
The goal of the implementation, which we call \newtoolname{},
is to infer vertex structure annotations
and graph types from ordinary, unannotated OCaml programs.
\toolname{} extends OCaml syntax with the keywords \lstinline{future} for
spawning
expressions into a future, \lstinline{touch} for joining a future handle's value to
the current thread, and a type \lstinline{'a future}.
Additionally, \newtoolname{} supports OCaml's user-definable recursive datatypes,
which were not previously supported by \toolname{} (there are some limitations,
which we discuss at the end of this section).
For example, we can define the \lstinline{'a pipe} type from
Sections~\ref{sec:intro} and~\ref{sec:overview}
using standard OCaml syntax:
\begin{lstlisting}
type 'a pipe = Pipe of 'a * 'a pipe future ;;
\end{lstlisting}
Our extension of \toolname{} successfully infers the corresponding vertex
structure annotations, for the type itself and for all of its uses in the
code in Figure~\ref{fig:pi}.

In addition, we implemented (by extending facilities existing in \toolname{})
a visualizer that uses several heuristics to generate a visualization of
a representative graph corresponding to each inferred
graph type.\footnote{Once the graph type is unrolled to generate the
representative graph, we output a file that can be turned into a visualization
using GraphViz~\citep{GansnerNo00}.}
This allows developers to see at a glance how their program will parallelize.
We have used \newtoolname{} to infer graph types
for all example programs in this paper.

The details of the implementation are out of the scope of the paper.
However, the main challenge in extending GML with support for algebraic data
types is generating the VS type corresponding to an ADT.
Our algorithm for this closely follows the presentation of
Section~\ref{sec:infer}.\footnote{As discussed in that section, the only
  non-algorithmic
detail of the presentation was constructing~$\vsty$ in~\rulename{F:Rec};
in the implementation, we add~$\kwkindarr{\vstyvar}{\kwtykind}$
to the context instead of~$\kwkindarr{\kwvscorec{\vstyvar}{\vsty}}{\kwtykind}$.
This means the context contains non-well-formed VS types, which makes the
theory more unwieldy but yields a convenient implementation.}
When processing a type declaration, \toolname{} generates the associated
VS type, and also generates a constructor and deconstructor function for
each constructor.
%
%For example, for the \lstinline{'a pipe} datatype, we would generate a
%constructor of type
%\lstinline{'a$\figannot{[\kwfst{\vvar}]}$ * 'a pipe future$\figannot{[\kwsnd{\vvar}]}$ -> 'a pipe$\figannot{[\vvar]}$}
%and a deconstructor of type
%\lstinline{'a pipe$\figannot{[\vvar]}$ -> 'a$\figannot{[\kwfst{\vvar}]}$ * 'a pipe future$\figannot{[\kwsnd{\vvar}]}$}.
%
Constructor applications are desugared to ordinary applications of the
constructor function and the deconstructor function is used during pattern
matching.
Another major challenge is implementing unification on vertex structures.
At the moment, our implementation uses a set of heuristics that are not
guaranteed to be complete (i.e., unification may fail for VSs that could be
unified, resulting in a spurious type error) but work well in practice on the
large examples tested.

In addition to extending the subset of OCaml supported by \toolname{}, we
have also substantially re-architected the code.
In \newtoolname{},
the graph type checker is completely separate from the type checker.
This simplifies the implementation and has a number of other benefits.
First, all futures in a program are known by the time graph checking begins.
This allows the implementation to infer graph types in several instances
where type annotations would previously have been required
(one such instance is noted in prior work~\citep{Muller22} as a limitation
of \toolname{},
which is not a limitation of \newtoolname{}).
Additionally, this architecture would simplify the process of integrating
graph checking as an extension of the OCaml compiler, as an additional pass
on type-checked ASTs.

\subsection{Examples}

To show the utility of \newtoolname{}, we discuss several example programs
for which it can infer and visualize graph types.

%% \begin{wrapfigure}{i}{0.5\textwidth}
%%   \vspace{-1.2cm}
%%   \hspace{2mm}
%% \begin{minipage}{0.3\textwidth}
%%     \begin{lstlisting}
%% let rec map fxs =
%%   let (f, xs) = fxs in
%%   match xs with
%%   | [] -> []
%%   | x :: xs ->
%%     let y = future (f x) in
%%     let ys = (map (f, xs)) in
%%     (touch y) :: ys
%% ;;
%% \end{lstlisting}%
%%    \end{minipage}%
%% %\hspace{.5in}
%% \begin{minipage}{0.25\textwidth}
%% \includegraphics[clip,scale=0.5, trim=.7in .7in .7in 1in]{map}%
%% \end{minipage}
%%   \caption{Parallelized map function}
%%   \label{fig:map}
%%  % \vspace{-0.75cm}
%% \end{wrapfigure}

\paragraph*{Produce-consume}
The producer-consumer example of \citet{BlellochRe97}, shown in
Figure~\ref{fig:sum},
is similar to the \lstinline{pi_pipeline} function of Figure \ref{fig:pi},
but allows the pipelined list to be finite (ending with \lstinline{FNil}).
As in the pipeline example, the \lstinline{FCons} constructor allows the
tail of the \lstinline{flist} to continue being computed in a future.
We compose \lstinline{produce}, which (for the sake of a simple example)
outputs a list of the numbers 1--n,
with the \lstinline{consume} function which calculates the sum of the list.
In the graph of the composed functions (right side of the figure), the
touches of \lstinline{consume} happen in parallel with the production of
the list.

\begin{figure}
  \begin{minipage}{0.45\textwidth}
    \begin{lstlisting}
type 'a flist =
| FNil
| FCons of 'a * ('a flist future);;

let rec produce n =
  if n < 0 then FNil
  else FCons (n, future (produce (n - 1)));;

let rec consume sumxs =
  let (sum, xs) = sumxs in
  match xs with
  | FNil -> sum
  | FCons (x, xs) -> consume (x + sum, force xs);;

consume (0, (produce n));;
\end{lstlisting}
  \end{minipage}%
  \hspace{.5in}
  \begin{minipage}{0.25\textwidth}
    \includegraphics[clip,scale=0.6, trim=.7in .7in .7in 1in]{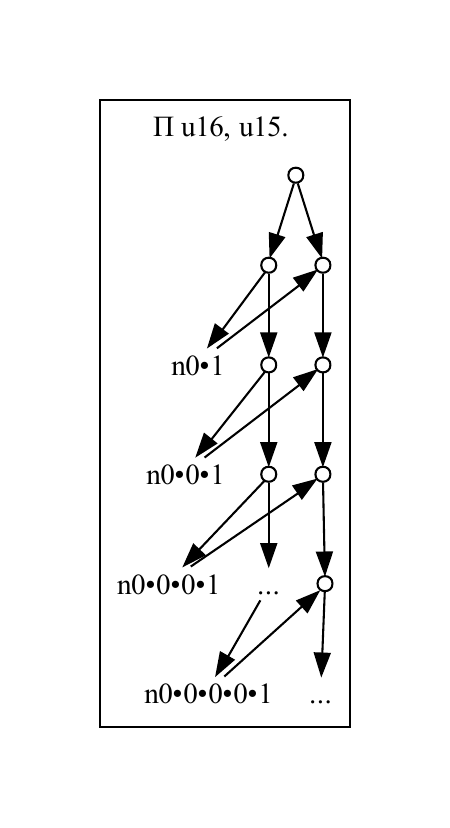}%
  \end{minipage}
  \caption{Blelloch-Reid-Miller produce-consume example}
  \label{fig:sum}
\end{figure}

\paragraph*{Tree Sum} In Figure \ref{fig:treesum}, we present operations
on a pipelined tree data structure~\citep{BlellochRe97}.
As with \lstinline{flist}, the two subtrees of an \lstinline{ftree} are futures, so they may be computed asynchronously while the value at the node is used.
The function
\lstinline{bst} generates a tree of numbers 0 to 10, then \lstinline{tree_sum}
calculates the sum of elements in the generated tree.
While the particular application of summing a binary tree is fairly simple,
one can imagine using the same structure for more complicated use-cases.
Because of the design of the data structure, \lstinline{bst}
immediately returns a future and then
\lstinline{tree_sum} can perform its calculation as later recursive steps of
\lstinline{bst} are still executing.
%
%\langname{} allows the programmer to define
%their data-structure with futures in positions that make the most sense for
%their application.

\begin{figure}
  \begin{minipage}{0.35\textwidth}
    \begin{lstlisting}
type ftree =
| Empty
| Node of int * ftree future * ftree future;;

let rec bst lohi =
  let (lo, hi) = lohi in
  if lo >= hi then Empty
  else
    let mid = (lo + hi) / 2 in
    Node (mid, future (bst (lo, mid)),
          future (bst (mid, hi)));;

let rec tree_sum tree =
  match tree with
  | Empty -> 0
  | Node (x, l, r) ->
    let left_sum_fut = future (tree_sum (touch l)) in
    let right_sum = tree_sum (touch r) in
    let left_sum = touch left_sum_fut in
    x + left_sum + right_sum;;

tree_sum (bst (0, 10));;
\end{lstlisting}
  \end{minipage}%
  \begin{minipage}{0.45\textwidth}
    \includegraphics[clip,scale=0.5]{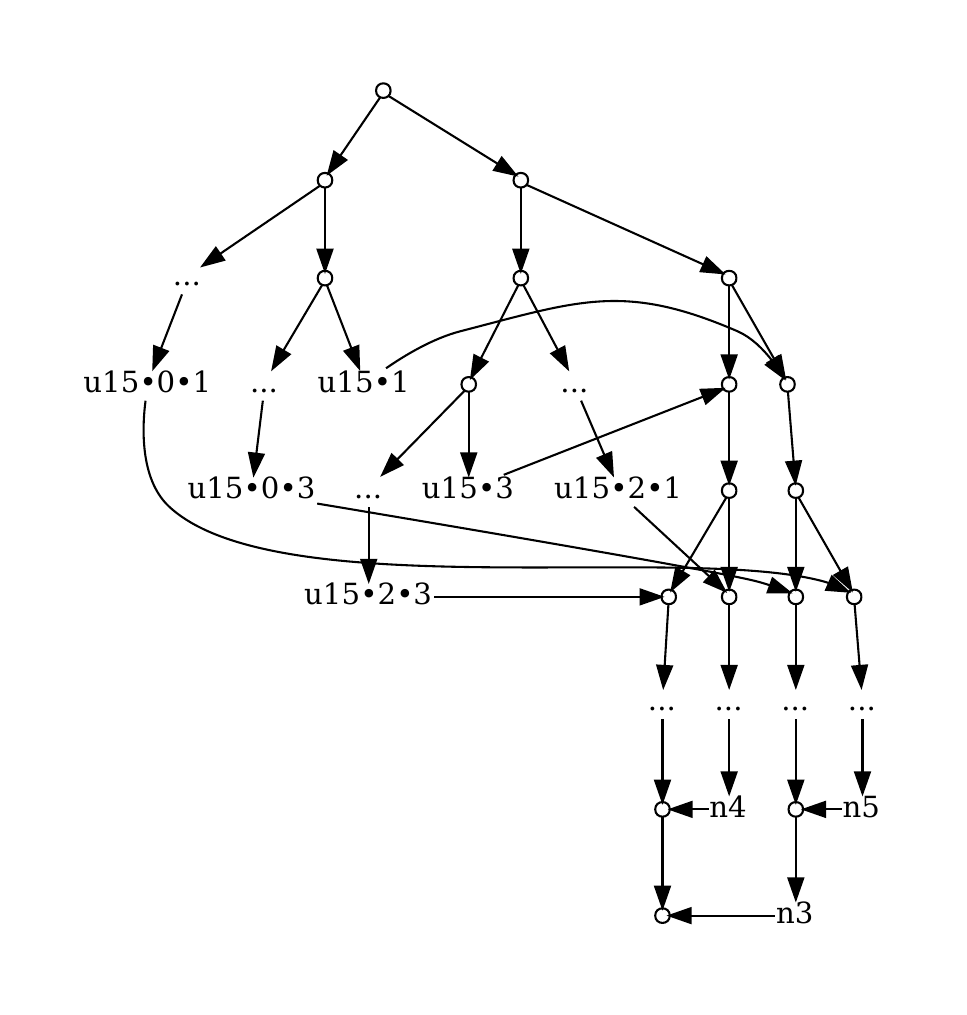}%
  \end{minipage}
  \caption{Sum over elements in a pipelined tree.}
  \label{fig:treesum}
\end{figure}

\paragraph*{Tree Reverse}

The function in Figure~\ref{fig:reverse} reverses a pipelined tree of
the type defined in Figure~\ref{fig:treesum}.
Here, the interesting feature of the output was not the visualization of the
function's graph type, which shows a similar structure to~\lstinline{tree_sum}
and~\lstinline{bst}, but the function type, which is
\[\kw{reverse} : \kwpi{\hastycl{\vvar_f}{\kw{vtree}}}
     {\hastycl{\vvar_t}{\kw{vtree}}}
     {\kwarrow{\kw{ftree}~\vvar_t}
     {\kw{ftree}~\vvar_f}
     {\graph}}
     \]
where~$\kw{vtree} = \kwvscorec{\vstyvar}
{\vstyvar \times \kwvty \times \vstyvar \times \kwvty}$.
We omit the graph type~$\graph$ for clarity.
The function takes two VS parameters and a tree indexed by~$\vvar_t$ and
returns a tree indexed by~$\vvar_f$.
%
%\begin{wrapfigure}{i}{0.45\textwidth}
\begin{figure}
  %\vspace{1cm}
  %\hspace{12mm}%
  \begin{minipage}{0.87\textwidth}
\begin{lstlisting}
let rec reverse tree =
  match tree with
  | Empty -> Empty
  | Node (x, l, r) -> Node (x, future (reverse (force r)), future (reverse (force l)))
\end{lstlisting}
\end{minipage}
  \caption{Reverse a pipelined tree.}
  \label{fig:reverse}
  %\end{wrapfigure}
\end{figure}
At first glance, this may seem imprecise because one might expect the VSs
parameterizing the input and output tree to be related (after all, the output
tree is the reverse of the input).
However, this is not correct: \lstinline{reverse}
touches (in a pipelined way) all of the
futures of the input tree and constructs a new tree with the reversed
{\em values} but {\em new futures} from the VS~$\vvar_f$.
Here, not just the graph type but the return type parameterized by its VS
can correct a misunderstanding about the parallel behavior of a program.

%% \paragraph*{Parallelized Map} \newtoolname{} In Figure \ref{fig:map}, we take the standard map
%% function, which maps a function~$f$ over all elements of a polymorphic list,
%% and parallelize it so all applications of $f$ happen in their own
%% futures.
%% %
%% In the visualization, each instance of \lstinline{'ga}, which is a parameter
%% for the graph type of the function~$f$, appears as a separate future split
%% off of the main thread.
%% %
%% These
%% vertices are then touched sequentially in order to build the new list back up.
%% %
%% Note that the instances of \lstinline{'ga} have no paths between them---this
%% indicates that all calls to~$f$ can happen in parallel, which achieves the
%% desired parallelization, particularly if~$f$ is expensive.
%% \jw{should I include the graph type produced by the implementation? I think
%% referring to 'ga without it may be confusing but the type itself is pretty ugly
%% because it's only for machines:}

\subsection{Limitations}
Though our inference algorithm checks most useful programs, there are some limitations.
%\begin{itemize}
%\item
First, polymorphic types cannot be instantiated with types
  that include futures.
  For example, a list of futures would have to be explicitly defined as a
  new type \lstinline{'a futlist} rather than by instantiating built-in lists
  to form the type \lstinline{'a future list}.
  This represents a design trade-off; graph types are not currently expressive
  enough to represent, say, a \lstinline{reverse} function on lists of futures.
  If the standard list type could be instantiated with a
  future, the polymorphic \lstinline{reverse} function would need to be
  assigned a general enough type to cover all instantiations of \lstinline{'a},
  which wouldn't be possible.
%  (instantiating \lstinline{'b}
%  in \lstinline{'b list} with \lstinline{'a future}); rather,
%  the programmer must define a separate
%  \lstinline{'a futlist} type.
%
%  This is because \langname{} does not allow for the ability to generalize over% vertex structures.
  %
% \item The user cannot annotate programs with their own graph types.
%   \skm{I really think of this as a strength (programmers don't have to annotate
%     programs...). If you can come up with an example where a programmer could
%     give an annotation in our existing system that we can't infer, maybe
%     mention this briefly. The example below is one, though it would be better
%     to have one that relates to more interesting VSs; otherwise it's just
%     a limitation we inherit from GML}
%   %
%   This limits the programs that can be checked to the programs that
%   can be inferred.
%   %
%   We made the decision to disallow user-annotated graphs in order to
%   keep developer overhead to a minimum and reduce code complexity.
%
  A limitation we inherit from \toolname{} is that functional arguments
  of higher order functions
  cannot spawn futures.
  This is possible in \langname{}, but cannot be inferred without annotations.
%\end{itemize}

% \skm{Here it would be worth showing off some of the outputs on the case studies
%   including the visualization.}

% \skm{Maybe leave this to the conclusion}
% In the future we hope to explore the ways that graph type analysis in
% \langname{} can catch and prevent bugs to do with

% \aaa{Status of the implementation/possible deliverables}

% What is missing:

% \begin{itemize}
% \item Recursive functions
%   \begin{itemize}
%   \item Hardwired combinators with axiomatized types
%   \item Check general recursive definitions
%   \end{itemize}
% \item Recursive data types
%   \begin{itemize}
%   \item Hardwired data types (e.g. lists of futures) \aaa{Stefan already had a
%       lists of futures implementation based on the old version}
%   \item General user-defined data types
%   \end{itemize}
% \item Useful polymorphic type definitions (i.e., can we make
%   \verb!int future list! be a useful type if we have defined \verb!'a list!
%   generically?)
%   \begin{itemize}
%   \item Maybe we still need a list type that is specialized to
%     non-vertex-structure-aware types, so that functions like list reverse are
%     expressible.
%   \end{itemize}
% \item Visualization
% \item Analysis of graph types (e.g. bug finding, parallelism analysis)
% \item Compiling and running code (probably not happening right now)
%   \begin{itemize}
%   \item Maybe integrate with OCaml 5?
%   \end{itemize}
% \end{itemize}

%% file: related.tex
\section{Related Work}

\paragraph{Graph Types and Related Analyses.}
The use of graphs to represent the parallel programs dates back
to at least the late 1960s~\citep{KarpMi66, Rodriguez69}.
Our notation is most directly inspired by the work
of~\citet{BlellochGr95,BlellochGr96} and~\citet{spoonhower09}, who extended
these graphs with notations for futures.
In this work, graphs were produced dynamically from programs using a {\em
  cost semantics}, which abstractly evaluates the program to form the graph
(or a family of graphs if execution is nondeterministically).
The first work we are aware of on statically approximating such graphs for
fine-grained parallel programs was the prior work of
one of us~\citep{Muller22}, which developed
a calculus~$\oldlangname$ and corresponding graph type system for inferring
graph types of parallel programs with futures.
Our work builds on~$\oldlangname$, including the use of an affine type system
to ensure that vertex names are unique and therefore do not appear twice
in a graph, which would result in an invalid graph.
However, the main thrust of this paper is overcoming the significant limitation
in~$\oldlangname$ that affine treatment of vertex names prevents building
collections of futures.

Dependency graphs are frequently used to represent
control dependencies in coarse-grained parallel programs and these have been
the target of several static analyses
(e.g.,~\citep{ChenXuZhYa02, Cheng93, KasaharaNoKaChUs95}) but such tools do
not contend with the substantial dynamicity inherent in fine-grained parallel
programs, especially those with futures.
Dependency graphs are also used to represent other dependencies in a program,
including data dependencies; analyzing the structure of such dependencies is
a form of program slicing (e.g.,~\citep{Weiser84, Korel87}).

As observed in prior work~\citep{Muller22}, graph type systems draw on ideas from
{\em region type systems}~\citep{TofteTa97}, where assigning a vertex to
a future corresponds to allocating an object within a {\em region} of memory,
in order to
aid in memory management and/or ensure safety (e.g.~\citep{FluetMoAh06}),
including in the presence of concurrency and complex, dynamic
data structures~\citep{MilanoTuMy22}.
It is not possible to list all of the
related work on regions and related systems, so we refer the interested reader
to the chapter by \citet{HengleinMaNi05}.
Two major differences with region systems are that vertex assignments must
be unique (whereas typically many objects are allocated within a single
region) and that, to generate useful graphs, we wish for vertex assignments
to be visible at a global scope (see the example from the Introduction of
why locally allocated vertices are not suitable for graph types of data
structures).

\paragraph{Heterogeneous and Indexed Data Structures.}
Indexed types~\citep{Zenger97, XiPf99}, a limited form of dependent types
in which a type is {\em indexed} by a value from a specified domain,
have long been used to add expressiveness to types---a classic example is a
type of vectors indexed with a natural number giving the vector's length.
We index recursive data types with a vertex structure
to assign unique vertices to futures in recursive data structures.
Vertex structures have a non-trivial semantics of their own but are not
first-class objects at the expression level, so computation on VSs
may be seen as an instance of {\em type-level computation}.
A similar indexing idea and type-level computation have been previously
combined to achieve heterogeneity in HList~\citep{KiselyovLaSc04}, a Haskell
library that expresses heterogeneous lists by indexing the list type
constructor with a type-level list.
Their work does not appear to generalize beyond lists or to infinite indices.

Richer forms of type-level computation have been explored, and could be used
to further generalize the theory of vertex structures.
\citet{YorgeyWeCrPJVyMa12} extend Haskell's kind system with features for
expressing a variety of type-level data structures.
As another example, we have considered extending vertex structures
with sums (so a vertex structure could, e.g., represent a
finite list) and using type-level
matching~\citep{BlanvillainBrKjOd22} to constrain the lengths of lists
by the length of the vertex structure parameter.
While this would expand the expressiveness of vertex structures, extending
VSs beyond tree-like corecursive structures causes a problem for inference
and so it seems likely that such an extension would require programmer
annotations in some cases.

We note that this paper enters a rich design space of combining data and
codata (e.g.~\citep{ThibodeauCaPi16}).
We have shown in Section~\ref{sec:infer} that any data type can be
``overapproximated by'' a codata type in the sense that there is a
straightforward, local mapping from nodes in the data type's AST to those
of the codata type (the vertex structure in our case).
Whether this has a deeper meaning in the theory of data and codata is left
to future work.

\paragraph{Affine Type Systems.}
We use an affine type system to handle vertex structures and ensure that
vertex names in output graphs are unique.
Affine type systems have been used in a number of languages to ensure safe
usage of resources (broadly construed), notably including
Cyclone~\citep{Cyclone} and Rust~\citep{Rust}.
Our notation for splitting and availability is inspired by
Cogent~\citep{OConnorCRJAKMSK21}, which uses these ideas for an affine
treatment of record types.

\paragraph{Encoding in Rust.} As a memory safety focused language, Rust's type
system would likely benefit from the features GML.
Though Rust is able to encode other type systems such as session
types~\citep{10.1145/2808098.2808100, 10.1007/978-3-030-50029-0_8}, we do not
believe GML could be usefully or at least conveniently encoded in Rust as is.
The main limiting factor we forsee is that we believe each vertex name would
need its own lifetime.
This is a problem because Rust requires all lifetimes to be declared statically,
and each piece of code can only refer to finitely many lifetimes.
However, in GML, vertex names are generated dynamically, therefore a function
might manipulate infinitely many vertices.

%% file: conclusion.tex
\section{Conclusion}
We have presented a type system for annotating parallel programs with futures
with {\em graph types}, which compactly represent the parallel structure of
the program.
Unlike prior work, we support complex data structures
containing futures.
%, which can be used to build pipelined algorithms that are
%asymptotically faster than versions without this feature.
%
As evidenced by our prototype implementation, it is possible to infer these
graph types automatically for examples that use this feature for efficient
pipelined algorithms.
Our implementation is also able to generate visualizations of the
resulting graph types, which can aid in understanding the structure of
parallel code and finding bugs.

In the future, we hope to expand on the types of program analysis and
bug-finding that can be done with these graph types.
For example, we could build on the prototype analyses of \citet{Muller22}
to study deadlock and asymptotic complexity in the
complex, pipelined graphs that arise from programs with data structures of
futures.
We also plan to scale the implementation up to support a larger subset of
OCaml, with the goal of integrating the analysis into the OCaml compiler.
Finally, because the graph type system itself does not depend on a particular
source language, we plan to explore implementing the graph type system in
front-ends for a variety of languages, so that more programmers
can benefit from graph types as an analysis tool and reasoning aid.

%% file: app.tex
\appendix

%This companion appendix contains full sets of rules, definitions, and proofs
%from the technical sections of the paper, organized by the corresponding
%section.

\section{From Section~\ref{sec:lang}}\label{app:lang-proofs}

\begin{figure}
  \centering
  \def \MathparLineskip {\lineskip=0.43cm}
  \begin{mathpar}
    \Rule{U:OmegaGen}
         {\strut}
         {\vsistype{\uspctx, \hastype{\genseed}{\vsty}}{\utctx}{\vertgen{\vsty}{\genseed}}{\vsty}}
    \and
    \Rule{U:PsiGen}
         {\strut}
         {\vsistype{\uspctx}{\utctx, \hastype{\genseed}{\vsty}}{\vertgen{\vsty}{\genseed}}{\vsty}}
    \and
    \Rule{OM:Gen}
         {\uspsplit{\uspctx}{\uspctx_1}{\uspctx_2}}
         {\uspsplit{\uspctx, \hastype{\genseed}{\vsty}}{\uspctx_1, \hastype{\genseed}{\vsty}}{\uspctx_2}}
    \and
	\Rule{OM:GenTypeSplit}
         {\uspsplit{\uspctx}{\uspctx_1}{\uspctx_2}\\
			\vstysplit{\vsty}{\vsty_1}{\vsty_2}}
         {\uspsplit{\uspctx, \hastype{\genseed}{\vsty}}
			{\uspctx_1, \hastype{\genseed}{\vsty_1}}
			{\uspctx_2, \hastype{\genseed}{\vsty_2}}}
  \end{mathpar}
  \caption{Extended rules for VS typing and $\uspctx$ context splitting with generators.}
  \label{fig:generator-additions}
\end{figure}

\skm{I moved these here to be close to the text. Lemma 1 doesn't seem to have text (maybe not necessary for the body, but then the lemma probably isn't necessary for the body either).}

Figure \ref{fig:generator-additions} gives the typing rules for generators and the $\uspctx$ context splitting rules for contexts containing generators, which are symmetric to their VS variable counterparts in Figures \ref{fig:vert-typing} and \ref{fig:uspctx-split} respectively. Generators behave identically to VS variables throughout this section of the appendix, with the exception that generators cannot be substituted for using the substitution lemma given in this section.

We will now describe some technical properties of VS type splitting and $\uspctx$ context splitting. 
Lemma \ref{lem:split-form} gives the canonical forms of VS types and $\uspctx$-contexts that are involved in splitting judgements.
Lemma \ref{lem:vsty-split-supertyping} states if $\vsty$ splits into $\vsty_1$ and $\vsty_2$, then all subtypes of $\vsty$ also split into $\vsty_1$ and $\vsty_2$ (i.e. $\vsty$ can be "strengthened" in the splitting judgement).
Lemma \ref{lem:split-move} states that VS type splitting and $\uspctx$ splitting are associative. For example, if $\vsty$ splits into three disjoint types, this is manifested by having $\vsty$ split into one of the three types and an intermediate $\vsty'$ that splits into the other two types. The associative nature of splitting allows any one of the three types to be split directly from $\vsty$ whie guaranteeing the existance of an intermediate $\vsty'$ that splits into the other two types. Loosely speaking, the order in which disjoint VS types and $\uspctx$s are split from a single source is inconsequential.

\begin{lemma}\label{lem:split-form}\strut
\begin{enumerate}
  \item If~$\vstysplit{\vsty}{\vsty_1}{\vsty_2}$
	then one of the following is true:
	\begin{enumerate}
		\item $\vsty = \vsprod{\vsty_3}{\isav}{\vsty_4}{\isav}$.
		\item $\vsty = \vsprod{\vsty_3}{\isav}{\vsty_4}{\isunav}$.
		\item $\vsty = \vsprod{\vsty_3}{\isunav}{\vsty_4}{\isav}$.
		\item $\vsty = \kwvscorec{\vstyvar}{\vsty_3}$.
	\end{enumerate}
    \label{lem:split-vsty-form}
  \item If~$\vstysplit{\vsprod{\vsty_1}{\isav}{\vsty_2}{\isunav}}{\vsty_3}{\vsty_4}$
	then there exists a $\vsty_1'$ and $\vsty_1''$
	such that~$\vstysplit{\vsty_1}{\vsty_1'}{\vsty_1''}$
	and~$\vstysubt{\vsprod{\vsty_1'}{\isav}{\vsty_2}{\isunav}}{\vsty_3}$
	and~$\vstysubt{\vsprod{\vsty_1''}{\isav}{\vsty_2}{\isunav}}{\vsty_4}$.
    \label{lem:split-left-prod}
  \item If~$\vstysplit{\vsprod{\vsty_1}{\isunav}{\vsty_2}{\isav}}{\vsty_3}{\vsty_4}$
	then there exists a $\vsty_2'$ and $\vsty_2''$
	such that~$\vstysplit{\vsty_2}{\vsty_2'}{\vsty_2''}$
	and~$\vstysubt{\vsprod{\vsty_1}{\isunav}{\vsty_2'}{\isav}}{\vsty_3}$
	and~$\vstysubt{\vsprod{\vsty_1}{\isunav}{\vsty_2''}{\isav}}{\vsty_4}$.
    \label{lem:split-right-prod}
  \item If~$\vstysplit{\vsprod{\vsty_1}{\isav}{\vsty_2}{\isav}}{\vsty_3}{\vsty_4}$
	then at least one of the following is true:
	\begin{enumerate}
		\item $\vstysubt{\vsprod{\vsty_1}{\isav}{\vsty_2}{\isunav}}{\vsty_3}$
			and $\vstysubt{\vsprod{\vsty_1}{\isunav}{\vsty_2}{\isav}}{\vsty_4}$.
		\item $\vstysubt{\vsprod{\vsty_1}{\isunav}{\vsty_2}{\isav}}{\vsty_3}$
			and $\vstysubt{\vsprod{\vsty_1}{\isav}{\vsty_2}{\isunav}}{\vsty_4}$.
		\item There exists a $\vsty_1'$ and $\vsty_1''$
			such that $\vstysplit{\vsty_1}{\vsty_1'}{\vsty_1''}$
			and $\vstysubt{\vsprod{\vsty_1'}{\isav}{\vsty_2}{\isav}}{\vsty_3}$
			and $\vstysubt{\vsprod{\vsty_1''}{\isav}{\vsty_2}{\isunav}}{\vsty_4}$.
		\item There exists a $\vsty_1'$ and $\vsty_1''$
			such that $\vstysplit{\vsty_1}{\vsty_1'}{\vsty_1''}$
			and $\vstysubt{\vsprod{\vsty_1'}{\isav}{\vsty_2}{\isunav}}{\vsty_3}$
			and $\vstysubt{\vsprod{\vsty_1''}{\isav}{\vsty_2}{\isav}}{\vsty_4}$.
		\item There exists a $\vsty_2'$ and $\vsty_2''$
			such that $\vstysplit{\vsty_2}{\vsty_2'}{\vsty_2''}$
			and $\vstysubt{\vsprod{\vsty_1}{\isav}{\vsty_2'}{\isav}}{\vsty_3}$
			and $\vstysubt{\vsprod{\vsty_1}{\isunav}{\vsty_2''}{\isav}}{\vsty_4}$.
		\item There exists a $\vsty_2'$ and $\vsty_2''$
			such that $\vstysplit{\vsty_2}{\vsty_2'}{\vsty_2''}$
			and $\vstysubt{\vsprod{\vsty_1}{\isunav}{\vsty_2'}{\isav}}{\vsty_3}$
			and $\vstysubt{\vsprod{\vsty_1}{\isav}{\vsty_2''}{\isav}}{\vsty_4}$.
		\item There exists a $\vsty_1'$ and $\vsty_1''$ and $\vsty_2'$ and $\vsty_2''$
			such that~$\vstysplit{\vsty_1}{\vsty_1'}{\vsty_1''}$
			and \\$\vstysplit{\vsty_2}{\vsty_2'}{\vsty_2''}$
			and $\vstysubt{\vsprod{\vsty_1'}{\isav}{\vsty_2'}{\isav}}{\vsty_3}$
			and $\vstysubt{\vsprod{\vsty_1''}{\isav}{\vsty_2''}{\isav}}{\vsty_4}$.
	\end{enumerate}
    \label{lem:split-full-prod}
  \item If~$\uspsplit{\uspctx, \hastype{\vvar}{\vsty}}{\uspctx_1}{\uspctx_2}$,
	then one of the following is true:
	\begin{enumerate}
		\item $\uspctx_1 = \uspctx_1', \hastype{\vvar}{\vsty}$
			and $\uspsplit{\uspctx}{\uspctx_1'}{\uspctx_2}$
			and $\vvar \notin \uspctx_2$.
		\item $\uspctx_2 = \uspctx_2', \hastype{\vvar}{\vsty}$
			and $\uspsplit{\uspctx}{\uspctx_1}{\uspctx_2'}$
			and $\vvar \notin \uspctx_1$.
		\item $\vstysplit{\vsty}{\vsty_1}{\vsty_2}$
			and $\uspctx_1 = \uspctx_1', \hastype{\vvar}{\vsty_1}$
			and $\uspctx_2 = \uspctx_2', \hastype{\vvar}{\vsty_2}$
			and $\uspsplit{\uspctx}{\uspctx_1'}{\uspctx_2'}$.
		\item $\vstysplit{\vsty}{\vsty_1}{\vsty_2}$
			and $\uspctx_1 = \uspctx_1', \hastype{\vvar}{\vsty_2}$
			and $\uspctx_2 = \uspctx_2', \hastype{\vvar}{\vsty_1}$
			and $\uspsplit{\uspctx}{\uspctx_1'}{\uspctx_2'}$.
	\end{enumerate}
    \label{lem:split-uspctx-var-form}
  \item If~$\uspsplit{\uspctx, \hastype{\genseed}{\vsty}}{\uspctx_1}{\uspctx_2}$,
	then one of the following is true:
	\begin{enumerate}
		\item $\uspctx_1 = \uspctx_1', \hastype{\genseed}{\vsty}$
			and $\uspsplit{\uspctx}{\uspctx_1'}{\uspctx_2}$
			and $\genseed \notin \uspctx_2$.
		\item $\uspctx_2 = \uspctx_2', \hastype{\genseed}{\vsty}$
			and $\uspsplit{\uspctx}{\uspctx_1}{\uspctx_2'}$
			and $\genseed \notin \uspctx_1$.
		\item $\vstysplit{\vsty}{\vsty_1}{\vsty_2}$
			and $\uspctx_1 = \uspctx_1', \hastype{\genseed}{\vsty_1}$
			and $\uspctx_2 = \uspctx_2', \hastype{\genseed}{\vsty_2}$
			and $\uspsplit{\uspctx}{\uspctx_1'}{\uspctx_2'}$.
		\item $\vstysplit{\vsty}{\vsty_1}{\vsty_2}$
			and $\uspctx_1 = \uspctx_1', \hastype{\genseed}{\vsty_2}$
			and $\uspctx_2 = \uspctx_2', \hastype{\genseed}{\vsty_1}$
			and $\uspsplit{\uspctx}{\uspctx_1'}{\uspctx_2'}$.
	\end{enumerate}
    \label{lem:split-uspctx-seed-form}
\end{enumerate}
\end{lemma}
\begin{proof}\strut
  \begin{enumerate}
\item By induction on the derivation of
  	$\vstysplit{\vsty}{\vsty_3}{\vsty_4}$.
\begin{itemize}
\item \rulename{US:Prod}.
	By inversion on the rule.

\item \rulename{US:SplitBoth}.
	By inversion on the rule.

\item \rulename{US:SplitLeft}.
	By inversion on the rule.

\item \rulename{US:SplitRight}.
	By inversion on the rule.

\item \rulename{US:Corecursive}.
	By inversion on the rule.

\item \rulename{US:Subtype}
	By induction.

\item \rulename{US:Commutative}.
	By induction.
\end{itemize}

\item By induction on the derivation of
  	$\vstysplit{\vsprod{\vsty_1}{\isav}{\vsty_2}{\isunav}}{\vsty_3}{\vsty_4}$.
\begin{itemize}
\item \rulename{US:SplitLeft}.
	By inversion on the rule.

\item \rulename{US:Subtype}.
	Then $\vstysplit{\vsprod{\vsty_1}{\isav}{\vsty_2}{\isunav}}{\vsty_3'}{\vsty_4}$
		and~$\vstysubt{\vsty_3'}{\vsty_3}$.
	By induction,
		\splitLprodexst{\vsty_1}{\vsty_1'}{\vsty_1''}{\vsty_2}{\vsty_3'}{\vsty_4}.
	Apply \rulename{UT:Transitive}.

\item \rulename{US:Commutative}.
	Then $\vstysplit{\vsprod{\vsty_1}{\isav}{\vsty_2}{\isunav}}{\vsty_4}{\vsty_3}$.
	By induction,
		\splitLprodexst{\vsty_1}{\vsty_1''}{\vsty_1'}{\vsty_2}{\vsty_4}{\vsty_3}.
	Apply \rulename{US:Commutative}.
\end{itemize}

\item By symmetry with part \ref{lem:split-left-prod}.

\item By induction on the derivation of
  	$\vstysplit{\vsprod{\vsty_1}{\isav}{\vsty_2}{\isav}}{\vsty_3}{\vsty_4}$
\begin{itemize}
\item \rulename{US:Prod}.
	By inversion on the rule.

\item \rulename{US:SplitBoth}.
	By inversion on the rule.

\item \rulename{US:SplitLeft}.
	By inversion on the rule.

\item \rulename{US:SplitRight}.
	By inversion on the rule.

\item \rulename{US:Subtype}.
	Then $\vstysplit{\vsprod{\vsty_1}{\isav}{\vsty_2}{\isav}}{\vsty_3'}{\vsty_4}$
		and~$\vstysubt{\vsty_3'}{\vsty_3}$.
	By induction and an application of \rulename{UT:Transitive}.

\item \rulename{US:Commutative}.
	Then $\vstysplit{\vsprod{\vsty_1}{\isav}{\vsty_2}{\isunav}}{\vsty_4}{\vsty_3}$.
	By induction, there are seven cases:
	\begin{enumerate}
		\item $\vstysubt{\vsprod{\vsty_1}{\isav}{\vsty_2}{\isunav}}{\vsty_4}$
				and $\vstysubt{\vsprod{\vsty_1}{\isunav}{\vsty_2}{\isav}}{\vsty_3}$.
		\item $\vstysubt{\vsprod{\vsty_1}{\isunav}{\vsty_2}{\isav}}{\vsty_4}$
				and $\vstysubt{\vsprod{\vsty_1}{\isav}{\vsty_2}{\isunav}}{\vsty_3}$.
		\item \splitLfullprodexstA{\vsty_1}{\vsty_1''}{\vsty_1'}{\vsty_2}{\vsty_4}{\vsty_3}.
			Apply US:Commutative.
		\item \splitLfullprodexstB{\vsty_1}{\vsty_1''}{\vsty_1'}{\vsty_2}{\vsty_4}{\vsty_3}.
			Apply US:Commutative.
		\item \splitRfullprodexstA{\vsty_2}{\vsty_2''}{\vsty_2'}{\vsty_1}{\vsty_4}{\vsty_3}.
			Apply US:Commutative.
		\item \splitRfullprodexstB{\vsty_2}{\vsty_2''}{\vsty_2'}{\vsty_1}{\vsty_4}{\vsty_3}.
			Apply US:Commutative.
		\item \splitbothprodexst{\vsty_1}{\vsty_1''}{\vsty_1'}{\vsty_2}{\vsty_2''}{\vsty_2'}{\vsty_4}{\vsty_3}.
			Apply US:Commutative twice.
	\end{enumerate}
\end{itemize}

\item By induction on the derivation of
  	$\uspsplit{\uspctx, \hastype{\vvar}{\vsty}}{\uspctx_1}{\uspctx_2}$.
\begin{itemize}
\item \rulename{OM:Var}.
	By inversion on the rule.
	Since $\uspctx, \hastype{\vvar}{\vsty}$ can contain only a single binding for $\vvar$,
		$\vvar \notin \uspctx$,
		so $\vvar \notin \uspctx_2$.

\item \rulename{OM:VarTypeSplit}.
	By inversion on the rule.

\item \rulename{OM:Commutative}.
	By induction, there are four cases:
	\begin{enumerate}
		\item \splituspLexstfull{\uspctx_2'}{\uspctx}{\uspctx_2}{\uspctx_1}{\vvar}{\vsty}.
			Apply \rulename{OM:Commutative}.
		\item \splituspRexstfull{\uspctx_1'}{\uspctx}{\uspctx_2}{\uspctx_1}{\vvar}{\vsty}.
			Apply \rulename{OM:Commutative}.
		\item \splituspbothexst{\uspctx_2'}{\uspctx_1'}{\uspctx}{\uspctx_2}{\uspctx_1}{\vvar}{\vsty}{\vsty_1}{\vsty_2}.
			Apply \rulename{OM:Commutative} and \rulename{US:Commutative}.
		\item \splituspbothcrossexst{\uspctx_2'}{\uspctx_1'}{\uspctx}{\uspctx_2}{\uspctx_1}{\vvar}{\vsty}{\vsty_1}{\vsty_2}.
			Apply \rulename{OM:Commutative}.
	\end{enumerate}
\end{itemize}

\item By symmetry with part \ref{lem:split-uspctx-var-form}.
  \end{enumerate}
\end{proof}

\begin{lemma}\label{lem:vsty-split-supertyping}\strut
\begin{enumerate}
  \item If~$\vstysplit{\kwvscorec{\vstyvar}{\vsty}}{\vsty_1}{\vsty_2}$
	then $\vstysplit{\vsub{\vsty}{\kwvscorec{\vstyvar}{\vsty}}{\vstyvar}}{\vsty_1'}{\vsty_2'}$.
    \label{lem:split-corec}
  \item If~$\vstysubt{\vsty}{\vsty'}$
	and~$\vstysplit{\vsty'}{\vsty_1}{\vsty_2}$
	then~$\vstysplit{\vsty}{\vsty_1}{\vsty_2}$.
    \label{lem:vsty-split-supertyping-last}
\end{enumerate}
\end{lemma}
\begin{proof}\strut
  \begin{enumerate}
\item By induction on the derivation of
  	$\vstysplit{\kwvscorec{\vstyvar}{\vsty}}{\vsty_1}{\vsty_2}$.
\begin{itemize}
\item \rulename{US:Corecursive}.
	By inversion on the rule.

\item \rulename{US:Subtype}.
	Then $\vstysplit{\kwvscorec{\vstyvar}{\vsty}}{\vsty_1'}{\vsty_2}$
		and~$\vstysubt{\vsty_1'}{\vsty_1}$.
	By induction,
		$\vstysplit{\vsub{\vsty}{\kwvscorec{\vstyvar}{\vsty}}{\vstyvar}}{\vsty_1'}{\vsty_2}$.
	Apply \rulename{US:Subtype}.

\item \rulename{US:Commutative}.
	Then $\vstysplit{\kwvscorec{\vstyvar}{\vsty}}{\vsty_2}{\vsty_1}$.
	By induction,
		$\vstysplit{\vsub{\vsty}{\kwvscorec{\vstyvar}{\vsty}}{\vstyvar}}{\vsty_2}{\vsty_1}$.
	Apply \rulename{US:Commutative}.
\end{itemize}

\item By induction on the derivation of
	$\vstysubt{\vsty}{\vsty'}$.
  \begin{itemize}
  \item \rulename{UT:ProdLeft}.
	Then $\vsty = \vsprod{\vsty_3}{\avail_1}{\vsty_4}{\avail_2}$
		and~$\vsty' = \vsprod{\vsty_5}{\isunav}{\vsty_4}{\avail_2}$.
	By Lemma \ref{lem:split-form},
		$\avail_2 = \isav$.
	By Lemma \ref{lem:split-form},
		\splitRprodexst{\vsty_4}{\vsty_4'}{\vsty_4''}{\vsty_5}{\vsty_1}{\vsty_2}.
	By \rulename{US:SplitRight},
		$\vstysplit{\vsty}{\vsprod{\vsty_3}{\avail_1}{\vsty_4'}{\isav}}{\vsprod{\vsty_3}{\isunav}{\vsty_4''}{\isav}}$.
	By \rulename{UT:PairLeft},
		$\vstysubt{\vsprod{\vsty_3}{\avail_1}{\vsty_4'}{\isav}}{\vsprod{\vsty_5}{\isunav}{\vsty_4'}{\isav}}$.
	By \rulename{UT:PairLeft},
		$\vstysubt{\vsprod{\vsty_3}{\isunav}{\vsty_4''}{\isav}}{\vsprod{\vsty_5}{\isunav}{\vsty_4''}{\isav}}$.
	Apply \rulename{US:Subtype} four times.

  \item \rulename{UT:ProdLeft}.
	By symmetry with the previous case.

\item \rulename{UT:Prod}.
	Then $\vsty = \vsprod{\vsty_3}{\avail_1}{\vsty_4}{\avail_2}$
		and~$\vsty' = \vsprod{\vsty_3'}{\avail_1}{\vsty_4'}{\avail_2}$
		and~$\vstysubt{\vsty_3}{\vsty_3'}$
		and~$\vstysubt{\vsty_4}{\vsty_4'}$.
	By Lemma \ref{lem:split-form}, there are three cases:
	\begin{enumerate}
		\item $\avail_1 = \isav$ and $\avail_2 = \isunav$.
			By Lemma \ref{lem:split-form},
				\splitLprodexst{\vsty_3'}{\vsty_3''}{\vsty_3'''}{\vsty_4'}{\vsty_1}{\vsty_2}.
			By induction,
				$\vstysplit{\vsty_3}{\vsty_3''}{\vsty_3'''}$.
			By \rulename{US:SplitLeft},
				$\vstysplit{\vsty}{\vsprod{\vsty_3''}{\isav}{\vsty_4}{\isunav}}{\vsprod{\vsty_3'''}{\isav}{\vsty_4}{\isunav}}$.
			By \rulename{UT:Pair},
				$\vstysubt{\vsprod{\vsty_3''}{\isav}{\vsty_4}{\isunav}}{\vsprod{\vsty_3''}{\isav}{\vsty_4'}{\isunav}}$.
			By \rulename{UT:Pair},
				$\vstysubt{\vsprod{\vsty_3'''}{\isav}{\vsty_4}{\isunav}}{\vsprod{\vsty_3'''}{\isav}{\vsty_4'}{\isunav}}$.
			Apply \rulename{US:Subtype} four times.
		\item $\avail_1 = \isunav$ and $\avail_2 = \isav$. By symmetry with the previous case.
		\item $\avail_1 = \avail_2 = \isav$.
			By Lemma \ref{lem:split-form}, there are seven cases:
			\begin{enumerate}
				\item $\vstysubt{\vsprod{\vsty_3'}{\isav}{\vsty_4'}{\isunav}}{\vsty_1}$
						and $\vstysubt{\vsprod{\vsty_3'}{\isunav}{\vsty_4'}{\isav}}{\vsty_2}$.
					By \rulename{US:Prod},
						$\vstysplit{\vsty}{\vsprod{\vsty_3}{\isav}{\vsty_4}{\isunav}}{\vsprod{\vsty_3}{\isunav}{\vsty_4}{\isav}}$.
					By \rulename{UT:Prod},
						$\vstysubt{\vsprod{\vsty_3}{\isav}{\vsty_4}{\isunav}}{\vsprod{\vsty_3'}{\isav}{\vsty_4'}{\isunav}}$.
					By \rulename{UT:Prod},
						$\vstysubt{\vsprod{\vsty_3}{\isunav}{\vsty_4}{\isav}}{\vsprod{\vsty_3'}{\isunav}{\vsty_4'}{\isav}}$.
					Apply \rulename{US:Subtype} twice.
				\item $\vstysubt{\vsprod{\vsty_3'}{\isunav}{\vsty_4'}{\isav}}{\vsty_1}$
						and $\vstysubt{\vsprod{\vsty_3'}{\isav}{\vsty_4'}{\isunav}}{\vsty_2}$.
					By symmetry with the previous case.
				\item \splitLfullprodexstA{\vsty_3'}{\vsty_3''}{\vsty_3'''}{\vsty_4'}{\vsty_1}{\vsty_2}.
					By induction,
						$\vstysplit{\vsty_3}{\vsty_3''}{\vsty_3'''}$.
					By \rulename{US:SplitLeft},
						$\vstysplit{\vsty}{\vsprod{\vsty_3''}{\isav}{\vsty_4}{\isav}}{\vsprod{\vsty_3'''}{\isav}{\vsty_4}{\isunav}}$.
					By \rulename{UT:Pair},
						$\vstysubt{\vsprod{\vsty_3''}{\isav}{\vsty_4}{\isav}}{\vsprod{\vsty_3''}{\isav}{\vsty_4'}{\isav}}$.
					By \rulename{UT:Pair},
						$\vstysubt{\vsprod{\vsty_3'''}{\isav}{\vsty_4}{\isunav}}{\vsprod{\vsty_3'''}{\isav}{\vsty_4'}{\isunav}}$.
					Apply \rulename{US:Subtype} four times.
				\item \splitLfullprodexstB{\vsty_3'}{\vsty_3''}{\vsty_3'''}{\vsty_4'}{\vsty_1}{\vsty_2}.
					By symmetry with the previous case.
				\item \splitRfullprodexstA{\vsty_4'}{\vsty_4''}{\vsty_4'''}{\vsty_3'}{\vsty_1}{\vsty_2}.
					By symmetry with the previous case.
				\item \splitRfullprodexstB{\vsty_4'}{\vsty_4''}{\vsty_4'''}{\vsty_3'}{\vsty_1}{\vsty_2}.
					By symmetry with the previous case.
				\item \splitbothprodexst{\vsty_3'}{\vsty_3''}{\vsty_3'''}{\vsty_4'}{\vsty_4''}{\vsty_4'''}{\vsty_1}{\vsty_2}.
					By induction,
						$\vstysplit{\vsty_3}{\vsty_3''}{\vsty_3'''}$
						and $\vstysplit{\vsty_4}{\vsty_4''}{\vsty_4'''}$.
					By \rulename{US:SplitBoth},
						$\vstysplit{\vsty}{\vsprod{\vsty_3''}{\isav}{\vsty_4''}{\isav}}{\vsprod{\vsty_3'''}{\isav}{\vsty_4'''}{\isav}}$.
					Apply \rulename{US:Subtype} twice.
			\end{enumerate}
	\end{enumerate}

\item \rulename{UT:Corec1}.
	Then $\vsty = \kwvscorec{\vstyvar}{\vsty''}$
		and~$\vsty' = \vsub{\vsty''}{\kwvscorec{\vstyvar}{\vsty''}}{\vstyvar}$.
	Apply US:Recursive.

\item \rulename{UT:Corec2}.
	Then $\vsty = \vsub{\vsty''}{\kwvscorec{\vstyvar}{\vsty''}}{\vstyvar}$
		and~$\vsty' = \kwvscorec{\vstyvar}{\vsty''}$.
	By part \ref{lem:split-corec}.

\item \rulename{UT:Transitive}.
	Then $\vstysubt{\vsty}{\vsty''}$
		and~$\vstysubt{\vsty''}{\vsty'}$.
	By induction,
		$\vstysplit{\vsty''}{\vsty_1}{\vsty_2}$.
	By induction.
\end{itemize}
\end{enumerate}
\end{proof}

\begin{lemma}\label{lem:split-move}\strut
\begin{enumerate}
  \item If~$\vstysplit{\vsty}{\vsty_1}{\vsty_2}$ then:
	\begin{enumerate}
		\item If~$\vstysplit{\vsty_1}{\vsty_1'}{\vsty_1''}$
			then there exists a $\vsty'$
			such that~$\vstysplit{\vsty}{\vsty'}{\vsty_1''}$
			and~$\vstysplit{\vsty'}{\vsty_1'}{\vsty_2}$.
		\item If~$\vstysplit{\vsty_2}{\vsty_2'}{\vsty_2''}$
			then there exists a $\vsty'$
			such that~$\vstysplit{\vsty}{\vsty'}{\vsty_2''}$
			and~$\vstysplit{\vsty'}{\vsty_2'}{\vsty_1}$.
	\end{enumerate}
    \label{lem:split-move-vsty}
  \item If~$\uspsplit{\uspctx}{\uspctx_1}{\uspctx_2}$ then:
	\begin{enumerate}
		\item If~$\uspsplit{\uspctx_1}{\uspctx_1'}{\uspctx_1''}$
			then there exists an $\uspctx'$
			such that~$\uspsplit{\uspctx}{\uspctx'}{\uspctx_1''}$
			and~$\uspsplit{\uspctx'}{\uspctx_1'}{\uspctx_2}$.
		\item If~$\uspsplit{\uspctx_2}{\uspctx_2'}{\uspctx_2''}$
			then there exists an $\uspctx'$
			such that~$\uspsplit{\uspctx}{\uspctx'}{\uspctx_2''}$
			and~$\uspsplit{\uspctx'}{\uspctx_2'}{\uspctx_1}$.
	\end{enumerate}
    \label{lem:split-move-usp}
  \item If~$\uspsplit{\uspctx}{\uspctx_1}{\uspctx_2}$
	and~$\uspsplit{\uspctx_1}{\uspctx_1'}{\uspctx_1''}$
	and~$\uspsplit{\uspctx_2}{\uspctx_2'}{\uspctx_2''}$
	then there exists an $\uspctx'$ and $\uspctx''$
	such that~$\uspsplit{\uspctx}{\uspctx'}{\uspctx''}$
	and~$\uspsplit{\uspctx'}{\uspctx_1'}{\uspctx_2'}$
	and~$\uspsplit{\uspctx''}{\uspctx_1''}{\uspctx_2''}$.
    \label{lem:split-move-two-usp}
\end{enumerate}
\end{lemma}
\begin{proof}\strut
  \begin{enumerate}
  \item By induction on the derivation of
  $\vstysplit{\vsty}{\vsty_1}{\vsty_2}$.
  \begin{itemize}
  \item \rulename{US:Prod}.
	Then $\vsty = \vsprod{\vsty_3}{\isav}{\vsty_4}{\isav}$
		and~$\vsty_1 = \vsprod{\vsty_3}{\isav}{\vsty_4}{\isunav}$
		and~$\vsty_2 = \vsprod{\vsty_3}{\isunav}{\vsty_4}{\isav}$.
	\begin{enumerate}
		\item If~$\vstysplit{\vsty_1}{\vsty_1'}{\vsty_1''}$,
			by Lemma \ref{lem:split-form},
				\splitLprodexst{\vsty_3}{\vsty_3'}{\vsty_3''}{\vsty_4}{\vsty_1'}{\vsty_1''}.
			By \rulename{US:SplitLeft},
				$\vstysplit{\vsty}{\vsprod{\vsty_3'}{\isav}{\vsty_4}{\isav}}{\vsprod{\vsty_3''}{\isav}{\vsty_4}{\isunav}}$.
			By \rulename{US:Subtype},
				$\vstysplit{\vsty}{\vsprod{\vsty_3'}{\isav}{\vsty_4}{\isav}}{\vsty_1''}$.
			By \rulename{US:Prod},
				$\vstysplit{\vsprod{\vsty_3'}{\isav}{\vsty_4}{\isav}}{\vsprod{\vsty_3'}{\isav}{\vsty_4}{\isunav}}
					{\vsprod{\vsty_3'}{\isunav}{\vsty_4}{\isav}}$.
			By \rulename{UT:ProdLeft},
				$\vstysubt{\vsprod{\vsty_3'}{\isunav}{\vsty_4}{\isav}}{\vsprod{\vsty_3}{\isunav}{\vsty_4}{\isav}}$.
			Apply \rulename{US:Subtype} twice.
		\item If~$\vstysplit{\vsty_2}{\vsty_2'}{\vsty_2''}$,
			by symmetry with the previous case.
	\end{enumerate}

  \item \rulename{US:SplitLeft}.
	Then $\vsty = \vsprod{\vsty_3}{\isav}{\vsty_4}{\avail}$
		and~$\vstysplit{\vsty_3}{\vsty_3'}{\vsty_3''}$
		and~$\vsty_1 = \vsprod{\vsty_3'}{\isav}{\vsty_4}{\avail}$
		and~$\vsty_2 = \vsprod{\vsty_3''}{\isav}{\vsty_4}{\isunav}$.
	\begin{enumerate}
		\item If~$\vstysplit{\vsty_1}{\vsty_1'}{\vsty_1''}$,
			there are two cases:
			\begin{enumerate}
				\item $\avail = \isunav$.
					By Lemma \ref{lem:split-form},
						\splitLprodexst{\vsty_3'}{\vsty_3'''}{\vsty_3''''}{\vsty_4}{\vsty_1'}{\vsty_1''}.
					By induction,
						\vstysplitexst{\vsty''}{\vsty_3}{\vsty_3'''}{\vsty_3''''}{\vsty_3''}.
					By \rulename{US:SplitLeft},
						$\vstysplit{\vsty}{\vsprod{\vsty''}{\isav}{\vsty_4}{\isunav}}{\vsprod{\vsty_3''''}{\isav}{\vsty_4}{\isunav}}$.
					By \rulename{US:Subtype},
						$\vstysplit{\vsty}{\vsprod{\vsty''}{\isav}{\vsty_4}{\isunav}}{\vsty_1''}$.
					By \rulename{US:SplitLeft},
						$\vstysplit{\vsprod{\vsty''}{\isav}{\vsty_4}{\isunav}}{\vsprod{\vsty_3'''}{\isav}{\vsty_4}{\isunav}}{\vsty_2}$.
					Apply \rulename{US:Subtype}.
				\item $\avail = \isav$.
					By Lemma \ref{lem:split-form}, there are seven cases:
					\begin{enumerate}
						\item $\vstysubt{\vsprod{\vsty_3'}{\isav}{\vsty_4}{\isunav}}{\vsty_1'}$
								and $\vstysubt{\vsprod{\vsty_3'}{\isunav}{\vsty_4}{\isav}}{\vsty_1''}$.
							By \rulename{US:Prod},
								$\vstysplit{\vsty}{\vsprod{\vsty_3}{\isav}{\vsty_4}{\isunav}}{\vsprod{\vsty_3}{\isunav}{\vsty_4}{\isav}}$.
							By \rulename{UT:ProdLeft},
								$\vstysubt{\vsprod{\vsty_3}{\isunav}{\vsty_4}{\isav}}{\vsprod{\vsty_3'}{\isunav}{\vsty_4}{\isav}}$.
							By two applications of \rulename{US:Subtype},
								$\vstysplit{\vsty}{\vsprod{\vsty_3}{\isav}{\vsty_4}{\isunav}}{\vsty_1''}$.
							By \rulename{US:SplitLeft},
								$\vstysplit{\vsprod{\vsty_3}{\isav}{\vsty_4}{\isunav}}{\vsprod{\vsty_3'}{\isav}{\vsty_4}{\isunav}}{\vsty_2}$.
							Apply \rulename{US:Subtype}.
						\item $\vstysubt{\vsprod{\vsty_3'}{\isunav}{\vsty_4}{\isav}}{\vsty_1'}$
								and $\vstysubt{\vsprod{\vsty_3'}{\isav}{\vsty_4}{\isunav}}{\vsty_1''}$.
							By symmetry with the previous case.
						\item \splitLfullprodexstA{\vsty_3'}{\vsty_3'''}{\vsty_3''''}{\vsty_4}{\vsty_1'}{\vsty_1''}.
							By induction,
								\vstysplitexst{\vsty''}{\vsty_3}{\vsty_3'''}{\vsty_3''''}{\vsty_3''}.
							By \rulename{US:SplitLeft},
								$\vstysplit{\vsty}{\vsprod{\vsty''}{\isav}{\vsty_4}{\isav}}{\vsprod{\vsty_3''''}{\isav}{\vsty_4}{\isunav}}$.
							By \rulename{US:Subtype},
								$\vstysplit{\vsty}{\vsprod{\vsty''}{\isav}{\vsty_4}{\isav}}{\vsty_1''}$.
							By \rulename{US:SplitLeft},
								$\vstysplit{\vsprod{\vsty''}{\isav}{\vsty_4}{\isav}}{\vsprod{\vsty_3'''}{\isav}{\vsty_4}{\isav}}{\vsty_2}$.
							Apply \rulename{US:Subtype}.
						\item \splitLfullprodexstB{\vsty_3'}{\vsty_3'''}{\vsty_3''''}{\vsty_4}{\vsty_1'}{\vsty_1''}.
							By symmetry with the previous case.
						\item \splitRfullprodexstA{\vsty_4}{\vsty_4'}{\vsty_4''}{\vsty_3'}{\vsty_1'}{\vsty_1''}.
							By \rulename{US:SplitRight},
								$\vstysplit{\vsty}{\vsprod{\vsty_3}{\isav}{\vsty_4'}{\isav}}{\vsprod{\vsty_3}{\isunav}{\vsty_4''}{\isav}}$.
							By \rulename{UT:ProdLeft},
								$\vstysubt{\vsprod{\vsty_3}{\isunav}{\vsty_4''}{\isav}}{\vsprod{\vsty_3'}{\isunav}{\vsty_4''}{\isav}}$.
							By two applications of \rulename{US:Subtype},
								$\vstysplit{\vsty}{\vsprod{\vsty_3}{\isav}{\vsty_4'}{\isav}}{\vsty_1''}$.
							By \rulename{US:SplitLeft},
								$\vstysplit{\vsprod{\vsty_3}{\isav}{\vsty_4'}{\isav}}{\vsprod{\vsty_3'}{\isav}{\vsty_4'}{\isav}}{\vsty_2}$.
							Apply \rulename{US:Subtype}.
						\item \splitRfullprodexstB{\vsty_4}{\vsty_4'}{\vsty_4''}{\vsty_3'}{\vsty_1'}{\vsty_1''}.
							By symmetry with the previous case.
						\item \splitbothprodexst{\vsty_3'}{\vsty_3'''}{\vsty_3''''}{\vsty_4}{\vsty_4'}{\vsty_4''}{\vsty_1'}{\vsty_1''}.
							By induction,
								\vstysplitexst{\vsty''}{\vsty_3}{\vsty_3'''}{\vsty_3''''}{\vsty_3''}.
							By \rulename{US:SplitBoth},
								$\vstysplit{\vsty}{\vsprod{\vsty''}{\isav}{\vsty_4'}{\isav}}{\vsprod{\vsty_3''''}{\isav}{\vsty_4''}{\isav}}$.
							By \rulename{US:Subtype},
								$\vstysplit{\vsty}{\vsprod{\vsty''}{\isav}{\vsty_4'}{\isav}}{\vsty_1''}$.
							By \rulename{US:SplitLeft},
								$\vstysplit{\vsprod{\vsty''}{\isav}{\vsty_4'}{\isav}}
									{\vsprod{\vsty_3'''}{\isav}{\vsty_4'}{\isav}}
									{\vsprod{\vsty_3''}{\isav}{\vsty_4'}{\isunav}}$.
							By \rulename{UT:ProdRight},
								$\vstysubt{\vsprod{\vsty_3''}{\isav}{\vsty_4'}{\isunav}}{\vsty_2}$.
							Apply \rulename{US:Subtype} twice.
					\end{enumerate}
			\end{enumerate}
		\item If~$\vstysplit{\vsty_2}{\vsty_2'}{\vsty_2''}$,
			by symmetry with the previous case when $\avail = \isunav$.

	\end{enumerate}

  \item \rulename{US:SplitRight}.
	By symmetry with the previous case.

  \item \rulename{US:SplitBoth}.
	Then $\vsty = \vsprod{\vsty_3}{\isav}{\vsty_4}{\isav}$
		and~$\vstysplit{\vsty_3}{\vsty_3'}{\vsty_3''}$
		and~$\vstysplit{\vsty_4}{\vsty_4'}{\vsty_4''}$
		and~$\vsty_1 = \vsprod{\vsty_3'}{\isav}{\vsty_4'}{\isav}$
		and~$\vsty_2 = \vsprod{\vsty_3''}{\isav}{\vsty_4''}{\isav}$.
	\begin{enumerate}
		\item If~$\vstysplit{\vsty_1}{\vsty_1'}{\vsty_1''}$,
			by Lemma \ref{lem:split-form}, there are seven cases:
			\begin{enumerate}
				\item $\vstysubt{\vsprod{\vsty_3'}{\isav}{\vsty_4'}{\isunav}}{\vsty_1'}$
						and $\vstysubt{\vsprod{\vsty_3'}{\isunav}{\vsty_4'}{\isav}}{\vsty_1''}$.
					By \rulename{US:Commutative},
						$\vstysplit{\vsty_4}{\vsty_4''}{\vsty_4'}$.
					By \rulename{US:SplitRight},
						$\vstysplit{\vsty}{\vsprod{\vsty_3}{\isav}{\vsty_4''}{\isav}}{\vsprod{\vsty_3}{\isunav}{\vsty_4'}{\isav}}$.
					By \rulename{UT:ProdLeft},
						$\vstysubt{\vsprod{\vsty_3}{\isunav}{\vsty_4'}{\isav}}{\vsprod{\vsty_3'}{\isunav}{\vsty_4'}{\isav}}$.
					By two applications of \rulename{US:Subtype},
						$\vstysplit{\vsty}{\vsprod{\vsty_3}{\isav}{\vsty_4''}{\isav}}{\vsty_1''}$.
					By \rulename{US:Commutative},
						$\vstysplit{\vsty_3}{\vsty_3''}{\vsty_3'}$.
					By \rulename{US:SplitLeft},
						$\vstysplit{\vsprod{\vsty_3}{\isav}{\vsty_4''}{\isav}}{\vsty_2}{\vsprod{\vsty_3'}{\isav}{\vsty_4''}{\isunav}}$.
					By \rulename{US:Commutative},
						$\vstysplit{\vsprod{\vsty_3}{\isav}{\vsty_4''}{\isav}}{\vsprod{\vsty_3'}{\isav}{\vsty_4''}{\isunav}}{\vsty_2}$.
					By \rulename{UT:ProdRight},
						$\vstysubt{\vsprod{\vsty_3'}{\isav}{\vsty_4''}{\isunav}}{\vsprod{\vsty_3'}{\isav}{\vsty_4'}{\isunav}}$.
					Apply \rulename{US:Subtype}.
				\item $\vstysubt{\vsprod{\vsty_3'}{\isunav}{\vsty_4'}{\isav}}{\vsty_1'}$
						and $\vstysubt{\vsprod{\vsty_3'}{\isav}{\vsty_4'}{\isunav}}{\vsty_1''}$.
					By symmetry with the previous case.
				\item \splitLfullprodexstA{\vsty_3'}{\vsty_3'''}{\vsty_3''''}{\vsty_4'}{\vsty_1'}{\vsty_1''}.
					By induction,
						\vstysplitexst{\vsty''}{\vsty_3}{\vsty_3'''}{\vsty_3''''}{\vsty_3''}.
					By \rulename{US:SplitLeft},
						$\vstysplit{\vsty}{\vsprod{\vsty''}{\isav}{\vsty_4}{\isav}}{\vsprod{\vsty_3''''}{\isav}{\vsty_4}{\isunav}}$.
					By \rulename{UT:ProdRight},
						$\vstysubt{\vsprod{\vsty_3''''}{\isav}{\vsty_4}{\isunav}}{\vsprod{\vsty_3''''}{\isav}{\vsty_4'}{\isunav}}$.
					By two applications of \rulename{US:Subtype},
						$\vstysplit{\vsty}{\vsprod{\vsty''}{\isav}{\vsty_4}{\isav}}{\vsty_1''}$.
					By \rulename{US:SplitBoth},
						$\vstysplit{\vsprod{\vsty''}{\isav}{\vsty_4}{\isav}}{\vsprod{\vsty_3'''}{\isav}{\vsty_4'}{\isav}}{\vsty_2}$.
					Apply \rulename{US:Subtype}.
				\item \splitLfullprodexstB{\vsty_3'}{\vsty_3'''}{\vsty_3''''}{\vsty_4'}{\vsty_1'}{\vsty_1''}.
					By symmetry with the previous case.
				\item \splitRfullprodexstA{\vsty_4'}{\vsty_4'''}{\vsty_4''''}{\vsty_3'}{\vsty_1'}{\vsty_1''}.
					By symmetry with the previous case.
				\item \splitRfullprodexstB{\vsty_4'}{\vsty_4'''}{\vsty_4''''}{\vsty_3'}{\vsty_1'}{\vsty_1''}.
					By symmetry with the previous case.
				\item \splitbothprodexst{\vsty_3'}{\vsty_3'''}{\vsty_3''''}{\vsty_4'}{\vsty_4'''}{\vsty_4''''}{\vsty_1'}{\vsty_1''}.
					By induction,
						\vstysplitexst{\vsty''}{\vsty_3}{\vsty_3'''}{\vsty_3''''}{\vsty_3''}
						and~\vstysplitexst{\vsty'''}{\vsty_4}{\vsty_4'''}{\vsty_4''''}{\vsty_4''}.
					By \rulename{US:SplitBoth},
						$\vstysplit{\vsty}{\vsprod{\vsty''}{\isav}{\vsty'''}{\isav}}{\vsprod{\vsty_3''''}{\isav}{\vsty_4''''}{\isav}}$.
					By \rulename{US:Subtype},
						$\vstysplit{\vsty}{\vsprod{\vsty''}{\isav}{\vsty'''}{\isav}}{\vsty_1''}$.
					By \rulename{US:SplitBoth},
						$\vstysplit{\vsprod{\vsty''}{\isav}{\vsty'''}{\isav}}{\vsprod{\vsty_3'''}{\isav}{\vsty_4'''}{\isav}}{\vsty_2}$.
					Apply \rulename{US:Subtype}.
			\end{enumerate}
		\item If~$\vstysplit{\vsty_2}{\vsty_2'}{\vsty_2''}$,
			by symmetry with the previous case.
	\end{enumerate}

  \item \rulename{US:Recursive}.
	Then $\vsty = \kwvsrec{\vstyvar}{\vsty''}$
		and~$\vstysplit{\vsub{\vsty''}{\kwvsrec{\vstyvar}{\vsty''}}{\vstyvar}}{\vsty_1}{\vsty_2}$.
	\begin{enumerate}
		\item If~$\vstysplit{\vsty_1}{\vsty_1'}{\vsty_1''}$,
			by induction,
				\vstysplitexst{\vsty'}{\vsub{\vsty''}{\kwvsrec{\vstyvar}{\vsty''}}{\vstyvar}}{\vsty_1'}{\vsty_1''}{\vsty_2}.
			Apply \rulename{US:Recursive}.
		\item If~$\vstysplit{\vsty_2}{\vsty_2'}{\vsty_2''}$,
			by symmetry with the previous case.
	\end{enumerate}

  \item \rulename{US:Corecursive}.
	By symmetry with the previous case.

  \item \rulename{US:Subtype}.
	Then $\vstysplit{\vsty}{\vsty_1'''}{\vsty_2}$
		and~$\vstysubt{\vsty_1'''}{\vsty_1}$.
	\begin{enumerate}
		\item If~$\vstysplit{\vsty_1}{\vsty_1'}{\vsty_1''}$,
			by Lemma \ref{lem:vsty-split-supertyping},
				$\vstysplit{\vsty_1'''}{\vsty_1'}{\vsty_1''}$.
			By induction.
		\item If~$\vstysplit{\vsty_2}{\vsty_2'}{\vsty_2''}$,
			by induction,
				\vstysplitexst{\vsty'}{\vsty}{\vsty_2'}{\vsty_2''}{\vsty_1'''}.
			By \rulename{US:Commutative} and \rulename{US:Subtype},
				$\vstysplit{\vsty'}{\vsty_1}{\vsty_2'}$.
			Apply \rulename{US:Commutative}.
	\end{enumerate}

  \item \rulename{US:Commutative}.
	By induction.
  \end{itemize}

\item By induction on the derivation of
  $\uspsplit{\uspctx}{\uspctx_1}{\uspctx_2}$.

  \begin{itemize}
  \item \rulename{OM:Empty}.
	Then $\uspctx = \uspctx_1 = \uspctx_2 = \ectx$.
	\begin{enumerate}
		\item If~$\uspsplit{\uspctx_1}{\uspctx_1'}{\uspctx_1''}$,
			by inversion on \rulename{OM:Empty},
				$\uspctx_1' = \uspctx_1'' = \ectx$.
			Let $\uspctx_3 = \ectx$.
			Apply \rulename{OM:Empty} twice.
		\item If~$\uspsplit{\uspctx_2}{\uspctx_2'}{\uspctx_2''}$,
			by symmetry with the previous case.
	\end{enumerate}

  \item \rulename{OM:Commutative}.
	By induction.

\item \rulename{OM:Var}.
	Then $\uspctx = \uspctx'', \hastype{\vvar}{\vsty}$
		and~$\uspctx_1 = \uspctx_3, \hastype{\vvar}{\vsty}$
		and~$\uspsplit{\uspctx''}{\uspctx_3}{\uspctx_2}$.
	\begin{enumerate}
		\item If~$\uspsplit{\uspctx_1}{\uspctx_1'}{\uspctx_1''}$,
			by Lemma \ref{lem:split-form}, there are four cases:
			\begin{enumerate}
				\item \splituspLexst{\uspctx_1'''}{\uspctx_3}{\uspctx_1'}{\uspctx_1''}{\vvar}{\vsty}.
					By induction,
						\uspsplitexst{\uspctx'''}{\uspctx''}{\uspctx_1'''}{\uspctx_1''}{\uspctx_2}.
					By \rulename{OM:Var},
						$\uspsplit{\uspctx}{\uspctx''', \hastype{\vvar}{\vsty}}{\uspctx_1''}$.
					By \rulename{OM:Var},
						$\uspsplit{\uspctx''', \hastype{\vvar}{\vsty}}{\uspctx_1'}{\uspctx_2}$.
				\item \splituspRexst{\uspctx_1''''}{\uspctx_3}{\uspctx_1'}{\uspctx_1''}{\vvar}{\vsty}.
					By symmetry with the previous case.
				\item \splituspbothexst{\uspctx_1'''}{\uspctx_1''''}{\uspctx_3}{\uspctx_1'}{\uspctx_1''}{\vvar}{\vsty}{\vsty_1}{\vsty_2}.
					By induction,
						\uspsplitexst{\uspctx'''}{\uspctx''}{\uspctx_1'''}{\uspctx_1''''}{\uspctx_2}.
					By \rulename{OM:VarTypeSplit},
						$\uspsplit{\uspctx}{\uspctx''', \hastype{\vvar}{\vsty_1}}{\uspctx_1''}$.
					By \rulename{OM:Var},
						$\uspsplit{\uspctx''', \hastype{\vvar}{\vsty}}{\uspctx_1'}{\uspctx_2}$.
				\item \splituspbothcrossexst{\uspctx_1'''}{\uspctx_1''''}{\uspctx_3}{\uspctx_1'}{\uspctx_1''}{\vvar}{\vsty}{\vsty_1}{\vsty_2}.
					By symmetry with the previous case.
			\end{enumerate}
		\item If~$\uspsplit{\uspctx_2}{\uspctx_2'}{\uspctx_2''}$,
			by induction,
				\uspsplitexst{\uspctx'''}{\uspctx''}{\uspctx_2'}{\uspctx_2''}{\uspctx_3}.
			By \rulename{OM:Var},
				$\uspsplit{\uspctx}{\uspctx''', \hastype{\vvar}{\vsty}}{\uspctx_2''}$.
			By \rulename{OM:Commutative} and \rulename{OM:Var},
						$\uspsplit{\uspctx''', \hastype{\vvar}{\vsty}}{\uspctx_1}{\uspctx_2'}$.
			Apply \rulename{OM:Commutative}.
	\end{enumerate}

\item \rulename{OM:Gen}.
	By symmetry with the previous case.

\item \rulename{OM:VarTypeSplit}.
	Then $\uspctx = \uspctx'', \hastype{\vvar}{\vsty}$
		and~$\vstysplit{\vsty}{\vsty_1}{\vsty_2}$
		and~$\uspctx_1 = \uspctx_3, \hastype{\vvar}{\vsty_1}$
		and~$\uspctx_2 = \uspctx_4, \hastype{\vvar}{\vsty_2}$
		and~$\uspsplit{\uspctx''}{\uspctx_3}{\uspctx_4}$.
	\begin{enumerate}
		\item If~$\uspsplit{\uspctx_1}{\uspctx_1'}{\uspctx_1''}$,
			by Lemma \ref{lem:split-form}, there are four cases:
			\begin{enumerate}
				\item \splituspLexst{\uspctx_1'''}{\uspctx_3}{\uspctx_1'}{\uspctx_1''}{\vvar}{\vsty_1}.
					By induction,
						\uspsplitexst{\uspctx'''}{\uspctx''}{\uspctx_1'''}{\uspctx_1''}{\uspctx_4}.
					By \rulename{OM:Var},
						$\uspsplit{\uspctx}{\uspctx''', \hastype{\vvar}{\vsty}}{\uspctx_1''}$.
					By \rulename{OM:VarTypeSplit},
						$\uspsplit{\uspctx''', \hastype{\vvar}{\vsty}}{\uspctx_1'}{\uspctx_2}$.
				\item \splituspRexst{\uspctx_1''''}{\uspctx_3}{\uspctx_1'}{\uspctx_1''}{\vvar}{\vsty}.
					By symmetry with the previous case.
				\item \splituspbothexst{\uspctx_1'''}{\uspctx_1''''}{\uspctx_3}{\uspctx_1'}{\uspctx_1''}{\vvar}{\vsty_1}{\vsty_1'}{\vsty_1''}.
					By part \ref{lem:split-move-vsty},
						\vstysplitexst{\vsty'}{\vsty}{\vsty_1'}{\vsty_1''}{\vsty_2}.
					By induction,
						\uspsplitexst{\uspctx'''}{\uspctx''}{\uspctx_1'''}{\uspctx_1''''}{\uspctx_4}.
					By \rulename{OM:VarTypeSplit},
						$\uspsplit{\uspctx}{\uspctx''', \hastype{\vvar}{\vsty'}}{\uspctx_1''}$.
					By \rulename{OM:VarTypeSplit},
						$\uspsplit{\uspctx''', \hastype{\vvar}{\vsty'}}{\uspctx_1'}{\uspctx_2}$.
				\item \splituspbothcrossexst{\uspctx_1'''}{\uspctx_1''''}{\uspctx_3}{\uspctx_1'}{\uspctx_1''}{\vvar}{\vsty_1}{\vsty_1'}{\vsty_1''}.
					By symmetry with the previous case.
			\end{enumerate}
		\item If~$\uspsplit{\uspctx_2}{\uspctx_2'}{\uspctx_2''}$,
			by symmetry with the previous case.
	\end{enumerate}

\item \rulename{OM:GenTypeSplit}.
	By symmetry with the previous case.
  \end{itemize}

%	If~$\uspsplit{\uspctx}{\uspctx_1}{\uspctx_2}$
%	and~$\uspsplit{\uspctx_1}{\uspctx_1'}{\uspctx_1''}$
%	and~$\uspsplit{\uspctx_2}{\uspctx_2'}{\uspctx_2''}$
%	then there exists an $\uspctx'$ and $\uspctx''$
%	such that~$\uspsplit{\uspctx}{\uspctx_3}{\uspctx_4}$
%	and~$\uspsplit{\uspctx'}{\uspctx_1'}{\uspctx_2'}$
%	and~$\uspsplit{\uspctx''}{\uspctx_1''}{\uspctx_2''}$.
\item
	By part \ref{lem:split-move-usp},
		\uspsplitexst{\uspctx'''}{\uspctx}{\uspctx_1'}{\uspctx_1''}{\uspctx_2}.
	By part \ref{lem:split-move-usp},
		\uspsplitexst{\uspctx'}{\uspctx'''}{\uspctx_2'}{\uspctx_2''}{\uspctx_1'}.
	By \rulename{OM:Commutative},
		$\uspsplit{\uspctx'}{\uspctx_1'}{\uspctx_2'}$.
	By \rulename{OM:Commutative},
		$\uspsplit{\uspctx'''}{\uspctx_2''}{\uspctx'}$.
	By part \ref{lem:split-move-usp},
		\uspsplitexst{\uspctx''}{\uspctx}{\uspctx_2''}{\uspctx'}{\uspctx_1''}.
	By \rulename{OM:Commutative},
		$\uspsplit{\uspctx}{\uspctx'}{\uspctx''}$.
	By \rulename{OM:Commutative},
		$\uspsplit{\uspctx''}{\uspctx_1''}{\uspctx_2''}$.

  \end{enumerate}
\end{proof}

With the help of Lemmas \ref{lem:vsty-split-subtypes} and \ref{lem:usp-split-upstream}, Lemma \ref{lem:usp-split-weakening} proves that $\uspctx$ context splitting maintains weakening in VS typing: if $\uspctx$ splits into $\uspctx'$ and $\vs$ is well-typed under $\uspctx'$, then $\vs$ is also well-typed under $\uspctx$. The notable complication is a result of rule \rulename{OM:VarTypeSplit}: a variable $\vvar$ may be mapped to $\vsty$ in $\uspctx$ while being mapped to $\vsty'$ in $\uspctx'$ if $\vsty$ splits into $\vsty'$. In order for $\uspctx'$ to be weakened to $\uspctx$, it must be valid for $\vvar$ to have type $\vsty$ when it was expected to have type $\vsty'$; in other words, $\vsty$ must be a subtype of $\vsty'$. This is accomplished by Lemma \ref{lem:vsty-split-subtypes}, which shows that VS types always split into supertypes. Lemma \ref{lem:usp-split-upstream} acts as intermidiate between \ref{lem:vsty-split-subtypes} and \ref{lem:usp-split-weakening}, stating that if $\uspctx$ splits into $\uspctx'$ and a variable $\vvar$ is mapped to $\vsty$ in $\uspctx'$, then $\vvar$ is mapped to a subtype of $\vsty$ in $\uspctx$. An additional consequence of Lemma \ref{lem:usp-split-upstream} is that $\uspctx$ context splitting cannot introduce new domain elements: if $\uspctx$ splits into $\uspctx'$, then only variables and generators within $\uspctx$ may appear within $\uspctx'$.
%\fr{I am now realizing that Lemma 6 implies Lemma 8. Before I remove Lemma 8, could someone double-check?}
%\skm{I believe Lemma 8 is (implied by) the contrapositive of Lemma 6. Constructively, you need to be careful when trying to use the contrapositives of things (they don't hold in general) but in this case it should be OK.}

\begin{lemma}\label{lem:vsty-split-subtypes}
	If $\vstysplit{\vsty}{\vsty_1}{\vsty_2}$
	then~$\vstysubt{\vsty}{\vsty_1}$
	and~$\vstysubt{\vsty}{\vsty_2}$.
\end{lemma}

\begin{proof}
	By induction on the derivation of $\vstysplit{\vsty}{\vsty_1}{\vsty_2}$.
	\begin{itemize}
  	\item \rulename{US:Prod}.
		Then $\vsty = \vsprod{\vsty_1'}{\isav}{\vsty_2'}{\isav}$
			and $\vsty_1 = \vsprod{\vsty_1'}{\isav}{\vsty_2'}{\isunav}$
			and $\vsty_2 = \vsprod{\vsty_1'}{\isunav}{\vsty_2'}{\isav}$.
		Apply \rulename{UT:ProdLeft} and \rulename{UT:ProdRight}.

  	\item \rulename{US:SplitBoth}.
		Then $\vsty = \vsprod{\vsty_1'}{\isav}{\vsty_2'}{\isav}$
			and $\vsty_1 = \vsprod{\vsty_1''}{\isav}{\vsty_2''}{\isav}$
			and $\vsty_2 = \vsprod{\vsty_1'''}{\isav}{\vsty_2'''}{\isav}$
			and $\vstysplit{\vsty_1'}{\vsty_1''}{\vsty_1'''}$
			and $\vstysplit{\vsty_2'}{\vsty_2''}{\vsty_2'''}$.
		By induction,
			$\vstysubt{\vsty_1'}{\vsty_1''}$
			and $\vstysubt{\vsty_1'}{\vsty_1'''}$
			and $\vstysubt{\vsty_2'}{\vsty_2''}$
			and $\vstysubt{\vsty_2'}{\vsty_2'''}$.
		Apply \rulename{UT:Prod} twice.

  	\item \rulename{US:SplitLeft}.
		Then $\vsty = \vsprod{\vsty_1'}{\isav}{\vsty_2'}{\avail}$
			and $\vsty_1 = \vsprod{\vsty_1''}{\isav}{\vsty_2'}{\avail}$
			and $\vsty_2 = \vsprod{\vsty_1'''}{\isav}{\vsty_2'}{\isunav}$
			and $\vstysplit{\vsty_1'}{\vsty_1''}{\vsty_1'''}$.
		By induction,
			$\vstysubt{\vsty_1'}{\vsty_1''}$
			and $\vstysubt{\vsty_1'}{\vsty_1'''}$.
		By \rulename{UT:Prod},
			$\vstysubt{\vsty}{\vsty_1}$
			and $\vstysubt{\vsty}{\vsprod{\vsty_1'''}{\isav}{\vsty_2'}{\avail}}$.
		By \rulename{UT:ProdRight},
			$\vstysubt{\vsprod{\vsty_1'''}{\isav}{\vsty_2'}{\avail}}{\vsty_2}$.
		Apply \rulename{UT:Transitive}.

	\item \rulename{US:SplitRight}.
		By symmetry with the previous case.

	\item \rulename{US:Corecursive}.
		Then $\vsty = \kwvscorec{\vstyvar}{\vsty'}$
			and $\vstysplit{\vsub{\vsty'}{\kwvscorec{\vstyvar}{\vsty'}}{\vstyvar}}{\vsty_1}{\vsty_2}$.
		By induction,\\
			$\vstysubt{\vsub{\vsty'}{\kwvscorec{\vstyvar}{\vsty'}}{\vstyvar}}{\vsty_1}$
			and $\vstysubt{\vsub{\vsty'}{\kwvscorec{\vstyvar}{\vsty'}}{\vstyvar}}{\vsty_2}$.
		By \rulename{UT:Rec1},
			$\vstysubt{\vsty}{\vsub{\vsty'}{\kwvscorec{\vstyvar}{\vsty'}}{\vstyvar}}$.
		Apply \rulename{UT:Transitive} twice.

	\item \rulename{US:Subtype}.
		Then $\vstysubt{\vsty_1'}{\vsty_1}$.
			and $\vstysplit{\vsty}{\vsty_1'}{\vsty_2}$.
		By induction,
			$\vstysubt{\vsty}{\vsty_1'}$
			and $\vstysubt{\vsty}{\vsty_2}$.
		Apply \rulename{UT:Transitive}.

	\item \rulename{US:Commutative}.
		By induction.
	\end{itemize}

\end{proof}

\begin{lemma}\label{lem:usp-split-upstream}\strut
	If $\uspsplit{\uspctx}{\uspctx_1}{\uspctx_2}$ then:
	\begin{enumerate}
		\item If~$\hastype{\vvar}{\vsty} \in \uspctx_1$
			then~$\hastype{\vvar}{\vsty'} \in \uspctx$
			and~$\vstysubt{\vsty'}{\vsty}$.
		\item If~$\hastype{\vvar}{\vsty} \in \uspctx_2$
			then~$\hastype{\vvar}{\vsty'} \in \uspctx$
			and~$\vstysubt{\vsty'}{\vsty}$.
		\item If~$\hastype{\genseed}{\vsty} \in \uspctx_1$
			then~$\hastype{\genseed}{\vsty'} \in \uspctx$
			and~$\vstysubt{\vsty'}{\vsty}$.
		\item If~$\hastype{\genseed}{\vsty} \in \uspctx_2$
			then~$\hastype{\genseed}{\vsty'} \in \uspctx$
			and~$\vstysubt{\vsty'}{\vsty}$.
	\end{enumerate}
\end{lemma}

\begin{proof}
	By induction on the derivation of $\uspsplit{\uspctx}{\uspctx_1}{\uspctx_2}$.
	\begin{itemize}
  	\item \rulename{OM:Commutative}.
		By induction.

	\item \rulename{OM:Var}.
		Then $\uspctx = \uspctx', \hastype{\vvar'}{\vsty''}$
			and $\uspctx_1 = \uspctx_1', \hastype{\vvar'}{\vsty''}$
			and $\uspsplit{\uspctx'}{\uspctx_1'}{\uspctx_2}$.
		\begin{enumerate}
			\item $\hastype{\vvar}{\vsty} \in \uspctx_1$.
				There are two cases:
					\begin{enumerate}
					\item $\vvar = \vvar'$.
						Then $\vsty = \vsty''$
							and $\hastype{\vvar}{\vsty} \in \uspctx$.
					\item $\vvar \neq \vvar'$.
						Then $\hastype{\vvar}{\vsty} \in \uspctx_1'$.
						By induction, $\hastype{\vvar}{\vsty'} \in \uspctx'$
							and~$\vstysubt{\vsty'}{\vsty}$.
						Therefore, $\hastype{\vvar}{\vsty'} \in \uspctx$.
					\end{enumerate}
			\item $\hastype{\vvar}{\vsty} \in \uspctx_2$.
				By induction, $\hastype{\vvar}{\vsty'} \in \uspctx'$
					and~$\vstysubt{\vsty'}{\vsty}$.
				Therefore, $\hastype{\vvar}{\vsty'} \in \uspctx$.
			\item $\hastype{\genseed}{\vsty} \in \uspctx_1$.
				Then $\hastype{\genseed}{\vsty} \in \uspctx_1'$.
				By induction, $\hastype{\genseed}{\vsty'} \in \uspctx'$
					and~$\vstysubt{\vsty'}{\vsty}$.
				Therefore, $\hastype{\genseed}{\vsty'} \in \uspctx$.
			\item $\hastype{\genseed}{\vsty} \in \uspctx_2$.
				By induction, $\hastype{\genseed}{\vsty'} \in \uspctx'$
					and~$\vstysubt{\vsty'}{\vsty}$.
				Therefore, $\hastype{\genseed}{\vsty'} \in \uspctx$.
		\end{enumerate}

	\item \rulename{OM:Gen}.
		By symmetry with the previous case.

	\item \rulename{OM:VarTypeSplit}.
		Then $\uspctx = \uspctx', \hastype{\vvar'}{\vsty''}$
			and $\vstysplit{\vsty''}{\vsty_1}{\vsty_2}$.
			and $\uspctx_1 = \uspctx_1', \hastype{\vvar'}{\vsty_1}$.
			and $\uspctx_2 = \uspctx_2', \hastype{\vvar'}{\vsty_2}$.
		\begin{enumerate}
			\item $\hastype{\vvar}{\vsty} \in \uspctx_1$.
				There are two cases:
					\begin{enumerate}
					\item $\vvar = \vvar'$.
						Then $\vsty = \vsty_1$
							and $\hastype{\vvar}{\vsty''} \in \uspctx$.
						By Lemma \ref{lem:vsty-split-subtypes},
							$\vstysubt{\vsty''}{\vsty}$.
					\item $\vvar \neq \vvar'$.
						Then $\hastype{\vvar}{\vsty} \in \uspctx_1'$.
						By induction, $\hastype{\vvar}{\vsty'} \in \uspctx'$
							and~$\vstysubt{\vsty'}{\vsty}$.
						Therefore, $\hastype{\vvar}{\vsty'} \in \uspctx$.
					\end{enumerate}
			\item $\hastype{\vvar}{\vsty} \in \uspctx_2$.
				By symmetry with the previous case.
			\item $\hastype{\genseed}{\vsty} \in \uspctx_1$.
				Then $\hastype{\genseed}{\vsty} \in \uspctx_1'$.
				By induction, $\hastype{\genseed}{\vsty'} \in \uspctx'$
					and~$\vstysubt{\vsty'}{\vsty}$.
				Therefore, $\hastype{\genseed}{\vsty'} \in \uspctx$.
			\item $\hastype{\genseed}{\vsty} \in \uspctx_2$.
				By symmetry with the previous case.
		\end{enumerate}

	\item \rulename{OM:GenTypeSplit}.
		By symmetry with the previous case.
	\end{itemize}

\end{proof}

\begin{lemma}\label{lem:usp-split-weakening}%\strut
%	\begin{enumerate}\item
	If $\uspsplit{\uspctx}{\uspctx_1}{\uspctx_0}$
		and~$\vsistype{\uspctx_1}{\utctx}{\vs}{\vsty}$
		then~$\vsistype{\uspctx}{\utctx}{\vs}{\vsty}$.
%		\label{lem:usp-split-weakening-vs}
%	\item If $\uspsplit{\uspctx}{\uspctx_1}{\uspctx_0}$
%		and~$\tywithdag{\gctx}{\uspctx_1}{\utctx}{\ctx}{e}{\tau}{\graph}$
%		then~$\tywithdag{\gctx}{\uspctx}{\utctx}{\ctx}{e}{\tau}{\graph}$.
%		\label{lem:usp-split-weakening-exp}
%	\end{enumerate}
\end{lemma}

\begin{proof}%\strut
%  \begin{enumerate}\item
	By induction on the derivation of $\vsistype{\uspctx_1}{\utctx}{\vs}{\vsty}$.
	\begin{itemize}
  	\item \rulename{U:OmegaVar}.
		Then $\vs = \vvar$
			and $\uspctx_1 = \uspctx_1', \hastype{\vvar}{\vsty}$.
		By Lemma \ref{lem:usp-split-upstream},
			$\hastype{\vvar}{\vsty'} \in \uspctx$
			and~$\vstysubt{\vsty'}{\vsty}$.
		By \rulename{U:OmegaVar},
			$\vsistype{\uspctx}{\utctx}{\vs}{\vsty'}$.
		Apply \rulename{U:Subtype}.

	\item \rulename{U:OmegaSeed}.
		By symmetry with the previous case.

  	\item \rulename{U:PsiVar}.
		Apply \rulename{U:PsiVar}.

	\item \rulename{U:PsiSeed}.
		By symmetry with the previous case.

  	\item \rulename{U:Pair}.
		Then $\vs = \kwpair{\vs_1}{\vs_2}$
			and $\vsty = \vsprod{\vsty_1}{\isav}{\vsty_2}{\isav}$
			and $\uspsplit{\uspctx_1}{\uspctx_1'}{\uspctx_1''}$
			and $\vsistype{\uspctx_1'}{\utctx}{\vs_1}{\vsty_1}$
			and $\vsistype{\uspctx_1''}{\utctx}{\vs_2}{\vsty_2}$.
		By Lemma \ref{lem:split-move},
			\uspsplitexst{\uspctx'}{\uspctx}{\uspctx_1'}{\uspctx_1''}{\uspctx_0}.
		By induction,
			$\vsistype{\uspctx'}{\utctx}{\vs_1}{\vsty_1}$.
		Apply \rulename{U:Pair}.

	\item All other cases are by induction and self-application.
	\end{itemize}

\end{proof}

\input{fig-dag-wf}
\input{fig-vert-equiv}
\input{fig-con-eval}

\begin{figure}
  \centering
  \def \MathparLineskip {\lineskip=0.43cm}
  \begin{mathpar}
    \Rule{DW:Par}
         {\uspsplit{\uspctx}{\uspctx_1}{\uspctx_2}\\
		   \dagwf{\gctx}{\uspctx_1}{\utctx}{\graph_1}{\kgraph}\\
           \dagwf{\gctx}{\uspctx_2}{\utctx}{\graph_2}{\kgraph}
         }
         {\dagwf{\gctx}{\uspctx}{\utctx}
           {\graph_1 \parcomp \graph_2}{\kgraph}}
    \and
    \Rule{C:Par}
         {\eval{\ppts}{e_1}{v_1}{\cgraph_1}\\
           \eval{\ppts}{e_2}{v_2}{\cgraph_2}}
         {\eval{\ppts}{\kwpar{e_1}{e_2}}{\kwpair{v_1}{v_2}}
           {\cgraph_1 \parcomp \cgraph_2}}
    \and
    \Rule{S:Par}
         {\uspsplit{\uspctx}{\uspctx_1}{\uspctx_2}\\
		   \tywithdag{\gctx}{\uspctx_1}{\utctx}{\ctx}{e_1}{\tau_1}{\graph_1}\\
           \tywithdag{\gctx}{\uspctx_2}{\utctx}{\ctx}{e_2}{\tau_2}{\graph_2}}
         {\tywithdag{\gctx}{\uspctx}{\utctx}{\ctx}{\kwpar{e_1}{e_2}}
           {\kwprod{\tau_1}{\tau_2}}
           {\graph_1 \parcomp \graph_2}}
  \end{mathpar}

  \begin{mathpar}
  \begin{array}{l r l l}
\\
    \mathit{Graph~Types} & \graph & \bnfdef &
    \dots \bnfalt
    \graph \parcomp \graph
    \\
    
    \mathit{Expressions} & e & \bnfdef &
    \dots \bnfalt
    \kwpar{e}{e}
  \end{array}
  \end{mathpar}

  \[
  \begin{array}{r l l l}
    \bnnf{\graph_1 \parcomp \graph_2} & \defeq &
    \bnnf{\graph_1} \parcomp \bnnf{\graph_2}
	\\
  \getgraphs{\graph_1 \parcomp \graph_2} & \defeq &
  \{\cgraph_1' \parcomp \cgraph_2' \mid
  \cgraph_1' \in \getgraphs{\graph_1}, \cgraph_2' \in \getgraphs{\graph_2} &
	\\
  \dagq{\vertices_1}{\edges_1}{\startv_1}{\sinkv_1} \parcomp
  \dagq{\vertices_2}{\edges_2}{\startv_2}{\sinkv_2} & \defeq &
  (\vertices_1 \cup \vertices_2 \cup \{\vertex_1, \vertex_2\}, \edges_1 \cup \edges_2 \cup
    & \vertex_1, \vertex_2 \fresh,\\
    & & ~~\{(\vertex_1, \startv_1), (\vertex_1, \startv_2),
         (\sinkv_1, \vertex_2), (\sinkv_2, \vertex_2)\}, \vertex_1, \vertex_2)
         & \vertices_1 \cap \vertices_2 = \emptyset
       \\
  \end{array}
  \]
  \caption{Adding fork-join parallelism to {\langname}.}
  \label{fig:fork-join-par}
\end{figure}

Figures \ref{fig:dag-wf}, \ref{fig:vert-equiv}, and \ref{fig:constructor evaluation} give the full rules for graph type formation, 
VS equivalence, and type constructor equivalence respectively.
Figure \ref{fig:fork-join-par} adds the syntax and rules for a fork-join parallel expression and graph type.
Fork-join parallellism in $\langname$ functions similar to how it does in \citet{Muller22}.
The expression $\kwpar{e_1}{e_2}$ executes $e_1$ and $e_2$ in parallel, and this expression
is given the graph type $\graph_1 \parcomp \graph_2$ where $e_1$ has graph type $\graph_1$
and $e_2$ has graph type $\graph_2$.

Lemma~\ref{lem:subst} shows that well-formed VSs, expressions, type
constructors, and graph types maintain their types and kinds when substituting
any free variables appropriately. In addition, substituting variables in
equivalent VSs and type constructors maintains their equivalence. When
substituting VS variables found in the affine context $\uspctx$, we require that
the term undergoing the substitution and the VS being substituted
 are well-typed under two disjoint contexts $\uspctx_1$ and
$\uspctx_2$, and the resulting term is well-typed under their conjoinment (any
$\uspctx$ where $\uspsplit{\uspctx}{\uspctx_1}{\uspctx_2}$). When substituting a
VS variable $\vvar$ with a VS $\vs$ in expressions and graph types, there are
two distinct cases: $\vvar$ is in both $\uspctx$ and $\utctx$ and $\vs$ is
well-typed under a disjoint $\uspctx'$ and separately well-typed under $\utctx$,
or $\vvar$ is only in $\utctx$ and $\vs$ is well-typed under $\utctx$ (when
assigning types to expressions and graph kinds to graph types, it never occurs
that a variable is only in $\uspctx$). We make these cases separate so that
substituting variables in $\utctx$ does not cause vertices only found within
$\utctx$ to appear within future spawns, but variables in $\utctx$ that do not
appear within future spawns can be more freely substituted for. Since we always
require that VSs are either typed under an empty $\uspctx$ or an empty $\utctx$
in the rules for other judgements,
we have two cases for substituting VS variables that reflect that: either the
variable is in $\uspctx$ and $\utctx$ is empty, or vice-versa. 

This substitution lemma requires Lemmas \ref{lem:vert-split-typing} and \ref{lem:weaken-affine-restriction}. 
Lemma \ref{lem:vert-split-typing} states that if a VS
$\vs$ is given a type $\vsty$ under some $\uspctx$, and $\vsty$
splits into $\vsty_1$ and $\vsty_2$, then $\uspctx$ must split into some
$\uspctx_1$ and $\uspctx_2$ where $\vs$ has type $\vsty_1$ under $\uspctx_1$ and
has type $\vsty_2$ under $\uspctx_2$.
Lemma \ref{lem:weaken-affine-restriction} states that the affine restriction can be removed 
from a VS typing judgement and that judgement still holds.

\begin{lemma}\label{lem:vert-split-typing}
  If~$\vsistype{\uspctx}{\ectx}{\vs}{\vsty}$
  and~$\vstysplit{\vsty}{\vsty_1}{\vsty_2}$
  then there exists an $\uspctx_1$ and $\uspctx_2$
  such that~$\uspsplit{\uspctx}{\uspctx_1}{\uspctx_2}$
  and~$\vsistype{\uspctx_1}{\ectx}{\vs}{\vsty_1}$
  and~$\vsistype{\uspctx_2}{\ectx}{\vs}{\vsty_2}$.
\end{lemma}
\begin{proof}
  By induction on the derivation of
  $\vsistype{\uspctx}{\ectx}{\vs}{\vsty}$.
  \iffull
  We prove some representative cases.
  \begin{itemize}
  \item \rulename{U:OmegaVar}.
	Then $\vs = \vvar$
		and~$\uspctx = \uspctx', \hastype{\vvar}{\vsty}$.
	Since $\uspsplit{\uspctx'}{\uspctx'}{\ectx}$,
		by \rulename{OM:TypeSplit},
		$\uspsplit{\uspctx}{\uspctx', \hastype{\vvar}{\vsty_1}}{\hastype{\vvar}{\vsty_2}}$.
	Apply \rulename{U:Omega} twice.

  \item \rulename{U:Pair}.
	Then $\vs = \kwpair{\vs_1}{\vs_2}$
		and~$\vsty = \vsprod{\vsty_3}{\isav}{\vsty_4}{\isav}$
		and~$\uspsplit{\uspctx}{\uspctx_3}{\uspctx_4}$
		and~$\vsistype{\uspctx_3}{\ectx}{\vs_1}{\vsty_3}$
		and~$\vsistype{\uspctx_4}{\ectx}{\vs_2}{\vsty_4}$.
	By Lemma \ref{lem:split-form}, there are seven cases:
	\begin{enumerate}
		\item $\vstysubt{\vsprod{\vsty_3}{\isav}{\vsty_4}{\isunav}}{\vsty_1}$
				and $\vstysubt{\vsprod{\vsty_3}{\isunav}{\vsty_4}{\isav}}{\vsty_2}$.
			By \rulename{U:OnlyLeftPair},
				$\vsistype{\uspctx_3}{\ectx}{\vs}{\vsprod{\vsty_3}{\isav}{\vsty_4}{\isunav}}$.
			By \rulename{U:OnlyRightPair},
				$\vsistype{\uspctx_4}{\ectx}{\vs}{\vsprod{\vsty_3}{\isunav}{\vsty_4}{\isav}}$.
			Apply \rulename{U:Subtype} twice.
		\item $\vstysubt{\vsprod{\vsty_3}{\isunav}{\vsty_4}{\isav}}{\vsty_1}$
				and $\vstysubt{\vsprod{\vsty_3}{\isav}{\vsty_4}{\isunav}}{\vsty_2}$.
			By symmetry with the previous case.
		\item \splitLfullprodexstA{\vsty_3}{\vsty_3'}{\vsty_3''}{\vsty_4}{\vsty_1}{\vsty_2}.
			By induction,
				\vstytouspsplitext{\uspctx_3}{\uspctx_3'}{\uspctx_3''}{\vs_1}{\vsty_3'}{\vsty_3''}.
			By Lemma \ref{lem:split-move},
				\uspsplitexst{\uspctx'}{\uspctx}{\uspctx_3'}{\uspctx_3''}{\uspctx_4}.
			By \rulename{U:Pair},
				$\vsistype{\uspctx'}{\ectx}{\vs}{\vsprod{\vsty_3'}{\isav}{\vsty_4}{\isav}}$.
			Apply \rulename{U:Subtype}.
			By \rulename{U:OnlyLeftPair},
				$\vsistype{\uspctx_3''}{\ectx}{\vs}{\vsprod{\vsty_3''}{\isav}{\vsty_4}{\isunav}}$.
			Apply \rulename{U:Subtype}.
		\item \splitLfullprodexstB{\vsty_3}{\vsty_3'}{\vsty_3''}{\vsty_4}{\vsty_1}{\vsty_2}.
			By symmetry with the previous case.
		\item \splitRfullprodexstA{\vsty_4}{\vsty_4'}{\vsty_4''}{\vsty_3}{\vsty_1}{\vsty_2}.
			By symmetry with the previous case.
		\item \splitRfullprodexstB{\vsty_4}{\vsty_4'}{\vsty_4''}{\vsty_3}{\vsty_1}{\vsty_2}.
			By symmetry with the previous case.
		\item \splitbothprodexst{\vsty_3}{\vsty_3'}{\vsty_3''}{\vsty_4}{\vsty_4'}{\vsty_4''}{\vsty_1}{\vsty_2}.
			By induction,
				\vstytouspsplitext{\uspctx_3}{\uspctx_3'}{\uspctx_3''}{\vs_1}{\vsty_3'}{\vsty_3''}
				and \vstytouspsplitext{\uspctx_4}{\uspctx_4'}{\uspctx_4''}{\vs_2}{\vsty_4'}{\vsty_4''}.
			By Lemma \ref{lem:split-move},
				\uspsplitbothexst{\uspctx'}{\uspctx''}{\uspctx}{\uspctx_3'}{\uspctx_3''}{\uspctx_4'}{\uspctx_4''}.
			By \rulename{U:Pair},
				$\vsistype{\uspctx'}{\ectx}{\vs}{\vsprod{\vsty_3'}{\isav}{\vsty_4'}{\isav}}$.
			Apply \rulename{U:Subtype}.
			By \rulename{U:Pair},
				$\vsistype{\uspctx''}{\ectx}{\vs}{\vsprod{\vsty_3''}{\isav}{\vsty_4''}{\isav}}$.
			Apply \rulename{U:Subtype}.
	\end{enumerate}

  \item \rulename{U:OnlyLeftPair}.
	Then $\vs = \kwpair{\vs_1}{\vs_2}$
		and~$\vsty = \vsprod{\vsty_3}{\isav}{\vsty_4}{\isunav}$
		and~$\vsistype{\uspctx}{\ectx}{\vs_1}{\vsty_3}$
		and~$\vsistype{\ectx}{\utctx}{\vs_2}{\vsty_4}$.
	By Lemma \ref{lem:split-form},
		\splitLprodexst{\vsty_3}{\vsty_3'}{\vsty_3''}{\vsty_4}{\vsty_1}{\vsty_2}.
	By induction,
		\vstytouspsplitext{\uspctx}{\uspctx_1}{\uspctx_2}{\vs}{\vsty_3'}{\vsty_3''}.
	By \rulename{U:OnlyLeftPair},
		$\vsistype{\uspctx_1}{\ectx}{\vs}{\vsprod{\vsty_3'}{\isav}{\vsty_4}{\isunav}}$.
	Apply \rulename{U:Subtype}.
	By \rulename{U:OnlyLeftPair},
		$\vsistype{\uspctx_2}{\ectx}{\vs}{\vsprod{\vsty_3''}{\isav}{\vsty_4}{\isunav}}$.
	Apply \rulename{U:Subtype}.

  \item \rulename{U:Fst}.
	Then $\vs = \kwfst{\vs'}$
		and~$\vsistype{\uspctx}{\ectx}{\vs'}{\vsprod{\vsty}{\isav}{\vsty_3}{\avail}}$.
	By \rulename{US:SplitLeft},
		$\vstysplit{\vsprod{\vsty}{\isav}{\vsty_3}{\avail}}
			{\vsprod{\vsty_1}{\isav}{\vsty_3}{\avail}}{\vsprod{\vsty_2}{\isav}{\vsty_3}{\isunav}}$.
	By induction,
		\vstytouspsplitext{\uspctx}{\uspctx_1}{\uspctx_2}{\vs'}
			{\vsprod{\vsty_1}{\isav}{\vsty_3}{\avail}}{\vsprod{\vsty_2}{\isav}{\vsty_3}{\isunav}}.
	Apply \rulename{U:Fst} twice.

  \item \rulename{U:Subtype}.
	Then $\vstysubt{\vsty'}{\vsty}$
		and~$\vsistype{\uspctx}{\ectx}{\vs}{\vsty'}$.
	By Lemma \ref{lem:vsty-split-supertyping},
		$\vstysplit{\vsty'}{\vsty_1}{\vsty_2}$.
	Apply induction.
  \end{itemize}
  \fi
\end{proof}

%\begin{lemma}\label{lem:usp-split-notin}\strut
%	\begin{enumerate}
%	\item If $\uspsplit{\uspctx}{\uspctx_1}{\uspctx_2}$
%		and~$\vvar \notin \uspctx$
%		then~$\vvar \notin \uspctx_1$
%		and~$\vvar \notin \uspctx_2$.
%		\label{lem:usp-split-notin-var}
%	\item If $\uspsplit{\uspctx}{\uspctx_1}{\uspctx_2}$
%		and~$\genseed \notin \uspctx$
%		then~$\genseed \notin \uspctx_1$
%		and~$\genseed \notin \uspctx_2$.
%		\label{lem:usp-split-notin-seed}
%	\end{enumerate}
%\end{lemma}
%\iffull
%\begin{proof}\strut
%  \begin{enumerate}
%	\item By induction on the derivation of $\uspsplit{\uspctx}{\uspctx_1}{\uspctx_2}$.
%	\begin{itemize}
%  	\item \rulename{OM:Empty}.
%		True since $\uspctx_1 = \uspctx_2 = \ectx$.
%
%  	\item \rulename{OM:Commutative}.
%		By induction.
%
%	\item \rulename{OM:Var}.
%		Then $\uspctx = \uspctx', \hastype{\vvar'}{\vsty}$
%			and $\uspctx_1 = \uspctx_1', \hastype{\vvar'}{\vsty}$
%			and $\uspsplit{\uspctx'}{\uspctx_1'}{\uspctx_2}$.
%		By induction,
%			$\vvar \notin \uspctx_1'$
%			and~$\vvar \notin \uspctx_2$.
%		Since $\vvar \neq \vvar'$,
%			$\vvar \notin \uspctx_1$.
%
%	\item \rulename{OM:Gen}, \rulename{OM:VarTypeSplit}, \rulename{OM:GenTypeSplit}.
%		Similar to the previous case.
%	\end{itemize}
%
%	\item By symmetry with the previous part.
%  \end{enumerate}
%  \fi
%\end{proof}

The proof for Lemma \ref{lem:weaken-affine-restriction} is straightforward, so it is ommited.

\begin{lemma}\label{lem:weaken-affine-restriction}
	If $\vsistype{\uspctx}{\ectx}{\vs}{\vsty}$
	then $\vsistype{\ectx}{\uspctx}{\vs}{\vsty}$.
\end{lemma}

\begin{lemma}\label{lem:subst}\strut
  \begin{enumerate}
  \item If~$\tywithdag{\gctx}{\uspctx}{\utctx}{\ctx,\hastype{x}{\tau'}}
    {e}{\tau}{\graph}$
    and~$\tywithdag{\ectx}{\ectx}{\utctx}{\ectx}{v}{\tau'}{\emptygraph}$
    then~$\tywithdag{\gctx}{\uspctx}{\utctx}{\ctx}
    {\sub{e}{v}{x}}{\tau}{\graph}$.
    \label{lem:subst-val-exp}
  %% \item If~$\dagwf{\gctx}{\uspctx,\vertex}{\utctx}{\graph}{\graphkind}$
  %%   and~$\vertex' \notin \uspctx$
  %%   then~$\dagwf{\gctx}{\uspctx,\vertex'}{\utctx}{\gsub{\graph}{\vertex'}{\vertex}}{\gsub{\graphkind}{\vertex'}{\vertex}}$.
  %% \item If~$\istype{\gctx}{\uspctx,\vertex}{\utctx}{\gctx}{\tau}$
  %%   and~$\vertex' \notin \uspctx$
  %%   then~$\istype{\gctx}{\uspctx,\vertex}{\utctx}{\gctx}{\tsub{\vertex'}{\vertex}{\tau}}$.
  %% \item If~$\istype{\gctx}{\uspctx}{\utctx,\vertex}{\gctx}{\tau}$
    %%   then~$\istype{\gctx}{\uspctx}{\utctx,\vertex}{\gctx}{\tsub{\vertex'}{\vertex}{\tau}}$.
  \item If~$\iskind{\gctx}{}{\utctx}{\vstctx, \haskind{\convar}{\kind'}}{\con}{\kind}$
    and~$\iskind{\gctx}{}{\utctx}{\vstctx}{\con'}{\kind'}$
    then~$\iskind{\gctx}{}{\utctx}{\vstctx}{\tsub{\con}{\con'}{\convar}}{\kind}$.
    \label{lem:subst-con-con}
  \item If~$\dagwf{\gctx,\hastype{\gvar}{\graphkind'}}{\uspctx}{\utctx}{\graph}
    {\graphkind}$
    and~$\dagwf{\gctx}{\ectx}{\utctx}{\graph'}{\graphkind'}$
    then~$\dagwf{\gctx}{\uspctx}{\utctx}{\gsub{\graph}{\graph'}{\gvar}}
    {\graphkind}$.
    \label{lem:subst-graph-graph}
  \item If~$\iskind{\gctx,\hastype{\gvar}{\graphkind}}{}{\utctx}{\vstctx}{\con}{\kind}$
    and~$\dagwf{\gctx}{\ectx}{\utctx}{\graph}{\graphkind}$
    then~$\iskind{\gctx}{}{\utctx}{\vstctx}{\tsub{\con}{\graph}{\gvar}}{\kind}$.
    \label{lem:subst-graph-con}
  \item If~$\coneq{\gctx,\hastype{\gvar}{\graphkind}}{\utctx}{\vstctx}{\con_1}{\con_2}{\kind}$
    and~$\dagwf{\gctx}{\ectx}{\utctx}{\graph}{\graphkind}$
    then~$\coneq{\gctx}{\utctx}{\vstctx}{\tsub{\con_1}{\graph}{\gvar}}{\tsub{\con_2}{\graph}{\gvar}}{\kind}$.
    \label{lem:subst-con-equivalence-graph}
%type evaluation subst
  \item If~$\tywithdag{\gctx,\hastype{\gvar}{\graphkind}}{\uspctx}{\utctx}{\ctx}{e}{\tau}{\graph}$
    and~$\dagwf{\gctx}{\ectx}{\utctx}{\graph'}{\graphkind}$
    then~$\tywithdag{\gctx}{\uspctx}{\utctx}{\tsub{\ctx}{\graph'}{\gvar}}
    	{e}{\tsub{\tau}{\graph'}{\gvar}}{\tsub{\graph}{\graph'}{\gvar}}$.
    \label{lem:subst-graph-exp}
  \item If~$\vsistype{\ectx}{\utctx, \hastype{\vvar}{\vsty'}}{\vs}{\vsty}$
    and~$\vsistype{\ectx}{\utctx}{\vs'}{\vsty'}$
    then~$\vsistype{\ectx}{\utctx}{\vsub{\vs}{\vs'}{\vvar}}{\vsty}$.
    \label{lem:subst-touchvert-vert}
  \item If~$\vsistype{\uspctx, \hastype{\vvar}{\vsty'}}{\ectx}{\vs}{\vsty}$
    and~$\vsistype{\uspctx'}{\ectx}{\vs'}{\vsty'}$
	and~$\uspsplit{\uspctx''}{\uspctx}{\uspctx'}$
    then~$\vsistype{\uspctx''}{\ectx}{\vsub{\vs}{\vs'}{\vvar}}{\vsty}$.
    \label{lem:subst-spawnvert-vert}
  \item If~$\dagwf{\gctx}{\uspctx}{\utctx, \hastype{\vvar}{\vsty}}{\graph}{\graphkind}$
    and~$\vsistype{\ectx}{\utctx}{\vs}{\vsty}$
    and~$\vvar \notin \uspctx$
    then~$\dagwf{\gctx}{\uspctx}{\utctx}{\tsub{\graph}{\vs}{\vvar}}{\graphkind}$.
    \label{lem:subst-touchvert-graph}
  \item If~$\dagwf{\gctx}{\uspctx, \hastype{\vvar}{\vsty}}{\utctx, \hastype{\vvar}{\vsty'}}{\graph}{\graphkind}$
    and~$\vsistype{\uspctx'}{\ectx}{\vs}{\vsty}$
    and~$\vsistype{\ectx}{\utctx}{\vs}{\vsty'}$
	and~$\uspsplit{\uspctx''}{\uspctx}{\uspctx'}$
    then~$\dagwf{\gctx}{\uspctx''}{\utctx}{\tsub{\graph}{\vs}{\vvar}}{\graphkind}$.
    \label{lem:subst-allvert-graph}
  \item If~$\iskind{\gctx}{}{\utctx, \hastype{\vvar}{\vsty}}{\vstctx}{\con}{\kind}$
    and~$\vsistype{\ectx}{\utctx}{\vs}{\vsty}$
    then~$\iskind{\gctx}{}{\utctx}{\vstctx}
    {\tsub{\con}{\vs}{\vvar}}{\kind}$.
    \label{lem:subst-touchvert-con}
  \item If~$\vseq{\utctx, \hastype{\vvar}{\vsty'}}{\vs_1}{\vs_2}{\vsty}$
    and~$\vsistype{\ectx}{\utctx}{\vs}{\vsty'}$
    then~$\vseq{\utctx}{\tsub{\vs_1}{\vs}{\vvar}}{\tsub{\vs_2}{\vs}{\vvar}}{\vsty}$.
    \label{lem:subst-vs-equivalence}
  \item If~$\coneq{\gctx}{\utctx, \hastype{\vvar}{\vsty}}{\vstctx}{\con_1}{\con_2}{\kind}$
    and~$\vsistype{\ectx}{\utctx}{\vs}{\vsty}$
    then~$\coneq{\gctx}{\utctx}{\vstctx}{\tsub{\con_1}{\vs}{\vvar}}{\tsub{\con_2}{\vs}{\vvar}}{\kind}$.
    \label{lem:subst-con-equivalence-vs}
  \item If~$\tywithdag{\gctx}{\uspctx}{\utctx, \hastype{\vvar}{\vsty}}{\ctx}{e}{\tau}{\graph}$
    and~$\vsistype{\ectx}{\utctx}{\vs}{\vsty}$
    and~$\vvar \notin \uspctx$
    then~$\tywithdag{\gctx}{\uspctx}{\utctx}{\tsub{\ctx}{\vs}{\vvar}}{\tsub{e}{\vs}{\vvar}}{\tsub{\tau}{\vs}{\vvar}}{\tsub{\graph}{\vs}{\vvar}}$.
    \label{lem:subst-touchvert-exp}
  \item If~$\tywithdag{\gctx}{\uspctx, \hastype{\vvar}{\vsty}}{\utctx, \hastype{\vvar}{\vsty'}}{\ctx}{e}{\tau}{\graph}$
    and~$\vsistype{\uspctx'}{\ectx}{\vs}{\vsty}$
    and~$\vsistype{\ectx}{\utctx}{\vs}{\vsty'}$
	and~$\uspsplit{\uspctx''}{\uspctx}{\uspctx'}$
    then~$\tywithdag{\gctx}{\uspctx''}{\utctx}{\tsub{\ctx}{\vs}{\vvar}}{\tsub{e}{\vs}{\vvar}}{\tsub{\tau}{\vs}{\vvar}}{\tsub{\graph}{\vs}{\vvar}}$.
    \label{lem:subst-allvert-exp}
  \end{enumerate}
\end{lemma}
\iffull
\begin{proof}
  We show the interesting cases for all parts except (\ref{lem:subst-val-exp}),
  which is standard.
  \begin{enumerate}

  \item By induction on the derivation of
    $\tywithdag{\gctx}{\uspctx}{\utctx}{\ctx,\hastype{x}{\tau'}}{e}{\tau}{\graph}$.

\item By induction on the derivation of
    $\iskind{\gctx}{}{\utctx}{\vstctx, \haskind{\convar}{\kind'}}{\con}{\kind}$.
	\begin{itemize}
    \item \rulename{K:Rec}.
      Then~$\con = \kwprec{\convar'}{\hastycl{\vvar}{\vsty}}{\con''}{\vs}$
      and~$\kind = \kwtykind$
      and~$\iskind{\gctx}{\uspctx}{\utctx, \hastype{\vvar}{\vsty}}
             {\vstctx, \haskind{\convar'}{\kwkindarr{\vsty}{\kwtykind}}, \haskind{\convar}{\kind'}}{\con''}{\kwtykind}$.
		Since $\convar'$ is bound inside $\con$, $\convar \neq \convar'$.
      By induction,
      $\iskind{\gctx}{\uspctx}{\utctx, \hastype{\vvar}{\vsty}}
             {\vstctx, \haskind{\convar'}{\kwkindarr{\vsty}{\kwtykind}}}{\tsub{\con''}{\con'}{\convar}}{\kwtykind}$.
      Apply \rulename{K:Rec}.
    \end{itemize}

  \item
    By induction on the derivation
    of~$\dagwf{\gctx,\hastype{\gvar}{\graphkind'}}{\uspctx}{\utctx}{\graph}{\graphkind}$.
    \begin{itemize}
    \item \rulename{DW:RecPi}.
      Then~$\graph = \dagrec{\gvar'}{\dagpi{\vvar_f}{\vvar_t}{\graph''}}{}$
      and~$\graphkind = \dagpi{\hastycl{\vvar_f}{\vsty_f}}{\hastycl{\vvar_t}{\vsty_t}}{\kgraph}$
      and~$\dagwf{\gctx, \hastype{\gvar'}{\dagpi{\hastycl{\vvar_f}{\vsty_f}}{\hastycl{\vvar_t}{\vsty_t}}{\kgraph}}, \hastype{\gvar}{\graphkind'}}
           {\hastype{\vvar_f}{\vsty_f}}{\utctx, \hastype{\vvar_f}{\vsty_f}, \hastype{\vvar_t}{\vsty_t}}
           {\graph}{\kgraph}$.
      By induction,
      $\dagwf{\gctx,\hastype{\gvar'}
           {\dagpi{\hastycl{\vvar_f}{\vsty_f}}{\hastycl{\vvar_t}{\vsty_t}}{\kgraph}}}
           {\hastype{\vvar_f}{\vsty_f}}{\utctx, \hastype{\vvar_f}{\vsty_f}, \hastype{\vvar_t}{\vsty_t}}
           {\gsub{\graph''}{\graph'}{\gvar}}{\kgraph}$.
      Apply \rulename{DW:RecPi}.
    \end{itemize}

  \item
    By induction on the derivation
    of~$\iskind{\gctx,\hastype{\gvar}{\graphkind}}{}{\utctx}{\vstctx}{\con}{\kind}$.
    \begin{itemize}
    \item \rulename{K:Fun}.
      Then~$\con = \kwpi{\hastycl{\vvar_f}{\vsty_f}}{\hastycl{\vvar_t}{\vsty_t}}
				{\kwarrow{\con_1}{\con_2}{\kwtapp{\graph'}{\vvar_f}{\vvar_t}}}$
      	~and $\iskind{\gctx,\hastype{\gvar}{\graphkind}}{\uspctx_1}{\utctx, \hastype{\vvar_f}{\vsty_f}, \hastype{\vvar_t}{\vsty_t}}{\vstctx}
             {\con_1}{\kwtykind}$
      	~and $\iskind{\gctx,\hastype{\gvar}{\graphkind}}{\uspctx_2}{\utctx, \hastype{\vvar_f}{\vsty_f}, \hastype{\vvar_t}{\vsty_t}}{\vstctx}
             {\con_2}{\kwtykind}$
      	~and $\dagwf{\gctx,\hastype{\gvar}{\graphkind}}{\ectx}{\utctx}{\graph'}
             {\dagpi{\hastycl{\vvar_f}{\vsty_f}}{\hastycl{\vvar_t}{\vsty_t}}{\kgraph}}$
      	~and $\kind = \kwtykind$.
      By induction,
      $\iskind{\gctx}{\uspctx_1}{\utctx, \hastype{\vvar_f}{\vsty_f}, \hastype{\vvar_t}{\vsty_t}}{\vstctx}
             {\tsub{\con_1}{\graph}{\gvar}}{\kwtykind}$
      ~and $\iskind{\gctx}{\uspctx_2}{\utctx, \hastype{\vvar_f}{\vsty_f}, \hastype{\vvar_t}{\vsty_t}}{\vstctx}
             {\tsub{\con_2}{\graph}{\gvar}}{\kwtykind}$.
      By part~\ref{lem:subst-graph-graph},
      $\dagwf{\gctx}{\ectx}{\utctx}{\tsub{\graph'}{\graph}{\gvar}}
             {\dagpi{\hastycl{\vvar_f}{\vsty_f}}{\hastycl{\vvar_t}{\vsty_t}}{\kgraph}}$.
      Apply~\rulename{K:Fun}.
    \end{itemize}

	\item By induction on the derivation of
    $\coneq{\gctx,\hastype{\gvar}{\graphkind}}{\utctx}{\vstctx}{\con_1}{\con_2}{\kind}$.
    \begin{itemize}
	\item \rulename{CE:Reflexive}.
		Then $\con_1 = \con_2$,
			and $\iskind{\gctx,\hastype{\gvar}{\graphkind}}{}{\utctx}{\vstctx}{\con_1}{\kind}$.
		By part \ref{lem:subst-graph-con},
			$\iskind{\gctx}{}{\utctx}{\vstctx}{\tsub{\con_1}{\graph}{\gvar}}{\kind}$.
		Apply~\rulename{CE:Reflexive}.

%	\item \rulename{CE:Commutative}.
%		By induction and \rulename{CE:Commutative}.
%
%	\item \rulename{CE:Transitive}.
%		By induction and \rulename{CE:Transitive}.

	\item \rulename{CE:Prod}.
		Then $\con_1 = \kwprod{\con_3}{\con_4}$
			and $\con_2 = \kwprod{\con_3'}{\con_4'}$
			and $\kind = \kwtykind$
			and $\coneq{\gctx,\hastype{\gvar}{\graphkind}}{\utctx}{\vstctx}{\con_3}{\con_3'}{\kwtykind}$
			and $\coneq{\gctx,\hastype{\gvar}{\graphkind}}{\utctx}{\vstctx}{\con_4}{\con_4'}{\kwtykind}$.
		By induction,
			$\coneq{\gctx}{\utctx}{\vstctx}{\tsub{\con_3}{\graph}{\gvar}}{\tsub{\con_3'}{\graph}{\gvar}}{\kwtykind}$
			and~$\coneq{\gctx}{\utctx}{\vstctx}{\tsub{\con_4}{\graph}{\gvar}}{\tsub{\con_4'}{\graph}{\gvar}}{\kwtykind}$.
		Apply~\rulename{CE:Prod}.
    \end{itemize}

  \item By induction on the derivation
    of~$\tywithdag{\gctx,\hastype{\gvar}{\graphkind}}{\uspctx}
    {\utctx}{\ctx}{e}{\tau}{\graph}$.
    \begin{itemize}
	\item \rulename{S:Fun}.
		Then $\tau = \kwpi{\hastycl{\vvar_f}{\vsty_f}}{\hastycl{\vvar_t}{\vsty_t}}
					             {\kwarrow{\tau_1}{\tau_2}
					               {\kwtapp{(\dagrec{\gvar'}{\dagpi{\hastype{\vvar_f}{\vsty_f}}{\hastype{\vvar_t}{\vsty_t}}
					               {\graph''}}{})}{\vvar_f}{\vvar_t}}}$
		and~$e = \kwfun{\vvar_f}{\vvar_t}{f}{x}{e'}$
		and~$\graph = \emptygraph$
		and~$\tywithdag{\gctx,\hastype{\gvar}{\graphkind},\hastype{\gvar'}{\dagpi{\hastycl{\vvar_f}{\vsty_f}}{\hastycl{\vvar_t}{\vsty_t}}{\kgraph}}}
           {\hastype{\vvar_f}{\vsty_f}}{\utctx, \hastype{\vvar_f}{\vsty_f}, \hastype{\vvar_t}{\vsty_t}}
           {\ctx,\hastype{f}{\kwpi{\hastycl{\vvar_f}{\vsty_f}}{\hastycl{\vvar_t}{\vsty_t}}
               {\kwarrow{\tau_1}{\tau_2}{\kwtapp{\gvar'}{\vvar_f}{\vvar_t}}}}, \hastype{x}{\tau_1}}{e'}{\tau_2}{\graph''}$
		and~$\iskind{\gctx,\hastype{\gvar}{\graphkind}}{\ectx}
			{\utctx, \hastype{\vvar_f}{\vsty_f}, \hastype{\vvar_t}{\vsty_t}}{\ectx}{\tau_1}{\kwtykind}$
		and~$\iskind{\gctx,\hastype{\gvar}{\graphkind}}{\ectx}
			{\utctx, \hastype{\vvar_f}{\vsty_f}, \hastype{\vvar_t}{\vsty_t}}{\ectx}{\tau_2}{\kwtykind}$.
	By induction,
		$\tywithdag{\gctx,\hastype{\gvar'}{\dagpi{\hastycl{\vvar_f}{\vsty_f}}{\hastycl{\vvar_t}{\vsty_t}}{\kgraph}}}
           {\hastype{\vvar_f}{\vsty_f}}{\utctx, \hastype{\vvar_f}{\vsty_f}, \hastype{\vvar_t}{\vsty_t}}
           {\tsub{\ctx}{\graph'}{\gvar},\hastype{f}{\kwpi{\hastycl{\vvar_f}{\vsty_f}}{\hastycl{\vvar_t}{\vsty_t}}
               {\kwarrow{\gsub{\tau_1}{\graph'}{\gvar}}{\gsub{\tau_2}{\graph'}{\gvar}}{\kwtapp{\gvar'}{\vvar_f}{\vvar_t}}}},
             \hastype{x}{\gsub{\tau_1}{\graph'}{\gvar}}}{e'}{\gsub{\tau_2}{\graph'}{\gvar}}{\gsub{\graph''}{\graph'}{\gvar}}$.
	By part \ref{lem:subst-graph-con},
		$\iskind{\gctx}{\ectx}{\utctx, \hastype{\vvar_f}{\vsty_f}, \hastype{\vvar_t}{\vsty_t}}{\ectx}{\gsub{\tau_1}{\graph'}{\gvar}}{\kwtykind}$
		and~$\iskind{\gctx}{\ectx}{\utctx, \hastype{\vvar_f}{\vsty_f}, \hastype{\vvar_t}{\vsty_t}}{\ectx}{\gsub{\tau_2}{\graph'}{\gvar}}{\kwtykind}$.
	Apply~\rulename{S:Fun}.

	\item \rulename{S:Type-Eq}.
		Then $\tywithdag{\gctx,\hastype{\gvar}{\graphkind}}{\uspctx}{\utctx}{\ctx}{e}{\tau_1}{\graph}$
			and $\coneq{\gctx,\hastype{\gvar}{\graphkind}}{\utctx}{\ectx}{\tau_1}{\tau}{\kwtykind}$.
		By induction,
			$\tywithdag{\gctx}{\uspctx}{\utctx}{\gsub{\ctx}{\graph'}{\gvar}}{e}{\tsub{\tau_1}{\graph'}{\gvar}}{\gsub{\graph}{\graph'}{\gvar}}$.
		By part \ref{lem:subst-con-equivalence-graph},
			$\coneq{\gctx}{\utctx}{\ectx}{\gsub{\tau_1}{\graph'}{\gvar}}{\vsub{\tau}{\graph'}{\gvar}}{\kwtykind}$.
		Apply \rulename{S:Type-Eq}.

    \item \rulename{S:App}.
      Then $e = \kwapp{}{\kwtapp{e_1}{\vs_f}{\vs_t}}{e_2}$
		and~$\tau = \tsub{\tsub{\tau_2}{\vs_f}{\vvar_f}}{\vs_t}{\vvar_t}$
		and~$\uspsplit{\uspctx}{\uspctx_1}{\uspctx'}$
		and~$\uspsplit{\uspctx'}{\uspctx_2}{\uspctx_3}$
		and~$\graph_1 \seqcomp \graph_2 \seqcomp
             \kwtapp{\graph_3}{\vs_f}{\vs_t}$
		and~$\tywithdag{\gctx,\hastype{\gvar}{\graphkind}}{\uspctx_1}{\utctx}{\ctx}{e_1}
           {\kwpi{\hastycl{\vvar_f}{\vsty_f}}{\hastycl{\vvar_t}{\vsty_t}}
             {\kwarrow{\tau_1}{\tau_2}{\kwtapp{\graph_3}{\vvar_f}{\vvar_t}}}}
           {\graph_1}$
		and~$\tywithdag{\gctx,\hastype{\gvar}{\graphkind}}{\uspctx_2}{\utctx}{\ctx}{e_2}
                     {\tsub{\tsub{\tau_1}{\vs_f}{\vvar_f}}{\vs_t}{\vvar_t}}
                     {\graph_2}$
		and~$\vsistype{\uspctx_3}{\ectx}{\vs_f}{\vsty_f}$
		and~$\vsistype{\ectx}{\utctx}{\vs_f}{\vsty_f}$
		and~$\vsistype{\ectx}{\utctx}{\vs_t}{\vsty_t}$.
	By induction,
		$\tywithdag{\gctx}{\uspctx_1}{\utctx}{\tsub{\ctx}{\graph'}{\gvar}}{e_1}
           {\kwpi{\hastycl{\vvar_f}{\vsty_f}}{\hastycl{\vvar_t}{\vsty_t}}
             {\kwarrow{\tsub{\tau_1}{\graph'}{\gvar}}{\tsub{\tau_2}{\graph'}{\gvar}}
				{\kwtapp{\gsub{\graph_3}{\graph'}{\gvar}}{\vvar_f}{\vvar_t}}}}{\gsub{\graph_1}{\graph'}{\gvar}}$
		and~$\tywithdag{\gctx}{\uspctx_2}{\utctx}{\tsub{\ctx}{\graph'}{\gvar}}{e_2}
                     {\tsub{\tsub{\tsub{\tau_1}{\graph'}{\gvar}}{\vs_f}{\vvar_f}}{\vs_t}{\vvar_t}}
                     {\gsub{\graph_2}{\graph'}{\gvar}}$.\\
	Apply~\rulename{S:App}.
    \end{itemize}

  \item
%	If~$\vsistype{\ectx}{\utctx, \hastype{\vvar}{\vsty'}}{\vs}{\vsty}$
%	    and~$\vsistype{\ectx}{\utctx}{\vs'}{\vsty'}$
%	    then~$\vsistype{\ectx}{\utctx}{\vsub{\vs}{\vs'}{\vvar}}{\vsty}$.
    By induction on the derivation of
    $\vsistype{\ectx}{\utctx, \hastype{\vvar}{\vsty'}}{\vs}{\vsty}$.
    \begin{itemize}
	\item \rulename{U:PsiVar}.
		Then $\vs = \vvar'$
			and~$\utctx = \utctx', \hastype{\vvar'}{\vsty}$.
		Since $\vvar \neq \vvar'$,
			$\vsub{\vvar'}{\vs'}{\vvar} = \vvar'$.
		Apply \rulename{U:PsiVar}.

	\item \rulename{U:Pair}.
		Then $\vs = \kwpair{\vs_1}{\vs_2}$
			and~$\vsty = \vsprod{\vsty_1}{\isav}{\vsty_2}{\isav}$
			and~$\uspsplit{\ectx}{\uspctx_1}{\uspctx_2}$
			and~$\vsistype{\uspctx_1}{\utctx, \hastype{\vvar}{\vsty'}}{\vs_1}{\vsty_1}$
			and~$\vsistype{\uspctx_2}{\utctx, \hastype{\vvar}{\vsty'}}{\vs_2}{\vsty_2}$.
		By inversion on \rulename{OM:Empty},
			$\uspctx_1 = \uspctx_2 = \ectx$.
		By induction,
			$\vsistype{\ectx}{\utctx}{\vsub{\vs_1}{\vs'}{\vvar}}{\vsty_1}$
			and~$\vsistype{\ectx}{\utctx}{\vsub{\vs_2}{\vs'}{\vvar}}{\vsty_2}$.
		Apply \rulename{U:Pair}.

	\item \rulename{U:OnlyLeftPair}.
		Then $\vs = \kwpair{\vs_1}{\vs_2}$
			and~$\vsty = \vsprod{\vsty_1}{\isav}{\vsty_2}{\isunav}$
			and~$\vsistype{\ectx}{\utctx, \hastype{\vvar}{\vsty'}}{\vs_1}{\vsty_1}$
			and~$\vsistype{\ectx}{\utctx'}{\vs_2}{\vsty_2}$.
		By weakening and induction,
			$\vsistype{\ectx}{\utctx}{\vsub{\vs_1}{\vs'}{\vvar}}{\vsty_1}$
			and~$\vsistype{\ectx}{\utctx', \utctx}{\vsub{\vs_2}{\vs'}{\vvar}}{\vsty_2}$.
		Apply \rulename{U:OnlyLeftPair}.

	\item \rulename{U:Fst}.
		Then $\vs = \kwfst{\vs''}$
			and~$\vsistype{\ectx}{\utctx, \hastype{\vvar}{\vsty'}}{\vs''}{\vsprod{\vsty}{\isav}{\vsty_2}{\avail}}$.
		By induction,
			$\vsistype{\ectx}{\utctx}{\vsub{\vs''}{\vs'}{\vvar}}{\vsprod{\vsty}{\isav}{\vsty_2}{\avail}}$.
		Apply \rulename{U:Fst}.
    \end{itemize}

  \item
    By induction on the derivation of
    $\vsistype{\uspctx, \hastype{\vvar}{\vsty'}}{\ectx}{\vs}{\vsty}$.
    \begin{itemize}
	\item \rulename{U:OmegaVar}.
		Then $\vs = \vvar'$
			and~$\uspctx = \uspctx''', \hastype{\vvar'}{\vsty}$.
		Since $\vvar \neq \vvar'$,
			$\vsub{\vs}{\vs'}{\vvar} = \vvar'$.
		Apply Lemma \ref{lem:usp-split-weakening}.

	\item \rulename{U:Pair}.
		Then $\vs = \kwpair{\vs_1}{\vs_2}$
			and~$\vsty = \vsprod{\vsty_1}{\isav}{\vsty_2}{\isav}$
			and~$\uspsplit{\uspctx, \hastype{\vvar}{\vsty'}}{\uspctx_1}{\uspctx_2}$
			and~$\vsistype{\uspctx_1}{\ectx}{\vs_1}{\vsty_1}$
			and~$\vsistype{\uspctx_2}{\ectx}{\vs_2}{\vsty_2}$.
		By Lemma \ref{lem:split-form} there are four cases:
			\begin{enumerate}
			\item \splituspLexstfull{\uspctx_1'}{\uspctx}{\uspctx_1}{\uspctx_2}{\vvar}{\vsty'}.
				Therefore,
					$\vs_2$ does not contain $\vvar$.
				By Lemma \ref{lem:split-move},
					\uspsplitexst{\uspctx'''}{\uspctx''}{\uspctx_1'}{\uspctx_2}{\uspctx'}.
				By induction,
					$\vsistype{\uspctx'''}{\ectx}{\vsub{\vs_1}{\vs'}{\vvar}}{\vsty_1}$.
				Apply \rulename{U:Pair}.
			\item \splituspRexstfull{\uspctx_2'}{\uspctx}{\uspctx_1}{\uspctx_2}{\vvar}{\vsty'}.
				By symmetry with the previous case.
			\item \splituspbothexst{\uspctx_1'}{\uspctx_2'}{\uspctx}{\uspctx_1}{\uspctx_2}{\vvar}{\vsty'}{\vsty_1'}{\vsty_2'}.
				By Lemma \ref{lem:vert-split-typing},
					\vstytouspsplitext{\uspctx'}{\uspctx_1''}{\uspctx_2''}{\vs'}{\vsty_1'}{\vsty_2'}.
				By Lemma \ref{lem:split-move},
					\uspsplitbothexst{\uspctx'''}{\uspctx''''}{\uspctx''}{\uspctx_1'}{\uspctx_2'}{\uspctx_1''}{\uspctx_2''}.
				By induction,
					$\vsistype{\uspctx'''}{\ectx}{\vsub{\vs_1}{\vs'}{\vvar}}{\vsty_1}$
					and~$\vsistype{\uspctx''''}{\ectx}{\vsub{\vs_2}{\vs'}{\vvar}}{\vsty_2}$.
				Apply \rulename{U:Pair}.
			\item \splituspbothcrossexst{\uspctx_1'}{\uspctx_2'}{\uspctx}{\uspctx_1}{\uspctx_2}{\vvar}{\vsty'}{\vsty_1'}{\vsty_2'}.
				By symmetry with the previous case.
			\end{enumerate}

	\item \rulename{U:OnlyLeftPair}.
		Then $\vs = \kwpair{\vs_1}{\vs_2}$
			and~$\vsty = \vsprod{\vsty_1}{\isav}{\vsty_2}{\isunav}$
			and~$\vsistype{\uspctx, \hastype{\vvar}{\vsty'}}{\ectx}{\vs_1}{\vsty_1}$
			and~$\vsistype{\ectx}{\utctx'}{\vs_2}{\vsty_2}$.
		By induction,
			$\vsistype{\uspctx''}{\ectx}{\vsub{\vs_1}{\vs'}{\vvar}}{\vsty_1}$.
		By weakening,
			$\vsistype{\ectx}{\utctx', \hastype{\vvar}{\vsty'}}{\vs_2}{\vsty_2}$.
		By Lemma \ref{lem:weaken-affine-restriction},
			$\vsistype{\ectx}{\uspctx'}{\vs'}{\vsty'}$.
		By weakening and part \ref{lem:subst-touchvert-vert},
			$\vsistype{\ectx}{\utctx', \uspctx'}{\vsub{\vs_2}{\vs'}{\vvar}}{\vsty_2}$.
		Apply \rulename{U:OnlyLeftPair}.

	\item \rulename{U:Fst}.
		Then $\vs = \kwfst{\vs''}$
			and~$\vsistype{\uspctx, \hastype{\vvar}{\vsty'}}{\ectx}{\vs''}{\vsprod{\vsty}{\isav}{\vsty_2}{\avail}}$.
		By induction,
			$\vsistype{\uspctx''}{\ectx}{\vsub{\vs''}{\vs'}{\vvar}}{\vsprod{\vsty}{\isav}{\vsty_2}{\avail}}$.
		Apply \rulename{U:Fst}.
    \end{itemize}

  \item
    By induction on the derivation of
    $\dagwf{\gctx}{\uspctx}{\utctx, \hastype{\vvar}{\vsty}}{\graph}{\graphkind}$.
    \begin{itemize}
	\item \rulename{DW:App}
		Then $\graph = \kwtapp{\graph'}{\vs_f}{\vs_t}$
			and~$\graphkind = \kgraph$
			and~$\uspsplit{\uspctx}{\uspctx_1}{\uspctx_2}$
			and~$\dagwf{\gctx}{\uspctx_1}{\utctx, \hastype{\vvar}{\vsty}}{\graph'}{\dagpi{\hastycl{\vvar_f}{\vsty_f}}{\hastycl{\vvar_t}{\vsty_t}}{\kgraph}}$
			and~$\vsistype{\uspctx_2}{\ectx}{\vs_f}{\vsty_f}$
			and~$\vsistype{\ectx}{\utctx, \hastype{\vvar}{\vsty}}{\vs_f}{\vsty_f}$
			and~$\vsistype{\ectx}{\utctx, \hastype{\vvar}{\vsty}}{\vs_t}{\vsty_t}$.
		By the contrapositive of Lemma \ref{lem:usp-split-upstream},
			$\vvar \notin \uspctx_1$
			and~$\vvar \notin \uspctx_2$.
		Therefore,
			$\vs_f$ does not contain $\vvar$.
		By induction,
			$\dagwf{\gctx}{\uspctx_1}{\utctx}{\vsub{\graph'}{\vs}{\vvar}}{\dagpi{\hastycl{\vvar_f}{\vsty_f}}{\hastycl{\vvar_t}{\vsty_t}}{\kgraph}}$.
		By part \ref{lem:subst-touchvert-vert},
			$\vsistype{\ectx}{\utctx}{\vs_f}{\vsty_f}$
			and$\vsistype{\ectx}{\utctx}{\vsub{\vs_t}{\vs}{\vvar}}{\vsty_t}$.
		Apply \rulename{DW:App}.
    \end{itemize}

  \item
    By induction on the derivation of
    $\dagwf{\gctx}{\uspctx, \hastype{\vvar}{\vsty}}{\utctx, \hastype{\vvar}{\vsty'}}{\graph}{\graphkind}$.
    \begin{itemize}

	\item \rulename{DW:RecPi}
		Then $\graph = \dagrec{\gvar}{\dagpi{\hastycl{\vvar_f}{\vsty_f}}{\hastycl{\vvar_t}{\vsty_t}}{\graph'}}{}$
			and~$\graphkind = \dagpi{\hastycl{\vvar_f}{\vsty_f}}{\hastycl{\vvar_t}{\vsty_t}}{\kgraph}$
			and~$\dagwf{\gctx,\hastype{\gvar}{\dagpi{\hastycl{\vvar_f}{\vsty_f}}{\hastycl{\vvar_t}{\vsty_t}}{\kgraph}}}
           				{\hastype{\vvar_f}{\vsty_f}}{\utctx, \hastype{\vvar_f}{\vsty_f}, \hastype{\vvar_t}{\vsty_t}, \hastype{\vvar}{\vsty'}}{\graph'}{\kgraph}$.
		By part \ref{lem:subst-touchvert-graph},
			$\dagwf{\gctx,\hastype{\gvar}{\dagpi{\hastycl{\vvar_f}{\vsty_f}}{\hastycl{\vvar_t}{\vsty_t}}{\kgraph}}}
           				{\hastype{\vvar_f}{\vsty_f}}{\utctx, \hastype{\vvar_f}{\vsty_f}, \hastype{\vvar_t}{\vsty_t}}{\vsub{\graph'}{\vs}{\vvar}}{\kgraph}$.
		Apply \rulename{DW:RecPi}.

	\item \rulename{DW:New}
		Then $\graph = \dagnew{\hastycl{\vvar'}{\vsty'}}{\graph'}$
			and~$\graphkind = \kgraph$
			and~$\dagwf{\gctx}{\uspctx, \hastype{\vvar'}{\vsty'}, \hastype{\vvar}{\vsty}}
						{\utctx, \hastype{\vvar'}{\vsty'}, \hastype{\vvar}{\vsty}}{\graph'}{\kgraph}$.
		By \rulename{OM:Var},
			$\uspsplit{\uspctx'', \hastype{\vvar'}{\vsty'}}{\uspctx, \hastype{\vvar'}{\vsty'}}{\uspctx'}$.
		By induction,
			$\dagwf{\gctx}{\uspctx'', \hastype{\vvar'}{\vsty'}}{\utctx, \hastype{\vvar'}{\vsty'}}{\vsub{\graph'}{\vs}{\vvar}}{\kgraph}$.
		Apply \rulename{DW:New}.

	\item \rulename{DW:App}
		Then $\graph = \kwtapp{\graph'}{\vs_f}{\vs_t}$
			and~$\graphkind = \kgraph$
			and~$\uspsplit{\uspctx, \hastype{\vvar}{\vsty}}{\uspctx_1}{\uspctx_2}$
			and~$\dagwf{\gctx}{\uspctx_1}{\utctx, \hastype{\vvar}{\vsty'}}{\graph'}{\dagpi{\hastycl{\vvar_f}{\vsty_f}}{\hastycl{\vvar_t}{\vsty_t}}{\kgraph}}$
			and~$\vsistype{\uspctx_2}{\ectx}{\vs_f}{\vsty_f}$
			and~$\vsistype{\ectx}{\utctx, \hastype{\vvar}{\vsty}}{\vs_f}{\vsty_f}$
			and~$\vsistype{\ectx}{\utctx, \hastype{\vvar}{\vsty'}}{\vs_t}{\vsty_t}$.
		By part \ref{lem:subst-touchvert-vert},
			$\vsistype{\ectx}{\utctx}{\vsub{\vs_f}{\vs}{\vvar}}{\vsty_f}$
			and~$\vsistype{\ectx}{\utctx}{\vsub{\vs_t}{\vs}{\vvar}}{\vsty_t}$.
		By Lemma \ref{lem:split-form} there are four cases:
			\begin{enumerate}
			\item \splituspLexstfull{\uspctx_1'}{\uspctx}{\uspctx_1}{\uspctx_2}{\vvar}{\vsty}.
				Therefore,
					$\vs_f$ does not contain $\vvar$.
				By Lemma \ref{lem:split-move},
					\uspsplitexst{\uspctx'''}{\uspctx''}{\uspctx_1'}{\uspctx_2}{\uspctx'}.
				By induction,
					$\dagwf{\gctx}{\uspctx'''}{\utctx}{\vsub{\graph'}{\vs}{\vvar}}{\dagpi{\hastycl{\vvar_f}{\vsty_f}}{\hastycl{\vvar_t}{\vsty_t}}{\kgraph}}$.
				Apply \rulename{DW:App}.
			\item \splituspRexstfull{\uspctx_2'}{\uspctx}{\uspctx_1}{\uspctx_2}{\vvar}{\vsty}.
				By Lemma \ref{lem:split-move},
					\uspsplitexst{\uspctx'''}{\uspctx''}{\uspctx_2'}{\uspctx_1}{\uspctx'}.
				By part \ref{lem:subst-touchvert-graph},
					$\dagwf{\gctx}{\uspctx_1}{\utctx}{\vsub{\graph'}{\vs}{\vvar}}{\dagpi{\hastycl{\vvar_f}{\vsty_f}}{\hastycl{\vvar_t}{\vsty_t}}{\kgraph}}$.
				By part \ref{lem:subst-spawnvert-vert},
					$\vsistype{\uspctx'''}{\ectx}{\vsub{\vs_f}{\vs}{\vvar}}{\vsty_f}$.
				Apply \rulename{DW:App}.
			\item \splituspbothexst{\uspctx_1'}{\uspctx_2'}{\uspctx}{\uspctx_1}{\uspctx_2}{\vvar}{\vsty}{\vsty_1}{\vsty_2}.
				By Lemma \ref{lem:vert-split-typing},
					\vstytouspsplitext{\uspctx'}{\uspctx_1''}{\uspctx_2''}{\vs}{\vsty_1}{\vsty_2}.
				By Lemma \ref{lem:split-move},
					\uspsplitbothexst{\uspctx'''}{\uspctx''''}{\uspctx''}{\uspctx_1'}{\uspctx_2'}{\uspctx_1''}{\uspctx_2''}.
				By induction,
					$\dagwf{\gctx}{\uspctx'''}{\utctx}{\vsub{\graph'}{\vs}{\vvar}}{\dagpi{\hastycl{\vvar_f}{\vsty_f}}{\hastycl{\vvar_t}{\vsty_t}}{\kgraph}}$.
				Apply \rulename{DW:App}.
			\item \splituspbothcrossexst{\uspctx_1'}{\uspctx_2'}{\uspctx}{\uspctx_1}{\uspctx_2}{\vvar}{\vsty}{\vsty_1}{\vsty_2}.
				By symmetry with the previous case.
			\end{enumerate}
    \end{itemize}

  \item
    By induction on the derivation of
    $\iskind{\gctx}{}{\utctx, \hastype{\vvar}{\vsty}}{\vstctx}{\con}{\kind}$.
    \begin{itemize}
	\item \rulename{K:Fun}.
		Then $\con = \kwpi{\hastycl{\vvar_f}{\vsty_f}}{\hastycl{\vvar_t}{\vsty_t}}
             			{\kwarrow{\con_1}{\con_2}{\kwtapp{\graph}{\vvar_f}{\vvar_t}}}$
			and~$\kind = \kwtykind$
			and~$\iskind{\gctx}{}{\utctx, \hastype{\vvar_f}{\vsty_f}, \hastype{\vvar_t}{\vsty_t}, \hastype{\vvar}{\vsty}}{\vstctx}{\con_1}{\kwtykind}$
			and~$\iskind{\gctx}{}{\utctx, \hastype{\vvar_f}{\vsty_f}, \hastype{\vvar_t}{\vsty_t}, \hastype{\vvar}{\vsty}}{\vstctx}{\con_2}{\kwtykind}$
			and~$\dagwf{\gctx}{\ectx}{\utctx, \hastype{\vvar}{\vsty}}{\graph}{\dagpi{\hastycl{\vvar_f}{\vsty_f}}{\hastycl{\vvar_t}{\vsty_t}}{\kgraph}}$.
		By induction,
			$\iskind{\gctx}{}{\utctx, \hastype{\vvar_t}{\vsty_t}}{\vstctx}{\tsub{\con_1}{\vs}{\vvar}}{\kwtykind}$
			and~$\iskind{\gctx}{}{\utctx, \hastype{\vvar_t}{\vsty_t}}{\vstctx}{\tsub{\con_2}{\vs}{\vvar}}{\kwtykind}$.
		By part \ref{lem:subst-touchvert-graph},
			$\dagwf{\gctx}{\ectx}{\utctx}{\vsub{\graph}{\vs}{\vvar}}{\dagpi{\hastycl{\vvar_f}{\vsty_f}}{\hastycl{\vvar_t}{\vsty_t}}{\kgraph}}$.
		Apply \rulename{K:Fun}.
    \end{itemize}

\item
    By induction on the derivation of
    $\vseq{\utctx, \hastype{\vvar}{\vsty'}}{\vs_1}{\vs_2}{\vsty}$.
    \begin{itemize}
	\item \rulename{UE:Reflexive}.
		Then $\vs_1 = \vs_2$
			and $\vsistype{\ectx}{\utctx, \hastype{\vvar}{\vsty'}}{\vs_1}{\vsty}$.
		By part \ref{lem:subst-touchvert-vert},
			$\vsistype{\ectx}{\utctx}{\tsub{\vs_1}{\vs}{\vvar}}{\vsty}$.
		Apply~\rulename{UE:Reflexive}.

	\item \rulename{UE:Pair}.
		Then $\vs_1 = \kwpair{\vs_3, \vs_4}$
			and $\vs_2 = \kwpair{\vs_3', \vs_4'}$
			and $\vsty = \vsprod{\vsty_3}{\isav}{\vsty_4}{\isav}$
			and $\vseq{\utctx, \hastype{\vvar}{\vsty'}}{\vs_3}{\vs_3'}{\vsty_3}$
			and $\vseq{\utctx, \hastype{\vvar}{\vsty'}}{\vs_4}{\vs_4'}{\vsty_4}$.
		By induction,
			$\vseq{\utctx}{\tsub{\vs_3}{\vs}{\vvar}}{\tsub{\vs_3'}{\vs}{\vvar}}{\vsty_3}$
			and $\vseq{\utctx}{\tsub{\vs_4}{\vs}{\vvar}}{\tsub{\vs_4'}{\vs}{\vvar}}{\vsty_4}$.
		Apply~\rulename{UE:Pair}.

	\item \rulename{UE:OnlyLeftPair}.
		Then $\vs_1 = \kwpair{\vs_3, \vs_4}$
			and $\vs_2 = \kwpair{\vs_3', \vs_4'}$
			and $\vsty = \vsprod{\vsty_3}{\isav}{\vsty_4}{\isunav}$
			and $\vseq{\utctx, \hastype{\vvar}{\vsty'}}{\vs_3}{\vs_3'}{\vsty_3}$
			and $\vseq{\utctx'}{\vs_4}{\vs_4'}{\vsty_4}$.
		By weakening and induction,
			$\vseq{\utctx}{\tsub{\vs_3}{\vs}{\vvar}}{\tsub{\vs_3'}{\vs}{\vvar}}{\vsty_3}$
			and $\vseq{\utctx',\utctx}{\tsub{\vs_4}{\vs}{\vvar}}{\tsub{\vs_4'}{\vs}{\vvar}}{\vsty_4}$.
		Apply~\rulename{UE:OnlyLeftPair}.
    \end{itemize}

  \item
    By induction on the derivation of
    $\coneq{\gctx}{\utctx, \hastype{\vvar}{\vsty}}{\vstctx}{\con_1}{\con_2}{\kind}$.
    \begin{itemize}
	\item \rulename{CE:Reflexive}.
		Then $\con_1 = \con_2$
			and $\iskind{\gctx}{}{\utctx, \hastype{\vvar}{\vsty}}{\vstctx}{\con_1}{\kind}$.
		By part \ref{lem:subst-touchvert-con},
			$\iskind{\gctx}{}{\utctx}{\vstctx}{\tsub{\con_1}{\vs}{\vvar}}{\kind}$.
		Apply~\rulename{CE:Reflexive}.

%	\item \rulename{CE:Commutative}.
%		By induction and \rulename{CE:Commutative}.
%
%	\item \rulename{CE:Transitive}.
%		By induction and \rulename{CE:Transitive}.

	\item \rulename{CE:Fut}.
		Then $\con_1 = \kwfutt{\con}{\vs'}$
			and $\con_2 = \kwfutt{\con'}{\vs''}$
			and $\kind = \kwtykind$
			and $\vseq{\utctx, \hastype{\vvar}{\vsty}}{\vs'}{\vs''}{\kwvty}$.
			and $\coneq{\gctx}{\utctx, \hastype{\vvar}{\vsty}}{\vstctx}{\con}{\con'}{\kwtykind}$.
		By induction,
			$\coneq{\gctx}{\utctx}{\vstctx}{\tsub{\con}{\vs}{\vvar}}{\tsub{\con'}{\vs}{\vvar}}{\kwtykind}$.
		By part \ref{lem:subst-vs-equivalence},
			$\vseq{\utctx}{\tsub{\vs'}{\vs}{\vvar}}{\tsub{\vs''}{\vs}{\vvar}}{\kwvty}$.
		Apply~\rulename{CE:Fut}.
    \end{itemize}

  \item
    By induction on the derivation of
    $\tywithdag{\gctx}{\uspctx}{\utctx, \hastype{\vvar}{\vsty}}{\ctx}{e}{\tau}{\graph}$.
    \begin{itemize}
%	\item \rulename{S:Var}.
%		Then $e = x$
%			and $\graph = \emptygraph$
%			and $\ctx = \ctx', \hastype{x}{\tau}$.
%		Apply \rulename{S:Var}.

	\item \rulename{S:Fun}.
		Then $\tau = \kwpi{\hastycl{\vvar_f}{\vsty_f}}{\hastycl{\vvar_t}{\vsty_t}}
				{\kwarrow{\tau_1}{\tau_2}
				{\kwtapp{(\dagrec{\gvar}{\dagpi{\hastype{\vvar_f}{\vsty_f}}{\hastype{\vvar_t}{\vsty_t}}
				{\graph'}}{})}{\vvar_f}{\vvar_t}}}$
			and $\graph = \emptygraph$
			and $e = \kwfun{\vvar_f}{\vvar_t}{f}{x}{e'}$
			and $\tywithdag{\gctx, \hastype{\gvar}{\dagpi{\hastycl{\vvar_f}{\vsty_f}}{\hastycl{\vvar_t}{\vsty_t}}{\kgraph}}}
				{\hastype{\vvar_f}{\vsty_f}}{\utctx, \hastype{\vvar_f}{\vsty_f}, \hastype{\vvar_t}{\vsty_t}, \hastype{\vvar}{\vsty}}
				{\ctx, \hastype{f}{\kwpi{\hastycl{\vvar_f}{\vsty_f}}{\hastycl{\vvar_t}{\vsty_t}}
					{\kwarrow{\tau_1}{\tau_2}{\kwtapp{\gvar}{\vvar_f}{\vvar_t}}}},\hastype{x}{\tau_1}}
				{e'}{\tau_2}{\graph'}$
			and $\iskind{\gctx}{\ectx}{\utctx, \hastype{\vvar_f}{\vsty_f}, \hastycl{\vvar_t}{\vsty_t}, \hastype{\vvar}{\vsty}}{\ectx}{\tau_1}{\kwtykind}$
			and $\iskind{\gctx}{\ectx}{\utctx, \hastype{\vvar_f}{\vsty_f}, \hastycl{\vvar_t}{\vsty_t}, \hastype{\vvar}{\vsty}}{\ectx}{\tau_2}{\kwtykind}$.
		By induction,
			$\tywithdag{\gctx, \hastype{\gvar}{\dagpi{\hastycl{\vvar_f}{\vsty_f}}{\hastycl{\vvar_t}{\vsty_t}}{\kgraph}}}
				{\hastype{\vvar_f}{\vsty_f}}{\utctx, \hastype{\vvar_f}{\vsty_f}, \hastype{\vvar_t}{\vsty_t}}
				{\vsub{\ctx}{\vs}{\vvar}, \hastype{f}{\kwpi{\hastycl{\vvar_f}{\vsty_f}}{\hastycl{\vvar_t}{\vsty_t}}
					{\kwarrow{\vsub{\tau_1}{\vs}{\vvar}}{\vsub{\tau_2}{\vs}{\vvar}}{\kwtapp{\gvar}{\vvar_f}{\vvar_t}}}},
					\hastype{x}{\vsub{\tau_1}{\vs}{\vvar}}}
				{\vsub{e'}{\vs}{\vvar}}{\vsub{\tau_2}{\vs}{\vvar}}{\vsub{\graph'}{\vs}{\vvar}}$.
		By part \ref{lem:subst-touchvert-con},
			$\iskind{\gctx}{\ectx}{\utctx, \hastype{\vvar_f}{\vsty_f}, \hastycl{\vvar_t}{\vsty_t}}{\ectx}{\vsub{\tau_1}{\vs}{\vvar}}{\kwtykind}$
			and $\iskind{\gctx}{\ectx}{\utctx, \hastype{\vvar_f}{\vsty_f}, \hastycl{\vvar_t}{\vsty_t}}{\ectx}{\vsub{\tau_2}{\vs}{\vvar}}{\kwtykind}$.
		Apply \rulename{S:Fun}.

	\item \rulename{S:Roll}.
		Then $e = \kwroll{e'}$
			and $\tau = \kwprec{\convar}{\vvar'}{\tau'}{\vs'}$
			and $\iskind{\gctx}{}{\utctx, \hastype{\vvar}{\vsty}}{\ectx}{\kwprec{\convar}{\vvar'}{\tau'}{\vs'}}{\kwtykind}$
			and $\tywithdag{\gctx}{\uspctx}{\utctx, \hastype{\vvar}{\vsty}}{\ctx}{e'}{\sub{\sub{\tau'}{\vs'}{\vvar'}}
				{\kwxi{\vvar''}{\kwprec{\convar}{\vvar'}{\tau'}{\vvar''}}}{\convar}}{\graph}$.
		By induction,
			$\tywithdag{\gctx}{\uspctx}{\utctx}{\vsub{\ctx}{\vs}{\vvar}}{\vsub{e'}{\vs}{\vvar}}
				{\sub{\sub{\vsub{\tau'}{\vs}{\vvar}}{\vsub{\vs'}{\vs}{\vvar}}{\vvar'}}
				{\kwxi{\vvar''}{\kwprec{\convar}{\vvar'}{\vsub{\tau'}{\vs}{\vvar}}{\vvar''}}}{\convar}}{\vsub{\graph}{\vs}{\vvar}}$.
		By part \ref{lem:subst-touchvert-con},
			$\iskind{\gctx}{}{\utctx}{\ectx}{\kwprec{\convar}{\vvar'}{\vsub{\tau'}{\vs}{\vvar}}{\vsub{\vs'}{\vs}{\vvar}}}{\kwtykind}$.
		Apply \rulename{S:Roll}.

	\item \rulename{S:Future}.
		Then $e = \kwfuture{\vs'}{e'}$
			and $\tau = \kwfutt{\tau'}{\vs'}$
			and $\graph = \leftcomp{\graph'}{\vs'}$
			and $\uspsplit{\uspctx}{\uspctx_1}{\uspctx_2}$
			and $\tywithdag{\gctx}{\uspctx_1}{\utctx, \hastype{\vvar}{\vsty}}{\ctx}{e'}{\tau'}{\graph'}$
			and $\vsistype{\uspctx_2}{\ectx}{\vs'}{\kwvty}$
			and $\vsistype{\ectx}{\utctx, \hastype{\vvar}{\vsty}}{\vs'}{\kwvty}$.
		By the contrapositive of Lemma \ref{lem:usp-split-upstream},
			$\vvar \notin \uspctx_1$
			and $\vvar \notin \uspctx_2$.
		Therefore,
			$\vs'$ does not contain $\vvar$.
		By part \ref{lem:subst-touchvert-vert},
			$\vsistype{\ectx}{\utctx}{\vs'}{\kwvty}$.
		By induction,
			$\tywithdag{\gctx}{\uspctx_1}{\utctx}{\vsub{\ctx}{\vs}{\vvar}}
				{\vsub{e'}{\vs}{\vvar}}{\vsub{\tau'}{\vs}{\vvar}}{\vsub{\graph'}{\vs}{\vvar}}$.
		Apply \rulename{S:Future}.

	\item \rulename{S:Type-Eq}.
		Then $\tywithdag{\gctx}{\uspctx}{\utctx, \hastype{\vvar}{\vsty}}{\ctx}{e}{\tau_1}{\graph}$
			and $\coneq{\gctx}{\utctx, \hastype{\vvar}{\vsty}}{\ectx}{\tau_1}{\tau}{\kwtykind}$.
		By induction,
			$\tywithdag{\gctx}{\uspctx}{\utctx}{\vsub{\ctx}{\vs}{\vvar}}{\vsub{e}{\vs}{\vvar}}{\vsub{\tau_1}{\vs}{\vvar}}{\vsub{\graph}{\vs}{\vvar}}$.
		By part \ref{lem:subst-con-equivalence-vs},
			$\coneq{\gctx}{\utctx}{\ectx}{\vsub{\tau_1}{\vs}{\vvar}}{\vsub{\tau}{\vs}{\vvar}}{\kwtykind}$.
		Apply \rulename{S:Type-Eq}.
    \end{itemize}

  \item
    By induction on the derivation of
    $\tywithdag{\gctx}{\uspctx, \hastype{\vvar}{\vsty}}{\utctx, \hastype{\vvar}{\vsty'}}{\ctx}{e}{\tau}{\graph}$.
    \begin{itemize}
	\item \rulename{S:App}.
		Then $e = \kwapp{}{\kwtapp{e_1}{\vs_f}{\vs_t}}{e_2}$
			and $\tau = \tsub{\tsub{\tau_2}{\vs_f}{\vvar_f}}{\vs_t}{\vvar_t}$
			and $\uspsplit{\uspctx, \hastype{\vvar}{\vsty}}{\uspctx_1}{\uspctx_0}$
			and $\uspsplit{\uspctx_0}{\uspctx_2}{\uspctx_3}$
			and $\graph = \graph_1 \seqcomp \graph_2 \seqcomp \kwtapp{\graph_3}{\vs_f}{\vs_t}$
			and $\tywithdag{\gctx}{\uspctx_1}{\utctx, \hastype{\vvar}{\vsty'}}{\ctx}{e_1}{\kwpi{\hastycl{\vvar_f}{\vsty_f}}{\hastycl{\vvar_t}{\vsty_t}}
				{\kwarrow{\tau_1}{\tau_2}{\kwtapp{\graph_3}{\vvar_f}{\vvar_t}}}}{\graph_1}$
			and $\tywithdag{\gctx}{\uspctx_2}{\utctx, \hastype{\vvar}{\vsty'}}{\ctx}{e_2}
				{\tsub{\tsub{\tau_1}{\vs_f}{\vvar_f}}{\vs_t}{\vvar_t}}{\graph_2}$
			and $\vsistype{\uspctx_3}{\ectx}{\vs_f}{\vsty_f}$
			and $\vsistype{\ectx}{\utctx, \hastype{\vvar}{\vsty'}}{\vs_f}{\vsty_f}$
			and $\vsistype{\ectx}{\utctx, \hastype{\vvar}{\vsty'}}{\vs_t}{\vsty_t}$.
		By part \ref{lem:subst-touchvert-vert},
			$\vsistype{\ectx}{\utctx}{\vsub{\vs_t}{\vs}{\vvar}}{\vsty_t}$
			and $\vsistype{\ectx}{\utctx}{\vsub{\vs_f}{\vs}{\vvar}}{\vsty_f}$.
		By Lemma \ref{lem:split-form} there are four cases:
			\begin{enumerate}
			\item \splituspLexstfull{\uspctx_1'}{\uspctx}{\uspctx_1}{\uspctx_0}{\vvar}{\vsty}.
				By the contrapositive of Lemma \ref{lem:usp-split-upstream},
					$\vvar \notin \uspctx_2$
					and $\vvar \notin \uspctx_3$.
				Therefore,
					$\vs_f$ does not contain $\vvar$.
				By Lemma \ref{lem:split-move},
					\uspsplitexst{\uspctx'''}{\uspctx''}{\uspctx_1'}{\uspctx_0}{\uspctx'}.
				By induction,
					$\tywithdag{\gctx}{\uspctx'''}{\utctx}{\vsub{\ctx}{\vs}{\vvar}}{\vsub{e_1}{\vs}{\vvar}}
						{\kwpi{\hastycl{\vvar_f}{\vsty_f}}{\hastycl{\vvar_t}{\vsty_t}}{\kwarrow{\vsub{\tau_1}{\vs}{\vvar}}{\vsub{\tau_2}{\vs}{\vvar}}
							{\kwtapp{\vsub{\graph_3}{\vs}{\vvar}}{\vvar_f}{\vvar_t}}}}{\vsub{\graph_1}{\vs}{\vvar}}$.
				By part \ref{lem:subst-touchvert-exp},
					\[\tywithdag{\gctx}{\uspctx_2}{\utctx}{\vsub{\ctx}{\vs}{\vvar}}{\vsub{e_2}{\vs}{\vvar}}
						{\tsub{\tsub{\vsub{\tau_1}{\vs}{\vvar}}{\vsub{\vs_f}{\vs}{\vvar}}{\vvar_f}}{\vsub{\vs_t}{\vs}{\vvar}}{\vvar_t}}
						{\vsub{\graph_2}{\vs}{\vvar}}\]
				Apply \rulename{S:App}.
			\item \splituspRexstfull{\uspctx_0'}{\uspctx}{\uspctx_1}{\uspctx_0}{\vvar}{\vsty}.
				By Lemma \ref{lem:split-move},
					\uspsplitexst{\uspctx'''}{\uspctx''}{\uspctx_0'}{\uspctx_1}{\uspctx'}.
				By part \ref{lem:subst-touchvert-exp},
					$\tywithdag{\gctx}{\uspctx_1}{\utctx}{\vsub{\ctx}{\vs}{\vvar}}{\vsub{e_1}{\vs}{\vvar}}
						{\kwpi{\hastycl{\vvar_f}{\vsty_f}}{\hastycl{\vvar_t}{\vsty_t}}{\kwarrow{\vsub{\tau_1}{\vs}{\vvar}}{\vsub{\tau_2}{\vs}{\vvar}}
							{\kwtapp{\vsub{\graph_3}{\vs}{\vvar}}{\vvar_f}{\vvar_t}}}}{\vsub{\graph_1}{\vs}{\vvar}}$.
				By Lemma \ref{lem:split-form} there are four cases:
					\begin{enumerate}
					\item \splituspLexstfull{\uspctx_2'}{\uspctx_0'}{\uspctx_2}{\uspctx_3}{\vvar}{\vsty}.
						Therefore,
							$\vs_f$ does not contain $\vvar$.
				    	By Lemma \ref{lem:split-move},
							\uspsplitexst{\uspctx''''}{\uspctx'''}{\uspctx_2'}{\uspctx_3}{\uspctx'}.\\
						By induction,
							\[\tywithdag{\gctx}{\uspctx''''}{\utctx}{\vsub{\ctx}{\vs}{\vvar}}{\vsub{e_2}{\vs}{\vvar}}
								{\tsub{\tsub{\vsub{\tau_1}{\vs}{\vvar}}{\vsub{\vs_f}{\vs}{\vvar}}{\vvar_f}}{\vsub{\vs_t}{\vs}{\vvar}}{\vvar_t}}
								{\vsub{\graph_2}{\vs}{\vvar}}\]
						Apply \rulename{S:App}.
					\item \splituspRexstfull{\uspctx_3'}{\uspctx_0'}{\uspctx_2}{\uspctx_3}{\vvar}{\vsty}.
						By Lemma \ref{lem:split-move},
							\uspsplitexst{\uspctx''''}{\uspctx'''}{\uspctx_3'}{\uspctx_2}{\uspctx'}.
						By part \ref{lem:subst-touchvert-exp},
							$\tywithdag{\gctx}{\uspctx_2}{\utctx}{\vsub{\ctx}{\vs}{\vvar}}{\vsub{e_2}{\vs}{\vvar}}
								{\tsub{\tsub{\vsub{\tau_1}{\vs}{\vvar}}{\vsub{\vs_f}{\vs}{\vvar}}{\vvar_f}}{\vsub{\vs_t}{\vs}{\vvar}}{\vvar_t}}
								{\vsub{\graph_2}{\vs}{\vvar}}$.
						By part \ref{lem:subst-spawnvert-vert},
							$\vsistype{\uspctx''''}{\ectx}{\vsub{\vs_f}{\vs}{\vvar}}{\vsty_f}$.
						Apply \rulename{S:App}.
					\item \splituspbothexst{\uspctx_2'}{\uspctx_3'}{\uspctx_0'}{\uspctx_2}{\uspctx_3}{\vvar}{\vsty}{\vsty_2}{\vsty_3}.
						By Lemma \ref{lem:vert-split-typing},
							\vstytouspsplitext{\uspctx'}{\uspctx_2''}{\uspctx_3''}{\vs}{\vsty_2}{\vsty_3}.
						By Lemma \ref{lem:split-move},
							\uspsplitbothexst{\uspctx''''}{\uspctx'''''}{\uspctx'''}{\uspctx_2'}{\uspctx_3'}{\uspctx_2''}{\uspctx_3''}.
						By induction,
							$\tywithdag{\gctx}{\uspctx''''}{\utctx}{\vsub{\ctx}{\vs}{\vvar}}{\vsub{e_2}{\vs}{\vvar}}
								{\tsub{\tsub{\vsub{\tau_1}{\vs}{\vvar}}{\vsub{\vs_f}{\vs}{\vvar}}{\vvar_f}}{\vsub{\vs_t}{\vs}{\vvar}}{\vvar_t}}
								{\vsub{\graph_2}{\vs}{\vvar}}$.
						By part \ref{lem:subst-spawnvert-vert},
							$\vsistype{\uspctx'''''}{\ectx}{\vsub{\vs_f}{\vs}{\vvar}}{\vsty_f}$.
						Apply \rulename{S:App}.
					\item \splituspbothcrossexst{\uspctx_2'}{\uspctx_3'}{\uspctx_0'}{\uspctx_2}{\uspctx_3}{\vvar}{\vsty}{\vsty_2}{\vsty_3}.
						By symmetry with the previous case.
					\end{enumerate}

			\item \splituspbothexst{\uspctx_1'}{\uspctx_0'}{\uspctx}{\uspctx_1}{\uspctx_0}{\vvar}{\vsty}{\vsty_1}{\vsty_0}.
				By Lemma \ref{lem:vert-split-typing},
					\vstytouspsplitext{\uspctx'}{\uspctx_1''}{\uspctx_0''}{\vs}{\vsty_1}{\vsty_0}.
				By Lemma \ref{lem:split-move},
					\uspsplitbothexst{\uspctx'''}{\uspctx''''}{\uspctx''}{\uspctx_1'}{\uspctx_0'}{\uspctx_1''}{\uspctx_0''}.
				By induction,
					$\tywithdag{\gctx}{\uspctx'''}{\utctx}{\vsub{\ctx}{\vs}{\vvar}}{\vsub{e_1}{\vs}{\vvar}}
						{\kwpi{\hastycl{\vvar_f}{\vsty_f}}{\hastycl{\vvar_t}{\vsty_t}}{\kwarrow{\vsub{\tau_1}{\vs}{\vvar}}{\vsub{\tau_2}{\vs}{\vvar}}
							{\kwtapp{\vsub{\graph_3}{\vs}{\vvar}}{\vvar_f}{\vvar_t}}}}{\vsub{\graph_1}{\vs}{\vvar}}$.
				By Lemma \ref{lem:split-form} there are four cases:
					\begin{enumerate}
					\item \splituspLexstfull{\uspctx_2'}{\uspctx_0'}{\uspctx_2}{\uspctx_3}{\vvar}{\vsty_0}.
						Therefore,
							$\vs_f$ does not contain $\vvar$.
				    	By Lemma \ref{lem:split-move},
							\uspsplitexst{\uspctx'''''}{\uspctx''''}{\uspctx_2'}{\uspctx_3}{\uspctx_0''}.
						By induction,
							\[\tywithdag{\gctx}{\uspctx'''''}{\utctx}{\vsub{\ctx}{\vs}{\vvar}}{\vsub{e_2}{\vs}{\vvar}}
								{\tsub{\tsub{\vsub{\tau_1}{\vs}{\vvar}}{\vsub{\vs_f}{\vs}{\vvar}}{\vvar_f}}{\vsub{\vs_t}{\vs}{\vvar}}{\vvar_t}}
								{\vsub{\graph_2}{\vs}{\vvar}}\]
						Apply \rulename{S:App}.
					\item \splituspRexstfull{\uspctx_3'}{\uspctx_0'}{\uspctx_2}{\uspctx_3}{\vvar}{\vsty_0}.
						By Lemma \ref{lem:split-move},
							\uspsplitexst{\uspctx'''''}{\uspctx''''}{\uspctx_3'}{\uspctx_2}{\uspctx_0''}.
						By part \ref{lem:subst-touchvert-exp},
							$\tywithdag{\gctx}{\uspctx_2}{\utctx}{\vsub{\ctx}{\vs}{\vvar}}{\vsub{e_2}{\vs}{\vvar}}
								{\tsub{\tsub{\vsub{\tau_1}{\vs}{\vvar}}{\vsub{\vs_f}{\vs}{\vvar}}{\vvar_f}}{\vsub{\vs_t}{\vs}{\vvar}}{\vvar_t}}
								{\vsub{\graph_2}{\vs}{\vvar}}$.
						By part \ref{lem:subst-spawnvert-vert},
							$\vsistype{\uspctx'''''}{\ectx}{\vsub{\vs_f}{\vs}{\vvar}}{\vsty_f}$.
						Apply \rulename{S:App}.
					\item \splituspbothexst{\uspctx_2'}{\uspctx_3'}{\uspctx_0'}{\uspctx_2}{\uspctx_3}{\vvar}{\vsty_0}{\vsty_2}{\vsty_3}.
						By Lemma \ref{lem:vert-split-typing},
							\vstytouspsplitext{\uspctx_0''}{\uspctx_2''}{\uspctx_3''}{\vs}{\vsty_2}{\vsty_3}.
						By Lemma \ref{lem:split-move},
							\uspsplitbothexst{\uspctx'''''}{\uspctx''''''}{\uspctx''''}{\uspctx_2'}{\uspctx_3'}{\uspctx_2''}{\uspctx_3''}.
						By induction,
							$\tywithdag{\gctx}{\uspctx'''''}{\utctx}{\vsub{\ctx}{\vs}{\vvar}}{\vsub{e_2}{\vs}{\vvar}}
								{\tsub{\tsub{\vsub{\tau_1}{\vs}{\vvar}}{\vsub{\vs_f}{\vs}{\vvar}}{\vvar_f}}{\vsub{\vs_t}{\vs}{\vvar}}{\vvar_t}}
								{\vsub{\graph_2}{\vs}{\vvar}}$.
						By part \ref{lem:subst-spawnvert-vert},\\
							$\vsistype{\uspctx''''''}{\ectx}{\vsub{\vs_f}{\vs}{\vvar}}{\vsty_f}$.
						Apply \rulename{S:App}.
					\item \splituspbothcrossexst{\uspctx_2'}{\uspctx_3'}{\uspctx_0'}{\uspctx_2}{\uspctx_3}{\vvar}{\vsty_0}{\vsty_2}{\vsty_3}.
						By symmetry with the previous case.
					\end{enumerate}

			\item \splituspbothcrossexst{\uspctx_1'}{\uspctx_0'}{\uspctx}{\uspctx_1}{\uspctx_0}{\vvar}{\vsty}{\vsty_1}{\vsty_0}.
				By symmetry with the previous case.
			\end{enumerate}
    \end{itemize}
  \end{enumerate}
\end{proof}

Lemmas \ref{lem:vs-preservation} shows that VS type assignments in VS equivalence judgements are consistent with the VS typing judgements:
if $\vs_1$ and $\vs_2$ are equivalent VSs with type $\vsty$, then $\vs_1$ and $\vs_2$ can both be given type $\vsty$ with a VS typing judgement.
Lemma \ref{lem:con-preservation} shows the same property for type constructor equivalence with regards to kinding.

\begin{lemma}\label{lem:vs-preservation}
  If~$\vseq{\utctx}{\vs}{\vs'}{\vsty}$
  then~$\vsistype{\ectx}{\utctx}{\vs}{\vsty}$
  and~$\vsistype{\ectx}{\utctx}{\vs'}{\vsty}$.
\end{lemma}
\begin{proof}
  By induction on the derivation of
  $\vseq{\utctx}{\vs}{\vs'}{\vsty}$.
  \iffull
  \begin{itemize}

\item \rulename{UE:Reflexive}.
	Then $\vs = \vs'$
		and $\vsistype{\ectx}{\utctx}{\vs}{\vsty}$.

\item \rulename{UE:Commutative}.
	Then $\vseq{\utctx}{\vs'}{\vs}{\vsty}$.
	Apply induction.

\item \rulename{UE:Transitive}.
	Then $\vseq{\utctx}{\vs}{\vs''}{\vsty}$
		and $\vseq{\utctx}{\vs''}{\vs'}{\vsty}$.
	Apply induction.

\item \rulename{UE:Pair}.
	Then~$\vs = \kwpair{\vs_1}{\vs_2}$
		and~$\vs' = \kwpair{\vs_1'}{\vs_2'}$
		and $\vsty = \vsprod{\vsty_1}{\isav}{\vsty_2}{\isav}$
		and~$\vseq{\utctx}{\vs_1}{\vs_1'}{\vsty_1}$
		and~$\vseq{\utctx}{\vs_2}{\vs_2'}{\vsty_2}$.
    By induction,
		$\vsistype{\ectx}{\utctx}{\vs_1}{\vsty_1}$
		and $\vsistype{\ectx}{\utctx}{\vs_1'}{\vsty_1}$
		and~$\vsistype{\ectx}{\utctx}{\vs_2}{\vsty_2}$
		and~$\vsistype{\ectx}{\utctx}{\vs_2'}{\vsty_2}$.
	Apply \rulename{OM:Empty} and \rulename{U:Pair} twice.

\item \rulename{UE:OnlyLeftPair}.
	Then~$\vs = \kwpair{\vs_1}{\vs_2}$
		and~$\vs' = \kwpair{\vs_1'}{\vs_2'}$
		and $\vsty = \vsprod{\vsty_1}{\isav}{\vsty_2}{\isunav}$
		and~$\vseq{\utctx}{\vs_1}{\vs_1'}{\vsty_1}$
		and~$\vseq{\utctx'}{\vs_2}{\vs_2'}{\vsty_2}$.
    By induction,
		$\vsistype{\ectx}{\utctx}{\vs_1}{\vsty_1}$
		and $\vsistype{\ectx}{\utctx}{\vs_1'}{\vsty_1}$
		and $\vsistype{\ectx}{\utctx'}{\vs_2}{\vsty_2}$
		and $\vsistype{\ectx}{\utctx'}{\vs_2'}{\vsty_2}$.
	Apply \rulename{U:OnlyLeftPair} twice.

\item \rulename{UE:Fst}.
	Then~$\vs = \kwfst{\vs''}$
		and~$\vs' = \kwfst{\vs'''}$
		and~$\vseq{\utctx}{\vs''}{\vs'''}{\vsprod{\vsty}{\isav}{\vsty_2}{\avail}}$.
    By induction,
		$\vsistype{\ectx}{\utctx}{\vs''}{\vsprod{\vsty}{\isav}{\vsty_2}{\avail}}$
		and $\vsistype{\ectx}{\utctx}{\vs'''}{\vsprod{\vsty}{\isav}{\vsty_2}{\avail}}$.
	Apply~\rulename{U:Fst} twice.

\item \rulename{UE:FstPair}.
	Then~$\vs = \kwfst{\vs''}$
		and~$\vseq{\utctx}{\vs''}{\kwpair{\vs'}{\vs_2}}{\vsprod{\vsty}{\isav}{\vsty_2}{\avail}}$.
    By induction,
		$\vsistype{\ectx}{\utctx}{\vs''}{\vsprod{\vsty}{\isav}{\vsty_2}{\avail}}$
		and $\vsistype{\ectx}{\utctx}{\kwpair{\vs'}{\vs_2}}{\vsprod{\vsty}{\isav}{\vsty_2}{\avail}}$.
	Apply~\rulename{U:Fst}.

	By inversion on the VS typing rules, there are two cases:
		\begin{enumerate}
		\item By inversion on \rulename{U:Subtype} and~\rulename{U:Pair},
				$\vstysubt{\vsprod{\vsty'}{\isav}{\vsty_2'}{\isav}}{\vsprod{\vsty}{\isav}{\vsty_2}{\avail}}$
				and~$\uspsplit{\ectx}{\uspctx_1}{\uspctx_2}$
				and~$\vsistype{\uspctx_1}{\utctx}{\vs'}{\vsty'}$.
			By inversion on \rulename{OM:Empty},
				$\uspctx_1 = \ectx$.
		    By inversion on the VS subtyping rules,
				$\vstysubt{\vsty'}{\vsty}$.
			Apply \rulename{U:Subtype}.
		\item By inversion on \rulename{U:Subtype} and~\rulename{U:OnlyLeftPair},
				$\vstysubt{\vsprod{\vsty'}{\isav}{\vsty_2'}{\isunav}}{\vsprod{\vsty}{\isav}{\vsty_2}{\avail}}$
				and~$\vsistype{\ectx}{\utctx}{\vs'}{\vsty'}$.
		    By inversion on the VS subtyping rules,
				$\vstysubt{\vsty'}{\vsty}$.
			Apply \rulename{U:Subtype}.
		\end{enumerate}

\item \rulename{UE:Subtype}.
	Then $\vseq{\utctx}{\vs}{\vs'}{\vsty'}$
		and $\vstysubt{\vsty'}{\vsty}$.
	By induction,
		$\vsistype{\ectx}{\utctx}{\vs}{\vsty'}$
  		and $\vsistype{\ectx}{\utctx}{\vs'}{\vsty'}$.
	Apply \rulename{U:Subtype} twice.

\item All other cases are by symmetry with another case of the proof.
  \end{itemize}
  \fi
\end{proof}

\begin{lemma}\label{lem:con-preservation}
  If~$\coneq{\gctx}{\utctx}{\vstctx}{\con}{\con'}{\kind}$
  then~$\iskind{\gctx}{}{\utctx}{\vstctx}{\con}{\kind}$
  and~$\iskind{\gctx}{}{\utctx}{\vstctx}{\con'}{\kind}$.
\end{lemma}
\begin{proof}
  By induction on the derivation of
  $\coneq{\gctx}{\utctx}{\vstctx}{\con}{\con'}{\kind}$.
  \iffull
  \begin{itemize}

\item \rulename{CE:Reflexive}.
	Then $\con = \con'$
		and $\iskind{\gctx}{}{\utctx}{\vstctx}{\con}{\kind}$.

\item \rulename{CE:Commutative}.
	By induction.

\item \rulename{CE:Transitive}.
	By induction.

\item \rulename{CE:Prod}.
	Then $\con = \kwprod{\con_1}{\con_2}$
		and $\con' = \kwprod{\con_1'}{\con_2'}$
		and $\kind = \kwtykind$
		and $\coneq{\gctx}{\utctx}{\vstctx}{\con_1}{\con_1'}{\kwtykind}$
		and $\coneq{\gctx}{\utctx}{\vstctx}{\con_2}{\con_2'}{\kwtykind}$.
	By induction,
		$\iskind{\gctx}{}{\utctx}{\vstctx}{\con_1}{\kwtykind}$
		and $\iskind{\gctx}{}{\utctx}{\vstctx}{\con_1'}{\kwtykind}$
		and $\iskind{\gctx}{}{\utctx}{\vstctx}{\con_2}{\kwtykind}$
		and $\iskind{\gctx}{}{\utctx}{\vstctx}{\con_2'}{\kwtykind}$.
	Apply \rulename{K:Prod} twice.

\item \rulename{CE:Sum}.
	Similar to the previous case.

\item \rulename{CE:Fut}.
	Then~$\con = \kwfutt{\con''}{\vs}$
		and~$\con' = \kwfutt{\con'''}{\vs'}$
		and $\kind = \kwtykind$
		and~$\coneq{\gctx}{\utctx}{\vstctx}{\con''}{\con'''}{\kwtykind}$
		and~$\vseq{\utctx}{\vs}{\vs'}{\kwvty}$.
	By induction,
		$\iskind{\gctx}{}{\utctx}{\vstctx}{\con''}{\kwtykind}$
		and $\iskind{\gctx}{}{\utctx}{\vstctx}{\con'''}{\kwtykind}$.
	By Lemma~\ref{lem:vs-preservation},
		$\vsistype{\ectx}{\utctx}{\vs}{\kwvty}$
		and $\vsistype{\ectx}{\utctx}{\vs'}{\kwvty}$.
	Apply~\rulename{K:Fut} twice.

\item \rulename{CE:Lambda}.
	Then~$\con = \kwxi{\hastycl{\vvar}{\vsty}}{\con''}$
		and~$\con' = \kwxi{\hastycl{\vvar}{\vsty}}{\con'''}$
		and $\kind = \kwkindarr{\vsty}{\kwtykind}$
		and~$\coneq{\gctx}{\utctx, \hastype{\vvar}{\vsty}}{\vstctx}{\con''}{\con'''}{\kwtykind}$.
	By induction,
		$\iskind{\gctx}{}{\utctx, \hastype{\vvar}{\vsty}}{\vstctx}{\con''}{\kwtykind}$
		and $\iskind{\gctx}{}{\utctx, \hastype{\vvar}{\vsty}}{\vstctx}{\con'''}{\kwtykind}$.
	Apply \rulename{K:Lambda} twice.

\item \rulename{CE:App}.
	Then~$\con = \kwvapp{\con''}{\vs}$
		and~$\con' = \kwvapp{\con'''}{\vs'}$
		and $\kind = \kwtykind$
		and~$\coneq{\gctx}{\utctx}{\vstctx}{\con''}{\con'''}{\kwkindarr{\vsty}{\kwtykind}}$.
		and $\vseq{\utctx}{\vs}{\vs'}{\vsty}$.
	By induction,
		$\iskind{\gctx}{}{\utctx}{\vstctx}{\con''}{\kwkindarr{\vsty}{\kwtykind}}$
		and $\iskind{\gctx}{}{\utctx}{\vstctx}{\con'''}{\kwkindarr{\vsty}{\kwtykind}}$.
	By Lemma~\ref{lem:vs-preservation},
		$\vsistype{\ectx}{\utctx}{\vs}{\vsty}$
		and $\vsistype{\ectx}{\utctx}{\vs'}{\vsty}$.
	Apply \rulename{K:App} twice.

\item \rulename{CE:BetaEq}.
	Then~$\con = \kwvapp{(\kwxi{\hastycl{\vvar}{\vsty}}{\con''})}{\vs}$
		and~$\con' = \vsub{\con''}{\vs}{\vvar}$
		and $\kind = \kwtykind$
		and~$\iskind{\gctx}{}{\utctx, \hastype{\vvar}{\vsty}}{\vstctx}{\con''}{\kwtykind}$
		and $\vsistype{\ectx}{\utctx}{\vs}{\vsty}$.
	By \rulename{K:Lambda},
		$\iskind{\gctx}{}{\utctx}{\vstctx}{\kwxi{\hastycl{\vvar}{\vsty}}{\con''}}{\kwkindarr{\vsty}{\kwtykind}}$.
	Apply \rulename{K:App}.
	Apply Lemma \ref{lem:subst}.

\item \rulename{CE:Rec}.
	Then~$\con = \kwprec{\convar}{\hastycl{\vvar}{\vsty}}{\con''}{\vs}$
		and~$\con' = \kwprec{\convar}{\hastycl{\vvar}{\vsty}}{\con'''}{\vs'}$
		and~$\coneq{\gctx}{\utctx, \hastype{\vvar}{\vsty}}
             {\vstctx, \haskind{\convar}{\kwkindarr{\vsty}{\kwtykind}}}{\con''}{\con'''}{\kwtykind}$
		and $\kind = \kwtykind$
		and $\vseq{\utctx}{\vs}{\vs'}{\vsty}$.
	By induction,
		$\iskind{\gctx}{}{\utctx, \hastype{\vvar}{\vsty}}
             {\vstctx, \haskind{\convar}{\kwkindarr{\vsty}{\kwtykind}}}{\con''}{\kwtykind}$
		and $\iskind{\gctx}{}{\utctx, \hastype{\vvar}{\vsty}}
             {\vstctx, \haskind{\convar}{\kwkindarr{\vsty}{\kwtykind}}}{\con'''}{\kwtykind}$.
	By Lemma~\ref{lem:vs-preservation},
		$\vsistype{\ectx}{\utctx}{\vs}{\vsty}$
		and $\vsistype{\ectx}{\utctx}{\vs'}{\vsty}$.
	Apply~\rulename{K:Rec} twice.
  \end{itemize}
  \fi
\end{proof}

Lemma \ref{lem:graph-formed} establishes that the type system for \langname{} only gives well-formed types and graph types to expressions. However, this is only true for expressions typed under well-formed contexts. To formalize what it means for a context to be well-formed, we extend the definition of type formation to $\ctx$ contexts in the natural way: $\ctx$ is well-formed if all of the types in it are. The contexts $\gctx$, $\uspctx$, and $\utctx$ are always well-formed. The judgment~$\isctx{\gctx}{\uspctx}{\utctx}{\ctx}$ denotes that $\ctx$ is well-formed under $\gctx$ and $\utctx$. Using the rules for this judgement, we can now prove that if an expression has type~$\tau$ and
graph type~$\graph$ under well-formed contexts, then~$\tau$ has kind $\kwtykind$ and $\graph$ has graph kind $\kgraph$.

\iffull

\begin{mathpar}
  \Rule{CN:Empty}
    {\strut}
              {\isctx{\gctx}{}{\utctx}{\ectx}}
              \and
  \Rule{CN:Elem}
    {\isctx{\gctx}{}{\utctx}{\ctx}\\
    \iskind{\gctx}{}{\utctx}{\ectx}{\tau}{\kwtykind}}
              {\isctx{\gctx}{}{\utctx}{\ctx,\hastype{x}{\tau}}}
\end{mathpar}
\fi

%We can now show that the various static semantics are consistent:
%assuming a well-formed context,

\begin{lemma}\label{lem:graph-formed}
  If~$\isctx{\gctx}{\uspctx}{\utctx}{\ctx}$
  and~$\tywithdag{\gctx}{\uspctx}{\utctx}{\ctx}{e}{\tau}{\graph}$
  then~$\iskind{\gctx}{}{\utctx}{\ectx}{\tau}{\kwtykind}$
  and~$\dagwf{\gctx}{\uspctx}{\utctx}{\graph}{\kgraph}$.
\end{lemma}
\begin{proof}
  By induction on the derivation of
  $\tywithdag{\gctx}{\uspctx}{\utctx}{\ctx}{e}{\tau}{\graph}$.
  \iffull
  We prove some representative cases.
  \begin{itemize}

 \item \rulename{S:Var}.
   Then~$e = x$
		and $\graph = \emptygraph$
		and $\ctx = \ctx', \hastype{x}{\tau}$.
   By inversion on~\rulename{CN:Elem},
   $\iskind{\gctx}{}{\utctx}{\ectx}{\tau}{\kwtykind}$.
   Apply \rulename{DW:Empty}.

  \item \rulename{S:Fun}.
    Then~$\tau = \kwpi{\hastycl{\vvar_f}{\vsty_f}}{\hastycl{\vvar_t}{\vsty_t}}{\kwarrow{\tau_1}{\tau_2}
		{\kwtapp{(\dagrec{\gvar}{\dagpi{\hastype{\vvar_f}{\vsty_f}}{\hastype{\vvar_t}{\vsty_t}}{\graph'}}{})}{\vvar_f}{\vvar_t}}}$
		and $\graph = \emptygraph$
		and $e = \kwfun{\vvar_f}{\vvar_t}{f}{x}{e'}$
    	~and $\tywithdag{\gctx,\hastype{\gvar}{\dagpi{\hastycl{\vvar_f}{\vsty_f}}{\hastycl{\vvar_t}{\vsty_t}}{\kgraph}}}
           {\hastype{\vvar_f}{\vsty_f}}{\utctx, \hastype{\vvar_f}{\vsty_f},\hastype{\vvar_t}{\vsty_t}}
           {\ctx,\hastype{f}{\kwpi{\hastycl{\vvar_f}{\vsty_f}}{\hastycl{\vvar_t}{\vsty_t}}
               {\kwarrow{\tau_1}{\tau_2}{\kwtapp{\gvar}{\vvar_f}{\vvar_t}}}}, \hastype{x}{\tau_1}}
           {e'}{\tau_2}{\graph'}$
    	~and $\iskind{\gctx}{}{\utctx, \hastype{\vvar_f}{\vsty_f}, \hastycl{\vvar_t}{\vsty_t}}{\ectx}{\tau_1}{\kwtykind}$
    	~and $\iskind{\gctx}{}{\utctx, \hastype{\vvar_f}{\vsty_f}, \hastycl{\vvar_t}{\vsty_t}}{\ectx}{\tau_2}{\kwtykind}$.
    By~\rulename{DW:Empty},
    $\dagwf{\gctx}{\uspctx}{\utctx}{\emptygraph}{\kgraph}$.
	By~\rulename{DW:Var},
	$\dagwf{\gctx,\hastype{\gvar}{\dagpi{\hastycl{\vvar_f}{\vsty_f}}
		{\hastycl{\vvar_t}{\vsty_t}}{\kgraph}}}
	{\ectx}{\utctx}{\gvar}{\dagpi{\hastycl{\vvar_f}{\vsty_f}}
		{\hastycl{\vvar_t}{\vsty_t}}{\kgraph}}$.
    By~\rulename{K:Fun},
    $\iskind{\gctx,\hastype{\gvar}{\dagpi{\hastycl{\vvar_f}{\vsty_f}}{\hastycl{\vvar_t}{\vsty_t}}{\kgraph}}}
           {}{\utctx}{\ectx}{\kwarrow{\tau_1}{\tau_2}{\kwtapp{\gvar}{\vvar_f}{\vvar_t}}}{\kwtykind}$.
    By weakening and two applications of~\rulename{CN:Elem},
    $\isctx{\gctx,\hastype{\gvar}{\dagpi{\hastycl{\vvar_f}{\vsty_f}}{\hastycl{\vvar_t}{\vsty_t}}{\kgraph}}}
           {}{\utctx, \hastype{\vvar_f}{\vsty_f}, \hastype{\vvar_t}{\vsty_t}}
           {\ctx,\hastype{f}{\kwpi{\hastycl{\vvar_f}{\vsty_f}}{\hastycl{\vvar_t}{\vsty_t}}
               {\kwarrow{\tau_1}{\tau_2}{\kwtapp{\gvar}{\vvar_f}{\vvar_t}}}}, \hastype{x}{\tau_1}}$.
    By induction,
    $\dagwf{\gctx,\hastype{\gvar}{\dagpi{\hastycl{\vvar_f}{\vsty_f}}{\hastycl{\vvar_t}{\vsty_t}}{\kgraph}}}
           {\hastype{\vvar_f}{\vsty_f}}{\utctx, \hastype{\vvar_f}{\vsty_f}, \hastype{\vvar_t}{\vsty_t}}{\graph'}{\kgraph}$.
    By~\rulename{DW:RecPi},
    $\dagwf{\gctx}{\uspctx}{\utctx}
        {\dagrec{\gvar}{\dagpi{\hastycl{\vvar_f}{\vsty_f}}{\hastycl{\vvar_t}{\vsty_t}}{\graph'}}{}}
           {\dagpi{\hastycl{\vvar_f}{\vsty_f}}{\hastycl{\vvar_t}{\vsty_t}}{\kgraph}}$.
    Apply \rulename{K:Fun}.

  \item \rulename{S:App}.
    Then~$e = \kwapp{}{\kwtapp{e_1}{\vs_f}{\vs_t}}{e_2}$
        and~$\tau = \tsub{\tsub{\tau_2}{\vs_f}{\vvar_f}}{\vs_t}{\vvar_t}$
		and~$\uspsplit{\uspctx}{\uspctx_1}{\uspctx'}$
		and~$\uspsplit{\uspctx'}{\uspctx_2}{\uspctx_3}$
        and~$\graph = \graph_1 \seqcomp \graph_2 \seqcomp
    		 \kwtapp{\graph_3}{\vs_f}{\vs_t}$
        and~$\tywithdag{\gctx}{\uspctx_1}{\utctx}{\ctx}{e_1}{\kwpi{\hastycl{\vvar_f}{\vsty_f}}{\hastycl{\vvar_t}{\vsty_t}}
            {\kwarrow{\tau_1}{\tau_2}{\kwtapp{\graph_3}{\vvar_f}{\vvar_t}}}}{\graph_1}$
        and~$\tywithdag{\gctx}{\uspctx_2}{\utctx}{\ctx}{e_2}{\tsub{\tsub{\tau_1}{\vs_f}{\vvar_f}}{\vs_t}{\vvar_t}}{\graph_2}$
        and~$\vsistype{\uspctx_3}{\ectx}{\vs_f}{\vsty_f}$
        and~$\vsistype{\ectx}{\utctx}{\vs_f}{\vsty_f}$
        and~$\vsistype{\ectx}{\utctx}{\vs_t}{\vsty_t}$.
   By induction,
        $\iskind{\gctx}{}{\utctx}{\ectx}{\kwpi{\hastycl{\vvar_f}{\vsty_f}}{\hastycl{\vvar_t}{\vsty_t}}
            {\kwarrow{\tau_1}{\tau_2}{\kwtapp{\graph_3}{\vvar_f}{\vvar_t}}}}{\kwtykind}$
        and~$\dagwf{\gctx}{\uspctx_1}{\utctx}{\graph_1}{\kgraph}$
        and~$\dagwf{\gctx}{\uspctx_2}{\utctx}{\graph_2}{\kgraph}$.
   By inversion on~\rulename{K:Fun},
        $\iskind{\gctx}{}{\utctx, \hastype{\vvar_f}{\vsty_f}, \hastype{\vvar_t}{\vsty_t}}{\ectx}{\tau_2}{\kwtykind}$
        and~$\dagwf{\gctx}{\ectx}{\utctx}{\graph_3}{\dagpi{\hastycl{\vvar_f}{\vsty_f}}{\hastycl{\vvar_t}{\vsty_t}}{\kgraph}}$.
   By two applications of Lemma~\ref{lem:subst},
        $\iskind{\gctx}{}{\utctx}{\ectx}{\tau}{\kwtykind}$.
%%  By inversion on~\rulename{DW:RecPi},
%%	$\graph_3 = \dagpi{\hastycl{\vvar_f}{\vsty_f}}{\hastycl{\vvar_t}{\vsty_t}}{\graph_4}$
%%	and~$\dagwf{\gctx,\hastype{\gvar}{\dagpi{\hastycl{\vvar_f}{\vsty_f}}{\hastycl{\vvar_t}{\vsty_t}}{\kgraph}}}
%%           {\hastype{\vvar_f}{\vsty_f}}{\utctx, \hastype{\vvar_f}{\vsty_f}, \hastype{\vvar_t}{\vsty_t}}{\graph_4}{\kgraph}$.
%%  By~\rulename{DW:Pi},
%%	$\dagwf{\gctx,\hastype{\gvar}
%%	{\dagpi{\hastycl{\vvar_f}{\vsty_f}}{\hastycl{\vvar_t}{\vsty_t}}{\kgraph}}}
%%           {\ectx}{\utctx}{\graph_3}
%%			{\dagpi{\hastycl{\vvar_f}{\vsty_f}}{\hastycl{\vvar_t}{\vsty_t}}{\kgraph}}$.
%%   By Lemma~\ref{lem:subst},
%%        $\dagwf{\gctx}{\ectx}{\utctx}
%%		{\gsub{\graph_3}{\dagrec{\gvar}{\graph_3}{}}{\gvar}}
%%			{\dagpi{\hastycl{\vvar_f}{\vsty_f}}{\hastycl{\vvar_t}{\vsty_t}}{\kgraph}}$.
   By~\rulename{DW:App},
        $\dagwf{\gctx}{\uspctx_3}{\utctx}
		{\kwtapp{\graph_3}{\vs_f}{\vs_t}}{\kgraph}$.
   Apply~\rulename{DW:Seq} twice.

  \item \rulename{S:Roll}.
	Then~$e = \kwroll{e'}$
        and~$\tau = \kwprec{\convar}{\hastycl{\vvar}{\vsty}}{\tau'}{\vs}$
        and~$\tywithdag{\gctx}{\uspctx}{\utctx}{\ctx}{e}
           {\sub{\sub{\tau'}{\vs}{\vvar}}
             {\kwxi{\hastycl{\vvar'}{\vsty}}{\kwprec{\convar}{\hastycl{\vvar}{\vsty}}{\tau'}{\vvar'}}}
				{\convar}}{\graph}$
		and~$\iskind{\gctx}{}{\utctx}{\ectx}{\tau}{\kwtykind}$.
	By induction.

  \item \rulename{S:Unroll}.
    Then~$e = \kwunroll{e'}$
        and~$\tau = \sub{\sub{\tau'}{\vs}{\vvar}}{\kwxi{\hastycl{\vvar'}{\vsty}}{\kwprec{\convar}{\hastycl{\vvar}{\vsty}}{\tau'}{\vvar'}}}{\convar}$
        and~$\tywithdag{\gctx}{\uspctx}{\utctx}{\ctx}{e'}{\kwprec{\convar}{\hastycl{\vvar}{\vsty}}{\tau'}{\vs}}{\graph}$.
    By induction,
        $\iskind{\gctx}{}{\utctx}{\ectx}{\kwprec{\convar}{\hastycl{\vvar}{\vsty}}{\tau'}{\vs}}{\kwtykind}$
        ~and $\dagwf{\gctx}{\uspctx}{\utctx}{\graph}{\kgraph}$.
    By inversion on~\rulename{K:Rec},
        $\iskind{\gctx}{\uspctx}{\utctx, \hastype{\vvar}{\vsty}}
             {\haskind{\convar}{\kwkindarr{\vsty}{\kwtykind}}}{\tau'}{\kwtykind}$
           and~$\vsistype{\ectx}{\utctx}{\vs}{\vsty}$.
	By \rulename{U:PsiVar},
		$\vsistype{\ectx}{\utctx, \hastype{\vvar'}{\vsty}}{\vvar'}{\vsty}$.
	By weakening and \rulename{K:Rec},
		$\iskind{\gctx}{}{\utctx, \hastype{\vvar'}{\vsty}}{\ectx}{\kwprec{\convar}{\hastycl{\vvar}{\vsty}}{\tau'}{\vvar'}}{\kwtykind}$.
	By \rulename{K:Lambda},
		$\iskind{\gctx}{}{\utctx}{\ectx}{\kwxi{\hastycl{\vvar'}{\vsty}}{\kwprec{\convar}{\hastycl{\vvar}{\vsty}}{\tau'}{\vvar'}}}{\kwkindarr{\vsty}{\kwtykind}}$.
   By two applications of Lemma~\ref{lem:subst}.

  \item \rulename{S:Future}.
    Then $e = \kwfuture{\vs}{e'}$
    ~and $\tau = \kwfutt{\tau'}{\vs}$
    ~and $\graph = \leftcomp{\graph'}{\vs}$
    ~and $\uspsplit{\uspctx}{\uspctx_1}{\uspctx_2}$
    ~and $\tywithdag{\gctx}{\uspctx_1}{\utctx}{\ctx}{e'}{\tau'}{\graph'}$
    ~and $\vsistype{\uspctx_2}{\ectx}{\vs}{\kwvty}$
    ~and $\vsistype{\ectx}{\utctx}{\vs}{\kwvty}$.
    By induction,
        $\iskind{\gctx}{}{\utctx}{\ectx}{\tau'}{\kwtykind}$
        ~and $\dagwf{\gctx}{\uspctx_1}{\utctx}{\graph'}{\kgraph}$.
    Apply~\rulename{K:Fut}.
    Apply~\rulename{DW:Spawn}.

  \item \rulename{S:Touch}.
    Then $e = \kwforce{e'}$
    ~and $\graph = \graph' \seqcomp \touchcomp{\vs}$
    ~and $\tywithdag{\gctx}{\uspctx}{\utctx}{\ctx}{e'}{\kwfutt{\tau}{\vs}}{\graph'}$.
    By induction,
        $\iskind{\gctx}{}{\utctx}{\ectx}{\kwfutt{\tau}{\vs}}{\kwtykind}$
        ~and $\dagwf{\gctx}{\uspctx}{\utctx}{\graph'}{\kgraph}$.
    By inversion on~\rulename{K:Fut},
        $\iskind{\gctx}{}{\utctx}{\ectx}{\tau}{\kwtykind}$
        ~and $\vsistype{\ectx}{\utctx}{\vs}{\kwvty}$.
    By~\rulename{DW:Touch},
        $\dagwf{\gctx}{\ectx}{\utctx}{\touchcomp{\vs}}{\kgraph}$.
    Apply~\rulename{DW:Seq}.

  \item \rulename{S:New}.
    Then $e = \kwnewf{\vvar}{\vsty}{e'}$
    ~and $\graph = \dagnew{\hastycl{\vvar}{\vsty}}{\graph'}$
    ~and $\tywithdag{\gctx}{\uspctx, \hastype{\vvar}{\vsty}}{\utctx, \hastype{\vvar}{\vsty}}{\ctx}{e'}{\tau}{\graph'}$
    ~and $\iskind{\gctx}{}{\utctx}{\ectx}{\tau}{\kwtykind}$.
    By induction,
        $\dagwf{\gctx}{\uspctx, \hastype{\vvar}{\vsty}}{\utctx, \hastype{\vvar}{\vsty}}{\graph'}{\kgraph}$.
    Apply~\rulename{DW:New}.

  \item \rulename{S:Type-Eq}.
    Then $\coneq{\gctx}{\utctx}{\ectx}{\tau_1}{\tau}{\kwtykind}$
    	and $\tywithdag{\gctx}{\uspctx}{\utctx}{\ctx}{e}{\tau_1}{\graph}$.
    By Lemma~\ref{lem:con-preservation},
		$\iskind{\gctx}{}{\utctx}{\ectx}{\tau}{\kwtykind}$.
    By induction.
  \end{itemize}
  \fi
\end{proof}

\section{From Section~\ref{sec:soundness}}\label{app:soundness-proofs}

\begin{figure}
  \centering
  \def \MathparLineskip {\lineskip=0.43cm}
  \begin{mathpar}
    \Rule{UV:SndPair}
         {\vseval{\vs}{\kwpair{\vs_1}{\vs_2}}}
         {\vseval{\kwsnd{\vs}}{\vs_2}}
    \and
    \Rule{UV:SndNotPair}
         {\vseval{\vs}{\vs'}\\
		\vs' \neq \kwpair{\vs_1}{\vs_2}}
         {\vseval{\kwsnd{\vs}}{\kwsnd{\vs'}}}
  \end{mathpar}
  \caption{Rules for evaluating $\kwsnd{\vs}$.}
  \label{fig:eval-snd-vs}
\end{figure}

Figure \ref{fig:eval-snd-vs} gives the rules for evaluating $\kw{snd}$ projections of VS
which are symmetric to those for $\kw{fst}$ projections in Figure \ref{fig:vs-eval-abbr}.
Together, Figures \ref{fig:vs-eval-abbr} and \ref{fig:eval-snd-vs} form the complete ruleset for VS evaluation. 

We now prove some properties involving the interactions between VS typing, VS evaluation, and VS equivalence.
Lemma \ref{lem:vs-progress} proves a combined progress and preservation property for VS evaluation.
Lemma \ref{lem:vs-eval-equiv} shows that evaluation implies equivalence: 
a well-formed VS $\vs$ is equivalent to any VS that $\vs$ evaluates to.
Lemma \ref{lem:vs-eval-transitive} shows that equivalent VSs evaluate to the same VSs: 
if $\vs$ evaluates to $\vs'$, then any VSs equivalent to $\vs$ also evaluate to $\vs'$.

\begin{lemma}\label{lem:vs-progress}
  If~$\vsistype{\uspctx}{\utctx}{\vs}{\vsty}$
  then there exists a $\vs'$ such that $\vseval{\vs}{\vs'}$
  and $\vsistype{\uspctx}{\utctx}{\vs'}{\vsty}$.
\end{lemma}
\begin{proof}
  By induction on the derivation of
  $\vsistype{\uspctx}{\utctx}{\vs}{\vsty}$.
  \iffull
  \begin{itemize}

\item \rulename{U:OmegaVar}, \rulename{U:PsiVar}.
  Then~$\vs = \vvar$. By \rulename{UV:Var}, $\vseval{\vs}{\vs}$. True by premise.

\item \rulename{U:OmegaGen}, \rulename{U:PsiGen}.
  Then~$\vs = \vertgen{\vsty}{\genseed}$. By \rulename{UV:Path}, $\vseval{\vs}{\vs}$. True by premise.

\item \rulename{U:Pair}.
	Then~$\vs = \kwpair{\vs_1}{\vs_2}$
		and $\vsty = \vsprod{\vsty_1}{\isav}{\vsty_2}{\isav}$
		and~$\uspsplit{\uspctx}{\uspctx_1}{\uspctx_2}$
		and~$\vsistype{\uspctx_1}{\utctx}{\vs_1}{\vsty_1}$
		and~$\vsistype{\uspctx_2}{\utctx}{\vs_2}{\vsty_2}$.
	By induction,
		$\vseval{\vs_1}{\vs'_1}$
		and~$\vseval{\vs_2}{\vs'_2}$
		and $\vsistype{\uspctx_1}{\utctx}{\vs_1'}{\vsty_1}$
		and~$\vsistype{\uspctx_2}{\utctx}{\vs_2'}{\vsty_2}$.
	By \rulename{UV:Pair},
		$\vseval{\vs}{\kwpair{\vs_1'}{\vs_2'}}$.
	Apply \rulename{U:Pair}.

\item \rulename{U:OnlyLeftPair}.
	Then~$\vs = \kwpair{\vs_1}{\vs_2}$
		and $\vsty = \vsprod{\vsty_1}{\isav}{\vsty_2}{\isunav}$
		and~$\vsistype{\uspctx}{\utctx}{\vs_1}{\vsty_1}$
		and~$\vsistype{\ectx}{\utctx'}{\vs_2}{\vsty_2}$.
	By induction,
		$\vseval{\vs_1}{\vs'_1}$
		and~$\vseval{\vs_2}{\vs'_2}$
		and $\vsistype{\uspctx}{\utctx}{\vs_1'}{\vsty_1}$
		and~$\vsistype{\ectx}{\utctx'}{\vs_2'}{\vsty_2}$.
	By \rulename{UV:Pair},
		$\vseval{\vs}{\kwpair{\vs_1'}{\vs_2'}}$.
	Apply \rulename{U:OnlyLeftPair}.

\item \rulename{U:OnlyRightPair}.
	By symmetry with the previous case.

\item \rulename{U:Fst}.
	Then~$\vs = \kwfst{\vs_0}$
		and~$\vsistype{\uspctx}{\utctx}{\vs_0}{\vsprod{\vsty}{\isav}{\vsty_2}{\avail}}$.
	By induction,
		$\vseval{\vs_0}{\vs_0'}$ and~$\vsistype{\uspctx}{\utctx}{\vs_0'}{\vsprod{\vsty}{\isav}{\vsty_2}{\avail}}$.
	There are two cases:
	\begin{enumerate}
	\item $\vs_0' \neq \kwpair{\vs_1}{\vs_2}$.
          By~\rulename{UV:FstNotPair},
			$\vseval{\vs}{\kwfst{\vs_0'}}$.
		Apply \rulename{U:Fst}.
	\item $\vs_0' = \kwpair{\vs_1}{\vs_2}$.
          By~\rulename{UV:FstPair},
			$\vseval{\vs}{\vs_1}$.
		By inversion on the VS typing rules, there are two cases:
		\begin{enumerate}
		\item By inversion on \rulename{U:Subtype} and~\rulename{U:Pair},
				$\vstysubt{\vsprod{\vsty'}{\isav}{\vsty_2'}{\isav}}{\vsprod{\vsty}{\isav}{\vsty_2}{\avail}}$
				and~$\uspsplit{\uspctx}{\uspctx_1}{\uspctx_2}$
				and~$\vsistype{\uspctx_1}{\utctx}{\vs_1}{\vsty'}$.
			By Lemma \ref{lem:usp-split-weakening},
				$\vsistype{\uspctx}{\utctx}{\vs_1}{\vsty'}$.
		    By inversion on the VS subtyping rules,
				$\vstysubt{\vsty'}{\vsty}$.
			Apply \rulename{U:Subtype}.
		\item By inversion on \rulename{U:Subtype} and~\rulename{U:OnlyLeftPair},
				$\vstysubt{\vsprod{\vsty'}{\isav}{\vsty_2'}{\isunav}}{\vsprod{\vsty}{\isav}{\vsty_2}{\avail}}$
				and~$\vsistype{\uspctx}{\utctx}{\vs_1}{\vsty'}$.
		    By inversion on the VS subtyping rules,
				$\vstysubt{\vsty'}{\vsty}$.
			Apply \rulename{U:Subtype}.
		\end{enumerate}
	\end{enumerate}

\item \rulename{U:Snd}. 
	By symmetry with the previous case.

\item \rulename{U:Subtype}.
	Then $\vsistype{\uspctx}{\utctx}{\vs}{\vsty'}$
		and $\vstysubt{\vsty'}{\vsty}$.
	By induction,
		$\vseval{\vs}{\vs'}$
		and $\vsistype{\uspctx}{\utctx}{\vs'}{\vsty'}$.
	Apply \rulename{U:Subtype}.
        
%%\item \rulename{U:Inl}.
%%	Then~$\vs = \kwinl{\vs'}$
%%		and~$\vsistype{\uspctx}{\utctx}{\vs'}{\vsty'}$.
%%	By induction,
%%		$\vseval{\vs'}{\vs''}$
%%		and $\vs''\vsnormal$.
%%	Apply~\rulename{UE:Inl} and \rulename{VN:Inl}.
%%
%%\item \rulename{U:InR}. Similar to above.
%
%%\item \rulename{U:Roll}.
%%	Then~$\vs = \kwroll{\vs_0}$
%%		and~$\vsistype{\uspctx}{\utctx}{\vs_0}{\vsty'}$.
%%	By induction,
%%	$\vseval{\vs_0}{\vs_0'}$
%%        and~$\vs_0' \vsnormal$.
%%	Apply~\rulename{UE:Roll} and~\rulename{VN:Roll}.
%%        
%%\item \rulename{U:UnrollRec}.
%%	Then~$\vs = \kwunroll{\vs_0}$
%%		and~$\vsistype{\uspctx}{\utctx}{\vs_0}{\kwvsrec{\vstyvar}{\vsty'}}$.
%%	By induction,
%%		$\vseval{\vs_0}{\vs_0'}$ and~$\vs_0'\vsnormal$.
%%	By Lemma \ref{lem:vs-preservation},
%%		$\vsistype{\uspctx}{\utctx}{\vs_0'}{\kwvsrec{\vstyvar}{\vsty'}}$.
%%	By inversion, there are two cases:
%%	\begin{enumerate}
%%	\item $\vs_0'\vsnormalneg$.
%%          By~\rulename{UE:Unroll},~$\vseval{\vs}{\kwunroll{\vs_0'}}$.
%%          By~\rulename{VNN:Unroll} and~\rulename{VN:VNN},
%%          $\kwunroll{\vs_0'}\vsnormal$.
%%	\item $\vs_0' = \kwroll{\vs_0''}$ and~$\vs_0''\vsnormal$.
%%	  By~\rulename{UE:UnrollRoll},~$\vseval{\vs}{\vs_0''}$.
%%	\end{enumerate}
%%
%%\item \rulename{U:UnrollCoRec}.
%%	Similar to above.
%
%\item \rulename{U:OmegaGen}, \rulename{U:PsiGen}.
%  Then~$\vs\vsnormal$ by rules~\rulename{VNN:Gen}
%  and~\rulename{VN:VNN}.

\end{itemize}
\fi
\end{proof}

\begin{lemma}\label{lem:vs-eval-equiv}
	If $\vsistype{\ectx}{\utctx}{\vs}{\vsty}$
	and~$\vseval{\vs}{\vs'}$
	then~$\vseq{\utctx}{\vs}{\vs'}{\vsty}$.
\end{lemma}
\begin{proof}
  By induction on the derivation of
  $\vsistype{\ectx}{\utctx}{\vs}{\vsty}$.
  \iffull
  \begin{itemize}

\item \rulename{U:OmegaVar}, \rulename{U:PsiVar}.
  Then~$\vs = \vvar$. By inversion on \rulename{UV:Var}, $\vs = \vs'$. Apply \rulename{UE:Reflexive}.

\item \rulename{U:OmegaGen}, \rulename{U:PsiGen}.
  Then~$\vs = \vertgen{\vsty}{\genseed}$. By inversion on \rulename{UV:Path}, $\vs = \vs'$. Apply \rulename{UE:Reflexive}.

\item \rulename{U:Pair}.
	Then~$\vs = \kwpair{\vs_1}{\vs_2}$
		and $\vsty = \vsprod{\vsty_1}{\isav}{\vsty_2}{\isav}$
		and~$\uspsplit{\ectx}{\uspctx_1}{\uspctx_2}$
		and~$\vsistype{\uspctx_1}{\utctx}{\vs_1}{\vsty_1}$
		and~$\vsistype{\uspctx_2}{\utctx}{\vs_2}{\vsty_2}$.
	By inversion on \rulename{OM:Empty},
		$\uspctx_1 = \uspctx_2 = \ectx$.
	By inversion on \rulename{UV:Pair}, 
		$\vs' = \kwpair{\vs_1'}{\vs_2'}$
		and $\vseval{\vs_1}{\vs_1'}$
		and $\vseval{\vs_2}{\vs_2'}$.
	By induction,
		$\vseq{\utctx}{\vs_1}{\vs_1'}{\vsty_1}$
		and $\vseq{\utctx}{\vs_2}{\vs_2'}{\vsty_2}$.
	Apply \rulename{UE:Pair}.

\item \rulename{U:OnlyLeftPair}.
	Then~$\vs = \kwpair{\vs_1}{\vs_2}$
		and $\vsty = \vsprod{\vsty_1}{\isav}{\vsty_2}{\isunav}$
		and~$\vsistype{\ectx}{\utctx}{\vs_1}{\vsty_1}$
		and~$\vsistype{\ectx}{\utctx'}{\vs_2}{\vsty_2}$.
	By inversion on \rulename{UV:Pair}, 
		$\vs' = \kwpair{\vs_1'}{\vs_2'}$
		and $\vseval{\vs_1}{\vs_1'}$
		and $\vseval{\vs_2}{\vs_2'}$.
	By induction,
		$\vseq{\utctx}{\vs_1}{\vs_1'}{\vsty_1}$
		and $\vseq{\utctx'}{\vs_2}{\vs_2'}{\vsty_2}$.
	Apply \rulename{UE:OnlyLeftPair}.

\item \rulename{U:OnlyRightPair}.
	By symmetry with the previous case.

\item \rulename{U:Fst}.
	Then~$\vs = \kwfst{\vs_0}$
		and~$\vsistype{\ectx}{\utctx}{\vs_0}{\vsprod{\vsty}{\isav}{\vsty_2}{\avail}}$.
	By inversion on the VS evaluation rules, there are two cases:
	\begin{enumerate}
	\item By inversion on \rulename{UV:FstNotPair},
		$\vs' = \kwfst{\vs_0'}$
			$\vseval{\vs_0}{\vs_0'}$.
		By induction,
			$\vseq{\utctx}{\vs_0}{\vs_0'}{\vsprod{\vsty}{\isav}{\vsty_2}{\avail}}$.
		Apply \rulename{UE:Fst}.
	\item By inversion on \rulename{UV:FstPair},
			$\vseval{\vs_0}{\kwpair{\vs'}{\vs_2}}$.
		By induction,
			$\vseq{\utctx}{\vs_0}{\kwpair{\vs'}{\vs_2}}{\vsprod{\vsty}{\isav}{\vsty_2}{\avail}}$.
		Apply \rulename{UE:FstPair}.
	\end{enumerate}

\item \rulename{U:Snd}. 
	By symmetry with the previous case.

\item \rulename{U:Subtype}.
	Then $\vsistype{\ectx}{\utctx}{\vs}{\vsty'}$
		and $\vstysubt{\vsty'}{\vsty}$.
	By induction,
		$\vseq{\utctx}{\vs}{\vs'}{\vsty'}$.
	Apply \rulename{UE:Subtype}.
\end{itemize}
\fi
\end{proof}

\begin{lemma}\label{lem:vs-eval-transitive}
	If~$\vseval{\vs}{\vs''}$
		and $\vseq{\utctx}{\vs}{\vs'}{\vsty}$
  	then~$\vseval{\vs'}{\vs''}$.
\end{lemma}

\begin{proof}
	By induction on the derivation of~$\vseq{\utctx}{\vs}{\vs'}{\vsty}$ or $\vseq{\utctx}{\vs'}{\vs}{\vsty}$.
\iffull
  \begin{itemize}
  \item \rulename{UE:Reflexive}. 
	Then~$\vs = \vs'$.
	True by premise of the lemma.

  \item \rulename{UE:Commutative}. 
	Apply induction.

\item \rulename{UE:Transitive}.
	\begin{itemize}
	\item $\vseq{\utctx}{\vs}{\vs'}{\vsty}$.
	    Then~$\vseq{\utctx}{\vs}{\vs'''}{\vsty}$
	    	and $\vseq{\utctx}{\vs'''}{\vs'}{\vsty}$.
		By induction,
			$\vseval{\vs'''}{\vs''}$.
		Apply induction.
	\item $\vseq{\utctx}{\vs'}{\vs}{\vsty}$.
	    By symmetry with the previous case.
	\end{itemize}

\item \rulename{UE:Pair}.
	\begin{itemize}
	\item $\vseq{\utctx}{\vs}{\vs'}{\vsty}$.
	    Then~$\vs = \kwpair{\vs_1}{\vs_2}$
		    and $\vs' = \kwpair{\vs_1'}{\vs_2'}$
	    	and $\vseq{\utctx}{\vs_1}{\vs_1'}{\vsty_1}$
	    	and $\vseq{\utctx}{\vs_2}{\vs_2'}{\vsty_2}$.
		By inversion on \rulename{UV:Pair},
			$\vs'' = \kwpair{\vs_1''}{\vs_2''}$
			and $\vseval{\vs_1}{\vs_1''}$
			and $\vseval{\vs_2}{\vs_2''}$.
		By induction,
			$\vseval{\vs_1'}{\vs_1''}$
			and $\vseval{\vs_2'}{\vs_2''}$.
		Apply \rulename{UV:Pair}.
	\item $\vseq{\utctx}{\vs'}{\vs}{\vsty}$.
	    By symmetry with the previous case.
	\end{itemize}

  \item \rulename{UE:Fst}.
	\begin{itemize}
	\item $\vseq{\utctx}{\vs}{\vs'}{\vsty}$.
	    Then~$\vs = \kwfst{\vs_0}$ 
			and~$\vs' = \kwfst{\vs_0'}$
	    	and~$\vseq{\utctx}{\vs_0}{\vs_0'}{\vsprod{\vsty}{\isav}{\vsty_2}{\avail}}$.
	    By inversion on the VS evaluation rules, there are two cases:
			\begin{itemize}
			\item By inversion on \rulename{UV:FstPair},
					$\vseval{\vs_0}{\kwpair{\vs''}{\vs_2}}$.
				By induction,
					$\vseval{\vs_0'}{\kwpair{\vs''}{\vs_2}}$.
			    Apply \rulename{UV:FstPair}.
			\item By inversion on \rulename{UV:FstNotPair},
				    $\vs'' = \kwfst{\vs_0''}$
					and $\vseval{\vs_0}{\vs_0''}$
					and $\vs_0'' \neq \kwpair{\vs_1}{\vs_2}$.
				By induction,
					$\vseval{\vs_0'}{\vs_0''}$.
				Apply \rulename{UV:FstNotPair}.
			\end{itemize}
	\item $\vseq{\utctx}{\vs'}{\vs}{\vsty}$.
	    By symmetry with the previous case.
	\end{itemize}

  \item \rulename{UE:FstPair}.
	\begin{itemize}
	\item $\vseq{\utctx}{\vs}{\vs'}{\vsty}$.
	    Then~$\vs = \kwfst{\vs_0}$ 
	    	and~$\vseq{\utctx}{\vs_0}{\kwpair{\vs'}{\vs_2}}{\vsprod{\vsty}{\isav}{\vsty_2}{\avail}}$.
	    By inversion on the VS evaluation rules, there are two cases:
			\begin{itemize}
			\item By inversion on \rulename{UV:FstPair},
					$\vseval{\vs_0}{\kwpair{\vs''}{\vs_2'}}$.
				By induction,
					$\vseval{\kwpair{\vs'}{\vs_2}}{\kwpair{\vs''}{\vs_2'}}$.
			    By inversion on \rulename{UV:Pair},
					$\vseval{\vs'}{\vs''}$.
			\item By inversion on \rulename{UV:FstNotPair},
				    $\vs'' = \kwfst{\vs_0''}$
					and $\vseval{\vs_0}{\vs_0''}$
					and $\vs_0'' \neq \kwpair{\vs_1}{\vs_2}$.
				By induction,
					$\vseval{\kwpair{\vs'}{\vs_2}}{\vs_0''}$.
			    By inversion on \rulename{UV:Pair},
					$\vs_0'' = \kwpair{\vs_1}{\vs_2}$,
					which contradicts $\vs_0'' \neq \kwpair{\vs_1}{\vs_2}$,
					so this case does not apply.
			\end{itemize}
	\item $\vseq{\utctx}{\vs'}{\vs}{\vsty}$.
	    Then~$\vs' = \kwfst{\vs_0'}$ 
	    	and~$\vseq{\utctx}{\vs_0'}{\kwpair{\vs}{\vs_2}}{\vsprod{\vsty}{\isav}{\vsty_2}{\avail}}$.
		By Lemma \ref{lem:vs-preservation},
			$\vsistype{\ectx}{\utctx}{\kwpair{\vs}{\vs_2}}{\vsprod{\vsty}{\isav}{\vsty_2}{\avail}}$.
		By inversion on \rulename{U:Subtype} and either \rulename{U:Pair} or \rulename{U:OnlyLeftPair},
			$\vsistype{\ectx}{\utctx}{\vs_2}{\vsty_2'}$.
		By Lemma \ref{lem:vs-progress},
			there exists a $\vs_2'$ such that $\vseval{\vs_2}{\vs_2'}$.
		By \rulename{UV:Pair},
			$\vseval{\kwpair{\vs}{\vs_2}}{\kwpair{\vs''}{\vs_2'}}$.
		By induction,
			$\vseval{\vs_0'}{\kwpair{\vs''}{\vs_2'}}$.
		Apply \rulename{UV:FstPair}.
	\end{itemize}

  \item \rulename{UE:Subtype}.
	Apply induction.

\item All other cases are by symmetry with another case of the proof.
  \end{itemize}
\fi
\end{proof}

We now show that the three-step normalization
process can be performed on any well-formed graph type and results in a set
of well-formed graphs.
The first step is to show that graph type unrolling and conversion to
NBNF both preserve well-formedness of graphs.
This in particular ensures that affine restrictions on vertices are maintained,
which is necessary to ensure that when a graph is unrolled and converted to
NBNF, it can still be expanded into a well-formed graph (recall that duplicate
vertices will make a graph ill-formed).
Lemma~\ref{lem:unr-preserves-typing} shows that graph type unrolling preserves
well-formedness.
Lemma~\ref{lem:nbnf-exists-typing} shows that for a well-formed graph type,
there exists an NBNF, and the NBNF of the graph type is still well-formed.
Finally, we show (Lemma~\ref{lem:exp-well-formed})
that if a well-formed graph type in NBNF is expanded, all of
the graphs in the resulting set are well-formed graphs (in particular, have no
duplicate vertices).

\begin{lemma}\label{lem:unr-preserves-typing}
  If~$\graph \unrstep \graph'$
  and~$\dagwf{\gctx}{\uspctx}{\utctx}{\graph}{\graphkind}$
  then~$\dagwf{\gctx}{\uspctx}{\utctx}{\graph'}{\graphkind}$.
\end{lemma}
\begin{proof}
  By induction on the derivation of~$\graph \unrstep \graph'$.
  \iffull
  \begin{itemize}
  \item \rulename{UR:Seq1}.
    Then~$\graph = \graph_1 \seqcomp \graph_2$
		and $\graph' = \graph_1' \seqcomp \graph_2$
		and $\graph_1 \unrstep \graph_1'$.
    By inversion on~\rulename{DW:Seq},
		$\graphkind = \kgraph$
		and $\uspsplit{\uspctx}{\uspctx_1}{\uspctx_2}$
    	and $\dagwf{\gctx}{\uspctx_1}{\utctx}{\graph_1}{\kgraph}$
    	and $\dagwf{\gctx}{\uspctx_2}{\utctx}{\graph_2}{\kgraph}$.
    By induction,
    	$\dagwf{\gctx}{\uspctx_1}{\utctx}{\graph_1'}{\kgraph}$.
    Apply rule~\rulename{DW:Seq}.
  \item \rulename{UR:Seq2}. Symmetric to above.
  \item \rulename{UR:Par1}. Then~$\graph = \graph_1 \parcomp \graph_2$
		and $\graph' = \graph_1' \parcomp \graph_2$
		and $\graph_1 \unrstep \graph_1'$.
    By inversion on~\rulename{DW:Par},
		$\graphkind = \kgraph$
		and $\uspsplit{\uspctx}{\uspctx_1}{\uspctx_2}$
    	and $\dagwf{\gctx}{\uspctx_1}{\utctx}{\graph_1}{\kgraph}$
    	and $\dagwf{\gctx}{\uspctx_2}{\utctx}{\graph_2}{\kgraph}$.
    By induction,
    	$\dagwf{\gctx}{\uspctx_1}{\utctx}{\graph_1'}{\kgraph}$.
    Apply rule~\rulename{DW:Par}.
  \item \rulename{UR:Par2}. Symmetric to above.
  \item \rulename{UR:Or1}. Then~$\graph = \graph_1 \dagor \graph_2$
		and $\graph' = \graph_1' \dagor \graph_2$
		and $\graph_1 \unrstep \graph_1'$.
    By inversion on~\rulename{DW:Or},
		$\graphkind = \kgraph$
    	and $\dagwf{\gctx}{\uspctx}{\utctx}{\graph_1}{\kgraph}$
    	and $\dagwf{\gctx}{\uspctx}{\utctx}{\graph_2}{\kgraph}$.
    By induction,
    	$\dagwf{\gctx}{\uspctx}{\utctx}{\graph_1'}{\kgraph}$.
    Apply rule~\rulename{DW:Or}.
  \item \rulename{UR:Or2}. Symmetric to above.
  \item \rulename{UR:Future}.
    Then~$\graph = \leftcomp{\graph_0}{\vs}$
    and~$\graph' = \leftcomp{\graph_0'}{\vs}$
    and~$\graph_0 \unrstep \graph_0'$.
    By inversion on~\rulename{DW:Spawn},
    $\uspsplit{\uspctx}{\uspctx_1}{\uspctx_2}$
    and~$\dagwf{\gctx}{\uspctx_1}{\utctx}{\graph_0}{\kgraph}$
    and~$\vsistype{\uspctx_2}{\ectx}{\vs}{\kwvty}$.
    By induction,~$\dagwf{\gctx}{\uspctx_1}{\utctx}{\graph_0'}{\kgraph}$.
    Apply~\rulename{DW:Spawn}.
  \item \rulename{UR:Rec}.
    Then~$\graph = \dagrec{\gvar}{\graph_0}{}$
    	and~$\graph' = \gsub{\graph_0}{\dagrec{\gvar}{\graph_0}{}}{\gvar}$.
    By inversion on~\rulename{DW:RecPi},
    	$\graphkind = \dagpi{\hastycl{\vvar_f}{\vsty_f}}{\hastycl{\vvar_t}{\vsty_t}}{\kgraph}$
		and $\graph_0 = \dagpi{\hastycl{\vvar_f}{\vsty_f}}{\hastycl{\vvar_t}{\vsty_t}}{\graph_0'}$
		and $\dagwf{\gctx,\hastype{\gvar}{\dagpi{\hastycl{\vvar_f}{\vsty_f}}{\hastycl{\vvar_t}{\vsty_t}}{\kgraph}}}
           {\hastype{\vvar_f}{\vsty_f}}{\utctx, \hastype{\vvar_f}{\vsty_f}, \hastype{\vvar_t}{\vsty_t}}{\graph_0'}{\kgraph}$.
    By \rulename{DW:RecPi},
		$\dagwf{\gctx}{\ectx}{\utctx}{\dagrec{\gvar}{\graph_0}{}}{\dagpi{\hastycl{\vvar_f}{\vsty_f}}{\hastycl{\vvar_t}{\vsty_t}}{\kgraph}}$.
	By \rulename{DW:Pi},
		$\dagwf{\gctx,\hastype{\gvar}{\dagpi{\hastycl{\vvar_f}{\vsty_f}}{\hastycl{\vvar_t}{\vsty_t}}{\kgraph}}}
           {\ectx}{\utctx}{\graph_0}{\dagpi{\hastycl{\vvar_f}{\vsty_f}}{\hastycl{\vvar_t}{\vsty_t}}{\kgraph}}$.
	By multiple applications of \rulename{OM:Var} and \rulename{OM:Gen},
		$\uspsplit{\uspctx}{\uspctx}{\ectx}$.
	Apply Lemma~\ref{lem:subst} and Lemma~\ref{lem:usp-split-weakening}.
    \item \rulename{UR:Pi}.
      Then~$\graph = \dagpi{\hastycl{\vvar_f}{\vsty_f}}{\hastycl{\vvar_t}{\vsty_t}}{\graph_0}$
      and~$\graph' = \dagpi{\hastycl{\vvar_f}{\vsty_f}}{\hastycl{\vvar_t}{\vsty_t}}{\graph_0'}$
      and~$\graph_0 \unrstep \graph_0'$.
       By inversion on~\rulename{DW:Pi},
     $\dagwf{\gctx}
           {\hastype{\vvar_f}{\vsty_f}}{\utctx, \hastype{\vvar_f}{\vsty_f}, \hastype{\vvar_t}{\vsty_t}}{\graph_0}{\kgraph}$.
           By induction,
           $\dagwf{\gctx}
           {\hastype{\vvar_f}{\vsty_f}}{\utctx, \hastype{\vvar_f}{\vsty_f}, \hastype{\vvar_t}{\vsty_t}}{\graph_0'}{\kgraph}$.
           Apply~\rulename{DW:Pi}.
      \item \rulename{UR:App}.
        Then~$\graph = \kwtapp{\graph_0}{\vs_f}{\vs_t}$
        and~$\graph' = \kwtapp{\graph_0'}{\vs_f}{\vs_t}$
        and~$\graph_0 \unrstep \graph_0'$.
        By inversion on~\rulename{DW:App},
        $\uspsplit{\uspctx}{\uspctx_1}{\uspctx_2}$
        and~$\dagwf{\gctx}{\uspctx_1}{\utctx}{\graph_0}
        {\dagpi{\hastycl{\vvar_f}{\vsty_f}}{\hastycl{\vvar_t}{\vsty_t}}{\kgraph}}$
   	     and~$\vsistype{\uspctx_2}{\ectx}{\vs_f}{\vsty_f}$
   	     and~$\vsistype{\ectx}{\utctx}{\vs_t}{\vsty_t}$.
        By induction,~$\dagwf{\gctx}{\uspctx_1}{\utctx}{\graph_0'}
        {\dagpi{\hastycl{\vvar_f}{\vsty_f}}{\hastycl{\vvar_t}{\vsty_t}}{\kgraph}}$.
        Apply~\rulename{DW:App}.
      \item \rulename{UR:New}.
        Then~$\graph = \dagnew{\hastycl{\vvar}{\vsty}}{\graph_0}$
        and~$\graph' = \dagnew{\hastycl{\vvar}{\vsty}}{\graph_0'}$
        and~$\graph_0 \unrstep \graph_0'$.
        By inversion on~\rulename{DW:New},
        $\dagwf{\gctx}{\uspctx, \hastycl{\vvar}{\vsty}}
        {\utctx, \hastycl{\vvar}{\vsty}}{\graph_0}{\kgraph}$.
        By induction,~$\dagwf{\gctx}{\uspctx, \hastycl{\vvar}{\vsty}}
        {\utctx, \hastycl{\vvar}{\vsty}}{\graph_0'}{\kgraph}$.
        Apply~\rulename{DW:New}.
  \end{itemize}
  \fi
\end{proof}

\begin{lemma}\label{lem:nbnf-exists-typing}
  If~$\dagwf{\ectx}{\guspctx}{\gutctx}{\graph}{\graphkind}$
  then there exist~$\graph'$,~$\guspctx_0$, and~$\gutctx_0$
  such that~$\graph' = \bnnf{\graph}$
  and~$\guspctx_0$ and~$\gutctx_0$ consist of only fresh generators,
  and~$\dagwf{\ectx}{\guspctx, \guspctx_0}{\gutctx, \gutctx_0}{\graph'}{\graphkind}$.
  %Furthermore, if~$\graphkind = \dagpi{\hastycl{\vvar_f}{\vsty_f}}
  %{\hastycl{\vvar_t}{\vsty_t}}{\graph_0}$,
  %then~$\graph' = \dagpi{\hastycl{\vvar_f}{\vsty_f}}
  %{\hastycl{\vvar_t}{\vsty_t}}{\graph_0'}$ for some~$\graph_0'$.
\end{lemma}
\begin{proof}
  By induction on the derivation
  of~$\dagwf{\ectx}{\guspctx}{\gutctx}{\graph}{\graphkind}$.
  \begin{itemize}
  \item \rulename{DW:Empty}.
    Then~$\graph = \emptygraph$.
    We have~$\bnnf{\graph} = \graph$.
  \item \rulename{DW:Var}. This case cannot occur since~$\gctx$ is empty.
  \item \rulename{DW:Seq}.
    Then~$\graph = \graph_1 \seqcomp \graph_2$
    and~$\uspsplit{\guspctx}{\guspctx_1}{\guspctx_2}$
    and~$\dagwf{\ectx}{\guspctx_1}{\gutctx}{\graph_1}{\kgraph}$
    and~$\dagwf{\ectx}{\guspctx_2}{\gutctx}{\graph_2}{\kgraph}$.
    By induction,
    there exist~$\graph_1'$, $\graph_2'$, $\guspctx_3$, $\guspctx_4$, $\gutctx_3$, and $\gutctx_4$
    such that~$\graph_1' = \bnnf{\graph_1}$
    and~$\graph_2' = \bnnf{\graph_2}$,
	and $\guspctx_3$, $\guspctx_4$, $\gutctx_3$, and $\gutctx_4$ consist of only fresh generators,
    and~$\dagwf{\ectx}{\guspctx_1, \guspctx_3}{\gutctx, \gutctx_3}{\graph_1'}{\kgraph}$
    and~$\dagwf{\ectx}{\guspctx_2, \guspctx_4}{\gutctx, \gutctx_4}{\graph_2'}{\kgraph}$.
    Let~$\graph' = \graph_1' \seqcomp \graph_2'$
    and~$\guspctx_0 = \guspctx_3,\guspctx_4$
    and~$\gutctx_0 = \gutctx_3,\gutctx_4$.
    We have~$\graph' = \bnnf{\graph}$.
    By multiple applications of \rulename{OM:Gen} and \rulename{OM:Commutative}, 
    $\uspsplit{\guspctx, \guspctx_0}{\guspctx_1, \guspctx_3}{\guspctx_2, \guspctx_4}$.
    Apply weakening and~\rulename{DW:Seq}.
  \item \rulename{DW:Par}, \rulename{DW:OR}. Similar to above.
  \item \rulename{DW:RecPi}.
    Then~$\graph = \dagrec{\gvar}{\dagpi{\hastycl{\vvar_f}{\vsty_f}}{\hastycl{\vvar_t}{\vsty_t}}{\graph'}}{}$.
    We have~$\bnnf{\graph} = \graph$.
  \item \rulename{DW:Spawn}.
    Then~$\graph = \leftcomp{\graph''}{\vs}$
    and~$\uspsplit{\guspctx}{\guspctx_1}{\guspctx_2}$
    and~$\dagwf{\gctx}{\guspctx_1}{\gutctx}{\graph''}{\kgraph}$
    and~$\vsistype{\guspctx_2}{\ectx}{\vs}{\kwvty}$.
    By induction, there exists~$\graph'''$, $\guspctx_3$, and $\gutctx_3$ such
    that~$\graph''' = \bnnf{\graph''}$,
	and $\guspctx_3$ and $\gutctx_3$ consist of only fresh generators,
    and~$\dagwf{\ectx}{\guspctx_1,\guspctx_3}{\gutctx, \gutctx_3}{\graph'''}{\kgraph}$.
    By Lemma~\ref{lem:vs-progress}, there exists~$\vs'$ such that
    $\vseval{\vs}{\vs'}$ and~$\vsistype{\guspctx_2}{\ectx}{\vs'}{\kwvty}$.
    We have~$\leftcomp{\graph'''}{\vs'} = \bnnf{\graph}$.
    By multiple applications of \rulename{OM:Gen},
	$\uspsplit{\guspctx, \guspctx_3}{\guspctx_1, \guspctx_3}{\guspctx_2}$.
    Apply~\rulename{DW:Spawn}.
  \item \rulename{DW:Touch}. 
	Then~$\vsistype{\ectx}{\gutctx}{\vs}{\kwvty}$.
    By Lemma~\ref{lem:vs-progress}, there exists~$\vs'$ such that
    $\vseval{\vs}{\vs'}$ and~$\vsistype{\ectx}{\gutctx}{\vs'}{\kwvty}$.
    We have~$\bnnf{\touchcomp{\vs}} = \touchcomp{\vs'}$.
    Apply~\rulename{DW:Touch}.
  \item \rulename{DW:New}.
    Then~$\graph = \dagnew{\hastycl{\vvar}{\vsty}}{\graph''}$
    and~$\graphkind = \kgraph$
    and~$\dagwf{\ectx}{\guspctx, \hastype{\vvar}{\vsty}}{\gutctx, \hastype{\vvar}{\vsty}}
      {\graph''}{\kgraph}$.
    By Lemma~\ref{lem:subst},
    $\dagwf{\ectx}{\guspctx, \hastype{\genseed}{\vsty}}{\gutctx, \hastype{\genseed}{\vsty}}{\graph''[\vertgen{\vsty}{\genseed}/\vvar]}
    {\kgraph}$.
    By induction, there exists~$\graph'$
    such that~$\graph' = \bnnf{\graph''[\vertgen{\vsty}{\genseed}/\vvar]}$
    and~$\dagwf{\ectx}{\guspctx, \hastype{\genseed}{\vsty}, \guspctx_0}{\gutctx, \hastype{\genseed}{\vsty}, \gutctx_0}{\graph'}{\kgraph}$.
  \item \rulename{DW:Pi}.
    Then~$\graph = \dagpi{\hastycl{\vvar_f}{\vsty_f}}{\hastycl{\vvar_t}{\vsty_t}}{\graph''}$.
	We have~$\bnnf{\graph} = \graph$.
  \item \rulename{DW:App}.
    Then~$\graph = \kwtapp{\graph''}{\vs_f}{\vs_t}$
    and~$\uspsplit{\guspctx}{\guspctx_1}{\guspctx_2}$
    and~$\dagwf{\ectx}{\guspctx_1}{\gutctx}{\graph''}
    {\dagpi{\hastycl{\vvar_f}{\vsty_f}}{\hastycl{\vvar_t}{\vsty_t}}{\kgraph}}$
    and~$\vsistype{\guspctx_2}{\ectx}{\vs_f}{\vsty_f}$
    and~$\vsistype{\ectx}{\gutctx}{\vs_f}{\vsty_f}$
    and~$\vsistype{\ectx}{\gutctx}{\vs_t}{\vsty_t}$.
    By induction, there exists~$\graph'''$, $\guspctx_3$, and $\gutctx_3$
    such that~$\graph''' = \bnnf{\graph''}$,
	and $\guspctx_3$ and $\gutctx_3$ consist of only fresh generators,
    and~$\dagwf{\ectx}{\guspctx_1,\guspctx_3}{\gutctx,\gutctx_3}{\graph'''}
    {\dagpi{\hastycl{\vvar_f}{\vsty_f}}{\hastycl{\vvar_t}{\vsty_t}}{\kgraph}}$.
	By inversion on the graph type formation rules, there are two cases:
		\begin{enumerate}
	    \item By inversion on~\rulename{DW:RecPi},
		    	$\graph''' = \dagrec{\gvar}{\dagpi{\hastycl{\vvar_f}{\vsty_f}}{\hastycl{\vvar_t}{\vsty_t}}{\graph_0}}{}$.
			We have~$\bnnf{\graph} = \graph$.
	    \item By inversion on~\rulename{DW:Pi},
				$\graph''' = \dagpi{\hastycl{\vvar_f}{\vsty_f}}{\hastycl{\vvar_t}{\vsty_t}}{\graph_0}$
				and $\dagwf{\ectx}{\guspctx_1,\guspctx_3, \hastype{\vvar_f}{\vsty_f}}
					{\gutctx,\gutctx_3, \hastype{\vvar_f}{\vsty_f}, \hastype{\vvar_t}{\vsty_t}}{\graph_0}{\kgraph}$.
		    By multiple applications of \rulename{OM:Gen}, 
		    $\uspsplit{\guspctx, \guspctx_3}{\guspctx_1, \guspctx_3}{\guspctx_2}$.
		    By Lemma~\ref{lem:subst},
		    $\dagwf{\ectx}{\guspctx, \guspctx_3}{\gutctx, \gutctx_3}{\graph_0[\vs_f,\vs_t/\vvar_f,\vvar_t]}{\kgraph}$. 
		    We have~$\bnnf{\graph} = \bnnf{\graph_0''[\vs_f,\vs_t/\vvar_f,\vvar_t]}$.
			Apply induction.
		\end{enumerate}
%%   \item \rulename{DW:Case-Vert-1}.
%%     Then~$\graph = \kwcase{\vs}{\vvar_1}{\graph_1}{\vvar_2}{\graph_2}$
%% 	    and~$\graphkind = \kgraph$
%% 	    and~$\uspsplit{\guspctx}{\guspctx_1}{\guspctx_2}$
%% 	    and~$\dagwf{\gctx}{\guspctx_1, \hastype{\vvar_1}{\vsty_1}}{\gutctx, \hastype{\vvar_1}{\vsty_1}}{\graph_1}{\kgraph}$
%% 	    and~$\dagwf{\gctx}{\guspctx_2, \hastype{\vvar_2}{\vsty_2}}{\gutctx, \hastype{\vvar_2}{\vsty_2}}{\graph_2}{\kgraph}$
%% 		and $\vsistype{\guspctx_2}{\ectx}{\vs}{\kwsum{\vsty_1}{\vsty_2}}$.
%% %    By Lemma~\ref{lem:vs-progress}
%% %    and Lemma~\ref{lem:vs-preservation},
%% %	there exists $\vs'$ such that
%% %    $\vseval{\vs}{\vs'}$
%% %	and $\vs'\normal$
%% %    and~$\vsistype{\guspctx_2}{\ectx}{\vs'}{\kwsum{\vsty_1}{\vsty_2}}$.
%%     By canonical forms, there are two cases:
%% 		\begin{enumerate}
%% 	    \item $\vs = \kwinl{\vs'}$.
%% 			We have~$\bnnf{\graph} = \bnnf{\vsub{\graph_1}{\vs'}{\vvar_1}}$.
%% 		    By inversion on \rulename{U:Subtype} and the VS subtyping rules and \rulename{U:Inl},
%% 				$\vstysubt{\kwsum{\vsty_1'}{\vsty_2'}}{\kwsum{\vsty_1}{\vsty_2}}$
%% 				and $\vstysubt{\vsty_1'}{\vsty_1}$
%% 				and $\vsistype{\guspctx_2}{\ectx}{\vs'}{\vsty_1'}$.
%% 			By \rulename{U:Subtype},
%% 				$\vsistype{\guspctx_2}{\ectx}{\vs'}{\vsty_1}$.
%% 			By Lemma \ref{lem:subst},
%% 				$\dagwf{\gctx}{\guspctx}{\gutctx}{\vsub{\graph_1}{\vs'}{\vvar_1}}{\kgraph}$.
%% 			By induction.
%% 	    \item $\vs = \kwinr{\vs'}$.
%% 			By symmetry with the previous case.
%% 		\end{enumerate}
%%   \item \rulename{DW:Case-Vert-2}. Similar to above.
  \end{itemize}
\end{proof}

Corollary \ref{lem:split-verts} proves an especially important
property of $\uspctx$ context splitting: it maintains the affine restriction
on well-typed vertices under $\guspctx$ contexts. If $\guspctx$ splits
into $\guspctx_1$ and $\guspctx_2$ (and $\guspctx$ contains at most one
mapping per generator), then the set of VSs $\vs$ where
$\vsistype{\guspctx_1}{\ectx}{\vs}{\kwvty}$ is completely distinct from
that where $\vsistype{\guspctx_2}{\ectx}{\vs}{\kwvty}$, meaning it is
impossible for a future spawned under $\guspctx_1$ to have the same
vertex as a future spawned under $\guspctx_2$.
In order to prove this, we first prove weaker lemmas that apply only to
VPs (Lemmas~\ref{lem:split-types-no-dups} and~\ref{lem:split-verts-vp}).
We then extend the results to general VSs by formally defining the relation
between VSs and VPs.
Figure~\ref{fig:vert-normal} defines the judgment~$\vs \vsnormal$
(by way of an auxiliary judgment~$\vs \vsnormalneg$), and
Lemma~\ref{lem:progress-norm} shows that a VS resulting from evaluation is
normal.
Furthermore (Lemma~\ref{lem:norm-vert-vp}), any normal VS of VS type~$\kwvty$
is a VP.
Because evaluation preserves typing and any well-formed VS under a context $\guspctx$ may be evaluated to
an equivalent VP, this completes the proof of the corollary.

\input{fig-vs-norm}

\begin{lemma}\label{lem:progress-norm}
  If~$\vseval{\vs}{\vs'}$
	then $\vs'\vsnormal$.
\end{lemma}
\begin{proof}
  By induction on the derivation of
  $\vseval{\vs}{\vs'}$.
  \iffull
  \begin{itemize}

\item \rulename{UV:Var}.
	Then $\vs' = \vvar$.
	Apply \rulename{VNN:Var} and \rulename{VN:VNN}.

\item \rulename{UV:Path}.
	Then $\vs' = \vertgen{}{\genseed}$.
	Apply \rulename{VNN:Gen} and \rulename{VN:VNN}.

\item \rulename{UV:Pair}.
	Then $\vs = \kwpair{\vs_1}{\vs_2}$
		and $\vs' = \kwpair{\vs_1'}{\vs_2'}$
		and $\vseval{\vs_1}{\vs_1'}$
		and $\vseval{\vs_2}{\vs_2'}$.
	By induction,
		$\vs_1'\vsnormal$
		and $\vs_2'\vsnormal$.
	Apply \rulename{VN:Pair}.

\item \rulename{UV:FstPair}.
	Then $\vs = \kwfst{\vs''}$
		and $\vseval{\vs''}{\kwpair{\vs'}{\vs_0}}$.
	By induction,
		$\kwpair{\vs'}{\vs_0}\vsnormal$.
	By inversion on \rulename{VN:Pair}
		(inversion on \rulename{VN:VNN} does not apply since a pair cannot be $\vsnormalneg$).

\item \rulename{UV:SndPair}.
	Similar to previous case.

\item \rulename{UV:FstNotPair}.
	Then $\vs = \kwfst{\vs''}$
	     and $\vs' = \kwfst{\vs'''}$
		and $\vseval{\vs''}{\vs'''}$
		and $\vs''' \neq \kwpair{\vs_1}{\vs_2}$.
	By induction,
		$\vs'''\vsnormal$.
	By inversion on \rulename{VN:VNN},
		$\vs'''\vsnormalneg$.
	Apply \rulename{VNN:Fst} and \rulename{VN:VNN}.

\item \rulename{UV:SndNotPair}.
	Similar to previous case.
\end{itemize}
\fi
\end{proof}

\begin{lemma}\label{lem:norm-vert-vp}
  If~$\vsistype{\guspctx}{\ectx}{\vs}{\kwvty}$
	and either $\vs\vsnormal$ or $\vs\vsnormalneg$
	then $\vs = \vspath$.
\end{lemma}
\begin{proof}
  By induction on the derivation of
  $\vs\vsnormal$ or $\vs\vsnormalneg$.
  \iffull
  \begin{itemize}

\item \rulename{VNN:Var}. 
	This case does not apply since it is not possible to type a variable under $\guspctx$.

\item \rulename{VNN:Gen}. 
	Then $\vs = \vertgen{}{\genseed}$.

\item \rulename{VNN:Fst}.
	Then $\vs = \kwfst{\vs'}$
		and $\vs'\vsnormalneg$.
	By induction,
		$\vs' = \vspath'$.

\item \rulename{VNN:Snd}.
	Similar to the previous case.

\item \rulename{VN:VNN}.
	Apply induction.

\item \rulename{VN:Pair}.
	By inversion on the VS typing rules,
		a pair cannot be given the type $\kwvty$,
		so this case does not apply.

\end{itemize}
\fi
\end{proof}

\begin{lemma}\label{lem:one-gen}
For every~$\vspath$, there exists a~$\genseed$ such that
if~$\vsistype{\guspctx}{\ectx}{\vspath}{\vsty}$, then
there exists a~$\vsty'$ such
that~$\hastycl{\genseed}{\vsty'} \in \guspctx$
and~$\vsistype{\hastycl{\genseed}{\vsty'}}{\ectx}{\vspath}{\vsty}$.
\end{lemma}
\begin{proof}
By induction on the structure of $\vspath$.
\begin{itemize}
	\item $\vspath = \vertgen{}{\genseed'}$. 
	Let $\genseed$ = $\genseed'$.
	If $\vsistype{\guspctx}{\ectx}{\vertgen{}{\genseed}}{\vsty}$,
		then by inversion on \rulename{U:Subtype} and \rulename{U:OmegaGen},
		$\hastycl{\genseed}{\vsty'} \in \guspctx$
		and $\vstysubt{\vsty'}{\vsty}$.
	Apply \rulename{U:OmegaGen} and \rulename{U:Subtype}.

	\item $\vspath = \kwfst{\vspath'}$.
	By induction, there exists a $\genseed$ such that
		if~$\vsistype{\guspctx}{\ectx}{\vspath'}{\vsty''}$, then
		there exists a~$\vsty'$ such
		that~$\hastycl{\genseed}{\vsty'} \in \guspctx$
		and~$\vsistype{\hastycl{\genseed}{\vsty'}}{\ectx}{\vspath'}{\vsty''}$.
	If $\vsistype{\guspctx}{\ectx}{\kwfst{\vspath'}}{\vsty}$,
		then by inversion on \rulename{U:Subtype} and \rulename{U:Fst},
		$\vsistype{\guspctx}{\ectx}{\vspath'}{\vsprod{\vsty_1}{\isav}{\vsty_2}{\avail}}$
		and $\vstysubt{\vsty_1}{\vsty}$.
	Therefore,
		there exists a~$\vsty'$ such
		that~$\hastycl{\genseed}{\vsty'} \in \guspctx$
		and~$\vsistype{\hastycl{\genseed}{\vsty'}}{\ectx}{\vspath'}{\vsprod{\vsty_1}{\isav}{\vsty_2}{\avail}}$.
	Apply \rulename{U:Fst} and \rulename{U:Subtype}.

	\item $\vspath = \kwfst{\vspath'}$.
		Similar to the previous case.
\end{itemize}
\end{proof}

%% \begin{lemma}{\label{lem:split-types-diamond}
%% If~$\vstysplit{\vsty}{\vsty_1}{\vsty_2}$
%% and~$\vstysubt{\vsty_1}{\vsty}$ and~$\vstysubt{\vsty_2}{\vsty}$
%% then~$\vsty$ is empty.
%% \end{lemma}
%% By induction on the derivation of~$\vstysplit{\vsty}{\vsty_1}{\vsty_2}$.
%% \begin{itemize}
%% \item \rulename{US:Prod}.
%% Then~$\vsty_1 = \vsprod{\vsty_a}{\isav}{\vsty_b}{\isunav}$
%% and~$\vsty_2 = \vsprod{\vsty_a'}{\isunav}{\vsty_b'}{\isav}$.
%% Proceed by inversion on~$\vstysubt{\vsty_1}{\vsty}$
%% and~$\vstysubt{\vsty_2}{\vsty}$
%% \begin{itemize}
%% \item \rulename{UT:ProdLeft}, \rulename{UT:ProdLeft}.
%% Then~$\vsty = \vsprod{\vsty_c}{\isunav}{\vsty_d}{\avail}$
%% where~$\avail = \isunav = \isav$, a contradiction.
%% \item \rulename{UT:ProdRight}, \rulename{UT:ProdRight}. Symmetric to above.
%% \item \rulename{UT:ProdLeft}, \rulename{UT:ProdRight}.
%% Then~$\vsty$ is empty.
%% \item \rulename{UT:ProdRight}, \rulename{UT:ProdLeft}.
%% Then~$\vsty$ is empty.
%% \item \rulename{UT:Prod}, \rulename{UT:ProdLeft}.
%% Then~$\vsty = \vsprod{\vsty_c}{\isav}{\vsty_d}{\isunav}$
%% and~$\isav = \isunav$.
%% \item \rulename{UT:Prod}, \rulename{UT:ProdRight}.
%% Then~$\vsty = \vsprod{\vsty_c}{\isav}{\vsty_d}{\isunav}$
%% and~$\isav = \isunav$.
%% \item \rulename{UT:ProdLeft}, \rulename{UT:Prod}. Symmetric
%% \item \rulename{UT:ProdRight}, \rulename{UT:Prod}. Symmetric.
%% \item \rulename{UT:Prod}, \rulename{UT:Prod}.
%% This is a contradiction because it requires~$\isav = \isunav$.
%% \end{itemize}
%% \end{itemize}
%% \end{proof}

\begin{lemma}\label{lem:split-types-no-dups}
If~$\vstysplit{\vsty}{\vsty_1}{\vsty_2}$,
then there is no~$\vspath$ such that
$\vsistype{\hastycl{\genseed}{\vsty_1}}{\ectx}{\vspath}{\kwvty}$ and
$\vsistype{\hastycl{\genseed}{\vsty_2}}{\ectx}{\vspath}{\kwvty}$.
\end{lemma}
\begin{proof}
By induction on the derivation of~$\vstysplit{\vsty}{\vsty_1}{\vsty_2}$.
\begin{itemize}
\item \rulename{US:Prod}.
Then~$\vsty = \vsprod{\vsty_a}{\isav}{\vsty_b}{\isav}$
and~$\vsty_1 = \vsprod{\vsty_a}{\isav}{\vsty_b}{\isunav}$
and~$\vsty_2 = \vsprod{\vsty_a}{\isunav}{\vsty_b}{\isav}$.
Consider possible prefixes of~$\vspath$.
\begin{itemize}
\item $\vspath = \vertgen{}{\genseed}$. This would not have the type~$\kwvty$.
\item $\vspath$ begins with~$\kwfst{\vertgen{}{\genseed}}$.
This is not possible because it is not possible to derive
$\vsistype{\hastycl{\genseed}{\vsty_2}}{\ectx}{\vspath}{\kwvty}$.
\item $\vspath$ begins with~$\kwsnd{\vertgen{}{\genseed}}$.
Symmetric to above.
\end{itemize}

\item \rulename{US:SplitBoth}.
Then~$\vsty = \vsprod{\vsty_a}{\isav}{\vsty_b}{\isav}$
and~$\vsty_1 = \vsprod{\vsty_a'}{\isav}{\vsty_b'}{\isav}$
and~$\vsty_2 = \vsprod{\vsty_a''}{\isav}{\vsty_b''}{\isav}$
where~$\vstysplit{\vsty_a}{\vsty_a'}{\vsty_a''}$
and~$\vstysplit{\vsty_b}{\vsty_b'}{\vsty_b''}$.
Consider possible prefixes of~$\vspath$.
\begin{itemize}
\item $\vspath = \vertgen{}{\genseed}$. This would not have the type~$\kwvty$.
\item $\vspath$ begins with $\kwfst{\vertgen{}{\genseed}}$.
We have~$\vsistype{\hastycl{\genseed}{\vsty_1}}{\ectx}{\kwfst{\vertgen{}{\genseed}}}{\vsty_a'}$
and~$\vsistype{\hastycl{\genseed}{\vsty_2}}{\ectx}{\kwfst{\vertgen{}{\genseed}}}{\vsty_a''}$.
Replacing~$\kwfst{\vertgen{}{\genseed}}$ in~$\vspath$
with~$\vertgen{}{\genseed}$ gives us derivations
$\vsistype{\hastycl{\genseed}{\vsty_a'}}{\ectx}{\vspath}{\kwvty}$
and~$\vsistype{\hastycl{\genseed}{\vsty_a''}}{\ectx}{\vspath}{\kwvty}$,
which is a contradiction by induction.
\item $\vspath$ begins with~$\kwsnd{\vertgen{}{\genseed}}$.
Symmetric to above.
\end{itemize}

\item \rulename{US:SplitLeft}.
Then~$\vsty = \vsprod{\vsty_a}{\isav}{\vsty_b}{\avail}$
and~$\vsty_1 = \vsprod{\vsty_a'}{\isav}{\vsty_b}{\avail}$
and~$\vsty_2 = \vsprod{\vsty_a''}{\isav}{\vsty_b}{\isunav}$
where~$\vstysplit{\vsty_a}{\vsty_a'}{\vsty_a''}$.
Consider possible prefixes of~$\vspath$.
\begin{itemize}
\item $\vspath = \vertgen{}{\genseed}$. This would not have the type~$\kwvty$.
\item $\vspath$ begins with $\kwfst{\vertgen{}{\genseed}}$.
We have~$\vsistype{\hastycl{\genseed}{\vsty_1}}{\ectx}{\kwfst{\vertgen{}{\genseed}}}{\vsty_a'}$
and~$\vsistype{\hastycl{\genseed}{\vsty_2}}{\ectx}{\kwfst{\vertgen{}{\genseed}}}{\vsty_a''}$.
Replacing~$\kwfst{\vertgen{}{\genseed}}$ in~$\vspath$
with~$\vertgen{}{\genseed}$ gives us derivations
$\vsistype{\hastycl{\genseed}{\vsty_a'}}{\ectx}{\vspath}{\kwvty}$
and~$\vsistype{\hastycl{\genseed}{\vsty_a''}}{\ectx}{\vspath}{\kwvty}$,
which is a contradiction by induction.
\item $\vspath$ begins with~$\kwsnd{\vertgen{}{\genseed}}$.
This is not possible because it is not possible to derive
$\vsistype{\hastycl{\genseed}{\vsty_1}}{\ectx}{\vspath}{\kwvty}$.
\end{itemize}

\item \rulename{US:SplitRight}. Symmetric to above.
\item \rulename{US:Corecursive}. By induction.

\item \rulename{US:Subtype}. 
	Then $\vstysplit{\vsty}{\vsty_1'}{\vsty_2}$
		and $\vstysubt{\vsty_1'}{\vsty_1}$.
	If a VP existed such that
		$\vsistype{\hastycl{\genseed}{\vsty_1}}{\ectx}{\vspath}{\kwvty}$
		and $\vsistype{\hastycl{\genseed}{\vsty_2}}{\ectx}{\vspath}{\kwvty}$,
		then $\vsistype{\hastycl{\genseed}{\vsty_1}}{\ectx}{\vertgen{}{\genseed}}{\vsty_1''}$ must be true
		in order to derive $\vsistype{\hastycl{\genseed}{\vsty_1}}{\ectx}{\vspath}{\kwvty}$.
		By inversion, $\vstysubt{\vsty_1}{\vsty_1''}$, 
		so $\vstysubt{\vsty_1'}{\vsty_1''}$ by \rulename{UT:Transitive}.
		By \rulename{U:Subtype}, $\vsistype{\hastycl{\genseed}{\vsty_1'}}{\ectx}{\vertgen{}{\genseed}}{\vsty_1''}$,
		so $\vsistype{\hastycl{\genseed}{\vsty_1'}}{\ectx}{\vspath}{\kwvty}$,
		which is a contradiction by induction.

\item \rulename{US:Commutative}. By induction.
\end{itemize}
\end{proof}

\begin{lemma}\label{lem:split-verts-vp}
If~$\uspsplit{\guspctx}{\guspctx_1}{\guspctx_2}$
%and~$\guspctx$ contains no duplicate generators,
then there is no~$\vspath$ such that
$\vsistype{\guspctx_1}{\ectx}{\vspath}{\kwvty}$ and
$\vsistype{\guspctx_2}{\ectx}{\vspath}{\kwvty}$.
\end{lemma}
\begin{proof}
Suppose that
$\vsistype{\guspctx_1}{\ectx}{\vspath}{\kwvty}$ and
$\vsistype{\guspctx_2}{\ectx}{\vspath}{\kwvty}$.
By Lemma~\ref{lem:one-gen}, there exists~$\genseed$ such that
$\hastycl{\genseed}{\vsty_1} \in \guspctx_1$
and~$\hastycl{\genseed}{\vsty_2} \in \guspctx_2$
and~$\vsistype{\hastycl{\genseed}{\vsty_1}}{\ectx}{\vspath}{\kwvty}$ and
$\vsistype{\hastycl{\genseed}{\vsty_2}}{\ectx}{\vspath}{\kwvty}$.
By Lemmas~\ref{lem:split-form} and~\ref{lem:usp-split-upstream},
there exists~$\vsty$
such that~$\vstysplit{\vsty}{\vsty_1}{\vsty_2}$
(one case requires an application of \rulename{US:Commutative}).
This contradicts Lemma~\ref{lem:split-types-no-dups}.
\end{proof}

%\begin{lemma}\label{lem:progress-to-vspath}
%  If~$\vsistype{\guspctx}{\ectx}{\vs}{\kwvty}$
%	and $\vseval{\vs}{\vs'}$
%	then there exists a $\vspath$ such that $\vs' = \vspath$.
%\end{lemma}
%\begin{proof}
%  By induction on the derivation of
%  $\vseval{\vs}{\vs'}$.
%  \iffull
%  \begin{itemize}
%
%\item \rulename{UV:Var}.
%	By inversion, this case does not apply since $\vvar$ cannot be typed.
%
%\item \rulename{UV:Path}.
%	Then $\vs' = \vertgen{}{\genseed}$.
%
%\item \rulename{UV:Pair}.
%	By inversion, this case does not apply since $\kwpair{\vs_1}{\vs_2}$ cannot be typed with VS type $\kwvty$.
%
%\item \rulename{UV:FstPair}.
%	Then $\vs = \kwfst{\vs''}$
%		and $\vseval{\vs''}{\kwpair{\vs'}{\vs_0}}$.
%	By inversion on \rulename{U:Fst} and \rulename{U:Subtype},
%		$\vsistype{\guspctx}{\ectx}{\vs''}{\vsprod{\vsty_1}{\isav}{\vsty_2}{\avail}}$
%		and $\vstysubt{\vsty_1}{\kwvty}$.
%
%\item \rulename{UV:SndPair}.
%	Similar to previous case.
%
%\item \rulename{UV:FstNotPair}.
%
%\item \rulename{UV:SndNotPair}.
%	Similar to previous case.
%\end{itemize}
%\fi
%\end{proof}

\begin{corollary}\label{lem:split-verts}
If~$\uspsplit{\guspctx}{\guspctx_1}{\guspctx_2}$
	  and~$\vsistype{\guspctx_1}{\ectx}{\vs}{\kwvty}$
	  and~$\vsistype{\guspctx_2}{\ectx}{\vs}{\kwvty}$
	  this contradicts our assumption that the domain of $\guspctx$ contains no duplicates.
\end{corollary}
\begin{proof}
By Lemmas~\ref{lem:vs-progress}--\ref{lem:vs-eval-transitive}, there
exists~$\vs'$ such that~$\vseval{\vs}{\vs'}$
	  and~$\vsistype{\guspctx_1}{\ectx}{\vs'}{\kwvty}$
	  and~$\vsistype{\guspctx_2}{\ectx}{\vs'}{\kwvty}$.
By Lemmas~\ref{lem:progress-norm} and \ref{lem:norm-vert-vp}, $\vs' = \vspath$.
This contradicts Lemma~\ref{lem:split-verts-vp},
given our assumption that~$\guspctx$ contains
no duplicate generators.
\end{proof}

In the statement of Lemma~\ref{lem:exp-well-formed}, which is the soundness
result for graph type expansion, we use the function~$\vertsof{\cdot}$
to get the set of vertices of a graph.
Formally,~$\vertsof{\dagq{\vertices}{\edges}{\startv}{\sinkv}} = \vertices$.
 
\begin{lemma}\label{lem:exp-well-formed}
  If~$\dagwf{\ectx}{\guspctx}{\gutctx}{\graph}{\kgraph}$
  and~$\graph$ is in NBNF
  and~$\cgraph \in \getgraphs{\graph}$,
  then~$\cgraph$ is a well-formed graph
  and for all~$\vs \in \vertsof{\cgraph}$,
  either~$\vsistype{\guspctx}{\ectx}{\vs}{\kwvty}$
  or~$\vs$ is a fresh vertex.
\end{lemma}
\begin{proof}
  By induction on the derivation
  of~$\dagwf{\ectx}{\guspctx}{\gutctx}{\graph}{\kgraph}$.
  \iffull
  \begin{itemize}
  \item \rulename{DW:Empty}.
    Then~$\graph = \emptygraph$.
    By inversion,~$\cgraph = \emptygraph$
    and~$\vertsof{\cgraph} = \{u\}$, where~$u\fresh$.

  \item \rulename{DW:Seq}.
    Then~$\graph = \graph_1 \seqcomp \graph_2$
    and~$\uspsplit{\guspctx}{\guspctx_1}{\guspctx_2}$
    and~$\dagwf{\ectx}{\guspctx_1}{\gutctx}{\graph_1}{\kgraph}$
    and~$\dagwf{\ectx}{\guspctx_2}{\gutctx}{\graph_2}{\kgraph}$.
	By inversion, $\graph_1$ and $\graph_2$ are in NBNF.
    By inversion,~$\cgraph = \cgraph_1 \seqcomp \cgraph_2$
    where~$\cgraph_1 \in \getgraphs{\graph_1}$
    and~$\cgraph_2 \in \getgraphs{\graph_2}$.
    By induction,~$\cgraph_1$ and~$\cgraph_2$ are well-formed graphs
    and for all non-fresh vertex paths~$\vs \in \vertsof{\cgraph_1}$,
    we have~$\vsistype{\guspctx_1}{\ectx}{\vs}{\kwvty}$
    and for all non-fresh vertex paths~$\vs \in \vertsof{\cgraph_2}$,
    we have~$\vsistype{\guspctx_2}{\ectx}{\vs}{\kwvty}$.
    By Corollary~\ref{lem:split-verts},
    there is no $\vs$ such that $\vsistype{\guspctx_1}{\ectx}{\vs}{\kwvty}$
    and $\vsistype{\guspctx_2}{\ectx}{\vs}{\kwvty}$.
    Without loss of generality, the sets of fresh vertices in~$\cgraph_1$
    and~$\cgraph_2$ are disjoint from each other and the
    contexts, so~$\vertsof{\cgraph_1}$
    and~$\vertsof{\cgraph_2}$ are disjoint,
    and so~$\cgraph$ is a well-formed graph.
    For all non-fresh vertex paths~$\vs \in \vertsof{\cgraph}$, 
	either $\vs \in \vertsof{\cgraph_1}$ or~$\vs \in \vertsof{\cgraph_2}$ 
	and so by Lemma~\ref{lem:usp-split-weakening},
    $\vsistype{\guspctx}{\ectx}{\vs}{\kwvty}$.

  %% \item \rulename{DW:Par}.
  %%   Then~$\graph = \graph_1 \parcomp \graph_2$
  %%   and~$\dagwf{\gctx}{\guspctx_1}{\gutctx}{\graph_1}{\kgraph}$
  %%   and~$\dagwf{\gctx}{\guspctx_2}{\gutctx}{\graph_2}{\kgraph}$
  %%   and~$\uspsplit{\guspctx}{\guspctx_1}{\guspctx_2}$.
  %%   By inversion,~$\cgraph = \cgraph_1 \parcomp \cgraph_2$
  %%   where~$\cgraph_1 \in \getgraphs{\graph_1}$
  %%   and~$\cgraph_2 \in \getgraphs{\graph_2}$.
  %%   By assumption,~$\guspctx_1$ and~$\guspctx_2$
  %%   contain no vertex path more than once.
  %%   By induction,~$\cgraph_1$ and~$\cgraph_2$ are well-formed graphs
  %%   and for all non-fresh vertex paths~$\vs \in \vertsof{\cgraph_1}$,
  %%   we have~$\vsistype{\guspctx_1}{\gutctx}{\vs}{\kwvty}$
  %%   and for all non-fresh vertex paths~$\vs \in \vertsof{\cgraph_2}$,
  %%   we have~$\vsistype{\guspctx_2}{\gutctx}{\vs}{\kwvty}$.
  %%   By Lemma~\ref{lem:split-verts},
  %%   if~$\vsistype{\guspctx_1}{\gutctx}{\vs}{\kwvty}$
  %%   and~$\vsistype{\guspctx_2}{\gutctx}{\vs}{\kwvty}$,
  %%   this would contradict~$\guspctx$ not containing generators more than once.
  %%   Without loss of generality, the sets of fresh vertices in~$\cgraph_1$
  %%   and~$\cgraph_2$ are disjoint from each other and the
  %%   contexts, so~$\vertsof{\cgraph_1}$
  %%   and~$\vertsof{\cgraph_2}$ are disjoint,
  %%   and so~$\cgraph$ is a well-formed graph.
  %%   For all vertex paths~$\vs \in \vertsof{\cgraph}$,
  %%   if~$\vs$ was not fresh, then~$\vs$ was either in~$\cgraph_1$
  %%   or~$\cgraph_2$ and so by XXX,
  %%   $\vsistype{\guspctx_1}{\gutctx}{\vs}{\kwvty}$.

  \item \rulename{DW:Or}.
    Then~$\graph = \graph_1 \dagor \graph_2$
    and~$\dagwf{\gctx}{\guspctx}{\gutctx}{\graph_1}{\kgraph}$
    and~$\dagwf{\gctx}{\guspctx}{\gutctx}{\graph_2}{\kgraph}$.
	By inversion, $\graph_1$ and $\graph_2$ are in NBNF.
    By inversion,~$\cgraph \in \getgraphs{\graph_1}$
    or~$\cgraph \in \getgraphs{\graph_2}$. 
	Apply induction in both cases.

  \item \rulename{DW:Spawn}.
    Then~$\graph = \leftcomp{\graph'}{\vs'}$
    and~$\dagwf{\gctx}{\guspctx_1}{\gutctx}{\graph'}{\kgraph}$
    and~$\vsistype{\guspctx_2}{\ectx}{\vs'}{\kwvty}$
    and~$\uspsplit{\guspctx}{\guspctx_1}{\guspctx_2}$.
	By inversion, $\graph'$ is in NBNF.
    By inversion,~$\cgraph = \leftcomp{\cgraph'}{\vs'}$
    where~$\cgraph' \in \getgraphs{\graph'}$.
    By induction,~$\cgraph'$ is a well-formed graph
    and for all non-fresh vertex paths~$\vs \in \vertsof{\cgraph'}$,
    we have~$\vsistype{\guspctx_1}{\ectx}{\vs}{\kwvty}$.
    By Corollary~\ref{lem:split-verts},
    there is no $\vs$ such that $\vsistype{\guspctx_1}{\ectx}{\vs}{\kwvty}$
    and $\vsistype{\guspctx_2}{\ectx}{\vs}{\kwvty}$.
    Therefore, $\vs' \not\in \vertsof{\cgraph'}$.
    Thus, $\cgraph$ is a well-formed graph.
    For all non-fresh vertex paths~$\vs \in \vertsof{\cgraph}$,
    either $\vs \in \vertsof{\cgraph'}$
    or $\vs = \vs'$ and so by Lemma~\ref{lem:usp-split-weakening},
    $\vsistype{\guspctx}{\ectx}{\vs}{\kwvty}$.

  \item \rulename{DW:Touch}.
    Then~$\graph = \touchcomp{\vs}$.
    By inversion,~$\cgraph = \touchcomp{\vs}$
    and~$\vertsof{\cgraph} = \{u\}$, where~$u\fresh$.

  \item \rulename{DW:New}. 
	Then $\graph = \dagnew{\hastycl{\vvar}{\vsty}}{\graph'}$. 
	By inversion, $\graph$ is not in NBNF, so this case does not apply.
    
  \item \rulename{DW:Pi}, \rulename{DW:RecPi}. 
	These cases do not apply because the graph kind is~$\kgraph$.

  \item \rulename{DW:App}.
	Then $\graph = \kwtapp{\graph'}{\vs_f}{\vs_t}$.
    This contradicts~$\cgraph \in \getgraphs{\graph}$
    since~$\getgraphs{\graph} = \emptyset$,
	so this case does not apply.
    
  \end{itemize}
  \fi
\end{proof}

Proof of Theorem~\ref{thm:norm-correct}:
\begin{proof}
  Follows directly from Lemmas~\ref{lem:unr-preserves-typing},
  \ref{lem:nbnf-exists-typing}
  and \ref{lem:exp-well-formed}.
\end{proof}

\input{fig-unroll-graph}

Proof of Theorem~\ref{thm:soundness}:

\iffull
To prove the theorem, we must first state and prove two technical lemmas about
graph type unrolling, whose full rules are given in Figure \ref{fig:unroll-graph}.
The first, Lemma \ref{lem:unroll-gluing}, is a ``framing'' property that allows us to unroll part of a graph
type multiple times while keeping the rest of the graph type the same.
The second, Lemma \ref{lem:graph-sub-unroll}, allows us to commute unrolling and substitution of VS variables.

\begin{lemma}\label{lem:unroll-gluing}\strut
  \begin{enumerate}
  \item If~$\graph_1 \unrstep^* \graph_1'$
    then~$\graph_1 \seqcomp \graph_2 \unrstep^* \graph_1' \seqcomp \graph_2$.
    \label{lem:glue-seq1}

  \item If~$\graph_2 \unrstep^* \graph_2'$
    then~$\graph_1 \seqcomp \graph_2 \unrstep^* \graph_1 \seqcomp \graph_2'$.
    \label{lem:glue-seq2}

  \item If~$\graph_1 \unrstep^* \graph_1'$
    then~$\graph_1 \parcomp \graph_2 \unrstep^* \graph_1' \parcomp \graph_2$.
    \label{lem:glue-par1}

  \item If~$\graph_2 \unrstep^* \graph_2'$
    then~$\graph_1 \parcomp \graph_2 \unrstep^* \graph_1 \parcomp \graph_2'$.
    \label{lem:glue-par2}

  \item If~$\graph_1 \unrstep^* \graph_1'$
    then~$\graph_1 \dagor \graph_2 \unrstep^* \graph_1' \dagor \graph_2$.
    \label{lem:glue-or1}

  \item If~$\graph_2 \unrstep^* \graph_2'$
    then~$\graph_1 \dagor \graph_2 \unrstep^* \graph_1 \dagor \graph_2'$.
    \label{lem:glue-or2}

  \item If~$\graph \unrstep^* \graph'$
    then~$\leftcomp{\graph}{\vs} \unrstep^* \leftcomp{\graph'}{\vs}$.
    \label{lem:glue-spawn}

  \item If~$\graph \unrstep^* \graph'$
    then~$\dagpi{\hastycl{\vvar_f}{\vsty}}{\hastycl{\vvar_t}{\vsty}}{\graph} \unrstep^* 
			\dagpi{\hastycl{\vvar_f}{\vsty}}{\hastycl{\vvar_t}{\vsty}}{\graph'}$.
    \label{lem:glue-pi}

  \item If~$\graph \unrstep^* \graph'$
    then~$\kwtapp{\graph}{\vs_f}{\vs_t} \unrstep^* \kwtapp{\graph'}{\vs_f}{\vs_t}$.
    \label{lem:glue-app}

  \item If~$\graph \unrstep^* \graph'$
    then~$\dagnew{\hastycl{\vvar}{\vsty}}{\graph} \unrstep^* 
			\dagnew{\hastycl{\vvar}{\vsty}}{\graph'}$.
    \label{lem:glue-new}
  \end{enumerate}
\end{lemma}
\begin{proof}\strut
  \begin{enumerate}
  	\item By induction on $\graph_1 \unrstep^* \graph_1'$.
		\begin{itemize}
		\item $\graph_1 = \graph_1'$.
			Therefore,
				$\graph_1 \seqcomp \graph_2 \unrstep^0 \graph_1' \seqcomp \graph_2$.
		\item $\graph_1 \unrstep \graph_1''$ and $\graph_1'' \unrstep^* \graph_1'$.
			By induction, 
				$\graph_1'' \seqcomp \graph_2 \unrstep^* \graph_1' \seqcomp \graph_2$.
			By \rulename{UR:Seq1}, 
				$\graph_1 \seqcomp \graph_2 \unrstep \graph_1'' \seqcomp \graph_2$.
			Therefore,
				$\graph_1 \seqcomp \graph_2 \unrstep^* \graph_1' \seqcomp \graph_2$.
		\end{itemize}
  \item All cases are similar.
  \end{enumerate}
\end{proof}

\begin{lemma}\label{lem:graph-sub-unroll}\strut
  \begin{enumerate}
  \item If~$\vsub{\graph}{\vs}{\vvar} \unrstep \graph'$
 		then there exists a~$\graph''$
		such that~$\graph \unrstep \graph''$
		and~$\vsub{\graph''}{\vs}{\vvar} = \graph'$.
		\label{lem:graph-sub-unroll-one}

  \item If~$\vsub{\graph}{\vs}{\vvar} \unrstep^* \graph'$
 		then there exists a~$\graph''$
		such that~$\graph \unrstep^* \graph''$
		and~$\vsub{\graph''}{\vs}{\vvar} = \graph'$.
		\label{lem:graph-sub-unroll-any}
  \end{enumerate}
\end{lemma}
\begin{proof}\strut
  \begin{enumerate}
  \item By induction on $\vsub{\graph}{\vs}{\vvar} \unrstep \graph'$.
	\begin{itemize}
	\item \rulename{UR:Rec}.
		Then~$\graph' = \gsub{\vsub{\graph_1}{\vs}{\vvar}}{\dagrec{\gvar}{\vsub{\graph_1}{\vs}{\vvar}}{}}{\gvar}
				= \vsub{\gsub{\graph_1}{\dagrec{\gvar}{\graph_1}{}}{\gvar}}{\vs}{\vvar}$
			and~$\vsub{\graph}{\vs}{\vvar} = \dagrec{\gvar}{\vsub{\graph_1}{\vs}{\vvar}}{}$
			and~$\graph = \dagrec{\gvar}{\graph_1}{}$.
		By \rulename{UR:Rec},
			$\graph \unrstep \gsub{\graph_1}{\dagrec{\gvar}{\graph_1}{}}{\gvar}$.

	\item \rulename{UR:Seq1}.
		Then~$\vsub{\graph}{\vs}{\vvar} = \vsub{\graph_1}{\vs}{\vvar} \seqcomp \vsub{\graph_2}{\vs}{\vvar}$
			and~$\graph = \graph_1 \seqcomp \graph_2$
			and~$\graph' = \graph_1' \seqcomp \vsub{\graph_2}{\vs}{\vvar}$
			and~$\vsub{\graph_1}{\vs}{\vvar} \unrstep \graph_1'$.
		By induction, 
			there exists a~$\graph_1''$ 
			such that~$\graph_1 \unrstep \graph_1''$
			and~$\vsub{\graph_1''}{\vs}{\vvar} = \graph_1'$.
		Therefore,
			$\graph' = \vsub{\graph_1''}{\vs}{\vvar} \seqcomp \vsub{\graph_2}{\vs}{\vvar}
				= \vsub{(\graph_1'' \seqcomp \graph_2)}{\vs}{\vvar}$.
		By \rulename{UR:Seq1},
			$\graph \unrstep \graph_1'' \seqcomp \graph_2$.

	\item \rulename{UR:Future}.
		Then~$\vsub{\graph}{\vs}{\vvar} = \leftcomp{(\vsub{\graph_1}{\vs}{\vvar})}{\vsub{\vs'}{\vs}{\vvar}}$
			and~$\graph = \leftcomp{\graph_1}{\vs'}$
			and~$\graph' = \leftcomp{\graph_1'}{\vsub{\vs'}{\vs}{\vvar}}$
			and~$\vsub{\graph_1}{\vs}{\vvar} \unrstep \graph_1'$.
		By induction, 
			there exists a~$\graph_1''$ 
			such that~$\graph_1 \unrstep \graph_1''$
			and~$\vsub{\graph_1''}{\vs}{\vvar} \\= \graph_1'$.
		Therefore,
			$\graph' = \leftcomp{\vsub{\graph_1''}{\vs}{\vvar}}{\vsub{\vs'}{\vs}{\vvar}}
				= \vsub{(\leftcomp{\graph_1''}{\vs'})}{\vs}{\vvar}$.
		By \rulename{UR:Future},
			$\graph \unrstep \leftcomp{\graph_1''}{\vs'}$.
	\end{itemize}

  \item By induction on $\vsub{\graph}{\vs}{\vvar} \unrstep^* \graph'$.
	\begin{itemize}
	\item $\vsub{\graph}{\vs}{\vvar} = \graph'$. 
		Therefore $\graph'' = \graph$.
	
	\item $\vsub{\graph}{\vs}{\vvar} \unrstep \graph_1$ and $\graph_1 \unrstep^* \graph'$.
		By part \ref{lem:graph-sub-unroll-one},
			there exists a~$\graph_1'$ 
			such that~$\graph \unrstep \graph_1'$
			and~$\vsub{\graph_1'}{\vs}{\vvar} = \graph_1$.
		Therefore
			$\vsub{\graph_1'}{\vs}{\vvar} \unrstep^* \graph'$.
		By induction,
			there exists a~$\graph_1''$ 
			such that~$\graph_1' \unrstep^* \graph_1''$
			and~$\vsub{\graph_1''}{\vs}{\vvar} = \graph'$.
		Therefore
			$\graph \unrstep^* \graph_1''$.
	\end{itemize}
  \end{enumerate}
\end{proof}

We now prove the soundness theorem.
\fi

\iffull \begin{proof}[Proof of Theorem~\ref{thm:soundness}]
  \else \begin{proof}
    \fi
  By induction on the derivation of
  $\eval{}{e}{v}{\cgraph}$.
  \iffull
  \begin{itemize}
\item \rulename{C:App}.
	Let $F = \kwfun{\vvar_f}{\vvar_t}{f}{x}{e'}$.
	Then~$e =\kwapp{}{\kwtapp{e_1}{\vs_f}{\vs_t}}{e_2}$
	and~$\cgraph = \cgraph_1 \seqcomp \cgraph_2 \seqcomp \cgraph_3$
	and~$\eval{}{e_1}{F}{\cgraph_1}$
	and~$\eval{}{e_2}{v'}{\cgraph_2}$
	and
        \[\eval{}{\sub{\vsub{\vsub{\sub{e'}{F}{f}}{\vs_f}{\vvar_f}}{\vs_t}{\vvar_t}}{v'}{x}}{v}{\cgraph_3}\]
	By inversion on \rulename{S:Type-Eq} and \rulename{S:App},
	$\coneq{\ectx}{\utctx}{\ectx}{\tau}{\tsub{\tsub{\tau_2}{\vs_f}{\vvar_f}}{\vs_t}{\vvar_t}}{\kwtykind}$
	and
        \[\graph = \graph_1 \seqcomp \graph_2 \seqcomp
			\kwtapp{\graph_3}{\vs_f}{\vs_t}\]
	and~$\uspsplit{\uspctx}{\uspctx_1}{\uspctx'}$
	and~$\uspsplit{\uspctx'}{\uspctx_2}{\uspctx_3}$
	and
        \[\tywithdag{\ectx}{\uspctx_1}{\utctx}{\ectx}{e_1}{\kwpi{\hastycl{\vvar_f}{\vsty_f}}{\hastycl{\vvar_t}{\vsty_t}}
             {\kwarrow{\tau_1}{\tau_2}{\kwtapp{\graph_3}{\vvar_f}{\vvar_t}}}}{\graph_1}\]
	and~$\tywithdag{\ectx}{\uspctx_2}{\utctx}{\ectx}{e_2}{\tsub{\tsub{\tau_1}{\vs_f}{\vvar_f}}{\vs_t}{\vvar_t}}{\graph_2}$
	and~$\vsistype{\uspctx_3}{\ectx}{\vs_f}{\vsty_f}$
	and~$\vsistype{\ectx}{\utctx}{\vs_f}{\vsty_f}$
	and~$\vsistype{\ectx}{\utctx}{\vs_t}{\vsty_t}$.
	By induction,
		\uctxexistsst{\utctx}{\gutctx_1}{F}{\kwpi{\hastycl{\vvar_f}{\vsty_f}}{\hastycl{\vvar_t}{\vsty_t}}
            {\kwarrow{\tau_1}{\tau_2}{\kwtapp{\graph_3}{\vvar_f}{\vvar_t}}}}
	and~\uctxexistsst{\utctx}{\gutctx_2}{v'}{\tsub{\tsub{\tau_1}{\vs_f}{\vvar_f}}{\vs_t}{\vvar_t}}
	and~\gexistsst{\graph_1'}{\graph_1}{\cgraph_1}
	and~\gexistsst{\graph_2'}{\graph_2}{\cgraph_2}.
	By inversion on \rulename{S:Fun},
	\[\graph_3 = \dagrec{\gvar}{\dagpi{\hastype{\vvar_f}{\vsty_f}}{\hastype{\vvar_t}{\vsty_t}}{\graph_4}}{}\]
	and
        \[\tywithdag{\gctx}
           {\hastype{\vvar_f}{\vsty_f}}{\utctx, \gutctx_1, \hastype{\vvar_f}{\vsty_f}, \hastype{\vvar_t}{\vsty_t}}
           {\hastype{f}{\kwpi{\hastycl{\vvar_f}{\vsty_f}}{\hastycl{\vvar_t}{\vsty_t}}
			{\kwarrow{\tau_1}{\tau_2}{\kwtapp{\gvar}{\vvar_f}{\vvar_t}}}}, \hastype{x}{\tau_1}}{e'}{\tau_2}{\graph_4}\]
       where~$\gctx = \hastype{\gvar}{\dagpi{\hastycl{\vvar_f}{\vsty_f}}{\hastycl{\vvar_t}{\vsty_t}}{\kgraph}}$
	and \[\iskind{\ectx}{\ectx}{\utctx, \gutctx_1, \hastype{\vvar_f}{\vsty_f}, \hastype{\vvar_t}{\vsty_t}}{\ectx}{\tau_1}{\kwtykind}\]
	and \[\iskind{\ectx}{\ectx}{\utctx, \gutctx_1, \hastype{\vvar_f}{\vsty_f}, \hastype{\vvar_t}{\vsty_t}}{\ectx}{\tau_2}{\kwtykind}\]
	By Lemma~\ref{lem:graph-formed},
		\[\iskind{\ectx}{}{\utctx, \gutctx_1}{\ectx}{\kwpi{\hastycl{\vvar_f}{\vsty_f}}{\hastycl{\vvar_t}{\vsty_t}}
			{\kwarrow{\tau_1}{\tau_2}{\kwtapp{\graph_3}{\vvar_f}{\vvar_t}}}}{\kwtykind}\]
	By inversion on~\rulename{K:Fun},
		$\dagwf{\ectx}{\ectx}{\utctx, \gutctx_1}{\graph_3}
			{\dagpi{\hastycl{\vvar_f}{\vsty_f}}{\hastycl{\vvar_t}{\vsty_t}}{\kgraph}}$.
%	By inversion on~\rulename{DW:RecPi},
%		$\dagwf{\hastype{\gvar}{\dagpi{\hastycl{\vvar_f}{\vsty_f}}{\hastycl{\vvar_t}{\vsty_t}}{\kgraph}}}
%		{\hastype{\vvar_f}{\vsty_f}}{\utctx_1', \hastype{\vvar_f}{\vsty_f}, \hastype{\vvar_t}{\vsty_t}}{\graph_4}{\kgraph}$.
	By weakening and 5 applications of Lemma~\ref{lem:subst},
		$\tywithdag{\ectx}{\ectx}{\utctx, \gutctx_1, \gutctx_2}{\ectx}{\sub{\vsub{\vsub{\sub{e'}{F}{f}}{\vs_f}{\vvar_f}}{\vs_t}{\vvar_t}}{v'}{x}}
			{\tsub{\tsub{\tau_2}{\vs_f}{\vvar_f}}{\vs_t}{\vvar_t}}
			{\gsub{\gsub{\gsub{\graph_4}{\graph_3}{\gvar}}{\vs_f}{\vvar_f}}{\vs_t}{\vvar_t}}$.
	(Since $\iskind{\ectx}{\ectx}{\utctx, \gutctx_1, \hastype{\vvar_f}{\vsty_f}, \hastype{\vvar_t}{\vsty_t}}{\ectx}{\tau_1}{\kwtykind}$,
			then $\tsub{\tau_1}{\graph_3}{\gvar} \\= \tau_1$, and similar for $\tau_2$.)
	\\By induction,
		\uctxexistsst{\utctx, \gutctx_1, \gutctx_2}{\gutctx_0}
			{v}{\tsub{\tsub{\tau_2}{\vs_f}{\vvar_f}}{\vs_t}{\vvar_t}}
		and~\gexistsst{\graph_4'}{\gsub{\gsub{\gsub{\graph_4}{\graph_3}{\gvar}}
			{\vs_f}{\vvar_f}}{\vs_t}{\vvar_t}}{\cgraph_3}.
	Apply weakening and \rulename{S:Type-Eq}.
	By 2 applications of Lemma \ref{lem:graph-sub-unroll},
        \[\gsubvsunrollexstTwo{\gsub{\graph_4}{\graph_3}{\gvar}}
			{\graph_4'}{\graph_4''}{\vs_f}{\vvar_f}{\vs_t}{\vvar_t}\]
	By 2 applications of Lemma \ref{lem:unroll-gluing},
	\[\kwtapp{(\dagpi{\hastycl{\vvar_f}{\vsty_f}}{\hastycl{\vvar_t}{\vsty_t}}       
				{\gsub{\graph_4}{\graph_3}{\gvar}})}{\vs_f}{\vs_t} \unrstep^* 
			\kwtapp{(\dagpi{\hastycl{\vvar_f}{\vsty_f}}{\hastycl{\vvar_t}{\vsty_t}}{\graph_4''})}{\vs_f}{\vs_t}\]
	By \rulename{UR:Rec},
		$\kwtapp{\graph_3}{\vs_f}{\vs_t} \unrstep \kwtapp{(\dagpi{\hastycl{\vvar_f}{\vsty_f}}{\hastycl{\vvar_t}{\vsty_t}}       
				{\gsub{\graph_4}{\graph_3}{\gvar}})}{\vs_f}{\vs_t}$, so
		\[\kwtapp{\graph_3}{\vs_f}{\vs_t} \unrstep^* 
			\kwtapp{(\dagpi{\hastycl{\vvar_f}{\vsty_f}}{\hastycl{\vvar_t}{\vsty_t}}{\graph_4''})}{\vs_f}{\vs_t}\]
	By 3 applications of Lemma \ref{lem:unroll-gluing},
	\[\graph_1 \seqcomp \graph_2 \seqcomp (\kwtapp{\graph_3}{\vs_f}{\vs_t})
			\unrstep^* \graph_1' \seqcomp \graph_2' \seqcomp 
					\kwtapp{(\dagpi{\hastycl{\vvar_f}{\vsty_f}}{\hastycl{\vvar_t}{\vsty_t}}{\graph_4''})}{\vs_f}{\vs_t}\]
	We have $\bnnf{\kwtapp{(\dagpi{\hastycl{\vvar_f}{\vsty_f}}{\hastycl{\vvar_t}{\vsty_t}}{\graph_4''})}{\vs_f}{\vs_t}} = 
			\bnnf{\gsub{\gsub{\graph_4''}{\vs_f}{\vvar_f}}{\vs_t}{\vvar_t}} = \bnnf{\graph_4'}$.	
	Therefore,
	\[
        \begin{array}{l l }
        & \bnnf{\graph_1' \seqcomp \graph_2' \seqcomp \kwtapp{(\dagpi{\hastycl{\vvar_f}{\vsty_f}}{\hastycl{\vvar_t}{\vsty_t}}{\graph_4''})}{\vs_f}{\vs_t}}\\
			= &\bnnf{\graph_1'} \seqcomp \bnnf{\graph_2'} \seqcomp 
							\bnnf{\kwtapp{(\dagpi{\hastycl{\vvar_f}{\vsty_f}}{\hastycl{\vvar_t}{\vsty_t}}{\graph_4''})}{\vs_f}{\vs_t}}\\
			= & \bnnf{\graph_1'} \seqcomp \bnnf{\graph_2'} \seqcomp \bnnf{\graph_4'}
    \end{array}\]
       We have  
\[
\begin{array}{l l}
& \getgraphs{\bnnf{\graph_1'} \seqcomp \bnnf{\graph_2'} \seqcomp \bnnf{\graph_4'}}\\
		= & \{\cgraph_1' \seqcomp \cgraph_2' \seqcomp \cgraph_3' \mid \cgraph_1' \in \getgraphs{\bnnf{\graph_1'}}, 
				\cgraph_2' \in \getgraphs{\bnnf{\graph_2'}}, \cgraph_3' \in \getgraphs{\bnnf{\graph_4'}}\}
                                \end{array}
                                \]
	Therefore, $\cgraph \in \getgraphs{\bnnf{\graph_1' \seqcomp \graph_2' \seqcomp 
			\kwtapp{(\dagpi{\hastycl{\vvar_f}{\vsty_f}}{\hastycl{\vvar_t}{\vsty_t}}{\graph_4''})}{\vs_f}{\vs_t}}}$.

\item \rulename{C:Pair}.
	Then~$e = \kwpair{e_1}{e_2}$
		and~$v = \kwpair{v_1}{v_2}$
		and~$\cgraph = \cgraph_1 \seqcomp \cgraph_2$
		and~$\eval{}{e_1}{v_1}{\cgraph_1}$
		and~$\eval{}{e_2}{v_2}{\cgraph_2}$.
	By inversion on \rulename{S:Type-Eq} and \rulename{S:Pair},
		$\coneq{\ectx}{\utctx}{\ectx}{\tau}{\kwprod{\tau_1}{\tau_2}}{\kwtykind}$
		and~$\graph = \graph_1 \seqcomp \graph_2$
		and~$\uspsplit{\uspctx}{\uspctx_1}{\uspctx_2}$
		and~$\tywithdag{\ectx}{\uspctx_1}{\utctx}{\ectx}{e_1}{\tau_1}{\graph_1}$
		and~$\tywithdag{\ectx}{\uspctx_2}{\utctx}{\ectx}{e_2}{\tau_2}{\graph_2}$.

	By induction,
		\uctxexistsst{\utctx}{\gutctx_1}{v_1}{\tau_1}
		and~\uctxexistsst{\utctx}{\gutctx_2}{v_2}{\tau_2}
		and~\gexistsst{\graph_1'}{\graph_1}{\cgraph_1}
		and~\gexistsst{\graph_2'}{\graph_2}{\cgraph_2}.
                
	By \rulename{OM:Empty},
		$\uspsplit{\ectx}{\ectx}{\ectx}$.
	By \rulename{S:Pair} and weakening,
		$\tywithdag{\ectx}{\ectx}{\utctx, \gutctx_1, \gutctx_2}{\ectx}{v}{\kwprod{\tau_1}{\tau_2}}{\emptygraph}$
		(Since $\emptygraph \seqcomp \emptygraph = \emptygraph$).
	Apply weakening and \rulename{S:Type-Eq}.
	By 2 applications of Lemma \ref{lem:unroll-gluing},
		$\graph \unrstep^* \graph_1' \seqcomp \graph_2'$.
	We have
        \[\bnnf{\graph_1' \seqcomp \graph_2'} = \bnnf{\graph_1'} \seqcomp \bnnf{\graph_2'}\]
        and
	\[\getgraphs{\bnnf{\graph_1'} \seqcomp \bnnf{\graph_2'}}
		= \{\cgraph_1' \seqcomp \cgraph_2' \mid \cgraph_1' \in \getgraphs{\bnnf{\graph_1'}}, 
				\cgraph_2' \in \getgraphs{\bnnf{\graph_2'}}\}\]
	Therefore, $\cgraph \in \getgraphs{\bnnf{\graph_1' \seqcomp \graph_2'}}$.

\item \rulename{C:Future}.
	Then~$e = \kwfuture{\vs}{e'}$
		and~$v = \kwhandle{\vspath}{v'}$
		and~$\cgraph = \leftcomp{\cgraph'}{\vspath}$
		and~$\eval{}{e'}{v'}{\cgraph'}$
		and~$\vseval{\vs}{\vspath}$.
                
	By inversion on \rulename{S:Type-Eq} and \rulename{S:Future},
		$\coneq{\ectx}{\utctx}{\ectx}{\tau}{\kwfutt{\tau'}{\vs}}{\kwtykind}$
		and~$\graph = \leftcomp{\graph''}{\vs}$
		and~$\uspsplit{\uspctx}{\uspctx_1}{\uspctx_2}$
		and~$\tywithdag{\ectx}{\uspctx_1}{\utctx}{\ectx}{e'}{\tau'}{\graph''}$
		and~$\vsistype{\uspctx_2}{\ectx}{\vs}{\kwvty}$
		and~$\vsistype{\ectx}{\utctx}{\vs}{\kwvty}$.
                
	By induction,
		\uctxexistsst{\utctx}{\gutctx}{v'}{\tau'}
		and~\gexistsst{\graph'''}{\graph''}{\cgraph'}.

	By Lemma \ref{lem:vs-eval-equiv},
		$\vseq{\utctx}{\vs}{\vspath}{\kwvty}$.
	By Lemma \ref{lem:vs-preservation} and weakening,
		$\vsistype{\ectx}{\utctx, \gutctx}{\vspath}{\kwvty}$.
	By \rulename{S:Handle},
		$\tywithdag{\ectx}{\ectx}{\utctx, \gutctx}{\ectx}{v}{\kwfutt{\tau'}{\vspath}}{\emptygraph}$.
                
	By \rulename{CE:Reflexive} and \rulename{CE:Fut},
		$\coneq{\ectx}{\utctx, \gutctx}{\ectx}{\kwfutt{\tau'}{\vs}}{\kwfutt{\tau'}{\vspath}}{\kwtykind}$.
	Apply weakening and \rulename{S:Type-Eq} twice.
        
	By Lemma \ref{lem:unroll-gluing},
		$\graph \unrstep^* \leftcomp{\graph'''}{\vs}$.

We have
\[\bnnf{\leftcomp{\graph'''}{\vs}} = \leftcomp{\bnnf{\graph'''}}{\vspath}
\]
and
\[\getgraphs{\leftcomp{\bnnf{\graph'''}}{\vspath}}
		= \{\leftcomp{\cgraph''}{\vspath} \mid \cgraph'' \in \getgraphs{\bnnf{\graph'''}}\}\]
	Therefore, $\cgraph \in \getgraphs{\bnnf{\leftcomp{\graph'''}{\vs}}}$.

\item \rulename{C:Touch}.
	Then~$e = \kwforce{e'}$
		and~$\cgraph = \cgraph' \seqcomp \touchcomp{\vspath}$
		and~$\eval{}{e'}{\kwhandle{\vspath}{v}}{\cgraph'}$.

	By inversion on \rulename{S:Type-Eq} and \rulename{S:Touch},
		$\coneq{\ectx}{\utctx}{\ectx}{\tau}{\tau'}{\kwtykind}$
		and $\graph = \graph'' \seqcomp \touchcomp{\vs}$
		and
               $\tywithdag{\ectx}{\uspctx}{\utctx}{\ectx}{e'}{\kwfutt{\tau'}{\vs}}{\graph''}$.
	By induction,
		\uctxexistsst{\utctx}{\gutctx}{\kwhandle{\vspath}{v}}{\kwfutt{\tau'}{\vs}}
		and~\gexistsst{\graph'''}{\graph''}{\cgraph'}.
	By inversion on \rulename{S:Type-Eq} and \rulename{S:Handle},
	\[\coneq{\ectx}{\utctx, \gutctx}{\ectx}{\kwfutt{\tau'}{\vs}}{\kwfutt{\tau''}{\vspath}}{\kwtykind}\]
	and
        \[\coneq{\ectx}{\utctx, \gutctx}{\ectx}{\tau'}{\tau''}{\kwtykind}\]
        and
        \[\tywithdag{\ectx}{\ectx}{\utctx, \gutctx}{\ectx}{v}{\tau''}{\emptygraph}\]
	Apply weakening and \rulename{S:Type-Eq} twice.

	By inversion on \rulename{CE:Fut},
		$\vseq{\utctx, \gutctx}{\vs}{\vspath}{\kwvty}$.
	By \rulename{UV:Path} and Lemma \ref{lem:vs-eval-transitive},
		$\vseval{\vs}{\vspath}$.
                
	By Lemma \ref{lem:unroll-gluing},
		$\graph \unrstep^* \graph''' \seqcomp \touchcomp{\vs}$.
        We have
	\[\bnnf{\graph''' \seqcomp \touchcomp{\vs}} = \bnnf{\graph'''} \seqcomp \touchcomp{\vspath}\]
        and
	\[\getgraphs{\bnnf{\graph'''} \seqcomp \touchcomp{\vspath}}
		= \{\cgraph'' \seqcomp \touchcomp{\vspath} \mid \cgraph'' \in \getgraphs{\bnnf{\graph'''}}\}\]
	Therefore, $\cgraph \in \getgraphs{\bnnf{\graph''' \seqcomp \touchcomp{\vs}}}$.

\item \rulename{C:New}.
	Then~$e = \kwnewf{\vvar}{\vsty}{e'}$
		and~$\eval{}{\sub{e'}{\vertgen{\vsty}{\genseed}}{\vvar}}{v}{\cgraph}$
		and~$\genseed' \fresh$.
                
	By inversion on \rulename{S:Type-Eq} and \rulename{S:New},
		$\coneq{\ectx}{\utctx}{\ectx}{\tau}{\tau'}{\kwtykind}$
		and $\graph = \dagnew{\hastycl{\vvar}{\vsty}}{\graph''}$
		and
        \[\tywithdag{\ectx}{\uspctx, \hastype{\vvar}{\vsty}}{\utctx, \hastype{\vvar}{\vsty}}{\ectx}{e'}{\tau'}{\graph''}\]
        
	By \rulename{U:OmegaGen},
		$\vsistype{\hastype{\genseed}{\vsty}}{\ectx}{\vertgen{\vsty}{\genseed}}{\vsty}$.
                
	By \rulename{U:PsiGen},
		$\vsistype{\ectx}{\utctx, \hastype{\genseed}{\vsty}}{\vertgen{\vsty}{\genseed}}{\vsty}$.
                
	By \rulename{OM:Gen},
		$\uspsplit{\uspctx, \hastype{\genseed}{\vsty}}{\uspctx}{\hastype{\genseed'}{\vsty}}$.
                
	By weakening and Lemma \ref{lem:subst},
		$\tywithdag{\ectx}{\uspctx, \hastype{\genseed}{\vsty}}{\utctx, \hastype{\genseed}{\vsty}}{\ectx}
			{\vsub{e'}{\vertgen{\vsty}{\genseed}}{\vvar}}{\tau'}{\vsub{\graph''}{\vertgen{\vsty}{\genseed}}{\vvar}}$.
		(Since $\vvar$ is local to $e$ and $\graph$, 
			$\vsub{\tau'}{\vertgen{\vsty}{\genseed}}{\vvar} = \tau'$).
                        
	By induction,
		\uctxexistsst{\utctx, \hastype{\genseed}{\vsty}}{\gutctx}{v}{\tau'}
		and \gexistsst{\graph'''}{\vsub{\graph''}{\vertgen{\vsty}{\genseed}}{\vvar}}{\cgraph}.
                
	Apply weakening and \rulename{S:Type-Eq}.
        
	By Lemma \ref{lem:graph-sub-unroll},
		$\gsubvsunrollexst{\graph''}{\graph'''}{\graph''''}{\vertgen{\vsty}{\genseed}}{\vvar}$.
                
	By Lemma \ref{lem:unroll-gluing},
		$\graph \unrstep^* \dagnew{\hastycl{\vvar}{\vsty}}{\graph''''}$.

We have \[\bnnf{\dagnew{\hastycl{\vvar}{\vsty}}{\graph''''}} = \bnnf{\vsub{\graph''''}{\vertgen{\vsty}{\genseed}}{\vvar}} = \bnnf{\graph'''}\]
	Therefore, $\cgraph \in \getgraphs{\bnnf{\dagnew{\hastycl{\vvar}{\vsty}}{\graph''''}}}$.
  \end{itemize}
  \fi
\end{proof}

\section{From Section~\ref{sec:infer}}\label{app:infer-proofs}

\begin{figure*}
  \small
  \centering
  \def \MathparLineskip {\lineskip=0.43cm}
  \begin{mathpar}
    \Rule{SV:Type-Eq}
         {\coneq{\ectx}{\uspctx}{\ectx}{\tau_1}{\tau_2}{\kwtykind}\\
            \affinetyped{\uspctx}{v}{\tau_1}}
         {\affinetyped{\uspctx}{v}{\tau_2}}
    \and
    \Rule{SV:Unit}
         {\strut}
         {\affinetyped{\uspctx}{\kwtriv}{\kwunit}}
    \and
    \Rule{SV:Fun}
         {\gctx = \hastype{\gvar}{\dagpi{\hastycl{\vvar_f}{\vsty_f}}{\hastycl{\vvar_t}{\vsty_t}}{\kgraph}}\\
		\ctx = \hastype{f}
	             {\kwpi{\hastycl{\vvar_f}{\vsty_f}}{\hastycl{\vvar_t}{\vsty_t}}
	               {\kwarrow{\tau_1}{\tau_2}{\kwtapp{\gvar}{\vvar_f}{\vvar_t}}}},
	             \hastype{x}{\tau_1}\\
           \gvar\fresh\\
		\utctx = \hastype{\vvar_f}{\vsty_f}, \hastype{\vvar_t}{\vsty_t}\\
           \tywithdag{\gctx}{\hastype{\vvar_f}{\vsty_f}}{\utctx}{\ctx}{e}{\tau_2}{\graph}\\
           \iskind{\ectx}{}{\utctx}{\ectx}{\tau_1}{\kwtykind}\\
           \iskind{\ectx}{}{\utctx}{\ectx}{\tau_2}{\kwtykind}}
         {\affinetyped{\uspctx}
           {\kwfun{\vvar_f}{\vvar_t}{f}{x}{e}}
           {\kwpi{\hastycl{\vvar_f}{\vsty_f}}{\hastycl{\vvar_t}{\vsty_t}}
             {\kwarrow{\tau_1}{\tau_2}
               {\kwtapp{(\dagrec{\gvar}{\dagpi{\hastype{\vvar_f}{\vsty_f}}{\hastype{\vvar_t}{\vsty_t}}
               {\graph}}{})}{\vvar_f}{\vvar_t}}}}
         }
    \and
    \Rule{SV:Pair}
         {\uspsplit{\uspctx}{\uspctx_1}{\uspctx_2}\\
	\affinetyped{\uspctx_1}{v_1}{\tau_1}\\
           \affinetyped{\uspctx_2}{v_2}{\tau_2}}
         {\affinetyped{\uspctx}{\kwpair{v_1}{v_2}}{\kwprod{\tau_1}{\tau_2}}}
    \and
    \Rule{SV:InL}
         {\affinetyped{\uspctx}{v}{\tau_1}
         }
         {\affinetyped{\uspctx}{\kwinl{v}}{\kwsum{\tau_1}{\tau_2}}
         }
    \and
    \Rule{SV:InR}
         {\affinetyped{\uspctx}{v}{\tau_2}
         }
         {\affinetyped{\uspctx}{\kwinr{v}}{\kwsum{\tau_1}{\tau_2}}}
    \and
    \Rule{SV:Roll}
         {\affinetyped{\uspctx}{v}
           {\sub{\sub{\tau}{\vs}{\vvar}}
             {\kwxi{\hastycl{\vvar'}{\vsty}}{\kwprec{\convar}{\hastycl{\vvar}{\vsty}}{\tau}{\vvar'}}}{\convar}}%\\
           %\iskind{\gctx}{}{\utctx}{\ectx}{\kwprec{\convar}{\hastycl{\vvar}{\vsty}}{\tau}{\vs}}{\kwtykind}
         }
         {\affinetyped{\uspctx}{\kwroll{v}}
           {\kwprec{\convar}{\hastycl{\vvar}{\vsty}}{\tau}{\vs}}}
    \and
    \Rule{SV:Handle}
         {\uspsplit{\uspctx}{\uspctx_1}{\uspctx_2}\\
           \affinetyped{\uspctx_1}{v}{\tau}\\
           \vsistype{\uspctx_2}{\ectx}{\vspath}{\kwvty}}
         {\affinetyped{\uspctx}{\kwhandle{\vspath}{v}}
           {\kwfutt{\tau}{\vspath}}}
  \end{mathpar}
  \caption{Affine typing rules for values.}
  \label{fig:affine-statics}
\end{figure*}

\begin{figure*}
  \small
  \centering
  \def \MathparLineskip {\lineskip=0.43cm}
  \begin{mathpar}
    \Rule{KU:Unit}
         {\strut}
         {\unannkind{\vstctx}{\kwunit}{\kwtykind}}
    \and
    \Rule{KU:Var}
         {\strut}
         {\unannkind{\vstctx, \haskind{\convar}{\kwtykind}}
           {\convar}{\kwtykind}}
    \and
    \Rule{KU:Fun}
         {\iskind{\ectx}{}{\hastype{\vvar_f}{\vsty_f}, \hastype{\vvar_t}{\vsty_t}}{\ectx}
             {\con_1}{\kwtykind}\\
           \iskind{\ectx}{}{\hastype{\vvar_f}{\vsty_f}, \hastype{\vvar_t}{\vsty_t}}{\ectx}
             {\con_2}{\kwtykind}\\
           \dagwf{\ectx}{\ectx}{\ectx}{\graph}
                 {\dagpi{\hastycl{\vvar_f}{\vsty_f}}{\hastycl{\vvar_t}{\vsty_t}}{\kgraph}}}
         {\unannkind{\vstctx}{\kwpi{\hastycl{\vvar_f}{\vsty_f}}{\hastycl{\vvar_t}{\vsty_t}}
             {\kwarrow{\con_1}{\con_2}{\kwtapp{\graph}{\vvar_f}{\vvar_t}}}}{\kwtykind}}
    \and
    \Rule{KU:Sum}
		{\unannkind{\vstctx}{\undecty_1}{\kwtykind}\\
           \unannkind{\vstctx}{\undecty_2}{\kwtykind}}
         {\unannkind{\vstctx}{\kwsum{\undecty_1}{\undecty_2}}{\kwtykind}}
    \and
    \Rule{KU:Prod}
         {\unannkind{\vstctx}{\undecty_1}{\kwtykind}\\
           \unannkind{\vstctx}{\undecty_2}{\kwtykind}}
         {\unannkind{\vstctx}{\kwprod{\undecty_1}{\undecty_2}}{\kwtykind}}
    \and
    \Rule{KU:Fut}
         {\unannkind{\vstctx}{\undecty}{\kwtykind}}
         {\unannkind{\vstctx}{\undecfutty{\undecty}}{\kwtykind}}
    \and
    \Rule{KU:Rec}
         {\unannkind{\vstctx, \haskind{\convar}{\kwtykind}}{\undecty}{\kwtykind}}
         {\unannkind{\vstctx}{\kwrec{\convar}{\undecty}}{\kwtykind}}
  \end{mathpar}
  \caption{Kinding rules for unannotated types.}
  \label{fig:unannotated-kinding}
\end{figure*}

\begin{figure*}
  \small
  \centering
  \def \MathparLineskip {\lineskip=0.43cm}
  \begin{mathpar}
    \Rule{SU:Fun}
         {\gctx = \hastype{\gvar}{\dagpi{\hastycl{\vvar_f}{\vsty_f}}{\hastycl{\vvar_t}{\vsty_t}}{\kgraph}}\\
		\ctx = \hastype{f}
	             {\kwpi{\hastycl{\vvar_f}{\vsty_f}}{\hastycl{\vvar_t}{\vsty_t}}
	               {\kwarrow{\tau_1}{\tau_2}{\kwtapp{\gvar}{\vvar_f}{\vvar_t}}}},
	             \hastype{x}{\tau_1}\\
           \gvar\fresh\\
		\utctx = \hastype{\vvar_f}{\vsty_f}, \hastype{\vvar_t}{\vsty_t}\\
           \tywithdag{\gctx}{\hastype{\vvar_f}{\vsty_f}}{\utctx}{\ctx}{e}{\tau_2}{\graph}\\
           \iskind{\ectx}{}{\utctx}{\ectx}{\tau_1}{\kwtykind}\\
           \iskind{\ectx}{}{\utctx}{\ectx}{\tau_2}{\kwtykind}}
         {\unannvaltype
           {\kwfun{\vvar_f}{\vvar_t}{f}{x}{e}}
           {\kwpi{\hastycl{\vvar_f}{\vsty_f}}{\hastycl{\vvar_t}{\vsty_t}}
             {\kwarrow{\tau_1}{\tau_2}
               {\kwtapp{(\dagrec{\gvar}{\dagpi{\hastype{\vvar_f}{\vsty_f}}{\hastype{\vvar_t}{\vsty_t}}
               {\graph}}{})}{\vvar_f}{\vvar_t}}}}
         }
    \and
    \Rule{SU:Unit}
         {\strut}
         {\unannvaltype{\kwtriv}{\kwunit}}
    \and
    \Rule{SU:Pair}
         {\unannvaltype{\undece_1}{\undecty_1}\\
           \unannvaltype{\undece_2}{\undecty_2}}
         {\unannvaltype{\kwpair{\undece_1}{\undece_2}}{\kwprod{\undecty_1}{\undecty_2}}}
    \and
    \Rule{SU:InL}
         {\unannvaltype{\undece}{\undecty_1}}
         {\unannvaltype{\kwinl{\undece}}{\kwsum{\undecty_1}{\undecty_2}}}
    \and
    \Rule{SU:InR}
         {\unannvaltype{\undece}{\undecty_2}}
         {\unannvaltype{\kwinr{\undece}}{\kwsum{\undecty_1}{\undecty_2}}}
    \and
    \Rule{SU:Roll}
         {\unannvaltype{\undece}{\sub{\undecty}{\kwrec{\convar}{\undecty}}{\convar}}}
         {\unannvaltype{\kwroll{\undece}}{\kwrec{\convar}{\undecty}}}
    \and
    \Rule{SU:Handle}
         {\unannvaltype{\undece}{\undecty}}
         {\unannvaltype{\undechandle{\undece}}{\undecfutty{\undecty}}}
  \end{mathpar}
  \caption{Typing rules for unannotated values.}
  \label{fig:unannotated-statics}
\end{figure*}

As stated in the main body of this paper, the main result of this section
requires an affine typing judgement for (annotated) values, 
a kinding judgement for unannotated types, and a typing judgement for unannotated values, 
whose rules are given in Figures \ref{fig:affine-statics}, \ref{fig:unannotated-kinding}, and \ref{fig:unannotated-statics} respectively.

%The addition of fork-join parallelism (given in Figure \ref{fig:fork-join-jar})
%is not supported by the inference algorithm presented in this paper, though we believe
%that adding this support would be straight-forward.

Before showing the main result of the section, we state and prove %two lemmas
a lemma
about how the two annotation judgments commute with substitution.
%
%Lemma~\ref{lem:annot-vp-sub} allows us to substitute a different vertex path
%in the type annotation judgment.
%
Lemma~\ref{lem:subst-annot-vstyvar} allows us to substitute any VS type for
a free VS type variable in the type annotation judgement.
This lemma is required for the proof of the main theorem, but is also useful
for thinking about annotating polymorphic types.
For example, to annotate a type~$\convar~\kw{list}$,
we can add~$\convar$ to the context with
kind~$\kwkindarr{\vstyvar}{\kwtykind}$ for a fresh vertex structure type
variable~$\vstyvar$.
The annotated type and vertex structure type will refer to~$\vstyvar$.
If we instantiate~$\convar$ with a concrete type~$\undecty'$,
the resulting annotated type and vertex structure type will strongly depend on
the particular type~$\undecty'$.
For example, no vertices need to be added if~$\undecty' = \kw{int}$, but we
need more vertices if~$\undecty' = \undecfutty{\kw{int}}$.
We also prove Lemma \ref{lem:subst-unannkind-vstyvar}, 
which states that well-formed unannotated types maintain their kinds 
when substituting any free type variables appropriately.

\begin{lemma}\label{lem:subst-annot-vstyvar}
  If~$\futurify{\vstctx}{\vs}{\undecty}{\tau}{\vsty}$
  	then~$\futurify{\vsub{\vstctx}{\vsty'}{\vstyvar}}{\vs}
			{\vsub{\undecty}{\vsty'}{\vstyvar}}{\vsub{\tau}{\vsty'}{\vstyvar}}{\vsub{\vsty}{\vsty'}{\vstyvar}}$.
\end{lemma}
\begin{proof}
  By induction on the derivation
  of~$\futurify{\vstctx}{\vs}{\undecty}{\tau}{\vsty}$.
  \begin{itemize}
\item \rulename{F:TyVar}.
	Then $\undecty = \convar$
		and $\tau = \kwapp{\convar}{\vs}$
		and $\vstctx = \vstctx', \haskind{\convar}{\kwkindarr{\vsty}{\kwtykind}}$.
	We have $\vsub{\vstctx}{\vsty'}{\vstyvar} = \vsub{\vstctx'}{\vsty'}{\vstyvar}, 
			\haskind{\convar}{\kwkindarr{\vsub{\vsty}{\vsty'}{\vstyvar}}{\kwtykind}}$.
	Apply \rulename{F:TyVar}.

\item \rulename{F:Unit}.
	Then $\undecty = \tau = \kwunit$
		and $\vsty = \kwvty$.
	Apply \rulename{F:Unit}.

\item \rulename{F:Fun}.
	Then $\undecty = \tau = \kwpi{\hastycl{\vvar_t}{\vsty_t}}{\hastycl{\vvar_f}{\vsty_f}}{\kwarrow{\tau_1}{\tau_2}{\graph}}$
		and $\vsty = \kwvty$.
	We have $\vsub{\undecty}{\vsty'}{\vstyvar} = \vsub{\tau}{\vsty'}{\vstyvar}
		= \undecty = \tau = \kwpi{\hastycl{\vvar_f}{\vsub{\vsty_f}{\vsty'}{\vstyvar}}}{\hastycl{\vvar_t}{\vsub{\vsty_t}{\vsty'}{\vstyvar}}}
				{\kwarrow{\vsub{\tau_1}{\vsty'}{\vstyvar}}{\vsub{\tau_2}{\vsty'}{\vstyvar}}{\vsub{\graph}{\vsty'}{\vstyvar}}}$.
	Apply \rulename{F:Fun}.

\item \rulename{F:Prod}.
	Then $\undecty = \kwprod{\undecty_1}{\undecty_2}$
		and $\tau = \kwprod{\tau_1}{\tau_2}$
		and $\vsty = \vsprod{\vsty_1}{\isav}{\vsty_2}{\isav}$
		and $\futurify{\vstctx}{\kwfst{\vs}}{\undecty_1}{\tau_1}{\vsty_1}$
		and $\futurify{\vstctx}{\kwsnd{\vs}}{\undecty_2}{\tau_2}{\vsty_2}$.
	By induction,
		$\futurify{\vsub{\vstctx}{\vsty'}{\vstyvar}}{\kwfst{\vs}}
			{\vsub{\undecty_1}{\vsty'}{\vstyvar}}{\vsub{\tau_1}{\vsty'}{\vstyvar}}{\vsub{\vsty_1}{\vsty'}{\vstyvar}}$
		and $\futurify{\vsub{\vstctx}{\vsty'}{\vstyvar}}{\kwsnd{\vs}}
			{\vsub{\undecty_2}{\vsty'}{\vstyvar}}{\vsub{\tau_2}{\vsty'}{\vstyvar}}{\vsub{\vsty_2}{\vsty'}{\vstyvar}}$.
	Apply \rulename{F:Prod}.

\item \rulename{F:Sum}.
	By symmetry with the previous case.

\item \rulename{F:Fut}.
	Then $\undecty = \undecfutty{\undecty'}$
		and $\tau = \kwfutt{\tau'}{\kwsnd{\vs}}$
		and $\vsty = \vsprod{\vsty''}{\isav}{\kwvty}{\isav}$
		and $\futurify{\vstctx}{\kwfst{\vs}}{\undecty'}{\tau'}{\vsty''}$.
	By induction,
		$\futurify{\vsub{\vstctx}{\vsty'}{\vstyvar}}{\kwfst{\vs}}
			{\vsub{\undecty'}{\vsty'}{\vstyvar}}{\vsub{\tau'}{\vsty'}{\vstyvar}}{\vsub{\vsty''}{\vsty'}{\vstyvar}}$.
	Apply \rulename{F:Fut}.

\item \rulename{F:Rec}.
	Then $\undecty = \kwrec{\convar}{\undecty'}$
		and $\tau = \kwprec{\convar}{\hastype{\vvar}{\kwvscorec{\vstyvar'}{\vsty''}}}{\tau'}{\vs}$
		and $\vsty = \kwvscorec{\vstyvar'}{\vsty''}$
		and $\futurify{\vstctx, \haskind{\convar}{\kwkindarr{\kwvscorec{\vstyvar'}{\vsty''}}{\kwtykind}}}
				{\vvar}{\undecty'}{\tau'}{\sub{\vsty''}{\kwvscorec{\vstyvar'}{\vsty''}}{\vstyvar'}}$.
	By induction,
		$\futurify{\vsub{\vstctx}{\vsty'}{\vstyvar}, \haskind{\convar}{\kwkindarr{\kwvscorec{\vstyvar'}{\vsub{\vsty''}{\vsty'}{\vstyvar}}}{\kwtykind}}}
				{\vvar}{\vsub{\undecty'}{\vsty'}{\vstyvar}}{\vsub{\tau'}{\vsty'}{\vstyvar}}
				{\sub{\vsub{\vsty''}{\vsty'}{\vstyvar}}{\kwvscorec{\vstyvar'}{\vsub{\vsty''}{\vsty'}{\vstyvar}}}{\vstyvar'}}$
		(since $\vstyvar'$ is bound within $\vsty$, $\vstyvar \neq \vstyvar'$).
	Apply \rulename{F:Rec}.
\end{itemize}
\end{proof}

\begin{lemma}\label{lem:subst-unannkind-vstyvar}
%  \item If~$\iskind{\gctx}{}{\utctx}{\vstctx}{\tsub{\con}{\vs}{\vvar}}{\kind}$
%    and~$\vsistype{\ectx}{\utctx}{\vs}{\vsty}$
%    then~$\iskind{\gctx}{}{\utctx, \hastype{\vvar}{\vsty}}{\vstctx}{\con}{\kind}$.

%  \item If~$\iskind{\ectx}{}{\utctx}{\vstctx}{\con}{\kind}$
%  	then~$\iskind{\ectx}{}{\vsub{\utctx}{\vsty'}{\vstyvar}}{\vsub{\vstctx}{\vsty'}{\vstyvar}}{\con}{\vsub{\kind}{\vsty'}{\vstyvar}}$.
  If $\unannkind{\vstctx, \haskind{\convar}{\kwtykind}}{\undecty}{\kwtykind}$
		and $\unannkind{\vstctx}{\undecty'}{\kwtykind}$
	then $\unannkind{\vstctx}{\sub{\undecty}{\undecty'}{\convar}}{\kwtykind}$.
\end{lemma}
\begin{proof}
%  \item
%  By induction on the derivation
%  of~$\iskind{\gctx}{}{\utctx}{\vstctx}{\tsub{\con}{\vs}{\vvar}}{\kind}$.
%  \begin{itemize}
%\item \rulename{K:Unit}.
%	Then $\tsub{\con}{\vs}{\vvar} = \con = \kwunit$.
%  \end{itemize}
    By induction on the derivation of~$\unannkind{\vstctx, \haskind{\convar}{\kwtykind}}{\undecty}{\kwtykind}$.
    \begin{itemize}
\item \rulename{KU:Unit}. 
	Then $\undecty = \kwunit$.
	Apply \rulename{KU:Unit}.

\item \rulename{KU:Var}. 
	Then $\undecty = \convar'$.
	There are two cases:
		\begin{enumerate}
		\item $\convar \neq \convar'$.
			Then $\vstctx = \vstctx', \haskind{\convar'}{\kwtykind}$.
			Apply \rulename{KU:Var}. 
		\item $\convar = \convar'$.
			Then $\sub{\undecty}{\undecty'}{\convar} = \undecty'$.
			True by premise.
		\end{enumerate}

\item \rulename{KU:Fun}.
	Then $\undecty = \kwpi{\hastycl{\vvar_f}{\vsty_f}}{\hastycl{\vvar_t}{\vsty_t}}
             {\kwarrow{\con_1}{\con_2}{\kwtapp{\graph}{\vvar_f}{\vvar_t}}}$
		and $\iskind{\ectx}{}{\hastype{\vvar_f}{\vsty_f}, \hastype{\vvar_t}{\vsty_t}}{\ectx}{\con_1}{\kwtykind}$
		and $\iskind{\ectx}{}{\hastype{\vvar_f}{\vsty_f}, \hastype{\vvar_t}{\vsty_t}}{\ectx}{\con_2}{\kwtykind}$.
	Therefore, $\con_1$ and $\con_2$ contain no instances of $\convar$, 
		so $\sub{\undecty}{\undecty'}{\convar} = \undecty$
	Apply \rulename{KU:Fun}

\item \rulename{KU:Prod}.
	Then $\undecty = \kwprod{\undecty_1}{\undecty_2}$
		and~$\unannkind{\vstctx, \haskind{\convar}{\kwtykind}}{\undecty_1}{\kwtykind}$
		and~$\unannkind{\vstctx, \haskind{\convar}{\kwtykind}}{\undecty_2}{\kwtykind}$.
	By induction,
		$\unannkind{\vstctx}{\sub{\undecty_1}{\undecty'}{\convar}}{\kwtykind}$
		and $\unannkind{\vstctx}{\sub{\undecty_2}{\undecty'}{\convar}}{\kwtykind}$.
	Apply \rulename{KU:Prod}.

\item \rulename{KU:Sum}.
	By symmetry with the previous case.

\item \rulename{KU:Fut}.
	Then $\undecty = \undecfutty{\undecty''}$
		and $\unannkind{\vstctx, \haskind{\convar}{\kwtykind}}{\undecty''}{\kwtykind}$.
	By induction,
		$\unannkind{\vstctx}{\sub{\undecty''}{\undecty'}{\convar}}{\kwtykind}$.
	Apply \rulename{F:Fut}.

\item \rulename{KU:Rec}.
	Then $\undecty = \kwrec{\convar'}{\undecty''}$
		and $\unannkind{\vstctx, \haskind{\convar'}{\kwtykind}, \haskind{\convar}{\kwtykind}}{\undecty''}{\kwtykind}$.
	By induction,
		$\unannkind{\vstctx, \haskind{\convar'}{\kwtykind}}{\sub{\undecty''}{\undecty'}{\convar}}{\kwtykind}$.
	Apply \rulename{KU:Rec}.
  \end{itemize}
\end{proof}

The main result of this section additionally requires the ability to commute both recursive type annotation and recursive type kinding with the unrolling of recursive types. Lemma \ref{lem:subst-rec-annot} allows us to do this: loosely speaking, if a type $\undecty$ annotates to a well-formed $\tau$ where $\undecty$ and $\tau$ are the bodies of recursive types $\undecty'$ and $\tau'$ respectively, then the unrolling of $\undecty'$ annotates to a well-formed unrolling of $\tau'$. Note that the statement of Lemma \ref{lem:subst-rec-annot} is stronger than this description.

\begin{lemma}\label{lem:subst-rec-annot}
  If~$\futurify{\vstctx, \haskind{\convar}{\kwkindarr{\kwvscorec{\vstyvar}{\vsty}}{\kwtykind}}}{\vs}{\undecty}{\tau}{\vsty'}$
  and~$\futurify{\vstctx, \haskind{\convar}{\kwkindarr{\kwvscorec{\vstyvar}{\vsty}}{\kwtykind}}}{\vvar}
			{\undecty'}{\tau'}{\vsub{\vsty}{\kwvscorec{\vstyvar}{\vsty}}{\vstyvar}}$
	and $\unannkind{\unannvstctx{\vstctx}, \haskind{\convar}{\kwtykind}}{\undecty}{\kwtykind}$,
  then~$\futurify{\vstctx}{\vsub{\vs}{\vs'}{\vvar}}{\sub{\undecty}{\kwrec{\convar}{\undecty'}}{\convar}}{\tau''}{\vsty'}$,
	and for any $\utctx$ such that 
%		$\vsistype{\ectx}{\utctx, \hastype{\vvar}{\kwvscorec{\vstyvar}{\vsty}}}{\vs}{\vsty'}$
		$\vsistype{\ectx}{\utctx}{\vsub{\vs}{\vs'}{\vvar}}{\vsty'}$
%		and $\iskind{\ectx}{}{\utctx}
%				{\vstctx, \haskind{\convar}{\kwkindarr{\kwvscorec{\vstyvar}{\vsty}}{\kwtykind}}}{\vsub{\tau}{\vs'}{\vvar}}{\kwtykind}$
		and $\iskind{\ectx}{}{\utctx, \hastype{\vvar}{\kwvscorec{\vstyvar}{\vsty}}}
				{\vstctx, \haskind{\convar}{\kwkindarr{\kwvscorec{\vstyvar}{\vsty}}{\kwtykind}}}{\tau'}{\kwtykind}$,
%		and $\vsistype{\ectx}{\utctx}{\vs}{\kwvscorec{\vstyvar}{\vsty}}$,
      it follows that $\coneq{\ectx}{\utctx}{\vstctx}{\tau''}{\sub{\sub{\tau}{\vs'}{\vvar}}
				{\kwxi{\hastycl{\vvar'}{\kwvscorec{\vstyvar}{\vsty}}}
					{\kwprec{\convar}{\hastycl{\vvar}{\kwvscorec{\vstyvar}{\vsty}}}{\tau'}{\vvar'}}}{\convar}}{\kwtykind}$.
\end{lemma}
\begin{proof}
  By induction on the derivation
  of~$\futurify{\vstctx, \haskind{\convar}{\kwkindarr{\kwvscorec{\vstyvar}{\vsty}}{\kwtykind}}}{\vs}{\undecty}{\tau}{\vsty'}$.
  \begin{itemize}
\item \rulename{F:TyVar}.
	Then $\undecty = \convar'$
		and $\tau = \kwvapp{\convar'}{\vs}$.
	There are two cases:
	\begin{enumerate}
		\item $\convar \neq \convar'$.
			Then $\vstctx = \vstctx', \haskind{\convar'}{\kwkindarr{\vsty'}{\kwtykind}}$.
			We have $\sub{\undecty}{\kwrec{\convar}{\undecty'}}{\convar} = \undecty$
				and $\sub{\sub{\tau}{\vs'}{\vvar}}{\kwxi{\hastycl{\vvar'}{\kwvscorec{\vstyvar}{\vsty}}}
						{\kwprec{\convar}{\hastycl{\vvar}{\kwvscorec{\vstyvar}{\vsty}}}{\tau'}{\vvar'}}}{\convar}
					= \kwvapp{\convar'}{(\vsub{\vs}{\vs'}{\vvar})}$.
			By \rulename{F:TyVar},
		        $\futurify{\vstctx}{\vsub{\vs}{\vs'}{\vvar}}{\undecty}
						{\\\kwvapp{\convar'}{(\vsub{\vs}{\vs'}{\vvar})}}{\vsty'}$.
			By \rulename{K:Var} and \rulename{K:App},
				$\iskind{\ectx}{}{\utctx}{\vstctx}
						{\kwvapp{\convar'}{(\vsub{\vs}{\vs'}{\vvar})}}{\kwtykind}$
				for all $\utctx$ such that
				$\vsistype{\ectx}{\utctx}{\vsub{\vs}{\vs'}{\vvar}}{\vsty'}$
				and $\iskind{\ectx}{}{\utctx, \hastype{\vvar}{\kwvscorec{\vstyvar}{\vsty}}}
						{\vstctx, \haskind{\convar}{\kwkindarr{\kwvscorec{\vstyvar}{\vsty}}{\kwtykind}}}{\tau'}{\kwtykind}$.
			Apply \rulename{CE:Reflexive}.
		\item $\convar = \convar'$.
			Then $\vsty' = \kwvscorec{\vstyvar}{\vsty}$.
			We have $\sub{\undecty}{\kwrec{\convar}{\undecty'}}{\convar} = \kwrec{\convar}{\undecty'}$
				and $\sub{\sub{\tau}{\vs'}{\vvar}}{\kwxi{\hastycl{\vvar'}{\kwvscorec{\vstyvar}{\vsty}}}
						{\kwprec{\convar}{\hastycl{\vvar}{\kwvscorec{\vstyvar}{\vsty}}}{\tau'}{\vvar'}}}{\convar}
					= \kwvapp{(\kwxi{\hastycl{\vvar'}{\kwvscorec{\vstyvar}{\vsty}}}
						{\kwprec{\convar}{\hastycl{\vvar}{\kwvscorec{\vstyvar}{\vsty}}}{\tau'}{\vvar'}})}{(\vsub{\vs}{\vs'}{\vvar})}$.
			By \rulename{F:Rec},
		        $\futurify{\vstctx}{\vsub{\vs}{\vs'}{\vvar}}{\kwrec{\convar}{\undecty'}}
						{\kwprec{\convar}{\hastycl{\vvar}{\kwvscorec{\vstyvar}{\vsty}}}{\tau'}{\vsub{\vs}{\vs'}{\vvar}}}{\vsty'}$.
			By \rulename{U:PsiVar},
				$\vsistype{\ectx}{\hastype{\vvar'}{\kwvscorec{\vstyvar}{\vsty}}}{\vvar'}{\kwvscorec{\vstyvar}{\vsty}}$.
			By weakening and \rulename{K:Rec},
				$\iskind{\ectx}{}{\utctx, \hastype{\vvar'}{\kwvscorec{\vstyvar}{\vsty}}}{\vstctx}
						{\kwprec{\convar}{\hastycl{\vvar}{\kwvscorec{\vstyvar}{\vsty}}}{\tau'}{\vvar'}}{\kwtykind}$
				for all $\utctx$ such that
				$\vsistype{\ectx}{\utctx}{\vsub{\vs}{\vs'}{\vvar}}{\vsty'}$
				and $\iskind{\ectx}{}{\utctx, \hastype{\vvar}{\kwvscorec{\vstyvar}{\vsty}}}
						{\vstctx, \haskind{\convar}{\kwkindarr{\kwvscorec{\vstyvar}{\vsty}}{\kwtykind}}}{\tau'}{\kwtykind}$.
			By \rulename{CE:BetaEq},
				$\coneq{\ectx}{\utctx}{\vstctx}{\kwvapp{(\kwxi{\hastycl{\vvar'}{\kwvscorec{\vstyvar}{\vsty}}}
							{\kwprec{\convar}{\hastycl{\vvar}{\kwvscorec{\vstyvar}{\vsty}}}{\tau'}{\vvar'}})}{(\vsub{\vs}{\vs'}{\vvar})}}
						{\kwprec{\convar}{\hastycl{\vvar}{\kwvscorec{\vstyvar}{\vsty}}}{\tau'}{\vsub{\vs}{\vs'}{\vvar}}}{\kwtykind}$
				(since $\vvar'$ is a fresh binding, $\tau'$ contains no instances of $\vvar'$).
			Apply \rulename{CE:Commutative}.
	\end{enumerate}

\item \rulename{F:Unit}.
	Then $\undecty = \tau = \sub{\undecty}{\kwrec{\convar}{\undecty'}}{\convar}
				= \sub{\sub{\tau}{\vs'}{\vvar}}{\kwxi{\hastycl{\vvar'}{\kwvscorec{\vstyvar}{\vsty}}}
					{\kwprec{\convar}{\hastycl{\vvar}{\kwvscorec{\vstyvar}{\vsty}}}{\tau'}{\vvar'}}}{\convar}
			= \kwunit$
		and $\vsty' = \kwvty$.
	Apply \rulename{F:Unit}.
	By \rulename{K:Unit} and \rulename{CE:Reflexive},
		$\coneq{\ectx}{\utctx}{\vstctx}{\kwunit}{\kwunit}{\kwtykind}$
		for all $\utctx$.

\item \rulename{F:Fun}.
	Then $\undecty = \tau = \kwpi{\hastycl{\vvar_f}{\vsty_f}}{\hastycl{\vvar_t}{\vsty_t}}{\kwarrow{\tau_1}{\tau_2}{\graph}}$
		and $\vsty' = \kwvty$.
	By inversion on \rulename{KU:Fun},
		$\iskind{\ectx}{}{\hastype{\vvar_f}{\vsty_f}, \hastype{\vvar_t}{\vsty_t}}{\ectx}{\tau_1}{\kwtykind}$
		and $\iskind{\ectx}{}{\hastype{\vvar_f}{\vsty_f}, \hastype{\vvar_t}{\vsty_t}}{\ectx}{\tau_2}{\kwtykind}$
		and $\dagwf{\ectx}{\ectx}{\ectx}{\graph}{\dagpi{\hastycl{\vvar_f}{\vsty_f}}{\hastycl{\vvar_t}{\vsty_t}}{\kgraph}}$.
	Therefore
		$\tau_1$, $\tau_2$, and $\graph$ contain no instances of $\vvar$ or $\convar$,
		so we have $\sub{\undecty}{\kwrec{\convar}{\undecty'}}{\convar}
				= \sub{\sub{\tau}{\vs'}{\vvar}}{\kwxi{\hastycl{\vvar'}{\kwvscorec{\vstyvar}{\vsty}}}
					{\kwprec{\convar}{\hastycl{\vvar}{\kwvscorec{\vstyvar}{\vsty}}}{\tau'}{\vvar'}}}{\convar}
			= \kwpi{\hastycl{\vvar_t}{\vsty_t}}{\hastycl{\vvar_f}{\vsty_f}}{\kwarrow{\tau_1}{\tau_2}{\graph}}$.
	Apply \rulename{F:Fun}.
	By \rulename{K:Fun} and \rulename{CE:Reflexive} and weakening,
		$\coneq{\ectx}{\utctx}{\vstctx}
				{\kwpi{\hastycl{\vvar_t}{\vsty_t}}{\hastycl{\vvar_f}{\vsty_f}}{\kwarrow{\tau_1}{\tau_2}{\graph}}}
				{\kwpi{\hastycl{\vvar_t}{\vsty_t}}{\hastycl{\vvar_f}{\vsty_f}}{\kwarrow{\tau_1}{\tau_2}{\graph}}}{\kwtykind}$
		for all $\utctx$.

\item \rulename{F:Prod}.
	Then $\undecty = \kwprod{\undecty_1}{\undecty_2}$
		and $\tau = \kwprod{\tau_1}{\tau_2}$
		and $\vsty' = \vsprod{\vsty_1}{\isav}{\vsty_2}{\isav}$
		and $\futurify{\vstctx, \haskind{\convar}{\kwkindarr{\kwvscorec{\vstyvar}{\vsty}}{\kwtykind}}}{\kwfst{\vs}}{\undecty_1}{\tau_1}{\vsty_1}$
		and $\futurify{\vstctx, \haskind{\convar}{\kwkindarr{\kwvscorec{\vstyvar}{\vsty}}{\kwtykind}}}{\kwsnd{\vs}}{\undecty_2}{\tau_2}{\vsty_2}$.
	By inversion on \rulename{KU:Prod},
		$\unannkind{\unannvstctx{\vstctx}, \haskind{\convar}{\kwtykind}}{\undecty_1}{\kwtykind}$
		and $\unannkind{\unannvstctx{\vstctx}, \haskind{\convar}{\kwtykind}}{\undecty_2}{\kwtykind}$.
	By induction,
		$\futurify{\vstctx}{\vsub{\kwfst{\vs}}{\vs'}{\vvar}}{\sub{\undecty_1}{\kwrec{\convar}{\undecty'}}{\convar}}{\tau_1''}{\vsty_1}$
			and $\futurify{\vstctx}{\vsub{\kwsnd{\vs}}{\vs'}{\vvar}}{\sub{\undecty_2}{\kwrec{\convar}{\undecty'}}{\convar}}{\tau_2''}{\vsty_2}$,
		and for any $\utctx_1$ such that 
		$\vsistype{\ectx}{\utctx_1}{\vsub{\kwfst{\vs}}{\vs'}{\vvar}}{\vsty_1}$
		and $\iskind{\ectx}{}{\utctx_1, \hastype{\vvar}{\kwvscorec{\vstyvar}{\vsty}}}
				{\vstctx, \haskind{\convar}{\kwkindarr{\kwvscorec{\vstyvar}{\vsty}}{\kwtykind}}}{\tau'}{\kwtykind}$,
      $\coneq{\ectx}{\utctx_1}{\vstctx}{\tau_1''}{\sub{\sub{\tau_1}{\vs'}{\vvar}}
				{\kwxi{\hastycl{\vvar'}{\kwvscorec{\vstyvar}{\vsty}}}
					{\kwprec{\convar}{\hastycl{\vvar}{\kwvscorec{\vstyvar}{\vsty}}}{\tau'}{\vvar'}}}{\convar}}{\kwtykind}$,
		and for any $\utctx_2$ such that 
		$\vsistype{\ectx}{\utctx_2}{\vsub{\kwsnd{\vs}}{\vs'}{\vvar}}{\vsty_2}$
		and $\iskind{\ectx}{}{\utctx_2, \hastype{\vvar}{\kwvscorec{\vstyvar}{\vsty}}}
				{\vstctx, \haskind{\convar}{\kwkindarr{\kwvscorec{\vstyvar}{\vsty}}{\kwtykind}}}{\tau'}{\kwtykind}$,
      $\coneq{\ectx}{\utctx_1}{\vstctx}{\tau_2''}{\sub{\sub{\tau_2}{\vs'}{\vvar}}
				{\kwxi{\hastycl{\vvar'}{\kwvscorec{\vstyvar}{\vsty}}}
					{\kwprec{\convar}{\hastycl{\vvar}{\kwvscorec{\vstyvar}{\vsty}}}{\tau'}{\vvar'}}}{\convar}}{\kwtykind}$.
	By \rulename{F:Prod},
		$\futurify{\vstctx}{\vsub{\vs}{\vs'}{\vvar}}{\sub{\undecty}{\kwrec{\convar}{\undecty'}}{\convar}}{\kwprod{\tau_1''}{\tau_2''}}{\vsty'}$.
	By \rulename{U:Fst} and \rulename{U:Snd},
		$\vsistype{\ectx}{\utctx}{\kwfst{\vsub{\vs}{\vs'}{\vvar}}}{\vsty_1}$
			and $\vsistype{\ectx}{\utctx}{\kwsnd{\vsub{\vs}{\vs'}{\vvar}}}{\vsty_2}$
		for all $\utctx$ such that
			$\vsistype{\ectx}{\utctx}{\vsub{\vs}{\vs'}{\vvar}}{\vsty'}$
			and $\iskind{\ectx}{}{\utctx, \hastype{\vvar}{\kwvscorec{\vstyvar}{\vsty}}}
					{\vstctx, \haskind{\convar}{\kwkindarr{\kwvscorec{\vstyvar}{\vsty}}{\kwtykind}}}{\tau'}{\kwtykind}$.
	Apply \rulename{CE:Prod}.

\item \rulename{F:Sum}.
	By symmetry with the previous case.

\item \rulename{F:Fut}.
	Then $\undecty = \undecfutty{\undecty_0}$
		and $\tau = \kwfutt{\tau_0}{\kwsnd{\vs}}$
		and $\vsty' = \vsprod{\vsty_0}{\isav}{\kwvty}{\isav}$
		and $\futurify{\vstctx, \haskind{\convar}{\kwkindarr{\kwvscorec{\vstyvar}{\vsty}}{\kwtykind}}}{\kwfst{\vs}}{\undecty_0}{\tau_0}{\vsty_0}$.
	By inversion on \rulename{KU:Fut},
		$\unannkind{\unannvstctx{\vstctx}, \haskind{\convar}{\kwtykind}}{\undecty_0}{\kwtykind}$.
	By induction,
		$\futurify{\vstctx}{\vsub{\kwfst{\vs}}{\vs'}{\vvar}}{\sub{\undecty_0}{\kwrec{\convar}{\undecty'}}{\convar}}{\tau_0''}{\vsty_0}$,
		and for any $\utctx_0$ such that 
		$\vsistype{\ectx}{\utctx_0}{\vsub{\kwfst{\vs}}{\vs'}{\vvar}}{\vsty_0}$
		and $\iskind{\ectx}{}{\utctx_0, \hastype{\vvar}{\kwvscorec{\vstyvar}{\vsty}}}
				{\vstctx, \haskind{\convar}{\kwkindarr{\kwvscorec{\vstyvar}{\vsty}}{\kwtykind}}}{\tau'}{\kwtykind}$,
      $\coneq{\ectx}{\utctx_0}{\vstctx}{\tau_0''}{\sub{\sub{\tau_0}{\vs'}{\vvar}}
				{\kwxi{\hastycl{\vvar'}{\kwvscorec{\vstyvar}{\vsty}}}
					{\kwprec{\convar}{\hastycl{\vvar}{\kwvscorec{\vstyvar}{\vsty}}}{\tau'}{\vvar'}}}{\convar}}{\kwtykind}$.
	By \rulename{F:Fut},
		$\futurify{\vstctx}{\vsub{\vs}{\vs'}{\vvar}}{\sub{\undecty}{\kwrec{\convar}{\undecty'}}{\convar}}
				{\kwfutt{\tau_0''}{\vsub{\kwsnd{\vs}}{\vs'}{\vvar}}}{\vsty'}$.
	By \rulename{U:Fst} and \rulename{U:Snd},
		$\vsistype{\ectx}{\utctx}{\kwfst{\vsub{\vs}{\vs'}{\vvar}}}{\vsty_0}$
			and $\vsistype{\ectx}{\utctx}{\kwsnd{\vsub{\vs}{\vs'}{\vvar}}}{\kwvty}$
		for all $\utctx$ such that
			$\vsistype{\ectx}{\utctx}{\vsub{\vs}{\vs'}{\vvar}}{\vsty'}$
			and $\iskind{\ectx}{}{\utctx, \hastype{\vvar}{\kwvscorec{\vstyvar}{\vsty}}}
					{\vstctx, \haskind{\convar}{\kwkindarr{\kwvscorec{\vstyvar}{\vsty}}{\kwtykind}}}{\tau'}{\kwtykind}$.
	Apply \rulename{UE:Reflexive} and \rulename{CE:Fut}.

\item \rulename{F:Rec}.
	Then $\undecty = \kwrec{\convar'}{\undecty_0}$
		and $\tau = \kwprec{\convar'}{\hastype{\vvar_0}{\vsty'}}{\tau_0}{\vs}$
		and $\vsty' = \kwvscorec{\vstyvar'}{\vsty_0}$
		and $\futurify{\vstctx, \haskind{\convar'}{\kwkindarr{\vsty'}{\kwtykind}}, 
				\haskind{\convar}{\kwkindarr{\kwvscorec{\vstyvar}{\vsty}}{\kwtykind}}}{\vvar_0}{\undecty_0}{\tau_0}{\sub{\vsty_0}{\vsty'}{\vstyvar'}}$.
	By inversion on \rulename{KU:Rec},
		$\unannkind{\unannvstctx{\vstctx}, \haskind{\convar'}{\kwtykind}, \haskind{\convar}{\kwtykind}}{\undecty_0}{\kwtykind}$.
	By induction,
		$\futurify{\vstctx, \haskind{\convar'}{\kwkindarr{\vsty'}{\kwtykind}}}{\vvar_0}
				{\sub{\undecty_0}{\kwrec{\convar}{\undecty'}}{\convar}}{\tau_0''}{\sub{\vsty_0}{\vsty'}{\vstyvar'}}$,
		and for any $\utctx_0$ such that
		$\vsistype{\ectx}{\utctx_0}{\vvar_0}{\sub{\vsty_0}{\vsty'}{\vstyvar'}}$
		and $\iskind{\ectx}{}{\utctx_0, \hastype{\vvar}{\kwvscorec{\vstyvar}{\vsty}}}
				{\vstctx, \haskind{\convar'}{\kwkindarr{\vsty'}{\kwtykind}}, \haskind{\convar}{\kwkindarr{\kwvscorec{\vstyvar}{\vsty}}{\kwtykind}}}
				{\tau'}{\kwtykind}$,
      $\coneq{\ectx}{\utctx_0}{\vstctx, \haskind{\convar'}{\kwkindarr{\vsty'}{\kwtykind}}}{\tau_0''}{\sub{\sub{\tau_0}{\vs'}{\vvar}}
				{\kwxi{\hastycl{\vvar'}{\kwvscorec{\vstyvar}{\vsty}}}
					{\kwprec{\convar}{\hastycl{\vvar}{\kwvscorec{\vstyvar}{\vsty}}}{\tau'}{\vvar'}}}{\convar}}{\kwtykind}$.
	By \rulename{F:Rec},
		$\futurify{\vstctx}{\vsub{\vs}{\vs'}{\vvar}}{\sub{\undecty}{\kwrec{\convar}{\undecty'}}{\convar}}
				{\kwprec{\convar'}{\hastype{\vvar_0}{\vsty'}}{\tau_0''}{\vsub{\vs}{\vs'}{\vvar}}}{\vsty'}$.
	By \rulename{U:PsiVar} and \rulename{UT:Corec1} and \rulename{U:Subtype} and weakening,
		$\vsistype{\ectx}{\utctx, \hastype{\vvar_0}{\vsty'}}{\vvar_0}{\sub{\vsty_0}{\vsty'}{\vstyvar'}}$
		and $\iskind{\ectx}{}{\utctx, \hastype{\vvar_0}{\vsty'}, \hastype{\vvar}{\kwvscorec{\vstyvar}{\vsty}}}
				{\vstctx, \haskind{\convar'}{\kwkindarr{\vsty'}{\kwtykind}}, \haskind{\convar}{\kwkindarr{\kwvscorec{\vstyvar}{\vsty}}{\kwtykind}}}
				{\tau'}{\kwtykind}$
		for all $\utctx$ such that
			$\vsistype{\ectx}{\utctx}{\vsub{\vs}{\vs'}{\vvar}}{\vsty'}$
			and $\iskind{\ectx}{}{\utctx, \hastype{\vvar}{\kwvscorec{\vstyvar}{\vsty}}}
					{\vstctx, \haskind{\convar'}{\kwkindarr{\vsty'}{\kwtykind}}, \haskind{\convar}{\kwkindarr{\kwvscorec{\vstyvar}{\vsty}}{\kwtykind}}}
					{\tau'}{\kwtykind}$.
	Apply \rulename{UE:Reflexive} and \rulename{CE:Rec}.
  \end{itemize}
\end{proof}

Proof of Theorem~\ref{thm:annot-complete}.

\begin{proof}\strut
  \begin{enumerate}
  \item
    By induction on the derivation of~$\unannkind{\vstctx}{\undecty}{\kwtykind}$.
    \begin{itemize}
\item \rulename{KU:Unit}. 
	Then $\undecty = \kwunit$.
	By \rulename{F:Unit},
		$\futurify{\annvstctx{\vstctx}}{\vs}{\undecty}{\kwunit}{\kwvty}$
		for all $\vs$.

\item \rulename{KU:Var}. 
	Then $\undecty = \convar$.
		and $\vstctx = \vstctx', \haskind{\convar}{\kwtykind}$.
	Thus
		$\annvstctx{\vstctx} = \annvstctx{\vstctx'}, \haskind{\convar}{\kwkindarr{\vstyvar}{\kwtykind}}$ where $\vstyvar \fresh$.
	By \rulename{F:TyVar},
		$\futurify{\annvstctx{\vstctx}}{\vs}{\undecty}{\kwvapp{\convar}{\vs}}{\vstyvar}$
		for all $\vs$.

\item \rulename{KU:Fun}.
	Then $\undecty = \kwpi{\hastycl{\vvar_f}{\vsty_f}}{\hastycl{\vvar_t}{\vsty_t}}
             {\kwarrow{\con_1}{\con_2}{\kwtapp{\graph}{\vvar_f}{\vvar_t}}}$.
	By \rulename{F:Fun},
		$\futurify{\annvstctx{\vstctx}}{\vs}{\undecty}
			{\kwpi{\hastycl{\vvar_f}{\vsty_f}}{\hastycl{\vvar_t}{\vsty_t}}
             	{\kwarrow{\con_1}{\con_2}{\kwtapp{\graph}{\vvar_f}{\vvar_t}}}}{\kwvty}$
		for all $\vs$.

\item \rulename{KU:Prod}.
	Then $\undecty = \kwprod{\undecty_1}{\undecty_2}$
		and~$\unannkind{\vstctx}{\undecty_1}{\kwtykind}$
		and~$\unannkind{\vstctx}{\undecty_2}{\kwtykind}$.
	By induction,
		for all $\vs$
		there exists $\tau_1$, $\tau_2$, $\vsty_1$, and $\vsty_2$
		such that $\futurify{\annvstctx{\vstctx}}{\kwfst{\vs}}{\undecty_1}{\tau_1}{\vsty_1}$
			and $\futurify{\annvstctx{\vstctx}}{\kwsnd{\vs}}{\undecty_2}{\tau_2}{\vsty_2}$.
	Apply \rulename{F:Prod}.

\item \rulename{KU:Sum}.
	By symmetry with the previous case.

\item \rulename{KU:Fut}.
	Then $\undecty = \undecfutty{\undecty'}$
		and $\unannkind{\vstctx}{\undecty'}{\kwtykind}$.
	By induction,
		for all $\vs$
		there exists $\tau'$ and $\vsty'$
		such that $\futurify{\annvstctx{\vstctx}}{\kwfst{\vs}}{\undecty'}{\tau'}{\vsty'}$.
	Apply \rulename{F:Fut}.

\item \rulename{KU:Rec}.
	Then $\undecty = \kwrec{\convar}{\undecty'}$
		and $\unannkind{\vstctx, \haskind{\convar}{\kwtykind}}{\undecty'}{\kwtykind}$.
	Thus
		$\annvstctx{\vstctx, \haskind{\convar}{\kwtykind}} = \annvstctx{\vstctx}, \haskind{\convar}{\kwkindarr{\vstyvar}{\kwtykind}}$ 
			where $\vstyvar \fresh$.
	By induction,
		there exists $\tau'$ and $\vsty'$
		such that $\futurify{\annvstctx{\vstctx}, \haskind{\convar}{\kwkindarr{\vstyvar}{\kwtykind}}}{\vvar}{\undecty'}{\tau'}{\vsty'}$.
	By Lemma \ref{lem:subst-annot-vstyvar},
		$\futurify{\annvstctx{\vstctx}, \haskind{\convar}{\kwkindarr{\kwvscorec{\vstyvar}{\vsty'}}{\kwtykind}}}
				{\vvar}{\undecty'}{\vsub{\tau'}{\kwvscorec{\vstyvar}{\vsty'}}{\vstyvar}}{\vsub{\vsty'}{\kwvscorec{\vstyvar}{\vsty'}}{\vstyvar}}$
		(since $\vstyvar \fresh$, 
			$\vsub{\annvstctx{\vstctx}}{\kwvscorec{\vstyvar}{\vsty'}}{\vstyvar} = \annvstctx{\vstctx}$
			and $\vsub{\undecty'}{\kwvscorec{\vstyvar}{\vsty'}}{\vstyvar} \\= \undecty'$).
	By \rulename{F:Rec},
		$\futurify{\annvstctx{\vstctx}}{\vs}{\undecty}{\kwprec{\convar}
				{\hastype{\vvar}{\kwvscorec{\vstyvar}{\vsty'}}}{\vsub{\tau'}{\kwvscorec{\vstyvar}{\vsty'}}{\vstyvar}}{\vs}}
				{\kwvscorec{\vstyvar}{\vsty'}}$
		for all $\vs$.

    \end{itemize}

  \item
    By induction on the derivation
    of~$\futurify{\vstctx}{\vs}{\undecty}{\tau}{\vsty}$.
	\begin{itemize}
\item \rulename{F:TyVar}.
	Then $\tau = \kwvapp{\convar}{\vs}$
		and $\vstctx = \vstctx', \haskind{\convar}{\kwkindarr{\vsty}{\kwtykind}}$.
	Apply \rulename{K:Var} and \rulename{K:App}.

\item \rulename{F:Unit}.
	Then $\tau = \kwunit$.
	Apply \rulename{K:Unit}.

\item \rulename{F:Fun}.
	Then $\undecty = \tau = \kwpi{\hastycl{\vvar_f}{\vsty_f}}{\hastycl{\vvar_t}{\vsty_t}}{\kwarrow{\con_1}{\con_2}{\graph}}$.
	By inversion on \rulename{KU:Fun},
		$\iskind{\ectx}{}{\hastype{\vvar_f}{\vsty_f}, \hastype{\vvar_t}{\vsty_t}}{\ectx}{\con_1}{\kwtykind}$
		and $\iskind{\ectx}{}{\hastype{\vvar_f}{\vsty_f}, \hastype{\vvar_t}{\vsty_t}}{\ectx}{\con_2}{\kwtykind}$
		and $\dagwf{\ectx}{\ectx}{\ectx}{\graph}{\dagpi{\hastycl{\vvar_f}{\vsty_f}}{\hastycl{\vvar_t}{\vsty_t}}{\kgraph}}$.
	Apply \rulename{K:Fun} and weakening.

\item \rulename{F:Prod}.
	Then $\undecty = \kwprod{\undecty_1}{\undecty_2}$
		and $\tau = \kwprod{\tau_1}{\tau_2}$
		and $\vsty = \vsprod{\vsty_1}{\isav}{\vsty_2}{\isav}$
		and $\futurify{\vstctx}{\kwfst{\vs}}{\undecty_1}{\tau_1}{\vsty_1}$
		and $\futurify{\vstctx}{\kwsnd{\vs}}{\undecty_2}{\tau_2}{\vsty_2}$.
	By inversion on \rulename{KU:Prod},
		$\unannkind{\unannvstctx{\vstctx}}{\undecty_1}{\kwtykind}$
		and $\unannkind{\unannvstctx{\vstctx}}{\undecty_2}{\kwtykind}$.
	By \rulename{U:Fst} and \rulename{U:Snd},
		$\vsistype{\ectx}{\utctx}{\kwfst{\vs}}{\vsty_1}$
		and $\vsistype{\ectx}{\utctx}{\kwsnd{\vs}}{\vsty_2}$.
	By induction,
		$\iskind{\ectx}{}{\utctx}{\vstctx}{\tau_1}{\kwtykind}$
		and $\iskind{\ectx}{}{\utctx}{\vstctx}{\tau_2}{\kwtykind}$.
	Apply \rulename{K:Prod}.

\item \rulename{F:Sum}.
	By symmetry with the previous case.

\item \rulename{F:Future}.
	Then $\undecty = \undecfutty{\undecty'}$
		and $\tau = \kwfutt{\tau'}{\kwsnd{\vs}}$
		and $\vsty = \vsprod{\vsty'}{\isav}{\kwvty}{\isav}$
		and $\futurify{\vstctx}{\kwfst{\vs}}{\undecty'}{\tau'}{\vsty'}$.
	By inversion on \rulename{KU:Fut},
		$\unannkind{\unannvstctx{\vstctx}}{\undecty'}{\kwtykind}$.
	By \rulename{U:Fst} and \rulename{U:Snd},
		$\vsistype{\ectx}{\utctx}{\kwfst{\vs}}{\vsty'}$
		and $\vsistype{\ectx}{\utctx}{\kwsnd{\vs}}{\kwvty}$.
	By induction,
		$\iskind{\ectx}{}{\utctx}{\vstctx}{\tau'}{\kwtykind}$.
	Apply \rulename{K:Fut}.

\item \rulename{F:Rec}.
	Then $\undecty = \kwrec{\convar}{\undecty'}$
		and $\tau = \kwprec{\convar}{\hastype{\vvar}{\kwvscorec{\vstyvar}{\vsty'}}}{\tau'}{\vs}$
		and $\vsty = \kwvscorec{\vstyvar}{\vsty'}$
		and $\futurify{\vstctx, \haskind{\convar}{\kwkindarr{\vsty}{\kwtykind}}}{\vvar}
				{\undecty'}{\tau'}{\sub{\vsty'}{\kwvscorec{\vstyvar}{\vsty'}}{\vstyvar}}$.
	By inversion on \rulename{KU:Rec},
		$\unannkind{\unannvstctx{\vstctx}, \haskind{\convar}{\kwtykind}}{\undecty'}{\kwtykind}$.
	By \rulename{U:PsiVar},
		$\vsistype{\ectx}{\utctx, \hastype{\vvar}{\vsty}}{\vvar}{\vsty}$.
	By \rulename{UT:Corec1} and \rulename{U:Subtype},
		$\vsistype{\ectx}{\utctx, \hastype{\vvar}{\vsty}}{\vvar}{\sub{\vsty'}{\kwvscorec{\vstyvar}{\vsty'}}{\vstyvar}}$.
	By induction,
		$\iskind{\ectx}{}{\utctx, \hastype{\vvar}{\vsty}}{\vstctx, \haskind{\convar}{\kwkindarr{\vsty}{\kwtykind}}}{\tau'}{\kwtykind}$.
	Apply \rulename{K:Rec}.
    \end{itemize}

  \item
    By induction on the derivation
    of~$\futurify{\ectx}{\vspath}{\undecty}{\tau}{\vsty}$.
	\begin{itemize}
\item \rulename{F:Unit}.
	Then $\undecty = \tau = \kwunit$.
	By inversion on \rulename{SU:Unit},
		$\undece = \kwtriv$.
	By \rulename{FE:Unit},
		$\futurifye{\vspath}{\undece}{\kwtriv}$.
	Apply \rulename{SV:Unit}
		for all $\uspctx$.

\item \rulename{F:Fun}.
	Then $\undecty = \tau = \kwpi{\hastycl{\vvar_f}{\vsty_f}}{\hastycl{\vvar_t}{\vsty_t}}{\kwarrow{\tau_1}{\tau_2}{\graph}}$.
	By inversion on \rulename{SU:Fun},
		$\undece = \kwfun{\vvar_f}{\vvar_t}{f}{x}{e}$
		and $\graph = \kwtapp{(\dagrec{\gvar}{\dagpi{\hastype{\vvar_f}{\vsty_f}}{\hastype{\vvar_t}{\vsty_t}}
               {\graph'}}{})}{\vvar_f}{\vvar_t}$
		and $\gvar\fresh$
		and $\tywithdag{\hastype{\gvar}{\dagpi{\hastycl{\vvar_f}{\vsty_f}}{\hastycl{\vvar_t}{\vsty_t}}{\kgraph}}}
				{\hastype{\vvar_f}{\vsty_f}}{\hastype{\vvar_f}{\vsty_f}, \hastype{\vvar_t}{\vsty_t}}
				{\hastype{f}{\kwpi{\hastycl{\vvar_f}{\vsty_f}}{\hastycl{\vvar_t}{\vsty_t}}
					{\kwarrow{\tau_1}{\tau_2}{\kwtapp{\gvar}{\vvar_f}{\vvar_t}}}}, \hastype{x}{\tau_1}}
				{e}{\tau_2}{\graph'}$
		and $\iskind{\ectx}{}{\hastype{\vvar_f}{\vsty_f}, \hastype{\vvar_t}{\vsty_t}}{\ectx}{\tau_1}{\kwtykind}$
		and $\iskind{\ectx}{}{\hastype{\vvar_f}{\vsty_f}, \hastype{\vvar_t}{\vsty_t}}{\ectx}{\tau_2}{\kwtykind}$.
	By \rulename{FE:Fun},
		$\futurifye{\vspath}{\undece}{\kwfun{\vvar_f}{\vvar_t}{f}{x}{e}}$.
	Apply \rulename{SV:Fun}
		for all $\uspctx$.

\item \rulename{F:Prod}.
	Then $\undecty = \kwprod{\undecty_1}{\undecty_2}$
		and $\tau = \kwprod{\tau_1}{\tau_2}$
		and $\vsty = \vsprod{\vsty_1}{\isav}{\vsty_2}{\isav}$
		and $\futurify{\ectx}{\kwfst{\vspath}}{\undecty_1}{\tau_1}{\vsty_1}$
		and $\futurify{\ectx}{\kwsnd{\vspath}}{\undecty_2}{\tau_2}{\vsty_2}$.
	By inversion on \rulename{SU:Pair},
		$\undece = \kwpair{\undece_1}{\undece_2}$
		and $\unannvaltype{\undece_1}{\undecty_1}$
		and $\unannvaltype{\undece_2}{\undecty_2}$.
	By inversion on \rulename{KU:Prod},
		$\unannkind{\ectx}{\undecty_1}{\kwtykind}$
		and $\unannkind{\ectx}{\undecty_2}{\kwtykind}$.
	By induction,
		there exist $v_1$ and $v_2$
		such that $\futurifye{\kwfst{\vspath}}{\undece_1}{v_1}$
			and $\futurifye{\kwsnd{\vspath}}{\undece_2}{v_2}$,
		and for all~$\uspctx_1$ such that~$\vsistype{\uspctx_1}{\ectx}{\kwfst{\vspath}}{\vsty_1}$,
		we have~$\affinetyped{\uspctx_1}{v_1}{\tau_1}$,
		and for all~$\uspctx_2$ such that~$\vsistype{\uspctx_2}{\ectx}{\kwsnd{\vspath}}{\vsty_2}$,
		we have~$\affinetyped{\uspctx_2}{v_2}{\tau_2}$.
	By \rulename{FE:Pair},
		$\futurifye{\vspath}{\undece}{\kwpair{v_1}{v_2}}$.
	By \rulename{US:Prod},
		$\vstysplit{\vsty}{\vsprod{\vsty_1}{\isav}{\vsty_2}{\isunav}}{\vsprod{\vsty_1}{\isunav}{\vsty_2}{\isav}}$.
	By Lemma \ref{lem:vert-split-typing},
		\vstytouspsplitext{\uspctx}{\uspctx_1}{\uspctx_2}{\vspath}
				{\vsprod{\vsty_1}{\isav}{\vsty_2}{\isunav}}{\vsprod{\vsty_1}{\isunav}{\vsty_2}{\isav}}
		for all $\uspctx$ such that $\vsistype{\uspctx}{\ectx}{\vspath}{\vsty}$.
	By \rulename{U:Fst} and \rulename{U:Snd},
		$\vsistype{\uspctx_1}{\ectx}{\kwfst{\vspath}}{\vsty_1}$
		and $\vsistype{\uspctx_2}{\ectx}{\kwsnd{\vspath}}{\vsty_2}$.
	Apply \rulename{SV:Pair}.

\item \rulename{F:Sum}.
	Then $\undecty = \kwsum{\undecty_1}{\undecty_2}$
		and $\tau = \kwsum{\tau_1}{\tau_2}$
		and $\vsty = \vsprod{\vsty_1}{\isav}{\vsty_2}{\isav}$
		and $\futurify{\ectx}{\kwfst{\vspath}}{\undecty_1}{\tau_1}{\vsty_1}$
		and $\futurify{\ectx}{\kwsnd{\vspath}}{\undecty_2}{\tau_2}{\vsty_2}$.
	By inversion on either \rulename{SU:InL} or \rulename{SU:InR}, 
		either $\undece = \kwinl{\undece'}$
				and $\unannvaltype{\undece'}{\undecty_1}$
			or $\undece = \kwinr{\undece'}$
				and $\unannvaltype{\undece'}{\undecty_2}$.
		We take the former case, the other is similar.
	By inversion on \rulename{KU:Sum},
		$\unannkind{\ectx}{\undecty_1}{\kwtykind}$.
	By induction,
		there exists $v'$
		such that $\futurifye{\kwfst{\vspath}}{\undece'}{v'}$,
		and for all~$\uspctx'$ such that~$\vsistype{\uspctx'}{\ectx}{\kwfst{\vspath}}{\vsty_1}$,
		we have~$\affinetyped{\uspctx'}{v'}{\tau_1}$.
	By \rulename{FE:InL},
		$\futurifye{\vspath}{\undece}{\kwinl{v'}}$.
	By \rulename{U:Fst},
		$\vsistype{\uspctx}{\ectx}{\kwfst{\vspath}}{\vsty_1}$
		for all $\uspctx$ such that $\vsistype{\uspctx}{\ectx}{\vspath}{\vsty}$.
	Apply \rulename{SV:InL}.

\item \rulename{F:Future}.
	Then $\undecty = \undecfutty{\undecty'}$
		and $\tau = \kwfutt{\tau'}{\kwsnd{\vspath}}$
		and $\vsty = \vsprod{\vsty'}{\isav}{\kwvty}{\isav}$
		and $\futurify{\ectx}{\kwfst{\vspath}}{\undecty'}{\tau'}{\vsty'}$.
	By inversion on \rulename{SU:Handle}, 
		$\undece = \undechandle{\undece'}$
			and $\unannvaltype{\undece'}{\undecty'}$.
	By inversion on \rulename{KU:Fut},
		$\unannkind{\ectx}{\undecty'}{\kwtykind}$.
	By induction,
		there exists $v'$
		such that $\futurifye{\kwfst{\vspath}}{\undece'}{v'}$,
		and for all~$\uspctx'$ such that~$\vsistype{\uspctx'}{\ectx}{\kwfst{\vspath}}{\vsty'}$,
		we have~$\affinetyped{\uspctx'}{v'}{\tau'}$.
	By \rulename{FE:Handle},
		$\futurifye{\vspath}{\undece}{\kwhandle{\kwsnd{\vspath}}{v'}}$.
	By \rulename{US:Prod},
		$\vstysplit{\vsty}{\vsprod{\vsty'}{\isav}{\kwvty}{\isunav}}{\vsprod{\vsty'}{\isunav}{\kwvty}{\isav}}$.
	By Lemma \ref{lem:vert-split-typing},
		for all $\uspctx$ such that $\vsistype{\uspctx}{\ectx}{\vspath}{\vsty}$,
		\vstytouspsplitext{\uspctx}{\uspctx_1}{\uspctx_2}{\vspath}
				{\vsprod{\vsty'}{\isav}{\kwvty}{\isunav}}{\vsprod{\vsty'}{\isunav}{\kwvty}{\isav}}.
	By \rulename{U:Fst} and \rulename{U:Snd},
		$\vsistype{\uspctx_1}{\ectx}{\kwfst{\vspath}}{\vsty'}$
		and $\vsistype{\uspctx_2}{\ectx}{\kwsnd{\vspath}}{\kwvty}$.
	Apply \rulename{SV:Handle}.

\item \rulename{F:Rec}.
	Then $\undecty = \kwrec{\convar}{\undecty'}$
		and $\tau = \kwprec{\convar}{\hastype{\vvar}{\kwvscorec{\vstyvar}{\vsty'}}}{\tau'}{\vspath}$
		and $\vsty = \kwvscorec{\vstyvar}{\vsty'}$
		and $\futurify{\haskind{\convar}{\kwkindarr{\kwvscorec{\vstyvar}{\vsty'}}{\kwtykind}}}{\vvar}
				{\undecty'}{\tau'}{\sub{\vsty'}{\kwvscorec{\vstyvar}{\vsty'}}{\vstyvar}}$.
	By inversion on \rulename{SU:Roll}, 
		$\undece = \kwroll{\undece'}$
			and $\unannvaltype{\undece'}{\sub{\undecty'}{\kwrec{\convar}{\undecty'}}{\convar}}$.
	By inversion on \rulename{KU:Rec},
		$\unannkind{\haskind{\convar}{\kwtykind}}{\undecty'}{\kwtykind}$.
	By Lemma \ref{lem:subst-unannkind-vstyvar},
		$\unannkind{\ectx}{\sub{\undecty'}{\kwrec{\convar}{\undecty'}}{\convar}}{\kwtykind}$.
	By Lemma \ref{lem:subst-rec-annot},
		$\futurify{\ectx}{\vspath}{\sub{\undecty'}{\kwrec{\convar}{\undecty'}}{\convar}}{\tau''}{\sub{\vsty'}{\kwvscorec{\vstyvar}{\vsty'}}{\vstyvar}}$,
		and for all $\utctx$ such that $\vsistype{\ectx}{\utctx}{\vspath}{\sub{\vsty'}{\kwvscorec{\vstyvar}{\vsty'}}{\vstyvar}}$
			and $\iskind{\ectx}{}{\utctx, \hastype{\vvar}{\kwvscorec{\vstyvar}{\vsty'}}}
					{\haskind{\convar}{\kwkindarr{\kwvscorec{\vstyvar}{\vsty'}}{\kwtykind}}}{\tau'}{\kwtykind}$,
		$\coneq{\ectx}{\utctx}{\ectx}{\tau''}{\sub{\sub{\tau'}{\vspath}{\vvar}}
				{\kwxi{\hastycl{\vvar'}{\kwvscorec{\vstyvar}{\vsty'}}}
					{\kwprec{\convar}{\hastycl{\vvar}{\kwvscorec{\vstyvar}{\vsty'}}}{\tau'}{\vvar'}}}{\convar}}{\kwtykind}$.
	By induction,
		there exists $v'$
		such that $\futurifye{\vspath}{\undece'}{v'}$,
		and for all~$\uspctx'$ such that $\vsistype{\uspctx'}{\ectx}{\vspath}{\sub{\vsty'}{\kwvscorec{\vstyvar}{\vsty'}}{\vstyvar}}$,
		we have~$\affinetyped{\uspctx}{v'}{\tau''}$.
	By \rulename{FE:Roll},
		$\futurifye{\vspath}{\undece}{\kwroll{v'}}$.
	By \rulename{UT:Corec1} and \rulename{U:Subtype},
		$\vsistype{\uspctx}{\ectx}{\vspath}{\sub{\vsty'}{\kwvscorec{\vstyvar}{\vsty'}}{\vstyvar}}$
		for all $\uspctx$ such that $\vsistype{\uspctx}{\ectx}{\vspath}{\vsty}$.
	By Lemma \ref{lem:weaken-affine-restriction},
		$\vsistype{\ectx}{\uspctx}{\vspath}{\sub{\vsty'}{\kwvscorec{\vstyvar}{\vsty'}}{\vstyvar}}$
		and $\vsistype{\ectx}{\uspctx}{\vspath}{\vsty}$.
	By part \ref{lem:annot-preservation},
		$\iskind{\ectx}{}{\uspctx}{\vstctx}{\tau}{\kwtykind}$.
	By inversion on \rulename{K:Rec},
		$\iskind{\ectx}{}{\uspctx, \hastype{\vvar}{\kwvscorec{\vstyvar}{\vsty'}}}{\haskind{\convar}{\kwvscorec{\vstyvar}{\vsty'}}}{\tau'}{\kwtykind}$.
	Thus $\coneq{\ectx}{\uspctx}{\ectx}{\tau''}{\sub{\sub{\tau'}{\vspath}{\vvar}}
				{\kwxi{\hastycl{\vvar'}{\kwvscorec{\vstyvar}{\vsty'}}}
					{\kwprec{\convar}{\hastycl{\vvar}{\kwvscorec{\vstyvar}{\vsty'}}}{\tau'}{\vvar'}}}{\convar}}{\kwtykind}$.
	Apply \rulename{SV:Type-Eq} and \rulename{SV:Roll}.
    \end{itemize}
  \end{enumerate}
\end{proof}

%% file: fig-dag-wf.tex
\begin{figure*}
  \small
  \centering
  \def \MathparLineskip {\lineskip=0.43cm}
  \begin{mathpar}
    %\iffull
    \Rule{DW:Empty}
         {\strut}
         {\dagwf{\gctx}{\uspctx}{\utctx}{\emptygraph}{\kgraph}}
         \and
         %\fi
    \Rule{DW:Var}
         {\strut}
         {\dagwf{\gctx,\hastype{\gvar}{\graphkind}}
           {\uspctx}{\utctx}{\gvar}{\graphkind}}
    \and
    \Rule{DW:Seq}
         {\uspsplit{\uspctx}{\uspctx_1}{\uspctx_2}\\
		   \dagwf{\gctx}{\uspctx_1}{\utctx}{\graph_1}{\kgraph}\\
           \dagwf{\gctx}{\uspctx_2}{\utctx}{\graph_2}{\kgraph}\\
           %\vertices_1 \cap \vertices_2 = \emptyset\\
         }
         {\dagwf{\gctx}{\uspctx}{\utctx}
           {\graph_1 \seqcomp \graph_2}{\kgraph}}
    \and
    %\iffull
%%    \Rule{DW:Par}
%%         {\uspsplit{\uspctx}{\uspctx_1}{\uspctx_2}\\
%%		   \dagwf{\gctx}{\uspctx_1}{\utctx}{\graph_1}{\kgraph}\\
%%           \dagwf{\gctx}{\uspctx_2}{\utctx}{\graph_2}{\kgraph}\\
%%%           \vertices_1 \cap \vertices_2 = \emptyset\\
%%         }
%%         {\dagwf{\gctx}{\uspctx}{\utctx}
%%           {\graph_1 \parcomp \graph_2}{\kgraph}}
%%    \and
    %\fi
    \Rule{DW:Or}
         {\dagwf{\gctx}{\uspctx}{\utctx}{\graph_1}{\kgraph}\\
           \dagwf{\gctx}{\uspctx}{\utctx}{\graph_2}{\kgraph}\\
         }
         {\dagwf{\gctx}{\uspctx}{\utctx}{\graph_1 \dagor \graph_2}
           {\kgraph}}
    \and
    \Rule{DW:Spawn}
         {\uspsplit{\uspctx}{\uspctx_1}{\uspctx_2}\\
		  \dagwf{\gctx}{\uspctx_1}{\utctx}{\graph}{\kgraph}\\
           \vsistype{\uspctx_2}{\ectx}{\vs}{\kwvty}\\
%           \vertex \not\in \vertices
         }
         {\dagwf{\gctx}{\uspctx}{\utctx}
           {\leftcomp{\graph}{\vs}}{\kgraph}}
    \and
    \Rule{DW:Touch}
         {\vsistype{\ectx}{\utctx}{\vs}{\kwvty}}
         {\dagwf{\gctx}{\uspctx}{\utctx}{\touchcomp{\vs}}{\kgraph}}
    \and
    \Rule{DW:New}
         {\dagwf{\gctx}{\uspctx, \hastype{\vvar}{\vsty}}{\utctx, \hastype{\vvar}{\vsty}}
             {\graph}{\kgraph}\\
%           \vsttree{\ectx}{\vsty}
		}
         {\dagwf{\gctx}{\uspctx}{\utctx}{\dagnew{\hastycl{\vvar}{\vsty}}{\graph}}{\kgraph}}
    \and
    %\iffull
    \Rule{DW:Pi}
         {\dagwf{\gctx}{\uspctx, \hastype{\vvar_f}{\vsty_f}}{\utctx, \hastype{\vvar_f}{\vsty_f}, \hastype{\vvar_t}{\vsty_t}}
           {\graph}{\kgraph}}
         {\dagwf{\gctx}{\uspctx}{\utctx}
           {\dagpi{\hastycl{\vvar_f}{\vsty_f}}{\hastycl{\vvar_t}{\vsty_t}}{\graph}}
           {\dagpi{\hastycl{\vvar_f}{\vsty_f}}{\hastycl{\vvar_t}{\vsty_t}}{\kgraph}}}
    \and
    \Rule{DW:RecPi}
         {\dagwf{\gctx,\hastype{\gvar}
             {\dagpi{\hastycl{\vvar_f}{\vsty_f}}{\hastycl{\vvar_t}{\vsty_t}}{\kgraph}}}
           {\hastype{\vvar_f}{\vsty_f}}{\utctx, \hastype{\vvar_f}{\vsty_f}, \hastype{\vvar_t}{\vsty_t}}{\graph}
           {\kgraph}}
         {\dagwf{\gctx}{\uspctx}{\utctx}
           {\dagrec{\gvar}{\dagpi{\hastycl{\vvar_f}{\vsty_f}}{\hastycl{\vvar_t}{\vsty_t}}{\graph}}{}}
           {\dagpi{\hastycl{\vvar_f}{\vsty_f}}{\hastycl{\vvar_t}{\vsty_t}}{\kgraph}}}
    \and
    %\fi
    \Rule{DW:App}
         {\uspsplit{\uspctx}{\uspctx_1}{\uspctx_2}\\
		   \dagwf{\gctx}{\uspctx_1}{\utctx}{\graph}
           {\dagpi{\hastycl{\vvar_f}{\vsty_f}}{\hastycl{\vvar_t}{\vsty_t}}{\kgraph}}\\
             \vsistype{\uspctx_2}{\ectx}{\vs_f}{\vsty_f}\\
		\vsistype{\ectx}{\utctx}{\vs_f}{\vsty_f}\\
             \vsistype{\ectx}{\utctx}{\vs_t}{\vsty_t}}
         {\dagwf{\gctx}{\uspctx}{\utctx}
           {\kwtapp{\graph}{\vs_f}{\vs_t}}
           {\kgraph}
             %\tsub{\tsub{\graphkind}{\verts_f'}{\verts_f}}{\verts_t'}{\verts_t}}
         }
  \end{mathpar}
  \caption{Rules for graph type formation.}
  \label{fig:dag-wf}
\end{figure*}

%% file: fig-vert-equiv.tex
\begin{figure*}
  \small
  \centering
  \def \MathparLineskip {\lineskip=0.43cm}
  \begin{mathpar}
    \Rule{UE:Reflexive}
         {\vsistype{\ectx}{\utctx}{\vs}{\vsty}}
         {\vseq{\utctx}{\vs}{\vs}{\vsty}}
	\and
    \Rule{UE:Commutative}
         {\vseq{\utctx}{\vs'}{\vs}{\vsty}}
         {\vseq{\utctx}{\vs}{\vs'}{\vsty}}
	\and
    \Rule{UE:Transitive}
         {\vseq{\utctx}{\vs}{\vs''}{\vsty}\\
			\vseq{\utctx}{\vs''}{\vs'}{\vsty}}
         {\vseq{\utctx}{\vs}{\vs'}{\vsty}}
    \and
    \Rule{UE:OnlyLeftPair}
         {\vseq{\utctx}{\vs_1}{\vs_1'}{\vsty_1}\\
			\vseq{\utctx'}{\vs_2}{\vs_2'}{\vsty_2}}
         {\vseq{\utctx}{\kwpair{\vs_1}{\vs_2}}{\kwpair{\vs_1'}{\vs_2'}}{\vsprod{\vsty_1}{\isav}{\vsty_2}{\isunav}}}
    \and
    \Rule{UE:OnlyRightPair}
         {\vseq{\utctx'}{\vs_1}{\vs_1'}{\vsty_1}\\
			\vseq{\utctx}{\vs_2}{\vs_2'}{\vsty_2}}
         {\vseq{\utctx}{\kwpair{\vs_1}{\vs_2}}{\kwpair{\vs_1}{\vs_2'}}{\vsprod{\vsty_1}{\isunav}{\vsty_2}{\isav}}}
    \and
    \Rule{UE:Pair}
         {\vseq{\utctx}{\vs_1}{\vs_1'}{\vsty_1}\\
           \vseq{\utctx}{\vs_2}{\vs_2'}{\vsty_2}}
         {\vseq{\utctx}{\kwpair{\vs_1}{\vs_2}}{\kwpair{\vs_1'}{\vs_2'}}{\vsprod{\vsty_1}{\isav}{\vsty_2}{\isav}}}
    \and
    \Rule{UE:Fst}
         {\vseq{\utctx}{\vs}{\vs'}{\vsprod{\vsty_1}{\isav}{\vsty_2}{\avail}}}
         {\vseq{\utctx}{\kwfst{\vs}}{\kwfst{\vs'}}{\vsty_1}}
    \and
    \Rule{UE:Snd}
         {\vseq{\utctx}{\vs}{\vs'}{\vsprod{\vsty_1}{\avail}{\vsty_2}{\isav}}}
         {\vseq{\utctx}{\kwsnd{\vs}}{\kwsnd{\vs'}}{\vsty_2}}
    \and
    \Rule{UE:FstPair}
         {\vseq{\utctx}{\vs}{\kwpair{\vs_1}{\vs_2}}{\vsprod{\vsty_1}{\isav}{\vsty_2}{\avail}}}
         {\vseq{\utctx}{\kwfst{\vs}}{\vs_1}{\vsty_1}}
    \and
    \Rule{UE:SndPair}
         {\vseq{\utctx}{\vs}{\kwpair{\vs_1}{\vs_2}}{\vsprod{\vsty_1}{\avail}{\vsty_2}{\isav}}}
         {\vseq{\utctx}{\kwsnd{\vs}}{\vs_2}{\vsty_2}}
    \and
    \Rule{UE:Subtype}
         {\vseq{\utctx}{\vs}{\vs'}{\vsty'}\\
			\vstysubt{\vsty'}{\vsty}}
         {\vseq{\utctx}{\vs}{\vs'}{\vsty}}
  \end{mathpar}
  \caption{Vertex structure equivalence.}
  \label{fig:vert-equiv}
\end{figure*}

%% file: fig-con-eval.tex
\begin{figure*}
  \small
  \centering
  \def \MathparLineskip {\lineskip=0.43cm}
  \begin{mathpar}
%    \Rule{CE:Unit}
%         {\strut}
%         {\coneq{\gctx}{\utctx}{\vstctx}{\kwunit}{\kwunit}}
%    \and
%    \Rule{CE:Fun}
%         {\strut}
%         {\coneq{\gctx}{\utctx}{\vstctx}
%			{\kwpi{\hastycl{\vvar_f}{\vsty_f}}{\hastycl{\vvar_t}{\vsty_t}}
%            	{\kwarrow{\con_1}{\con_2}{\graph}}}
%			{\kwpi{\hastycl{\vvar_f}{\vsty_f}}{\hastycl{\vvar_t}{\vsty_t}}
%            	{\kwarrow{\con_1}{\con_2}{\graph}}}}
%    \and
    \Rule{CE:Reflexive}
         {\iskind{\gctx}{}{\utctx}{\vstctx}{\con}{\kind}}
         {\coneq{\gctx}{\utctx}{\vstctx}{\con}{\con}{\kind}}
    \and
    \Rule{CE:Commutative}
         {\coneq{\gctx}{\utctx}{\vstctx}{\con'}{\con}{\kind}}
         {\coneq{\gctx}{\utctx}{\vstctx}{\con}{\con'}{\kind}}
    \and
    \Rule{CE:Transitive}
         {\coneq{\gctx}{\utctx}{\vstctx}{\con}{\con''}{\kind}\\
			\coneq{\gctx}{\utctx}{\vstctx}{\con''}{\con'}{\kind}}
         {\coneq{\gctx}{\utctx}{\vstctx}{\con}{\con'}{\kind}}
    \and
    \Rule{CE:Prod}
         {\coneq{\gctx}{\utctx}{\vstctx}{\con_1}{\con_1'}{\kwtykind}\\
			\coneq{\gctx}{\utctx}{\vstctx}{\con_2}{\con_2'}{\kwtykind}}
         {\coneq{\gctx}{\utctx}{\vstctx}{\kwprod{\con_1}{\con_2}}{\kwprod{\con_1'}{\con_2'}}{\kwtykind}}
    \and
    \Rule{CE:Sum}
         {\coneq{\gctx}{\utctx}{\vstctx}{\con_1}{\con_1'}{\kwtykind}\\
			\coneq{\gctx}{\utctx}{\vstctx}{\con_2}{\con_2'}{\kwtykind}}
         {\coneq{\gctx}{\utctx}{\vstctx}{\kwsum{\con_1}{\con_2}}{\kwsum{\con_1'}{\con_2'}}{\kwtykind}}
    \and
    \Rule{CE:Fut}
         {\coneq{\gctx}{\utctx}{\vstctx}{\con}{\con'}{\kwtykind}\\
			\vseq{\utctx}{\vs}{\vs'}{\kwvty}}
         {\coneq{\gctx}{\utctx}{\vstctx}{\kwfutt{\con}{\vs}}{\kwfutt{\con'}{\vs'}}{\kwtykind}}
%    \and
%    \Rule{CE:Var}
%         {\strut}
%         {\coneq{\gctx}{\utctx}{\vstctx}{\convar}{\convar}}
    \and
    \Rule{CE:Lambda}
         {\coneq{\gctx}{\utctx, \hastype{\vvar}{\vsty}}
             {\vstctx}{\con}{\con'}{\kwtykind}}
         {\coneq{\gctx}{\utctx}{\vstctx}{\kwxi{\hastycl{\vvar}{\vsty}}{\con}}{\kwxi{\hastycl{\vvar}{\vsty}}{\con'}}{\kwkindarr{\vsty}{\kwtykind}}}
    \and
    \Rule{CE:App}
         {\coneq{\gctx}{\utctx}{\vstctx}{\con}{\con'}{\kwkindarr{\vsty}{\kwtykind}}\\
			\vseq{\utctx}{\vs}{\vs'}{\vsty}}
         {\coneq{\gctx}{\utctx}{\vstctx}{\kwvapp{\con}{\vs}}{\kwvapp{\con'}{\vs'}}{\kwtykind}}
    \and
    \Rule{CE:BetaEq}
         {\iskind{\gctx}{}{\utctx, \hastype{\vvar}{\vsty}}{\vstctx}{\con}{\kwtykind}\\
			\vsistype{\ectx}{\utctx}{\vs}{\vsty}}
         {\coneq{\gctx}{\utctx}{\vstctx}{\kwvapp{(\kwxi{\hastycl{\vvar}{\vsty}}{\con})}{\vs}}
           {\vsub{\con}{\vs}{\vvar}}{\kwtykind}}
    \and
    \Rule{CE:Rec}
         {\coneq{\gctx}{\utctx, \hastype{\vvar}{\vsty}}
             {\vstctx, \haskind{\convar}{\kwkindarr{\vsty}{\kwtykind}}}{\con}{\con'}{\kwtykind}\\
			\vseq{\utctx}{\vs}{\vs'}{\vsty}}
         {\coneq{\gctx}{\utctx}{\vstctx}{\kwprec{\convar}{\hastycl{\vvar}{\vsty}}{\con}{\vs}}
				{\kwprec{\convar}{\hastycl{\vvar}{\vsty}}{\con'}{\vs'}}{\kwtykind}}
%    \and
%    \Rule{CE:Let}
%         {\vseval{\kwfst{\vs}}{\vs_1}\\
%			\vseval{\kwsnd{\vs}}{\vs_2}}
%         {\coneq{\gctx}{\utctx}{\vstctx}{\kwletpair{\vvar_1}{\vvar_2}{\vs}{\con}}{\vsub{\vsub{\con}{\vs_1}{\vvar_1}}{\vs_2}{\vvar_2}}}
%    \and
%    \Rule{CE:Case-Left}
%         {\vseval{\vs}{\kwinl{\vs'}}}
%         {\coneq{\gctx}{\utctx}{\vstctx}{\kwvcase{\vs}{\vvar_1}{\con_1}{\vvar_2}{\con_2}}{\vsub{\con_1}{\vs'}{\vvar_1}}}
%    \and
%    \Rule{CE:Case-Right}
%         {\vseval{\vs}{\kwinr{\vs'}}}
%         {\coneq{\gctx}{\utctx}{\vstctx}{\kwvcase{\vs}{\vvar_1}{\con_1}{\vvar_2}{\con_2}}{\vsub{\con_2}{\vs'}{\vvar_2}}}
  \end{mathpar}
  \caption{Type constructor equivalence rules.}
  \label{fig:constructor evaluation}
\end{figure*}

%% file: fig-vs-norm.tex
\begin{figure*}
  \centering
  \def \MathparLineskip {\lineskip=0.43cm}
  \begin{mathpar}
    \Rule{VNN:Var}
         {\strut}
         {\vvar \vsnormalneg}
    \and
    \Rule{VNN:Gen}
         {\strut}
         {\vertgen{\vsty}{\genseed} \vsnormalneg}
    \and
    \Rule{VNN:Fst}
         {\vs \vsnormalneg}
         {\kwfst{\vs} \vsnormalneg}
    \and
    \Rule{VNN:Snd}
         {\vs \vsnormalneg}
         {\kwsnd{\vs} \vsnormalneg}
%    \and
%    \Rule{VNN:Unroll}
%         {\vs \vsnormalneg}
%         {\kwunroll{\vs} \vsnormalneg}
    \and
    \Rule{VN:VNN}
         {\vs \vsnormalneg}
         {\vs \vsnormal}
    \and
    \Rule{VN:Pair}
         {\vs_1 \vsnormal\\
           \vs_2 \vsnormal}
         {\kwpair{\vs_1}{\vs_2} \vsnormal}
%    \and
%    \Rule{VN:Inl}
%         {\vs \vsnormal}
%         {\kwinl{\vs} \vsnormal}
%    \and
%    \Rule{VN:Inr}
%         {\vs \vsnormal}
%         {\kwinr{\vs} \vsnormal}
%    \and
%    \Rule{VN:Roll}
%         {\vs \vsnormal}
%         {\kwvsroll{\vs} \vsnormal}
  \end{mathpar}
  \caption{Definition of normal vertex structures.}
  \label{fig:vert-normal}
\end{figure*}

%% file: fig-unroll-graph.tex
\begin{figure*}
  \centering
  \def \MathparLineskip {\lineskip=0.43cm}
  \begin{mathpar}
    \Rule{UR:Seq1}
         {\graph_1 \unrstep \graph_1'}
         {\graph_1 \seqcomp \graph_2
           \unrstep
           \graph_1' \seqcomp \graph_2}
    \and
    \Rule{UR:Seq2}
         {\graph_2 \unrstep \graph_2'}
         {\graph_1 \seqcomp \graph_2
           \unrstep
           \graph_1 \seqcomp \graph_2'}
    \and
    \Rule{UR:Par1}
         {\graph_1 \unrstep \graph_1'}
         {\graph_1 \parcomp \graph_2
           \unrstep
           \graph_1' \parcomp \graph_2}
    \and
    \Rule{UR:Par2}
         {\graph_2 \unrstep \graph_2'}
         {\graph_1 \parcomp \graph_2
           \unrstep
           \graph_1 \parcomp \graph_2'}
    \and
    \Rule{UR:Or1}
         {\graph_1 \unrstep \graph_1'}
         {\graph_1 \dagor \graph_2
           \unrstep
           \graph_1' \dagor \graph_2}
    \and
    \Rule{UR:Or2}
         {\graph_2 \unrstep \graph_2'}
         {\graph_1 \dagor \graph_2
           \unrstep
           \graph_1 \dagor \graph_2'}
    \and
    \Rule{UR:Future}
         {\graph \unrstep \graph'}
         {\leftcomp{\graph}{\vs} \unrstep \leftcomp{\graph'}{\vs}}
    \and
    \Rule{UR:Rec}
         {\strut}
         {\dagrec{\gvar}{\graph}{} \unrstep
           \gsub{\graph}{\dagrec{\gvar}{\graph}{}}{\gvar}
         }
     \and
     \Rule{UR:Pi}
         {\graph \unrstep \graph'}
         {\dagpi{\vvar_f}{\vvar_t}{\graph} \unrstep
           \dagpi{\vvar_f}{\vvar_t}{\graph'}
         }
     \and
     \Rule{UR:App}
         {\graph \unrstep \graph'}
         {\kwtapp{\graph}{\vs_f}{\vs_t} \unrstep
           \kwtapp{\graph'}{\vs_f}{\vs_t}
         }
     \and
     \Rule{UR:New}
         {\graph \unrstep \graph'}
         {\dagnew{\vertex}{\graph} \unrstep
           \dagnew{\vertex}{\graph'}
         }
%     \and
%     \Rule{UR:Let}
%          {\graph \unrstep \graph'}
%          {\kwletpair{\vvar_1}{\vvar_2}{\vs}{\graph}
%            \unrstep
%            \kwletpair{\vvar_1}{\vvar_2}{\vs}{\graph'}
%          }
%     \and
%     \Rule{UR:Case1}
%          {\graph_1 \unrstep \graph_1'}
%          {\kwcase{\vs}{\vvar_1}{\graph_1}{\vvar_2}{\graph_2}
%            \unrstep
%            \kwcase{\vs}{\vvar_1}{\graph_1'}{\vvar_2}{\graph_2}
%          }
%     \and
%     \Rule{UR:Case2}
%          {\graph_2 \unrstep \graph_2'}
%          {\kwcase{\vs}{\vvar_1}{\graph_1}{\vvar_2}{\graph_2}
%            \unrstep
%            \kwcase{\vs}{\vvar_1}{\graph_1}{\vvar_2}{\graph_2'}
%          }
  \end{mathpar}
  \caption{Graph type unrolling.}
  \label{fig:unroll-graph}
\end{figure*}